\documentclass[
	final,
	paper=a4,
	pagesize=auto,
	fontsize=10pt,
	version=last,
]{scrartcl}

\usepackage[utf8]{inputenc}
\usepackage[T1]{fontenc} 
\usepackage{ae}
\usepackage{aecompl}
\usepackage[
	english,
]{babel}

\KOMAoptions{%
   DIV=10,% (Size of Text Body, higher values = greater textbody)
}%
\KOMAoptions{%
   numbers=noenddot
}%
\KOMAoptions{% (most options are for package typearea)
   twoside=semi,  % two side layout (non alternating margins!)
   twocolumn=false, % (true)
   headinclude=false,%
   footinclude=false,%
   mpinclude=false,%
   headlines=2.1,%
   headsepline=true,%
   footsepline=false,%
   cleardoublepage=empty %plain, headings
}%
\KOMAoptions{%
   parskip=relative, % change indentation according to fontsize (recommanded)
}%
\KOMAoptions{%
}%  
\KOMAoptions{%
}%  
\KOMAoptions{%
   bibliography=totoc % add Bibliography to TOC without number
}%
\KOMAoptions{%
}%
\KOMAoptions{%
}%
\KOMAoptions{% 	
}%

\makeatletter
\newcommand{\LoadPackagesNow}{}
\newcommand{\LoadPackageLater}[2][]{%
   \g@addto@macro{\LoadPackagesNow}{%
      \usepackage[#1]{#2}%
   }%
}
\makeatother

\usepackage{xspace}
\usepackage{xargs}
\usepackage{ifthen}
\usepackage{etoolbox}

\usepackage{calc}
\usepackage[
	a4paper,
	margin=3cm,
	]{geometry}
\usepackage{framed}
\usepackage{mdframed}
\usepackage{pbox}
\usepackage{enumitem}
\usepackage{array}
\usepackage{multirow}
\usepackage{longtable}
\usepackage{hhline}
\usepackage{afterpage}
\usepackage{pdflscape}
\usepackage{rotating}
\usepackage{booktabs}
\usepackage{setspace}
\usepackage[
   bottom,      % Footnotes appear always on bottom. This is necessary
   stable,      % Make footnotes stable in section titles
   multiple,    % rearrange multiple footnotes intelligent in the text.
   flushmargin,
   hang,
]{footmisc}
\usepackage{mathpazo}
\usepackage[%
   headsepline,
   automark,
]{scrlayer-scrpage}

\clearscrheadings
\clearscrplain
\lehead{\pagemark}
\rohead{\pagemark}
\rehead{\headmark}
\lohead{\shortauthor}
\automark[section]{section} %[rechts]{links}
\usepackage{microtype}

\usepackage{relsize}
\usepackage[normalem]{ulem}      
\usepackage{soul}		           
\usepackage{url}
\usepackage{lipsum}
\usepackage[table]{xcolor}
\usepackage[%
]{graphicx}
\usepackage{epstopdf}
\usepackage{wrapfig}
\usepackage{caption}
\usepackage{subcaption}
\usepackage[
   tbtags,    % 'Top-or-bottom tags': For a split equation, place equation numbers level
   sumlimits,  %(default) Place the subscripts and superscripts of summation
   nointlimits, % (default) Opposite of intlimits.
   namelimits, % (default) Like sumlimits, but for certain 'operator names' such as
   reqno,     % Place equation numbers on the right.
]{amsmath} %

\usepackage{amsfonts}
\usepackage{mathrsfs} %% Schreibschriftbuchstaben für den Mathematiksatz (nur Großbuchstaben)
\usepackage{dsfont}
\usepackage{amssymb}
\usepackage{units}
\LoadPackageLater{amsthm}
\LoadPackageLater{thmtools}
\usepackage[fixamsmath,disallowspaces]{mathtools}
\mathtoolsset{showonlyrefs}
\mathtoolsset{centercolon=true}
\usepackage{bm} 
\makeatletter
\g@addto@macro\bfseries{\boldmath}
\makeatother
\allowdisplaybreaks[4]
\numberwithin{equation}{section}
\usepackage{cancel}

\usepackage[toc]{appendix}

\usepackage{csquotes}
\usepackage[
	bibstyle=alphabetic,
	citestyle=alphabetic,
	sorting=ynt,
	sortcites=true,
	giveninits=true,
	maxbibnames=99,
	backend=biber,
	maxalphanames=5,
]{biblatex}

\renewbibmacro{in:}{}
\DeclareFieldFormat{pages}{#1}
\usepackage[textsize=tiny,english,colorinlistoftodos,
	disable,
	]{todonotes}
\definecolor{pdfurlcolor}{rgb}{0,0,0.6}
\definecolor{pdffilecolor}{rgb}{0.7,0,0}
\definecolor{pdflinkcolor}{rgb}{0,0,0.6}
\definecolor{pdfcitecolor}{rgb}{0,0,0.6}
\usepackage[
  colorlinks=true,         % Links erhalten Farben statt Kaeten
  urlcolor=pdfurlcolor,    % \href{...}{...} external (URL)
  filecolor=pdffilecolor,  % \href{...} local file
  linkcolor=pdflinkcolor,  %\ref{...} and \pageref{...}
  citecolor=pdfcitecolor,  %
  raiselinks=true,			 % calculate real height of the link
  breaklinks,              % Links berstehen Zeilenumbruch
  verbose,
  hyperindex=true,         % backlinkex index
  linktocpage=true,        % Inhaltsverzeichnis verlinkt Seiten
  hyperfootnotes=false,     % Keine Links auf Fussnoten
  bookmarks=true,          % Erzeugung von Bookmarks fuer PDF-Viewer
  bookmarksopenlevel=1,    % Gliederungstiefe der Bookmarks
  bookmarksopen=true,      % Expandierte Untermenues in Bookmarks
  bookmarksnumbered=true,  % Nummerierung der Bookmarks
  bookmarkstype=toc,       % Art der Verzeichnisses
  plainpages=false,        % Anchors even on plain pages ?
  pageanchor=true,         % Pages are linkable
  pdfdisplaydoctitle=true, % Dokumententitel statt Dateiname im Fenstertitel
  pdfstartview=FitH,       % Dokument wird Fit Width geaefnet
  pdfpagemode=UseOutlines, % Bookmarks im Viewer anzeigen
  pdfpagelabels=true,           % set PDF page labels
  pdfpagelayout=OneColumn, % zweiseitige Darstellung: ungerade Seiten
]{hyperref}

\LoadPackagesNow

\newcommand{\ifargdef}[3][{}]{\ifthenelse{\equal{#2}{}}{#1}{#3}}

\newlength{\hangind}
\newcommand{\myhangindent}[1]{\settowidth{\hangind}{\widthof{#1}}\hangindent=\the\hangind}

\setkomafont{section}{\Large\rmfamily\bfseries}
\setkomafont{subsection}{\large\rmfamily\bfseries}
\setkomafont{subsubsection}{\rmfamily\bfseries}
\setkomafont{paragraph}{\rmfamily\bfseries}
\setkomafont{pageheadfoot}{\smaller\scshape\rmfamily}

\hypersetup{
	pdfauthor={Martin Genzel, Gitta Kutyniok, and Maximilian März},
	pdftitle={\texorpdfstring{$\l{1}$}{L1}-Analysis Minimization and Generalized (Co-)Sparsity},
	pdfsubject={Article},
	pdfcreator={PDF-LaTeX},
}
\newenvironment{highlight}{\begin{quote}\itshape}{\end{quote}}

\newenvironment{properties}[2][2em]
{\begin{enumerate}[label={\textsc{(#2\arabic*)}},leftmargin=#1]}
{\end{enumerate}} 

\newenvironment{listingroman}
{\begin{enumerate}[label={(\roman*)}]}
{\end{enumerate}}

\newenvironment{rmklist}
{\begin{enumerate}[label={(\arabic*)},itemindent=2em,leftmargin=0em]}
{\end{enumerate}}

\newenvironment{deflist}
{\begin{enumerate}[label={(\arabic*)}]}
{\end{enumerate}}

\newenvironment{expstep}
{\begin{itemize}[label={$\blacktriangleright$},leftmargin=1.5em]}
{\end{itemize}}
\newenvironment{expsubstep}
{\begin{itemize}[label={$\vartriangleright$},leftmargin=3em]}
{\end{itemize}}

\newcommand{\expkwd}[1]{\noindent\textbf{#1}}

\renewcommand{\qedsymbol}{$_\blacksquare$}

\providecommand{\qedhere}{\hfill\qedsymbol}

\newtheoremstyle{claim}
	{\topsep}{\topsep}%
	{\itshape}%         Body font
	{}%         Indent amount (empty = no indent, \parindent = para indent)
	{}% Thm head font
	{}%        Punctuation after thm head
	{.5em}%     Space after thm head (\newline = linebreak)
	{{\bfseries\boldmath\thmname{#1} \thmnumber{#2}} \thmnote{(#3)}}%         Thm head spec

\newtheoremstyle{definition}
	{\topsep}{\topsep}%
	{}%         Body font
	{}%         Indent amount (empty = no indent, \parindent = para indent)
	{}% Thm head font
	{}%        Punctuation after thm head
	{.5em}%     Space after thm head (\newline = linebreak)
	{\textbf{\thmname{#1} \thmnumber{#2}} \thmnote{(#3)}}%         Thm head spec
	
\newtheoremstyle{algorithm}
	{\topsep}{\topsep}%
	{}%         Body font
	{}%         Indent amount (empty = no indent, \parindent = para indent)
	{\bfseries\boldmath}% Thm head font
	{}%        Punctuation after thm head
	{\newline}%     Space after thm head (\newline = linebreak)
	{\thmname{#1} \thmnumber{#2} \thmnote{(#3)}}%         Thm head spec

\mdfdefinestyle{emphframe}{linecolor=black, innertopmargin = 3pt, splittopskip = \topskip, skipbelow=6pt, skipabove=6pt}

\declaretheorem[style=claim,numberwithin=section]{theorem}
\declaretheorem[style=claim,sibling=theorem]{lemma}
\declaretheorem[style=claim,sibling=theorem]{corollary}
\declaretheorem[style=claim,sibling=theorem]{proposition}
\declaretheorem[style=definition,sibling=theorem]{definition}
\declaretheorem[style=definition,sibling=theorem,name={Definition \& Notation}]{defnotation}

\declaretheorem[style=definition,sibling=theorem]{model}
\declaretheorem[style=definition,sibling=theorem]{problem}

\declaretheorem[style=definition,sibling=theorem,qed=$\Diamond$]{remark}

\declaretheorem[style=algorithm,sibling=theorem,%
	preheadhook={\begin{mdframed}[style=emphframe] \setcounter{mpfootnote}{\value{footnote}}},%
	postfoothook=\end{mdframed}]{experiment}
\declaretheorem[style=algorithm,sibling=theorem,%
	preheadhook={\begin{mdframed}[style=emphframe] \setcounter{mpfootnote}{\value{footnote}}},%
	postfoothook=\end{mdframed}]{algorithm}
\newcommand{\opleft}[1]{\mathopen{}\left#1}
\newcommand{\opright}[1]{\right#1\mathclose{}}
\newcommandx{\braces}[4]{%
\ifstrequal{#3}{normal}{#1#4#2}{%
\ifstrequal{#3}{auto}{\left#1#4\right#2}{%
\ifstrequal{#3}{opauto}{\opleft#1#4\opright#2}{%
#3#1#4#3#2}}}%
}
\newcommandx{\opannot}[3][3=\downarrow]{\stackrel{\mathclap{\substack{#1 \\ #3 \vspace{2pt}}}}{#2}}
\newcommandx{\lineannot}[3][3=\rightarrow]{\mathllap{\boxed{\text{\textsmaller{#1}}} #3} #2}
\newcommandx{\multilineannot}[4][4=\rightarrow]{\mathllap{\boxed{\parbox{#1}{\RaggedRight\textsmaller{#2}}} #4} #3}
\newcommand{\N}{\mathbb{N}} % natural numbers
\newcommand{\R}{\mathbb{R}} % real numbers
\newcommand{\C}{\mathbb{C}} % complex numbers
\renewcommand{\iff}{\Leftrightarrow} % if and only if
\newcommand{\suchthat}[1][normal]{\ifstrequal{#1}{normal}{\mid}{#1|}} % seperator in sets (#1op = size)

\newcommand{\setcompl}[1]{#1^c} % complement of a set
\newcommand{\cardinality}[1]{\abs{#1}} % cardinality of a set
\newcommand{\intersec}{\cap} % intersection
\newcommandx{\intvcl}[3][1=normal]{\braces{[}{]}{#1}{#2, #3}} % closed interval (#1op=size, #2=left bound, #3=right bound)
\newcommandx{\intvop}[3][1=normal]{\braces{(}{)}{#1}{#2, #3}} % open interval
\newcommandx{\intvclop}[3][1=normal]{\braces{[}{)}{#1}{#2, #3}} % half-open interval (right)
\newcommandx{\intvopcl}[3][1=normal]{\braces{(}{]}{#1}{#2, #3}} % half-open interval (left)
\DeclareMathOperator*{\argmin}{argmin} % argmin
\DeclareMathOperator{\sign}{sign}
\newcommandx{\abs}[2][1=normal]{\braces{\lvert}{\rvert}{#1}{#2}} % absolute value
\newcommandx{\ceil}[2][1=normal]{\braces{\lceil}{\rceil}{#1}{#2}} % ceil
\newcommandx{\floor}[2][1=normal]{\braces{\lfloor}{\rfloor}{#1}{#2}} % floor
\newcommandx{\round}[2][1=normal]{\braces{[}{]}{#1}{#2}} % round
\newcommandx{\der}[1]{D^{#1}} % differential operator (#1 = multiindex)
\newcommandx{\gradient}{\nabla} % gradient
\newcommandx{\partder}[4][1={},4={}]{\frac{\partial^{#4} #2}{\partial #3^{#4}}\ifargdef{#1}{\Big|_{#1}}} % partial derivative (#1=point of evaluation, #2=function, #3=variable, #4=order)
\newcommandx{\integ}[4][1={},2={}]{\int_{#1}^{#2} #3 \, #4} % integral (#1op=lower bound, #2op=upper bound, #3=integrand, #4=differential form)
\newcommandx{\asympffaster}[2][1=normal]{o\braces{(}{)}{#1}{#2}} % asymptotically faster (proper) (#1op=size)
\newcommandx{\asympfaster}[2][1=normal]{O\braces{(}{)}{#1}{#2}} % asymptotically faster
\newcommandx{\asympeq}[2][1=normal]{\Theta\braces{(}{)}{#1}{#2}} % asymptotically equal
\newcommandx{\asympsslower}[2][1=normal]{\omega\braces{(}{)}{#1}{#2}} % asymptotically slower (proper)
\newcommandx{\asympslower}[2][1=normal]{\Omega\braces{(}{)}{#1}{#2}} % asymptotically slower

\DeclareMathOperator{\Id}{Id} % identity operator
\newcommand{\matr}[1]{\begin{bmatrix} #1 \end{bmatrix}} % matrix
\newcommand{\smallmatr}[1]{\left[\begin{smallmatrix} #1 \end{smallmatrix}\right]} % small matrix
\newcommandx{\norm}[2][1=normal]{\braces{\|}{\|}{#1}{#2}} % norm
\renewcommandx{\sp}[3][1=normal]{\braces{\langle}{\rangle}{#1}{#2, #3}} % inner product (#1op=size, #2=left, #3=right)
\newcommandx{\End}[2][2={}]{\mathcal{L}\opleft( #1 \ifargdef{#2}{, #2} \opright)} % endomorphism (#1=from, #2op=to)
\newcommand{\orthcompl}[1]{{#1}^\perp} % orthogonal complement
\DeclareMathOperator{\spann}{\operatorname{span}} % span
\newcommand{\T}{\mathsf{T}} % transposition (of a matrix)
\renewcommand{\vec}[1]{\boldsymbol{#1}} % vectors in boldface

\newcommandx{\measure}[2][1=normal]{\operatorname{vol}\braces{(}{)}{#1}{#2}} % Lebesgue-measure/volume of a set
\newcommand{\indset}[1]{\chi_{#1}} % indicator function for sets
\DeclareMathOperator{\supp}{supp} % support
\newcommandx{\Leb}[3][1={},3=normal]{L^{#2}\ifargdef{#1}{\braces{(}{)}{#3}{#1}}{}} % Lebesgue spaces (#1op=set, #2=exponent)
\newcommandx{\Lebnorm}[4][1=normal,3={2},4={}]{\norm[#1]{#2}_{#3}} % Lebesgue norm (#1op=size, #2=content, #3op=exponent, #4op=set)
\renewcommandx{\l}[3][1={},3=normal]{\ell^{#2}\ifargdef{#1}{\braces{(}{)}{#3}{#1}}} % lp sequence spaces (#1op=set, #2=exponent)
\newcommandx{\lnorm}[4][1=normal,3={2},4={}]{\norm[#1]{#2}_{#3}} % lp norm (#1op=size, #2=content, #3op=exponent, #4op=set)
\newcommandx{\Smooth}[4][1={},3={},4=normal]{C_{#3}^{#2}\ifargdef{#1}{\braces{(}{)}{#4}{#1}}} % space of differentiable functions (#1op=set, #2=order, #3op=modifier)
\newcommandx{\Schwartz}[2][1={},2=normal]{\mathscr{S}\ifargdef{#1}{\braces{(}{)}{#2}{#1}}} % space of Schwartz functions
\newcommandx{\Schwartzpoly}[2][1=normal]{\braces{\langle}{\rangle}{#1}{\abs[#1]{#2}} } % Schwartz polynomial
\newcommandx{\Tempdistr}[2][1={},2=normal]{\mathscr{S}'\ifargdef{#1}{\braces{(}{)}{#2}{#1}}} % tempered distributions
\newcommandx{\distrinp}[3][1=normal]{\braces{\langle}{\rangle}{#1}{#2, #3}} % evaluation of a tempered distribution (#1op=size, #2=distribution, #3=Schwartz function)
\newcommandx{\ft}[3][1=default,2=auto]{
\ifstrequal{#1}{default}{\widehat{#3}}{
\ifstrequal{#1}{long}{{\braces{(}{)}{#2}{#3}}^{\wedge}}{}}} % Fourier transform (hat-notation) (#1op=long expression mode, #2op=size, #3=content)
\newcommandx{\ift}[3][1=default,2=auto]{
\ifstrequal{#1}{default}{\check{#3}}{
\ifstrequal{#1}{long}{{\braces{(}{)}{#2}{#3}}^{\vee}}{}}} % inverse Fourier transform (hat-notation) (#1op=long expression mode, #2op=size, #3=content)
\newcommand{\erf}{\operatorname{erf}} % sinc function
\newcommand{\define}[1]{\emph{#1}}
\newcommand{\y}{\vec{y}} % measurements
\newcommand{\A}{\vec{A}} % meausrement matrix
\renewcommand{\a}{\vec{a}} % measurement vector
\newcommand{\x}{\vec{x}} % vector
\newcommand{\xsparse}{\bar{\vec{x}}} % vector
\newcommand{\grtr}{\vec{x}^\ast} % groundtruth signal
\newcommand{\sparse}{\vec{x}^\ast_{\text{sp}}} % groundtruth signal
\newcommand{\compr}{\vec{x}^\ast_{\text{c}}} % groundtruth signal
\newcommand{\grtrsparse}{\bar{\vec{x}}^\ast} % groundtruth signal
\newcommand{\comprsparse}{\bar{\vec{x}}^\ast_{\text{c}}} % groundtruth signal
\newcommand{\signalclass}{\mathcal{X}}

\newcommandx{\solu}[1][1={}]{\ifargdef[\hat{\vec{x}}]{#1}{\hat{#1}}} % minimizer
\newcommand{\h}{\vec{h}} % difference vector
\renewcommand{\v}{\vec{v}} % difference vector
\newcommand{\vnull}{\vec{0}} % zero vector
\newcommand{\sset}{K} % signal set
\newcommand{\dict}{\vec{D}} % dictionary
\newcommand{\aop}{\vec{\Psi}} % analysis operator
\newcommand{\aopr}{\vec{\Psi_{\texttt{rdwt}}}} % analysis operator of redundant wavelets
\newcommand{\aoprD}{\vec{\Psi_{\texttt{irdwt-2}}}} % analysis operator of redundant wavelets
\newcommand{\aopi}{\vec{\Psi_{\texttt{irdwt}}}} % analysis operator of redundant wavelets, inverse
\newcommand{\aopd}{\vec{\Psi_{\texttt{dwt}}}} % analysis operator of redundant wavelets, ONB
\newcommand{\aoptv}{\vec{\Psi_{\text{TV-1}}}} % analysis operator of 1D TV
\newcommand{\aoptvD}{\vec{\Psi_{\text{TV-2}}}} % analysis operator of 2D TV
\newcommand{\aopdct}{\vec{\Psi}_{\texttt{dct-2}}} % dct analysis operator 
\newcommand{\aoprand}{\vec{\Psi_{\text{rand}}}} % analysis operator of random frame

\newcommand{\aopart}{\tilde{\vec{\Psi}}} % artificial analysis operator 
\newcommand{\avecart}{\tilde{\vec{\psi}}} % artificial analysis vectors
\newcommand{\avec}{\vec{\psi}} % analysis operator
\newcommand{\gram}{\vec{G}} % Gram matrix
\newcommand{\proj}[1]{\vec{P}_{#1}} % (orthogonal) projection
\newcommand{\I}[1]{\vec{I}_{#1}} % identity matrix
\newcommand{\ssupp}[1][{}]{\mathcal{S}_{#1}} % Support of the analysis coefficients
\newcommand{\ssuppc}[1][{}]{\setcompl{\mathcal{S}}_{#1}} % Off-Support of the analysis coefficients
\newcommand{\lam}[1][{}]{\Lambda} % Off-Support of the analysis coefficients
\newcommand{\anasign}[1]{\vec{\sigma}_{#1}}
\newcommand{\spparam}[1]{S_{\ssupp[#1]}^{\anasign{#1}}}
\newcommand{\cospparam}[1]{L_{\ssuppc[#1]}}
\newcommand{\cospparamdg}[1]{L_{\ssuppc[#1]}^{\vec{\smallsetminus}}}
\newcommand{\samplecompl}[1]{\mathcal{M}(#1)}
\newcommand{\ed}[1]{\delta(#1)}
\newcommand{\fbu}{b} % upper frame bound
\newcommand{\fbl}{a} % lower frame bound

\newcommand{\ananorm}[1]{\lnorm{\aop #1}[1]}
\newcommandx{\subdiff}[2][2={}]{\partial #1\ifargdef{#2}{(#2)}{}}
\newcommand{\dv}{\vec{w}}
\newcommand{\dvparam}{\alpha}
\newcommandx{\clip}[3][1=normal]{\operatorname{clip}\braces{(}{)}{#1}{#2; #3}}
\newcommand{\projnsp}[1][{}]{U_{#1}}
\newcommand{\noise}{\vec{e}} % noise vector
\newcommand{\noiseparam}{\eta} % noise level (in basis pursuit)
\newcommand{\probsuccess}{u} %\epsilon

\newcommandx{\prob}[2][1={},2=normal]{\mathbb{P}\ifargdef{#1}{\braces{[}{]}{#2}{#1}}}

\newcommandx{\mean}[2][1={},2=normal]{\mathbb{E}\ifargdef{#1}{\braces{[}{]}{#2}{#1}}}
\newcommandx{\var}[2][1={},2=normal]{\mathbb{V}\ifargdef{#1}{\braces{[}{]}{#2}{#1}}}
\newcommand{\distributed}{\sim}
\newcommandx{\Normdistr}[3][1=normal]{\mathcal{N}\braces{(}{)}{#1}{#2, #3}} % Normal distribution
\newcommandx{\normsubg}[2][1=normal]{\norm[#1]{#2}_{\psi_2}} % sub-Gaussian
\newcommand{\subgparam}{\kappa} % sub-Gaussian parameter
\newcommand{\gaussian}{\vec{g}} % Gaussian random vector
\newcommand{\gaussianuniv}{g} % Gaussian random number

\newcommand{\Covmatr}{\vec{\Sigma}} % Covariance matrix

\newcommand{\objfunc}{f}
\newcommandx{\anorm}[3][1=normal,3={\sset}]{\norm[#1]{#2}_{#3}} % atomic norm 
\newcommand{\pospart}[1]{\left[#1\right]_+}

\newcommandx{\ball}[2][1={},2={}]{B_{#1}^{#2}} % lp unit ball
\renewcommand{\S}{S} % Euclidean unit sphere
\newcommand{\meanwidth}[2][{}]{w_{#1}(#2)} % mean width
\newcommand{\effdim}[2][{}]{w_{#1}^2(#2)} % effective dimension
\newcommand{\conic}{\wedge} % indicator for conic version
\newcommand{\cone}[1]{\operatorname{cone}(#1)} % cone
\newcommand{\descset}[1]{\mathscr{D}(#1)} % descent set
\newcommand{\desccone}[1]{\mathscr{D}_{\conic}(#1)} % descent cone
\newcommand{\subd}[1]{\partial #1} % subdifferential

\newcommand{\basis}{\vec{B}}

\graphicspath{{images/}}

\begin{document}

\pagestyle{scrheadings}

\begin{center}
	\bfseries\larger[2]{$\l{1}$-Analysis Minimization and Generalized (Co-)Sparsity:\\ When Does Recovery Succeed?}
\end{center}
% title
% \noindent{\huge\raggedright\bfseries A New Perspective on (Co)-Sparsity \\ in $\l{1}$-Analysis Recovery}
% \noindent{\bfseries\larger[2]{$\l{1}$-Analysis Minimization and Generalized (Co-)Sparsity:\\ When Does Recovery Succeed?}}
% Generalized (Co-)Sparsity in $\l{1}$-Analysis Recovery

\vspace{1\baselineskip}
\begin{addmargin}[2em]{2em}
\begin{center}
\noindent{\normalsize\bfseries{Martin Genzel \qquad Gitta Kutyniok \qquad Maximilian März}}

\vspace{.5\baselineskip}
{\smaller\noindent Technische Universit\"at Berlin, Department of Mathematics \\
Straße des 17. Juni 136, 10623 Berlin, Germany

% \vspace{.5\baselineskip}
\noindent E-Mail: \href{mailto:genzel@math.tu-berlin.de}{\texttt{[genzel,kutyniok,maerz]@math.tu-berlin.de}}}
\end{center}

% \linenumbers

% \begin{addmargin}[3em]{3em}
\vspace{1\baselineskip}
{\smaller
\noindent\textbf{Abstract.}
This paper investigates the problem of signal estimation from undersampled noisy sub-Gaussian measurements under the assumption of a cosparse model.
Based on generalized notions of sparsity, we derive novel recovery guarantees for the $\l{1}$-analysis basis pursuit, enabling accurate predictions of its sample complexity.
The corresponding bounds on the number of required measurements do explicitly depend on the Gram matrix of the analysis operator and therefore particularly account for its mutual coherence structure.
Our findings defy conventional wisdom which promotes the \mbox{sparsity} of analysis coefficients as the crucial quantity to study.
In fact, this common paradigm breaks down completely in many situations of practical interest, for instance, when applying a redundant (multilevel) frame as analysis prior.
By extensive numerical experiments, we demonstrate that, in contrast, our theoretical sampling-rate bounds reliably capture the recovery capability of various examples, such as redundant wavelets systems, total variation, or random frames.
The proofs of our main results build upon recent achievements in the convex geometry of data mining problems. More precisely, we establish a sophisticated upper bound on the conic Gaussian mean width that is associated with the underlying $\l{1}$-analysis polytope.
Due to a novel localization argument, it turns out that the presented framework naturally extends to stable recovery, allowing us to incorporate compressible coefficient sequences as well.

\vspace{.5\baselineskip}

\vspace{.25\baselineskip}

}
\end{addmargin}
\newcommand{\shortauthor}{Genzel, Kutyniok, and März: $\l{1}$-Analysis Minimization and Generalized (Co-)Sparsity}

% \linenumbers

% \maketitle
\thispagestyle{plain}

\section{Introduction}
\label{sec:intro}

Initiated by the pioneering works of Candès, Donoho, Romberg, and Tao \cite{candes2006cs,candes2006stable,donoho2006cs}, a considerable amount of research on \emph{compressed sensing} during the last decade has dramatically changed our methodology for exploiting structure in many signal processing tasks.
The classical setup in this field considers the problem of reconstructing an unknown \emph{sparse signal vector} $\grtr \in \R^n$ from \emph{non-adaptive, linear measurements} of the form
\begin{equation}\label{eq:intro:meas}
	\y = \A \grtr + \noise,
\end{equation}
where $\A \in \R^{m\times n}$ is a known \emph{sensing matrix} and $\noise \in \R^m$ captures potential distortions, usually due to noise.
The success of compressed sensing is based on the fundamental insight that---by explicitly incorporating a sparsity prior---this task becomes feasible even if $m \ll n$.
In particular, there exist numerous convex and greedy recovery methods that enjoy both efficient implementations and a rich theoretical foundation.
Among them, probably the most popular approach is the \emph{basis pursuit}~\cite{chen1998}:
\begin{equation}\label{eq:intro:bp}\tag{$\text{BP}_{\noiseparam}$}
	\min_{\x \in \R^n} \lnorm{\x}[1] \quad \text{subject to \quad $\lnorm{\A \x - \y} \leq \noiseparam$,}
\end{equation}
where $\noiseparam \geq 0$ is chosen such that $\lnorm{\noise} \leq \noiseparam$.
The crucial ingredient of \eqref{eq:intro:bp} is the $\l{1}$ objective functional which promotes sparse solutions of the minimization problem.
Remarkably, it can be shown that \eqref{eq:intro:bp} indeed recovers an $S$-sparse vector\footnote{That means, at most $S$ entries of $\grtr$ are non-zero, or more formally, $\lnorm{\grtr}[0] \coloneqq \cardinality{\supp(\grtr)} \leq S$.} $\grtr$  with the optimal sampling rate of $m = \asympfaster{S \cdot \log(2S/n)}$ if the measurement design $A$ is drawn according to an appropriate random distribution.

Although such types of theoretical guarantees are elegant and practically appealing, the traditional assumption of sparsity is not directly satisfied in most real-world applications. Fortunately, many signals-of-interest do at least exhibit a \emph{low-complexity representation} with respect to a certain transformation, for instance, the wavelet or Fourier transform.
In this work, we study the so-called \emph{analysis sparsity model} (also known as \emph{cosparse model}), which has gained increasing attention within the past years~\cite{candes2011csdict,nam2013,kabanava2015}.
The key idea of this approach is to test (``analyze'') the signal $\grtr \in \R^n$ with a collection of \emph{analysis vectors} $\avec_1, \dots, \avec_N \in \R^n$, i.e., one computes
\begin{equation}\label{eq:intro:anacoeff}
	\aop\grtr = (\sp{\avec_1}{\grtr}, \dots, \sp{\avec_N}{\grtr}) \in \R^N,
\end{equation}
where the matrix $\aop \coloneqq [\avec_1 \dots \avec_N]^\T \in \R^{N\times n}$ is called the \emph{analysis operator}. If $\aop$ is able to reflect the underlying structure of $\grtr$, one might expect that these \emph{analysis coefficients} are dominated by only a few large entries.
This hypothesis of transform sparsity motivates the following generalization of \eqref{eq:intro:bp} that is typically referred to as the \emph{analysis basis pursuit} (or \emph{$\l{1}$-analysis minimization}):
\begin{equation}\label{eq:intro:anabp}\tag{$\text{BP}_{\noiseparam}^{\aop}$}
	\min_{\x \in \R^n} \ananorm{\x} \quad \text{subject to \quad $\lnorm{\A \x - \y} \leq \noiseparam$.}
\end{equation}
\newcommand{\refabpnoiseless}[2]{(\hyperref[eq:intro:anabp]{$\text{BP}_{#1}^{#2}$})}%
In fact, the adapted strategy of \eqref{eq:intro:anabp} has turned out to work surprisingly well for numerous problem settings, such as in total variation minimization~\cite{rudin1992,chambolle2004} or for multiscale transforms in classical signal and image processing tasks~\cite{cai2010,plonka2011,lustig2008,vander2013shearlet}. Apart from that, it has become also relevant to regularizing physics-driven inverse problems~\cite{kitic2016,nam2012} and operator learning approaches based on test data sets \cite{rubinstein2013ksvd,hawe2013,chen2014learning,bian2016learning}. %\cite{rubinstein2013ksvd,hawe2013,chen2014learning,ravishankar2013,yaghoobi2013learning,ravishankar2011}.
More recently, the cosparse model has been connected to topics in deep neural networks as well, e.g., in the context of multi-layer sparse coding \cite{aberdam2019holistic}.
But despite these successful applications, many theoretical properties of the analysis basis pursuit remain unexplored and ``its rigorous understanding is still in
its infancy'' \cite[p.~174]{kabanava2015}.

\subsection{Does (Co-)Sparsity Explain the Success (or Failure) of the Analysis Basis Pursuit?}
\label{subsec:intro:issues}

As already foreshadowed by \eqref{eq:intro:anacoeff}, a central objective of the analysis formulation is to come up with an operator $\aop \in \R^{N \times n}$ that provides a coefficient vector $\aop \grtr$ of ``low complexity.'' The traditional theory of compressed sensing---where $\aop$ is just the identity---would suggest that the sparsity of $\aop \grtr$ is the key concept to look at. Indeed, a large part of the related literature precisely builds upon this intuition.
Although many of those approaches rely on different proof strategies, e.g., the $\dict$-RIP \cite{candes2011csdict} or conic geometry \cite{kabanava2015}, they eventually promote results of a very similar type: Recovery of $\grtr \in \R^n$ via \eqref{eq:intro:anabp} succeeds if the number of measurements obeys
\begin{equation}\label{eq:intro:boundliterature}
	m \geq C \cdot S \cdot \operatorname{PolyLog}(\tfrac{2N}{S}),
\end{equation}
where $S \coloneqq \lnorm{\aop \grtr}[0]$ and $C > 0$ is a constant that might depend on $\aop$.\footnote{In order to indicate that $S$ is associated with the space of analysis coefficients $\R^N$, we have decided to use a capital letter, whereas small letters are perhaps more common in the literature when referring to sparsity.}
This bound on the sampling rate clearly resembles classical guarantees for the basis pursuit \eqref{eq:intro:bp} and therefore forms a quite natural extension towards analysis sparsity.

Another important branch of research takes a somewhat contrary perspective, and identifies the \emph{cosparsity} $L \coloneqq N - S$ as a crucial quantity for the success of the analysis model.
A remarkable observation of \cite{nam2013} was that the location of vanishing coefficients in $\aop \grtr$ is the driving force behind the analysis formulation, rather than the number of non-zero components.
This viewpoint naturally leads to the so-called \emph{cosparse signal model}, which is typically described by \emph{union-of-subspaces} \cite{blumensath2009}.
Following this terminology, it has turned out that successful recovery via combinatorial searching can be guaranteed if the number of observations is of the order of the signal's manifold dimension \cite[Sec.~3]{nam2013}.
However, such a simple relationship does not seem to carry over to tractable methods like \eqref{eq:intro:anabp}, see \cite{giryes2015}.
Indeed, many sampling-rate bounds relying on cosparsity can be easily translated into sparsity-based statements, simply meaning that $S$ is replaced by $N - L$ in \eqref{eq:intro:boundliterature}, e.g., see \cite[Thm. 1]{kabanava2015} or \cite[Thm.~3.8]{giryes2014}.

\enlargethispage{.5\baselineskip}
Even though the above approaches are quite appealing due to their simplicity, interpretability, and consistency, it still remains unclear whether they are sufficient for a sound foundation of \eqref{eq:intro:anabp} in its general form.
The notions of analysis \mbox{(co-)sparsity} are completely determined by the support of $\aop \grtr$, which in turn does not account for the coherence structure of the individual analysis vectors $\avec_1, \dots, \avec_N$; in other words, the underlying ``geometry'' of $\aop$ is ignored.

% \pagebreak
To get a first glimpse of this issue, let us consider a simple example: Figure~\ref{fig:intro:1d} shows the results of a numerical simulation that reconstructs a block-signal (see Figure~\ref{fig:intro:1d:signal}) using three different analysis operators.
The plot of Figure~\ref{fig:intro:1d:pt_frame} exhibits the phase transition behavior of \refabpnoiseless{\noiseparam=0}{\aop} for a \emph{redundant, discrete Haar wavelet transform} $\aopr$ and the analysis operator $\aopi$ associated with the inverse wavelet transform (i.e., a certain dual frame of $\aopr$; see Subsection~\ref{subsec:intro:notation}\ref{item:intro:notation:frames}).
Although the \mbox{(co-)sparsity} ($S = 906$ and $L = 886$) are exactly the same for both choices, their recovery capability indeed differs dramatically!
In conclusion, just investigating the parameters $S$ and $L$ does by far not explain why the transition of $\aopi$ happens much earlier ($m \approx 85$) than the one of $\aopr$ ($m \approx 240$).
Even more striking, the prediction of \eqref{eq:intro:boundliterature} deviates from the truth by orders of magnitudes.

The plot of Figure~\ref{fig:intro:1d:pt_onb} reveals another insight: While \eqref{eq:intro:boundliterature} is reliable for an \emph{orthonormal Haar wavelet transform} $\aopd \in \R^{n \times n}$, a comparison with the actual recovery rate of \refabpnoiseless{\noiseparam=0}{\aopi} indicates that \emph{redundancy} can be beneficial in the analysis model.
Finally, it is also worth mentioning that a \emph{compressibility} argument does not ``save the day'' here: Figure~\ref{fig:intro:1d:decay} demonstrates that the transitions of \refabpnoiseless{\noiseparam=0}{\aopi} and \refabpnoiseless{\noiseparam=0}{\aopr} both take place far before the remaining coefficients would be negligibly small.

\begin{figure}
	\centering
	\begin{subfigure}[t]{0.45\textwidth}
		\centering
		\includegraphics[width=\textwidth]{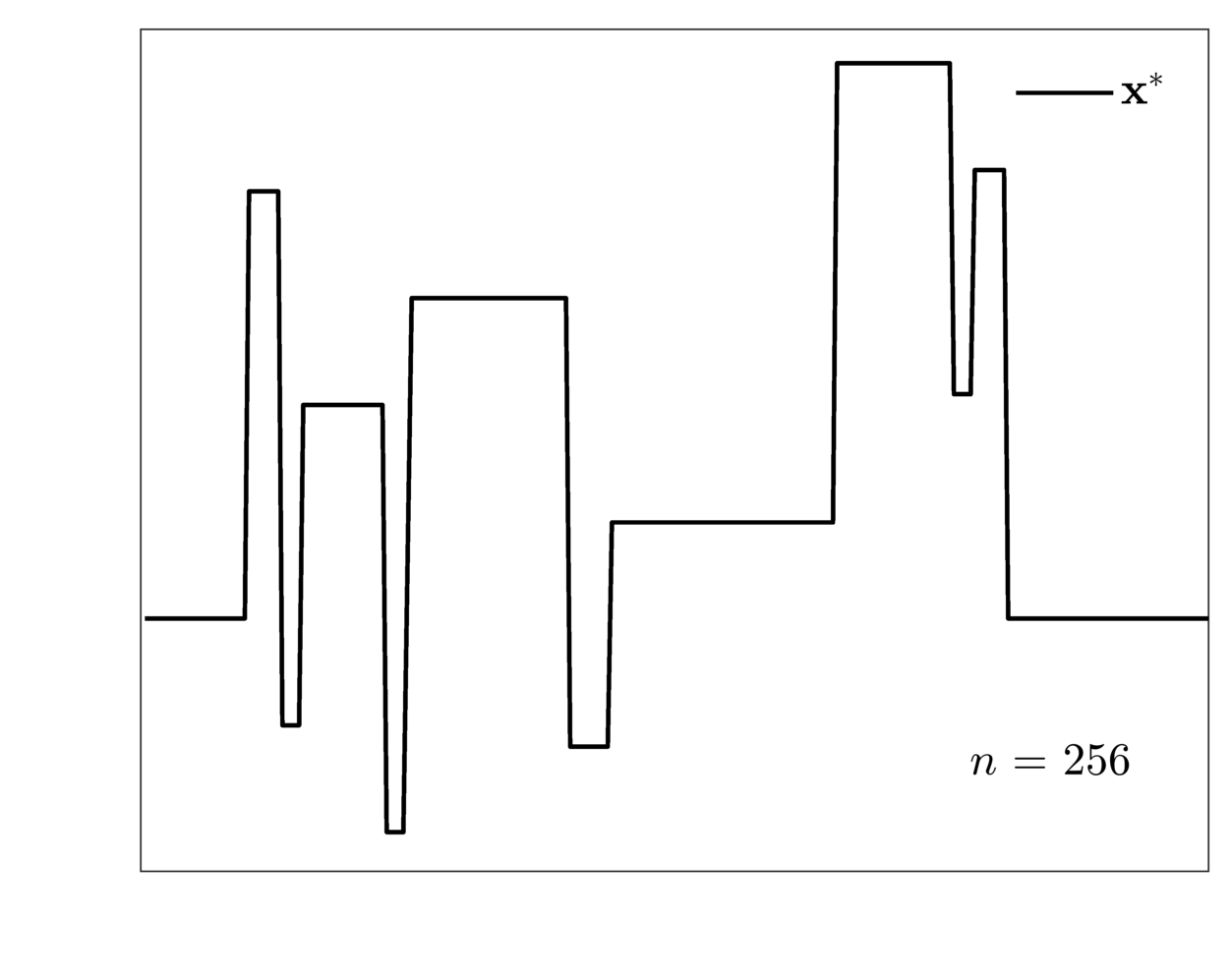}
		\caption{}
		\label{fig:intro:1d:signal}
	\end{subfigure}%
	\qquad
	\begin{subfigure}[t]{0.45\textwidth}
		\centering
		\includegraphics[width=\textwidth]{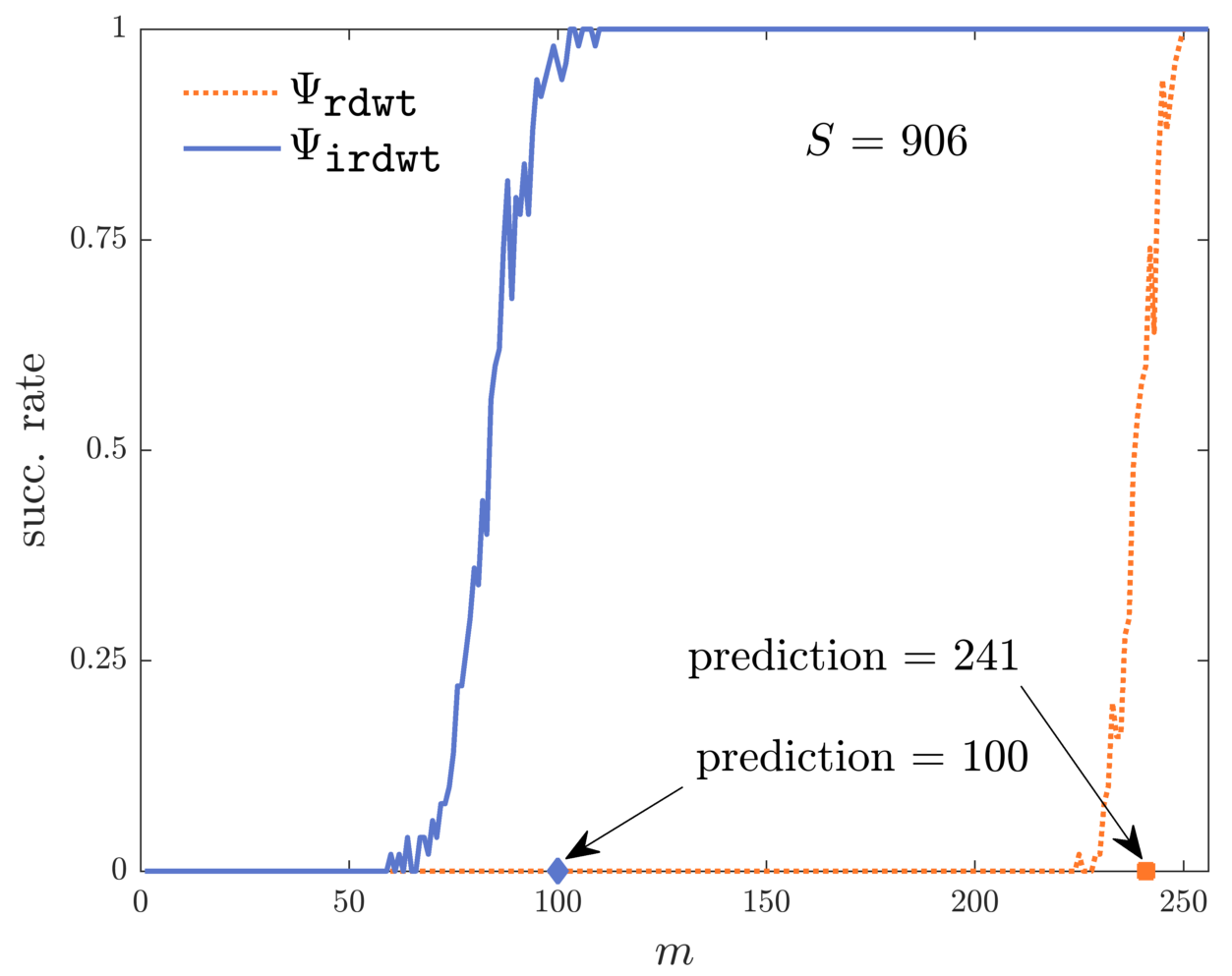}
		\caption{}
		\label{fig:intro:1d:pt_frame}
	\end{subfigure}%	
	
	\vspace{.5\baselineskip}
	\begin{subfigure}[t]{0.45\textwidth}
		\centering
		\includegraphics[width=\textwidth]{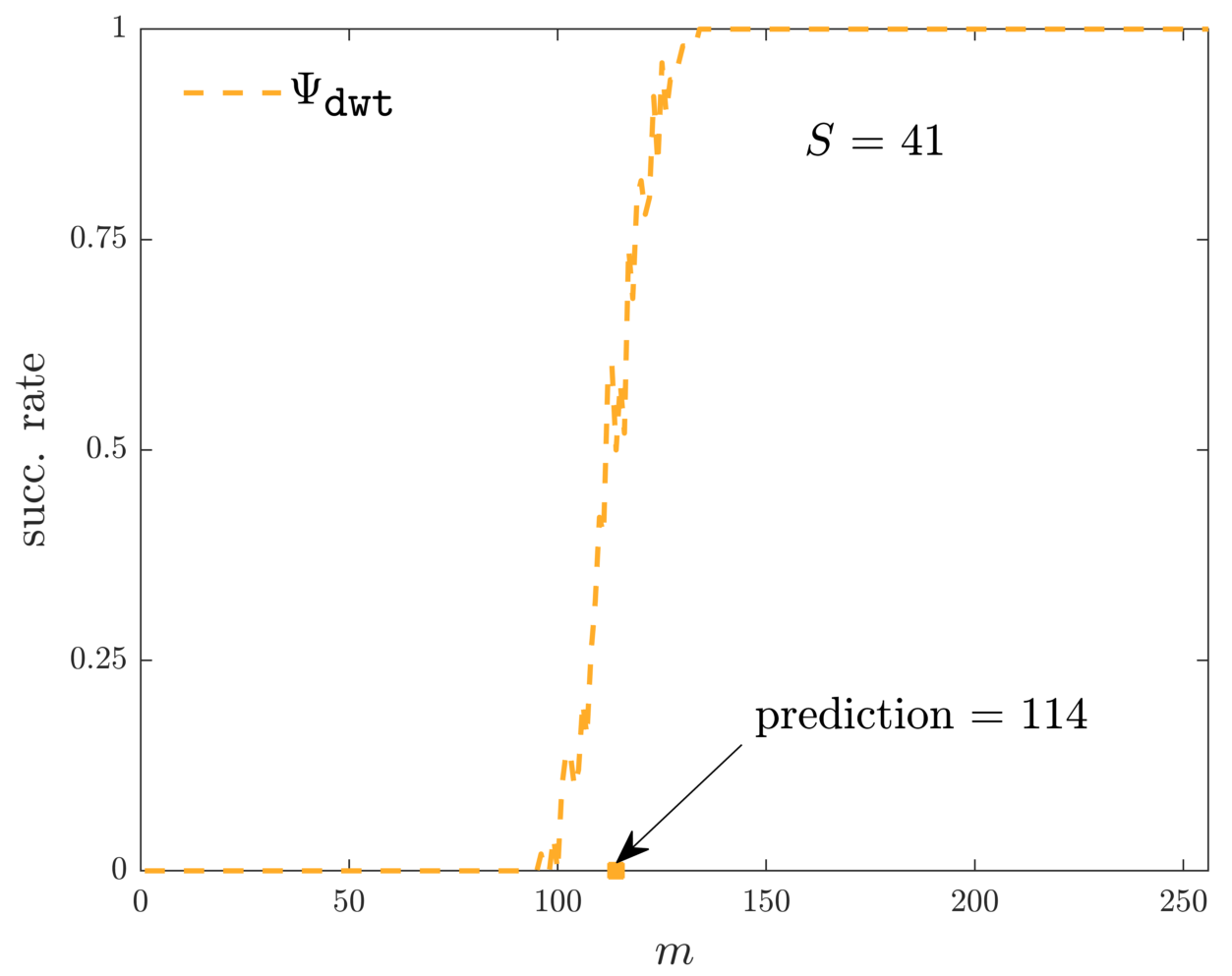}
		\caption{}
		\label{fig:intro:1d:pt_onb}
	\end{subfigure}%
	\qquad
	\begin{subfigure}[t]{0.45\textwidth}
		\centering
		\includegraphics[width=\textwidth]{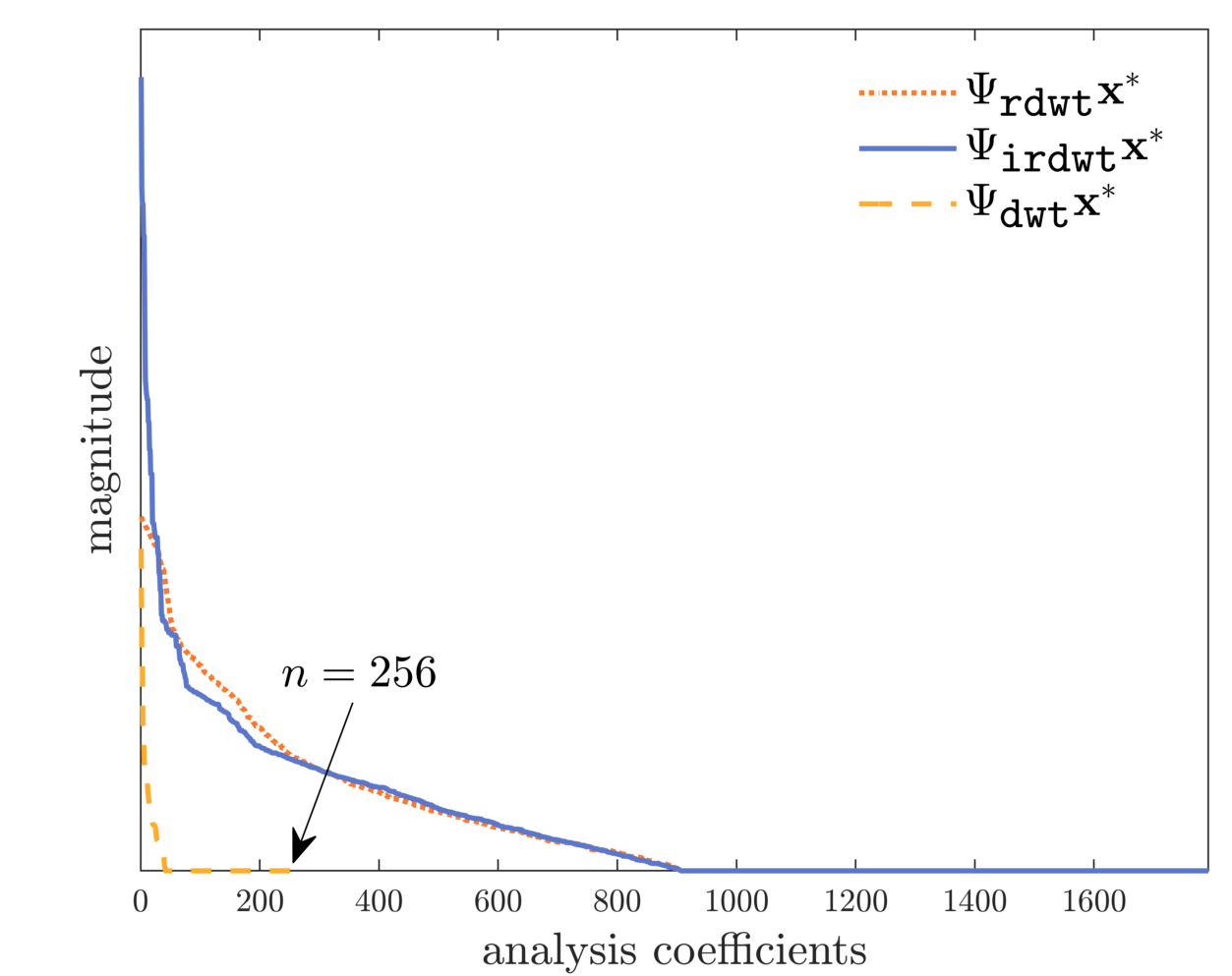}
		\caption{}
		\label{fig:intro:1d:decay}
	\end{subfigure}%	
	\caption{The phase transition of \refabpnoiseless{\noiseparam=0}{\aop} for wavelets in 1D with noiseless Gaussian measurements ($\noiseparam = 0$). For a detailed description of this experiment, see Subsection~\ref{subsec:appl:wavelets}. \subref{fig:intro:1d:signal}~Signal vector $\grtr \in \R^{256}$. \subref{fig:intro:1d:pt_frame}~Success rate of exact recovery via \refabpnoiseless{\noiseparam=0}{\aopr} and \refabpnoiseless{\noiseparam=0}{\aopi}, respectively. \subref{fig:intro:1d:pt_onb}~Success rate of exact recovery via \refabpnoiseless{\noiseparam=0}{\aopd}. \subref{fig:intro:1d:decay}~Decay behavior of the analysis coefficients; the sorted and $\l{1}$-normalized magnitudes of the coefficients $\aopr\grtr$ and $\aopi\grtr$ exhibit a very similar behavior; the sorted coefficients of $\aopd\grtr$ have been scaled manually to a comparable order of magnitude.}
	\label{fig:intro:1d}
\end{figure}

We emphasize that the above example is not too specific or artificial, but rather illustrates a scenario that often occurs in applications: 
Due to linear dependencies within $\aop$, the analysis sparsity oftentimes cannot become arbitrarily small.
For instance, if $\aop$ corresponds to a highly redundant dictionary, one typically has $S \gg n$, whereas the true sample complexity is relatively small (below the space dimension $n$).
This important phenomenon gives rise to the following fundamental question:
\begin{highlight}
	If \mbox{(co-)sparsity} does not fully explain what is happening, which general principles lead to success or failure of the analysis basis pursuit \eqref{eq:intro:anabp}?
\end{highlight}
% \pagebreak[0]
To be more precise about this concern, let us formulate three central issues that we will address in this paper:
\begin{properties}[3em]{Q}
\item\label{quest:intro:samplecompl}
	\emph{Sampling rate.} How many measurements are needed for an accurate estimate of $\grtr$ via \eqref{eq:intro:anabp}?
	What crucial quantities and parameters determine the \emph{sample complexity}?
\item\label{quest:intro:compressibility}
	\emph{Compressibility.} What if $\grtr$ is only \define{compressible}, i.e., its analysis coefficients $\aop\grtr$ are not exactly sparse but only close to a sparse vector? Is \eqref{eq:intro:anabp} stable under such model inaccuracies?
\item\label{quest:intro:interpretability}
	\emph{Interpretability.} What practical guidelines can we derive from our recovery results? Which characteristics of $\aop$ deserve special attention in applications?
\end{properties}

\subsection{Main Contributions and Overview}
\label{subsec:intro:contrib}

A major concern of this work is to shed more light on the problem of $\l{1}$-analysis recovery and to provide a deeper understanding of its underlying mechanisms.
Our first main result in Section~\ref{sec:results} (Theorem~\ref{thm:results:recovery}) focuses on the setup of noiseless Gaussian measurements, i.e., $\noise = \vnull$ in \eqref{eq:intro:meas} and the entries of $\A$ are independent standard Gaussians. In this prototypical situation, we can even expect an \emph{exact} reconstruction of $\grtr$ via \refabpnoiseless{\noiseparam=0}{\aop}, supposed that $m$ is sufficiently large.
In fact, Theorem~\ref{thm:results:recovery} states a \emph{non-asymptotic} and \emph{non-uniform} bound on the sampling rate which intends to meet the following desiderata:
\begin{listingroman}
\item\label{item:intro:desiderata:accurate}
	\emph{Accurate.} The number of required measurements should be close to the (optimal) sample complexity.
\item\label{item:intro:desiderata:computable}
	\emph{Computable.} The expressions should be explicit with respect to the underlying analysis operator $\aop$ and numerically evaluable. 
\item\label{item:intro:desiderata:interpretable}
	\emph{Interpretable.} The involved parameters should have a clear meaning and resemble well-known principles.
\item\label{item:intro:desiderata:generic}
	\emph{Generic.} The bound should not be tailored to a specific operator $\aop$ (e.g., wavelets or total variation), but apply in a general setting.
\end{listingroman}
The key ingredients of our bound are three parameters, generalizing the notions of sparsity and cosparsity (Definition~\ref{def:results:generalsparsity}). While the support of the coefficients $\aop\grtr$ still plays a major role, these novel quantities do also take account of the \emph{(mutual) coherence structure} of $\aop$ by incorporating the (off-diagonal) entries of its Gram matrix. Indeed, this refinement will allow us to satisfy the above ``wish list'' for many different examples of interest.
Besides the well-known case of orthonormal bases, this particularly includes \emph{highly redundant} analysis operators, which are only poorly understood in that context so far.

The proof of Theorem~\ref{thm:results:recovery} (provided in Section~\ref{sec:proofs}) loosely follows a recent technique proposed by \cite[Recipe~4.1]{amelunxen2014edge}. More precisely, we establish a highly non-trivial upper bound on the \emph{Gaussian mean width}\footnote{This is equivalent to bounding the so-called \emph{statistical dimension}, introduced in \cite{amelunxen2014edge}. But we will keep using the notion of \emph{Gaussian mean width}, which is more common in the literature.} of the descent cone of $\x \mapsto \lnorm{\aop \x}[1]$ at $\grtr$ (Theorem~\ref{thm:proofs:mainresults:scbound}), which is known to characterize the \emph{phase transition behavior} of \refabpnoiseless{\noiseparam=0}{\aop} for Gaussian measurements.
Noteworthy, Kabanava, Rauhut, and Zhang follow a very similar proof strategy in \cite{rkz2015}. It is therefore not surprising that their approach bears a certain resemblance to ours, but there are in fact crucial differences as a more detailed comparison in Section~\ref{sec:appl} and Subsection~\ref{subsec:discussion:further} will show.

Our second guarantee (Theorem~\ref{thm:results:robuststable}) tackles the problem of stability \ref{quest:intro:compressibility} in the setting of noisy sub-Gaussian measurements.
It is based on the simple geometric idea of approximating the ground truth $\grtr$ by another signal vector $\grtrsparse \in \R^n$ of lower complexity, in the sense that its associated Gaussian mean width is significantly smaller.
This novel result is achieved by means of a localized variant of the (conic) Gaussian mean width \cite{mendelson2007reconstruction} which allows us to quantify the mismatch caused by approximating $\grtr$.
Let us emphasize that this approach goes beyond the ``naive'' reasoning for compressible vectors that just suggests determining the best $S$-term approximation of the analysis coefficient vector $\aop\grtr$.

Section~\ref{sec:appl} presents several applications of the framework developed in Section~\ref{sec:results}.
A special focus here is on the little studied example of redundant wavelet frames, but the variety of our results will be also demonstrated for other important choices of $\aop$, such as the finite difference operator and random frames.
In this course, we assess the predictive quality of our theoretical findings by means of various numerical experiments.
Section~\ref{sec:discussion} is then dedicated to an overview of the related literature, revisiting the initial concern of Subsection~\ref{subsec:intro:issues}: \emph{Does \mbox{(co-)sparsity} explain the recovery performance of \eqref{eq:intro:anabp}?}
This discussion is concluded by Section~\ref{sec:concl}, where we address the issues of \ref{quest:intro:interpretability} in greater detail:
Even though the main objective of this work is a mathematical foundation, our novel perspective on \mbox{(co-)sparsity} might have practical implications as well, for instance, when designing or learning an analysis operator for a specific task.
Finally, we point out some unresolved challenges left for future works, in particular, the question of how close to optimal our prediction bounds are.

\subsection{General Notation}
\label{subsec:intro:notation}

Let us fix some notations and conventions that will be frequently used in the remainder of this paper:
\begin{rmklist}
\item\label{item:intro:notation:vecmatr}
	For an integer $d \in \N$, we set $[d] \coloneqq \{1, \dots, d\}$. If $\mathcal{I} \subset [d]$, the set complement of $\mathcal{I}$ in $[d]$ is given by $\setcompl{\mathcal{I}} \coloneqq [d] \setminus \mathcal{I}$. Vectors and matrices are denoted by lower- and uppercase boldface letters, respectively. Unless stated otherwise, their entries are indicated by subscript indices and lowercase letters, e.g., $\x = (x_1, \dots, x_d) \in \R^d$ for a vector and $\vec{B} = [b_{k,k'}] \in \R^{d\times d'}$ for a matrix.
	For $\mathcal{I} \subset [d]$, the vector $\x_{\mathcal{I}} \in \R^{\cardinality{\mathcal{I}}}$ is the restriction of $\x \in \R^d$ to the components of $\mathcal{I}$. Similarly, $\vec{B}_{\mathcal{I}} \in \R^{\cardinality{\mathcal{I}}\times d'}$ restricts a matrix $\vec{B} \in \R^{d \times d'}$ to the rows corresponding to $\mathcal{I}$.
\item\label{item:intro:notation:lp}
	Let $\x = (x_1, \dots, x_d) \in \R^d$. The \emph{support} of $\x$ is defined by the set of its non-zero entries $\supp(\x) \coloneqq \{ k \in [d] \suchthat x_k \neq 0 \}$ and the \emph{sparsity} of $\x$ is $\lnorm{\x}[0] \coloneqq \cardinality{\supp(\x)}$. 
	For $1 \leq p \leq \infty$, we denote the \emph{$\l{p}$-norm} on $\R^d$ by $\lnorm{\cdot}[p]$. The associated \emph{unit ball} is given by $\ball[p][d] \coloneqq \{ \x \in \R^d \suchthat \lnorm{\x}[p] \leq 1 \}$ and the \emph{Euclidean unit sphere} is $S^{d-1} \coloneqq \{ \x \in \R^d \suchthat \lnorm{\x} = 1 \}$. 
% 	For $1 \leq p \leq \infty$, the \emph{$\l{p}$-norm} is given by
% 	\begin{equation}
% 	\lnorm{\x}[p] \coloneqq \begin{cases} (\sum_{k = 1}^d \abs{x_k}^p)^{1/p}, & p < \infty, \\ \max_{k \in [d]} \abs{x_k}, & p = \infty.
% 	\end{cases}
% 	\end{equation}
% 	Then, $\ball[p][d] \coloneqq \{ \x \in \R^d \suchthat \lnorm{\x}[p] \leq 1 \}$ denotes the associated \emph{unit ball} and the \emph{(Euclidean) unit sphere} is $S^{d-1} \coloneqq \{ \x \in \R^d \suchthat \lnorm{\x} = 1 \}$. 
	If $\Sigma_S^d \coloneqq \{ \x \in \R^d \suchthat \lnorm{\x}[0] \leq S \}$ denotes the set of all \emph{$S$-sparse} vectors in $\R^d$, the \emph{best $S$-term approximation error} (with respect to the $\l{p}$-norm) of a vector $\x' \in \R^d$ is
	\begin{equation}
		\sigma_S(\x')_p \coloneqq \min_{\x \in \Sigma_S^d} \lnorm{\x' - \x}[p].
	\end{equation}
\item\label{item:intro:notation:convexgeo}
	We denote the \emph{conic hull} of a set $\sset \subset \R^d$ by $\cone{\sset}$. If $\projnsp \subset \R^d$ is a linear subspace, the associated \emph{orthogonal projection onto $\projnsp$} is denoted by $\proj{\projnsp} \in \R^{d \times d}$.
	Then, we have $\proj{\orthcompl{\projnsp}} = \I{d} - \proj{\projnsp}$, where $\orthcompl{\projnsp} \subset \R^d$ is the orthogonal complement of $\projnsp$ and $\I{d} \in \R^{d \times d}$ is the identity matrix.
\item\label{item:intro:notation:rndvec}
	If $\gaussian$ is a mean-zero Gaussian random vector in $\R^d$ with covariance matrix $\Covmatr \in \R^{d \times d}$, we just write $\gaussian \distributed \Normdistr{\vnull}{\Covmatr}$. Moreover, a random vector $\a$ in $\R^d$ is \emph{sub-Gaussian} if $\normsubg{\a} < \infty$, where $\normsubg{\cdot}$ is the sub-Gaussian norm; see \cite[Def.~5.22]{vershynin2012random}.
\item\label{item:intro:notation:frames}
	Let $\aop \in \R^{N \times n}$ be a matrix with row vectors $\avec_1, \dots, \avec_N \in \R^n$. Then, the collection $\mathcal{F} \coloneqq \{\avec_k\}_{k \in [N]} \subset \R^n$ forms a \emph{frame} for $\R^n$ with frame bounds $0 < \fbl \leq \fbu < \infty$ if
	\begin{equation}
		\fbl \cdot \lnorm{\x}^2 \leq \lnorm{\aop \x}^2 \leq \fbu \cdot \lnorm{\x}^2 \quad \text{for all $\x \in \R^n$.}
	\end{equation}
	If $\fbl = \fbu$, then $\mathcal{F}$ is said to be a \emph{tight frame}.
	If there is no danger of confusion, we will identify the \emph{analysis operator} $\aop$ with $\mathcal{F}$.
	Finally, we call a frame $\tilde{\aop} \in \R^{N \times n}$ a \emph{dual frame} of $\aop$ if ${\tilde{\aop}}^\T \aop = \I{n}$. 
	See \cite{casazza2012finite} for a detailed introduction to finite-dimensional frame theory.
\item\label{item:intro:notation:clip}
	For $s, t \in \R$, we define the \emph{clip function} $\clip{s}{t} \coloneqq \sign(s) \cdot \min\{\abs{s}, t\}$ and the \emph{positive part} $\pospart{s} \coloneqq \max\{s, 0\}$.
\item\label{item:intro:notation:constants}
	The letter $C$ is always reserved for a (generic) constant, whose value could change from time to time.
	We refer to $C$ as a \emph{numerical constant} if its value does not depend on any other involved parameter.
	If an \mbox{(in-)equality} holds true up to a numerical constant $C$, we sometimes simply write $A \lesssim B$ instead of $A \leq C \cdot B$. %, and if $C_1 \cdot A \leq B \leq C_2 \cdot A$ for numerical constants $C_1, C_2 > 0$, the abbreviation $A \asymp B$ is used.
\end{rmklist}

\section{Main Results}
\label{sec:results}

In this section, we present the main theoretical results of this work. Starting with some notations and model assumptions in Subsection~\ref{subsec:results:setup}, our novel sparsity parameters are introduced in Subsection~\ref{subsec:results:parameters}.
Based on these central notions, Subsection~\ref{subsec:results:samplecompl} then establishes an exact recovery result for noiseless Gaussian observations.
Finally, in Subsection~\ref{subsec:results:compressible}, we show that the analysis basis pursuit \eqref{eq:intro:anabp} is also robust against noise and stable under model inaccuracies, meaning that the analysis coefficient vector is allowed to be compressible.
For the sake of readability, all corresponding proofs are postponed to Section~\ref{sec:proofs}.

\subsection{Model Setup and Notation}
\label{subsec:results:setup}

In this part, we state the standing assumptions and model hypothesis of our recovery framework; see also Table~\ref{tab:results:notation} for a summary.
The following notations and conventions are supposed to hold true for the remainder of this section.
Let us first recall the linear measurement scheme of \eqref{eq:intro:meas} and set up a formal random observation model:
\begin{model}[Noisy Linear Measurements]\label{model:results:setup:meas}
	Let $\grtr \in \R^n$ be a fixed vector, which is typically referred to as the \emph{signal} (or \emph{source}).
	The \define{sensing vectors} $\a_1, \dots, \a_m \in \R^n$ are assumed to be independent copies of an isotropic, mean-zero, sub-Gaussian random vector $\a$ in $\R^n$ with $\normsubg{\a} \leq \subgparam$.
	These vectors form the rows of the \emph{sensing matrix} $\A \coloneqq [\a_1 \dots \a_m]^\T \in \R^{m \times n}$.
	The actual \define{measurements} of $\grtr$ are then given by
	\begin{equation}\label{eq:results:setup:meas}
		\y \coloneqq \A \grtr + \noise \in \R^m,
	\end{equation}
	where $\noise = (e_1, \dots, e_m) \in \R^m$ models \emph{noise}, which could be deterministic and systematic.
	We assume that the noise is $\l{2}$-bounded, i.e., $\lnorm{\noise} \leq \noiseparam$ for some $\noiseparam \geq 0$.
\end{model}
Next, we fix our notation for the analysis operator and coefficients. Note that the dimension of the analysis domain $\R^N$ does not necessarily have to be larger than the dimension of $\R^n$.
\begin{defnotation}\label{def:results:setup:analysis}
\begin{rmklist}
\item
	The matrix $\aop = [\psi_{k,j}] \in \R^{N\times n}$ is called an \define{analysis operator} (or \define{analysis matrix}) if none of its rows equals the zero vector.
	The rows of $\aop$, denoted by $\avec_1, \dots, \avec_N \in \R^{n}$, are called the \define{analysis vectors}.
	Moreover, we define the \define{Gram matrix} of $\aop$ as
	\begin{equation}
		\gram = [g_{k,k'}] \coloneqq \aop \aop^\T = [\sp{\avec_k}{\avec_{k'}}]  \in \R^{N \times N}.
	\end{equation}
\item
	The \define{analysis coefficients} of a vector $\x \in \R^n$ (with respect to $\aop$) are given by
	\begin{equation}
		\aop \x = (\sp{\avec_1}{\x}, \dots, \sp{\avec_N}{\x}) \in \R^{N}.
	\end{equation}
	The \define{analysis support} of $\x$ is denoted by $\ssupp[\x] \coloneqq \supp(\aop \x) \subset [N]$, and if $S = \cardinality{\ssupp[\x]}$, we say that $\x$ is \define{$S$-analysis-sparse}.\footnote{If there is no danger of confusion, we may omit the term ``analysis'' and just speak of \emph{coefficients}, \emph{support}, \emph{sparsity}, etc.} Analogously, we call the complement $\ssuppc[\x] \coloneqq \setcompl{\supp(\aop \x)} \subset [N]$ the \define{analysis cosupport} of $\x$ and speak of an \define{$L$-analysis-cosparse} vector if $L = \cardinality{\ssuppc[\x]} = N - S$.
\end{rmklist}
\end{defnotation}

\begin{table}
\centering
\begin{tabular}{l|l}
	\textbf{Term} & \textbf{Notation} \\
 	\hline\hline
	(Ground truth) Signal vector & $\grtr \in \R^n$ \\ \hline
	Sensing vectors & $\a_1, \dots, \a_m \in \R^{m}$ \\
	Sensing matrix & $\A  = [a_{i,j}] = [\a_1 \dots \a_m]^\T \in \R^{m \times n}$ \\ \hline
	Noise variables & $\noise = (e_1, \dots, e_m) \in \R^{m}$ \\ \hline
	Measurement variables & $\y = (y_1, \dots, y_m) = \A \grtr + \noise \in \R^{m}$ \\ \hline
	Analysis vectors & $\avec_1, \dots, \avec_N \in \R^{n}$ \\
	Analysis operator & $\aop = [\psi_{k,j}] = [\avec_1 \dots \avec_N]^\T \in \R^{N \times n}$ \\ \hline
	Gram matrix & $\gram = [g_{k,k'}] = \aop \aop^\T = [\sp{\avec_k}{\avec_{k'}}]  \in \R^{N \times N}$ \\ \hline
	Analysis coefficients of $\x \in \R^n$ & $\aop \x = (\sp{\avec_1}{\x}, \dots, \sp{\avec_N}{\x}) \in \R^{N}$ \\ \hline
	Analysis support and sparsity & $\ssupp[\x] = \supp(\aop \x) \subset [N]$ and $S = \cardinality{\ssupp[\x]}$ \\ \hline
	Analysis cosupport and cosparsity & $\ssuppc[\x] = \setcompl{\supp(\aop \x)} \subset [N]$ and $L = \cardinality{\ssuppc[\x]} = N - S$ \\ \hline
	Analysis sign vector & $\anasign{\x} = (\sigma_{\x}^{1}, \dots, \sigma_{\x}^{N}) = \sign(\aop \x) \in \{-1,0,1\}^N$ \\
\end{tabular}
\caption{A summary of important notations used in this paper.}
\label{tab:results:notation}
\end{table}

\subsection{Generalized (Co-)Sparsity}
\label{subsec:results:parameters}

Before presenting the actual recovery results, we need to introduce three adapted notions of sparsity, which will form the heart of our sampling-rate bounds:
\begin{definition}[Generalized (Co-)Sparsity]\label{def:results:generalsparsity}
	Let $\aop \in \R^{N \times n}$ be an analysis operator and let $\x \in \R^n$. 
	We define the \define{generalized sparsity} of $\x$ (with respect to $\aop$) by
	\begin{equation}\label{eq:results:generalsparsity:spparam}
		\spparam{\x} \coloneqq \sum_{k,k' \in \ssupp[\x]} \sigma_{\x}^k \cdot \sigma_{\x}^{k'} \cdot g_{k,k'} \ ,
	\end{equation}
	where%\footnote{Here, $\sign(\aop \x)$ means that the $\sign$-function is applied entry-wise to $\aop \x$.}
	\begin{equation}
		\sigma_{\x}^{k} \coloneqq \sign(\sp{\avec_k}{\x}) \in \{-1,0,1\}, \quad k = 1, \dots, N,
	\end{equation}
	and $\anasign{\x} \coloneqq (\sigma_{\x}^{1}, \dots, \sigma_{\x}^{N}) \in \{-1,0,1\}^N$ is called the \define{analysis sign vector} of $\x$.
	Moreover, we introduce the terms
	\begin{align}
		\cospparam{\x} &\coloneqq \sum_{k,k' \in \ssuppc[\x]} \frac{g_{k,k'}^2}{\sqrt{g_{k,k} \cdot g_{k',k'}}} \ , \\
		\cospparamdg{\x} &\coloneqq \sum_{k \in \ssuppc[\x]} \sqrt{g_{k,k}} \ , \label{eq:results:generalsparsity:cospparam}
	\end{align}
	which are both referred to as the \define{generalized cosparsity} of $\x$.\footnote{The label ``$\vec{\smallsetminus}$'' indicates that $\cospparamdg{\x}$ only operates on the \emph{diagonal} entries of $\gram$.}
\end{definition}
Note that, for the sake of readability, we have omitted the dependence on $\aop$.
Considering the canonical choice of an orthonormal basis, it becomes actually clear why we speak of \emph{generalized} sparsity:
Since $\gram = \aop \aop^\T = \I{n}$ in this case, one obtains $\spparam{\x} = S = \lnorm{\aop\x}[0]$ and $\cospparam{\x} = \cospparamdg{\x} = L = n - S$.
Hence, the notions of Definition~\ref{def:results:generalsparsity} precisely coincide with their traditional counterparts.

In general, this correspondence is more complicated. The definition of the generalized sparsity $\spparam{\x}$ still operates on the analysis support of $\x$, but also involves a weighted sum over all Gram matrix entries associated with $\ssupp[\x]$.
The same holds true for the generalized cosparsity term $\cospparam{\x}$, respectively.
Such an incorporation of the off-diagonal entries of $\gram$ is in fact quite appealing because, to a certain extent, it respects the mutual coherence structure of the analysis vectors $\avec_1, \dots, \avec_N$.

Finally, let us state two basic observations, characterizing when the above sparsity parameters do not vanish (see also Lemma~\ref{lem:proofs:scbound:subd}):
\begin{align}
	\spparam{\x} = \lnorm{\aop^\T \anasign{\x}}^2 > 0 \quad &\text{if and only if} \quad \x \not\in \ker\aop, \quad \text{and} \\
	\cospparam{\x}, \cospparamdg{\x} > 0 \quad &\text{if and only if} \quad \ssuppc[\x] \neq \emptyset. \label{eq:results:parameters:nondegen}
\end{align}

\subsection{Sampling-Rate Function and Exact Recovery}
\label{subsec:results:samplecompl}

To highlight the key features of our approach, we restrict our analysis to the simplified case of noiseless Gaussian observations in this part, i.e., $\noiseparam = 0$ and $\a \distributed \Normdistr{\vnull}{\I{n}}$ in Model~\ref{model:results:setup:meas}.
The analysis basis pursuit then takes the form
\begin{equation}\label{eq:intro:anabpnoiseless}\tag{$\text{BP}_{\noiseparam=0}^{\aop}$}
	\min_{\x \in \R^n}\ananorm{\x} \quad \text{subject to \quad $\A \x = \y$,}
\end{equation}
\renewcommand{\refabpnoiseless}[2]{(\hyperref[eq:intro:anabpnoiseless]{$\text{BP}_{#1}^{#2}$})}%
and we can even hope for an \emph{exact} retrieval of $\grtr \in \R^n$.
Indeed, the recent work of \cite{amelunxen2014edge} has made the remarkable observation that a convex program of this type typically undergoes a sharp \emph{phase transition} as $m$ varies:
Recovery of $\grtr$ fails with overwhelmingly high probability if $m$ is below a certain threshold.
But once $m$ exceeds a small transition region, recovery succeeds with overwhelmingly high probability; see Figure~\ref{fig:intro:1d:pt_frame} for an example.
This minimal number of required measurements (also depending on the desired probability of success) is often referred to as the \emph{sample complexity} (or \emph{optimal sampling rate}) of an estimation problem.
However, computing such a quantity in an \emph{explicit} way is usually a challenging task.
\begin{figure}
\centering
\footnotesize
\includegraphics[width=0.5\textwidth]{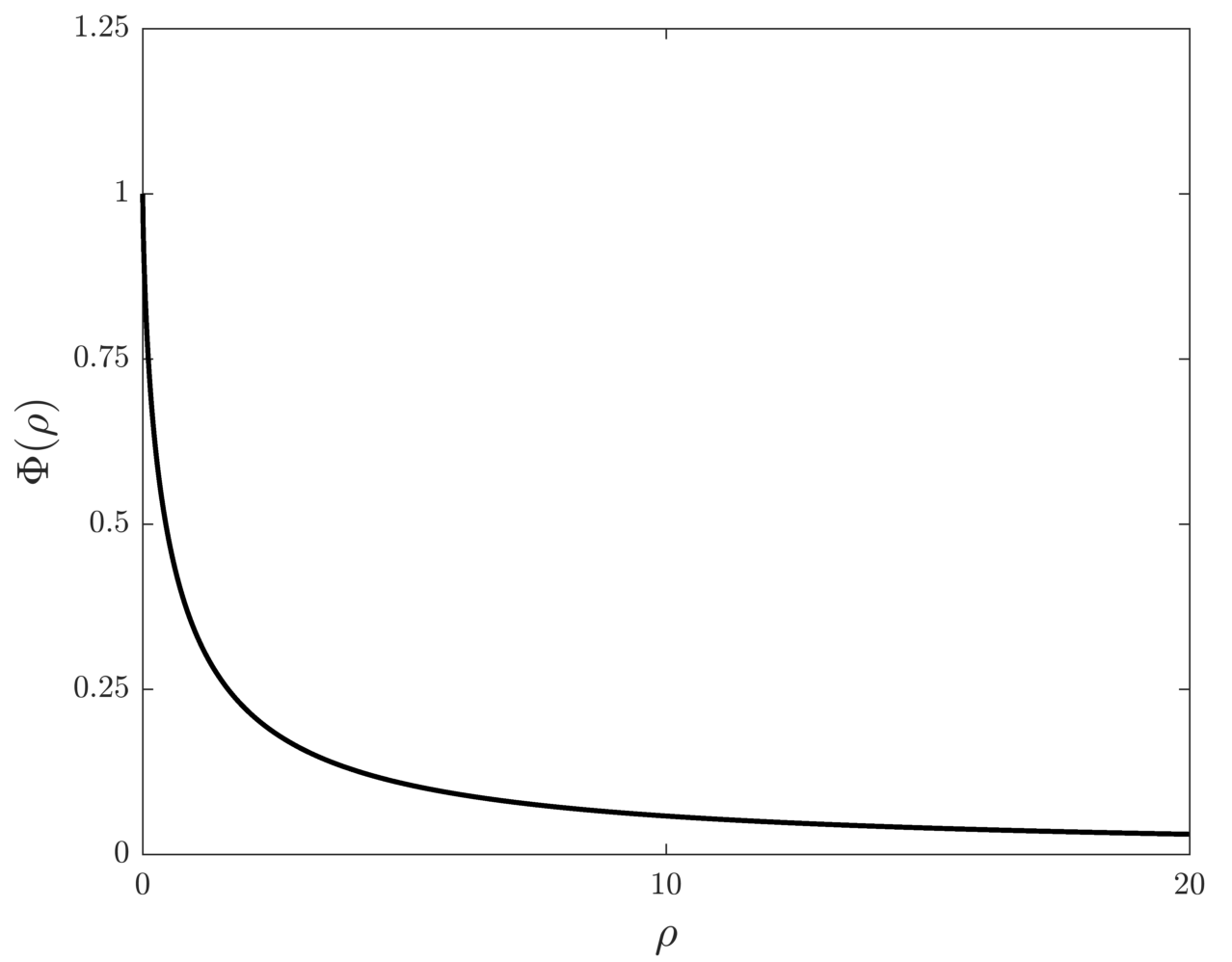} 
\caption[$\Phi$]{The graph of $\Phi$. This function is strictly monotonically decreasing (in particular, $\lim_{\rho \searrow 0}\Phi'(\rho) = -\infty$), and at its boundary points, we have $\lim_{\rho \searrow 0}\Phi(\rho) = 1$ and $\lim_{\rho \to \infty}\Phi(\rho) = 0$.}
\label{Fig:subsec:results:samplecompl:Phi}
\end{figure}

Our actual goal is therefore rather to come up with an upper bound on the sample complexity of \eqref{eq:intro:anabpnoiseless} that is tight in many situations.
For this purpose, let us introduce the following function, which essentially determines the sampling rate proposed by Theorem~\ref{thm:results:recovery} below:
\begin{definition}\label{def:results:samplingratefct}
	Let $\aop \in \R^{N \times n}$ be an analysis operator and let $\grtr \in \R^n$ with $\grtr \not\in \ker\aop$.
	Then, we define the \emph{sampling-rate function} of $\aop$ and $\grtr$ by
% 	Then we the \define{sample complexity parameter} (\define{SC-parameter}) of \eqref{eq:intro:anabp} with respect to $\grtr$ by
	\begin{equation}\label{eq:results:scfunc:samplecompl}
		\samplecompl{\aop, \grtr} \coloneqq \begin{cases} 
			\displaystyle n - \frac{(\cospparamdg{\grtr})^2}{\cospparam{\grtr}} \cdot \Phi\left( \frac{\spparam{\grtr}}{\cospparam{\grtr}} \right), & \text{if $\ssuppc[\grtr] \neq \emptyset$}, \\
			n\vphantom{\displaystyle\frac{1}{1}}, & \text{otherwise},
		\end{cases}
	\end{equation}
	where\footnote{Here, $\erf(\tau) := \tfrac{2}{\sqrt{\pi}} \int_{0}^\tau \exp(-x^2) \ dx $ denotes the \emph{error function}.}
	\begin{equation}\label{eq:results:scfunc:Phi}
		\Phi(\rho) \coloneqq \erf\left( \tfrac{h^{-1}(\rho)}{\sqrt{2}} \right) , \quad \rho > 0,
	\end{equation}
	with
	\begin{equation}\label{eq:results:scfunc:h}
		h\colon \intvop{0}{\infty} \to \intvop{0}{\infty}, \ \tau \mapsto \sqrt{\tfrac{2}{\pi}} \tfrac{e^{-\tau^2/2}}{\tau} + \erf(\tfrac{\tau}{\sqrt{2}}) - 1.
	\end{equation}
\end{definition}

% \pagebreak
It is not hard to verify that the univariate functions $h$, $h^{-1}$, and $\Phi$ ---each one mapping from $\intvop{0}{\infty}$ to $\intvop{0}{\infty}$ ---are well-defined; see first paragraph of Appendix~\ref{subsec:prelim:scbound:globalmin} for more details.
Together with \eqref{eq:results:parameters:nondegen}, this also implies the well-definedness of $\samplecompl{\aop, \grtr}$.
Figure~\ref{Fig:subsec:results:samplecompl:Phi} shows the graph of $\Phi$, visualizing how the ratio $\spparam{\grtr} / \cospparam{\grtr}$ affects the sampling-rate function.

With this notion at hand, we are now ready to state our first main result:
\begin{theorem}[Exact Recovery via \eqref{eq:intro:anabpnoiseless}]\label{thm:results:recovery}
	Assume that Model~\ref{model:results:setup:meas} is satisfied with $\a \distributed \Normdistr{\vnull}{\I{n}}$ and $\noiseparam = 0$. Let $\aop \in \R^{N \times n}$ be an analysis operator such that $\grtr \not\in \ker\aop$.
	Then, for every $\probsuccess > 0$, the following holds true with probability at least $1 - e^{-\probsuccess^2/2}$:
	If the number of measurements obeys
	\begin{equation}\label{eq:results:recovery:meas}
		m > \Big( \sqrt{\samplecompl{\aop, \grtr}} + \probsuccess \Big)^2 + 1,
	\end{equation}
	then \eqref{eq:intro:anabpnoiseless} recovers $\grtr$ exactly.
\end{theorem}

Theorem~\ref{thm:results:recovery} gives a first answer to \ref{quest:intro:samplecompl}: Roughly speaking, reconstruction succeeds with high probability as $m$ slightly exceeds the sampling-rate function $\samplecompl{\aop, \grtr}$.
Since $\samplecompl{\aop, \grtr}$ is completely determined by our generalized sparsity parameters from Definition~\ref{def:results:generalsparsity}, we can even conclude that the bound \eqref{eq:results:recovery:meas} meets the desiderata of \ref{item:intro:desiderata:computable}, \ref{item:intro:desiderata:interpretable}, and \ref{item:intro:desiderata:generic} requested in Subsection~\ref{subsec:intro:contrib}.

Unlike many approaches from the literature (cf. \eqref{eq:intro:boundliterature}), the statement of Theorem~\ref{thm:results:recovery} is highly signal-dependent and non-uniform. Indeed, the condition of \eqref{eq:results:recovery:meas} does not just involve the analysis sparsity $S = \lnorm{\aop\grtr}[0]$, but explicitly depends on the support $\ssupp[\grtr]$ as well as on the sign vector $\anasign{\grtr} = \sign(\aop\grtr)$.
We will return to this point in the course of our experiments in Section~\ref{sec:appl}, demonstrating that the performance of \eqref{eq:intro:anabpnoiseless} is oftentimes not fully explainable by means of $S$, even if $\aop$ is fixed.

On the other hand, our wish for an almost tight bound on the sample complexity has not been clarified yet (see \ref{item:intro:desiderata:accurate} in Subsection~\ref{subsec:intro:contrib}).
This issue unfortunately turns out to be very challenging, and in general, we do not have a quantitative error estimate for our prediction. 
However, we will give at least numerical evidence in Section~\ref{sec:appl} and discuss some theoretical aspects of optimality in Subsection~\ref{subsec:proofs:scbound} (see Remark~\ref{rmk:proofs:scbound:kkt}).

\begin{remark}
\begin{rmklist}
\item
	The assumption of $\grtr \not\in \ker\aop$ in Theorem~\ref{thm:results:recovery} is not very restrictive and rather of technical nature.
	If $\grtr \in \ker\aop$, the analysis basis pursuit \eqref{eq:intro:anabpnoiseless} uniquely recovers $\grtr$ if, and only if,
	\begin{equation}
		\ker\aop \intersec (\grtr + \ker\A) = \{\grtr\}.
	\end{equation}
	In the setting of Gaussian measurements, where $\ker\A$ can be identified with a random subspace of dimension $n-m$, the latter condition is fulfilled almost surely if and only if $m \geq \dim(\ker\aop)$.
	Thus, the dimension of $\ker\aop$ precisely yields the sample complexity of \eqref{eq:intro:anabpnoiseless} in this case.
\item
	An important feature of Theorem~\ref{thm:results:recovery} is \emph{scaling invariance}: Replacing $\aop$ by a scaled version $\lambda \aop$ with $\lambda \neq 0$ does not affect the sampling-rate function, i.e., $\samplecompl{\lambda\aop, \grtr} = \samplecompl{\aop, \grtr}$.
	This is due to the fact that the parameters $\spparam{\grtr}$, $\cospparam{\grtr}$, and $\cospparamdg{\grtr}$ do only appear in terms of appropriate fractions.
\item
	Since $\samplecompl{\aop, \grtr} \leq n$, the condition of \eqref{eq:results:recovery:meas} does not lead to situations where $m$ needs to be much larger than $n$ in order to achieve successful recovery.
	Such a requirement would be overly restrictive because the equation system $\y = \A\grtr$ is almost surely uniquely solvable if $m \geq n$, no matter what operator $\aop$ is applied.
	In contrast, this simple observation is not always reflected by a naive bound of the form \eqref{eq:intro:boundliterature}, at least when the domain of analysis coefficients $\R^N$ is much higher dimensional, i.e., $N \gg n$.
	\qedhere
\end{rmklist}
\end{remark}

The function $\Phi$ from Definition~\ref{def:results:samplingratefct} is quite easy to understand from an analytical perspective, but it is however non-elementary.
For that reason, let us conclude with a simplification of the bound \eqref{eq:results:recovery:meas} in Theorem~\ref{thm:results:recovery}.
This obviously comes along with a certain loss of accuracy but sheds more light on the asymptotic behavior of our sampling-rate function $\samplecompl{\aop, \grtr}$.
\begin{proposition}\label{prop:results:scfuncsimplyfied}
	Let $\aop \in \R^{N \times n}$ be an analysis operator and let $\grtr \in \R^n$ with $\grtr \not\in \ker\aop$ and $\ssuppc[\grtr] \neq \emptyset$.
	Setting $S \coloneqq \spparam{\grtr}$, $L \coloneqq \cospparam{\grtr}$, and $\bar{L} \coloneqq \cospparamdg{\grtr}$, we have
	\begin{equation}\label{eq:results:scfuncsimplyfied:bound}
		\samplecompl{\aop, \grtr} \leq \min\Big\{ n - \tfrac{\bar{L}^2}{L} + (\tfrac{\bar{L}}{L})^2 \cdot [2 S \cdot \log(\tfrac{S + L}{S}) + S], n - \tfrac{2}{\pi} \cdot \tfrac{\bar{L}^2}{S + L}  \Big\}.
	\end{equation}
	In particular, if $\aop \in \R^{n \times n}$ forms an orthonormal basis, we have $S = \lnorm{\aop\grtr}[0]$ and $L = \bar{L} = n - S$, so that
	\begin{equation}\label{eq:results:scfuncsimplyfied:classicalcs}
		\samplecompl{\aop, \grtr} \leq 2 S \cdot \log(\tfrac{n}{S}) + 2 S \lesssim S \cdot \log(\tfrac{2n}{S}).
	\end{equation}
\end{proposition}

The expression of \eqref{eq:results:scfuncsimplyfied:bound} is indeed much more explicit than $\samplecompl{\aop, \grtr}$.
The left branch particularly reminds us of a well-known property of traditional bounds, where the sampling rate is supposed to grow linearly with $S$ up to a logarithmic factor. This intuition is also confirmed by \eqref{eq:results:scfuncsimplyfied:classicalcs}, showing that our approach is consistent with the standard setup of compressed sensing.
On the contrary, the right branch of \eqref{eq:results:scfuncsimplyfied:bound} intends to mimic the asymptotics of $\samplecompl{\aop, \grtr}$ if the value of $\spparam{\grtr}$ is large compared to $\cospparam{\grtr}$, i.e., $\spparam{\grtr} / \cospparam{\grtr} \gg 1$.

\subsection{Stable and Robust Recovery}
\label{subsec:results:compressible}

In this part, we turn to our wish for stability as stated in \ref{quest:intro:compressibility}.
For illustration, let us recall the above situation of exact recovery and inspect the sampling-rate function $\grtr \mapsto \samplecompl{\aop, \grtr}$ more closely.
While this mapping actually insinuates a strong dependence on $\grtr$, a brief look at our generalized sparsity parameters reveals that it is already determined by the sign vector $\anasign{\grtr}$. In particular, the latter quantity does not take account of the magnitudes of the analysis coefficients. 
The size of $\samplecompl{\aop, \grtr}$ may therefore change dramatically (in a discontinuous manner) if $\grtr$ is slightly varied in such a way that the zero entries of $\aop\grtr$ become negligibly small but do not vanish anymore. 
For a fully populated coefficient vector $\aop\grtr$ with $\lnorm{\aop\grtr}[0] = N$, we would even have $\samplecompl{\aop, \grtr} = n$, which typically occurs in real-world examples, such as Figure~\ref{fig:results:compresible}.
In these situations, the ``binary'' statement proposed by Theorem~\ref{thm:results:recovery} is too pessimistic, meaning that one cannot expect exact recovery with $m \ll n$.
Instead of undergoing a sharp phase transition, the reconstruction error of \eqref{eq:intro:anabpnoiseless} will then decay rather smoothly as $m$ grows.

Our approach to this problem relies on a simple geometric idea: Approximate $\grtr \in \R^n$ by a vector $\grtrsparse \in \R^n$ whose analysis coefficients $\aop\grtrsparse$ are much sparser than $\aop\grtr$.
Such a ``surrogate'' of $\grtr$ is supposed to be of lower complexity in the sense that $\samplecompl{\aop, \grtrsparse}$ is small, and if $\grtrsparse$ and $\grtr$ are not too distant, we speak of \emph{analysis compressibility}.
In that case, one may hope that the actual performance of \eqref{eq:intro:anabp} is almost the same as if $\grtrsparse$ would be the ground truth signal, i.e., $\y = \A \grtrsparse$.
This desirable property of a recovery method is usually referred to as \emph{stability}.

\begin{figure}[!t]
	\centering
	\begin{subfigure}[t]{0.45\textwidth}
		\centering
		\includegraphics[width=\textwidth]{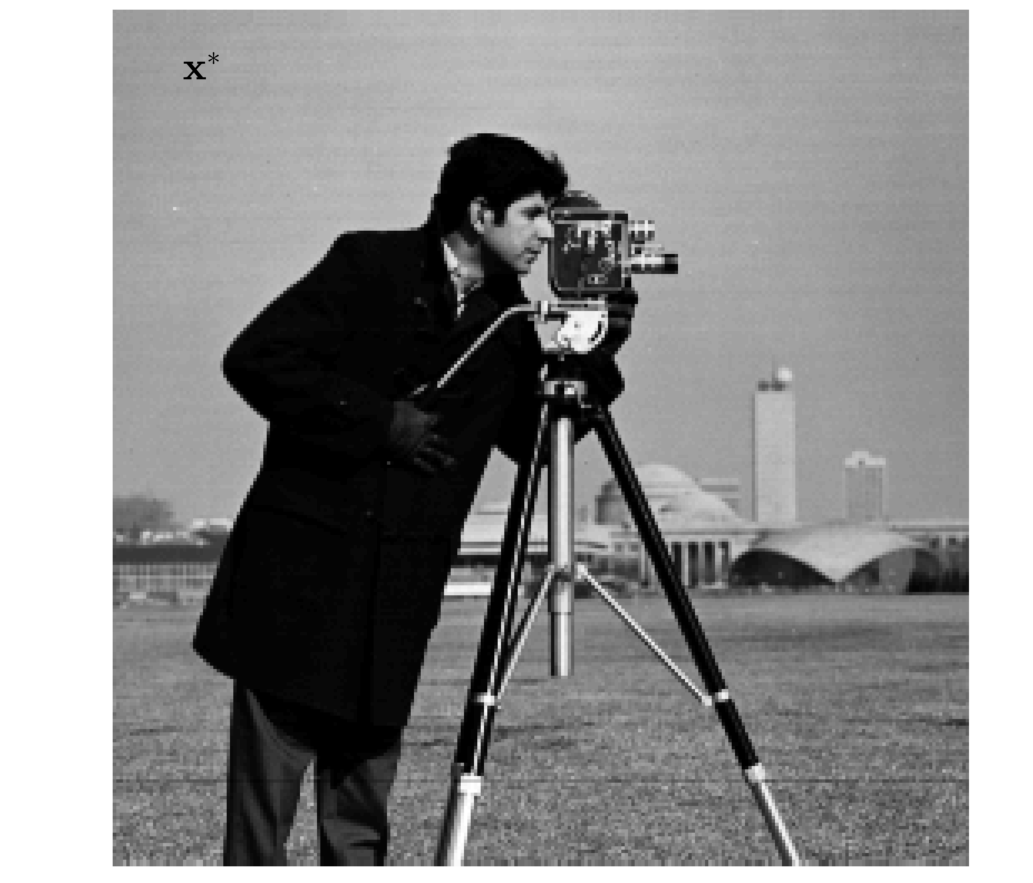}
		\caption{}
		\label{fig:results:compresible:cameraman}
	\end{subfigure}%
	\qquad
	\begin{subfigure}[t]{0.45\textwidth}
		\centering
		\includegraphics[width=\textwidth]{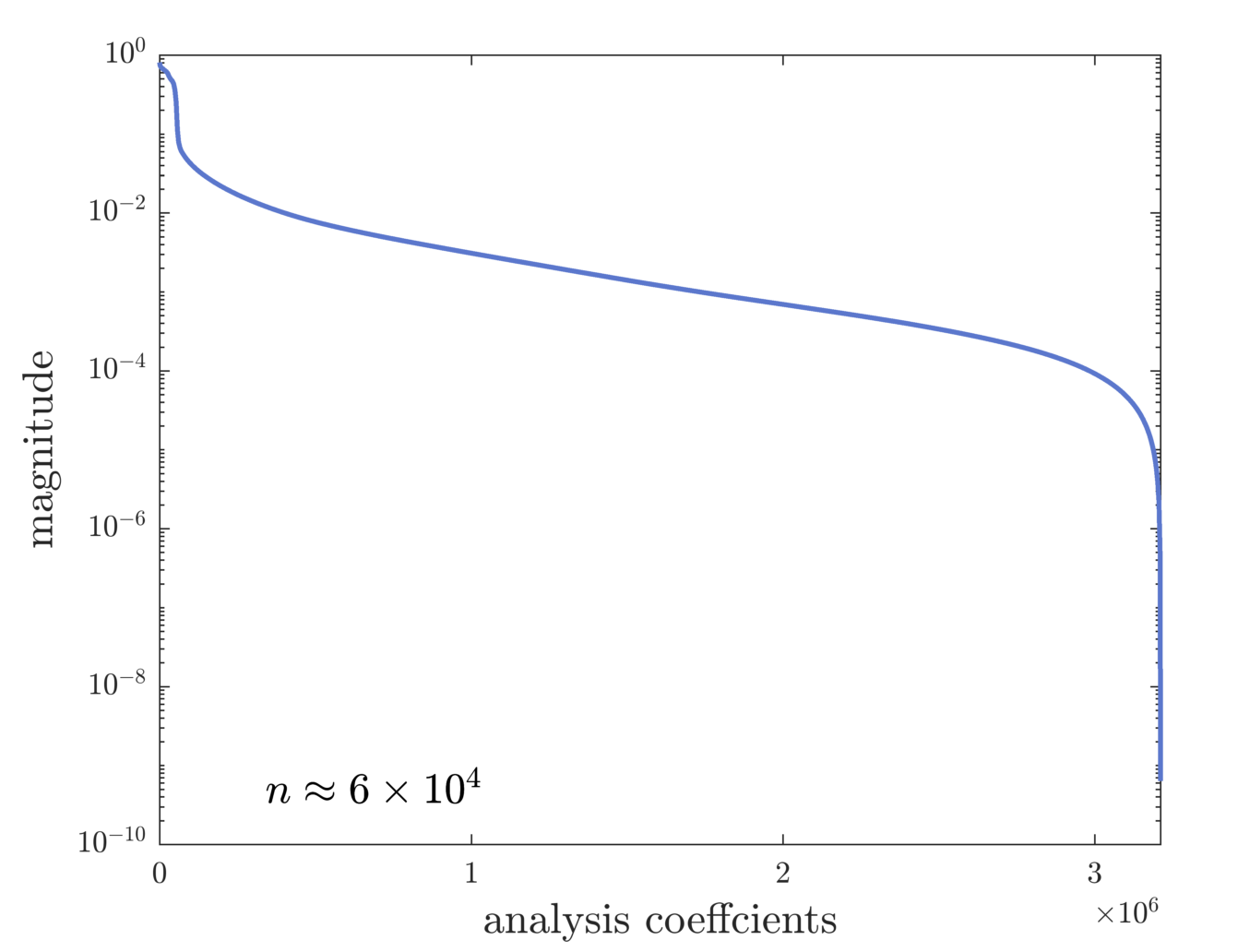}
		\caption{}
		\label{fig:results:compresible:decay}
	\end{subfigure}%	
	\caption{\subref{fig:results:compresible:cameraman}~Visualization of the ``cameraman'' $\grtr \in \R^{256\times 256}$. \subref{fig:results:compresible:decay}~Sorted $\l{1}$-normalized magnitudes (in logarithmic scale) of the analysis coefficients $\aop\grtr$, where $\aop$ corresponds to a compactly supported shearlet frame with four scales, provided by \texttt{Shearlab} \cite{kutyniok2016shearlab}.}
	\label{fig:results:compresible}
\end{figure}

A very natural, albeit naive strategy to come up with a good low-complexity approximation of~$\grtr$ is to solve
\begin{equation}\label{eq:results:compressible:optimalSterm}
	\min_{\xsparse \in \R^n} \lnorm{\grtr - \xsparse} \quad \text{subject to \quad $\cardinality{\ssupp[\xsparse]} = \lnorm{\aop \xsparse}[0] \leq S$}
\end{equation}
for some appropriately chosen sparsity threshold $S$. Any minimizer is then called a \define{best $S$-analysis sparse approximation of $\grtr$} (with respect to $\aop$).
Since the constraint set of \eqref{eq:results:compressible:optimalSterm} is actually a union of linear subspaces, it is also convenient to consider the following equivalent reformulation:
\begin{equation}\label{eq:results:compressible:optimalSterm:subspace}
	\min_{\substack{\ssupp \subset [N] \\ \cardinality{\ssupp} \leq S}} \lnorm{\grtr - \proj{\projnsp[\ssupp]} \grtr},
\end{equation}
where $\projnsp[\ssupp] \coloneqq \ker\aop_{\ssuppc} \subset \R^n$.
Interestingly, the same optimization problem was recently studied in \cite{giryes2014} under the name of \emph{optimal projections}, forming a crucial step in greedy-like algorithms for combinatorial cosparse modeling.
Although \eqref{eq:results:compressible:optimalSterm:subspace} is tractable in some simple cases, it eventually turned out to be NP-hard in general \cite{tillmann2014}, due to its combinatorial nature.
% However, selecting an optimal subset out of $\binom{N}{S}$ candidates unfortunately turns \eqref{eq:results:compressible:optimalSterm:subspace} into an NP-hard combinatorial problem in general, and moreover, the only solution could be just the null vector if $S$ is too small (cf. Subsection~???).
However, such algorithmic issues are only of minor importance to us, since we merely aim for a theoretical justification of stability in $\l{1}$-analysis recovery.

For these reasons, we shall state our second main result for Euclidean projections onto an arbitrary subspace $\projnsp \subset \R^n$, leaving space for applying less challenging (near-optimal) approximation methods.
Note that, compared to Theorem~\ref{thm:results:recovery}, we now also allow for sub-Gaussian measurements as well as the presence of noise.
\begin{theorem}[Stable and Robust Recovery via \eqref{eq:intro:anabp}]\label{thm:results:robuststable}
	Suppose that Model~\ref{model:results:setup:meas} is satisfied and let $\aop \in \R^{N \times n}$ be an analysis operator with $\grtr \not\in \ker\aop$.
	Let $\projnsp \subset \R^n$ be a linear subspace and assume that $\proj{\projnsp} \grtr \not\in \ker\aop$.
	Then
	\begin{equation}\label{eq:results:robuststable:grtrsparse}
				\grtrsparse \coloneqq \frac{\lnorm{\aop\grtr}[1]}{\lnorm{\aop\proj{\projnsp}\grtr}[1]} \cdot \proj{\projnsp} \grtr \in \R^n
	\end{equation}
	is well-defined and there exists a numerical constant $C > 0$ such that, for every $R > 0$ and $\probsuccess > 0$, the following holds true with probability at least $1 - e^{-\probsuccess^2/2}$:
	If the number of measurements obeys\footnote{One may also set $R = \infty$ here, with the convention that $\tfrac{R+1}{R} \coloneqq 1$ and $R \cdot 0 \coloneqq 0$.}
	\begin{equation}\label{eq:results:robuststable:meas}
		m > m_0 \coloneqq C \cdot \subgparam^4 \cdot \Big( \tfrac{R+1}{R} \cdot \big[\sqrt{\samplecompl{\aop, \grtrsparse}} + 1\big]  + \probsuccess \Big)^2 + 1,
	\end{equation}
	then any minimizer $\solu$ of \eqref{eq:intro:anabp} satisfies
	\begin{equation}\label{eq:results:robuststable:bound}
		\lnorm{\solu - \grtr} \leq \max\Big\{ R \lnorm{\grtr - \grtrsparse}, \tfrac{2\noiseparam}{\sqrt{m-1} - \sqrt{m_0 - 1}} \Big\} \ .
	\end{equation}
	In the Gaussian case of $\a \distributed \Normdistr{\vnull}{\I{n}}$, we particularly have $C = \subgparam = 1$.
\end{theorem}

The statement of Theorem~\ref{thm:results:robuststable} makes the above geometric principle precise:
Instead of evaluating the sampling-rate function $\samplecompl{\aop, \cdot}$ at the ground truth $\grtr$, we rather consider $\samplecompl{\aop, \grtrsparse}$ in \eqref{eq:results:robuststable:meas}, which could be significantly smaller. The price to pay is an additional error term in \eqref{eq:results:robuststable:bound}, governed by the distance between $\grtr$ and $\grtrsparse$.
Thus, one can formulate the following rule-of-thumb:
\begin{highlight}
	Recovery accuracy can be exchanged with signal complexity (taking fewer measurements), and vice versa.
\end{highlight}
The benignity of this trade-off clearly depends on the ansatz space $\projnsp$.
Above, we already pointed out that selecting $\projnsp$ is indeed very delicate and specific about the analysis operator $\aop$.
Moreover, even if one could identify an optimal support set $\ssupp$ according to \eqref{eq:results:compressible:optimalSterm:subspace}, it is still not clear whether setting $\projnsp = \projnsp[\ssupp] = \ker\aop_{\ssuppc}$ leads to the smallest sampling rate $\samplecompl{\aop, \grtrsparse}$ among all $S$-analysis-sparse approximations of $\grtr$.
For these reasons, we defer a general discussion of this issue to future works. 
A simple greedy strategy is however presented in Appendix~\ref{subsec:implementation:greedy}, producing satisfactory results for redundant wavelets frames (see Subsection~\ref{subsec:appl:compr}).

Compared to the setup of the previous subsection, the sub-Gaussianity of $\A$ comes along with an extra factor of $C \cdot \subgparam^4$ for the sampling rate, whereas the noise level $\noiseparam$ affects the actual error estimate. Roughly speaking, \eqref{eq:results:robuststable:bound} states that we can expect reliable outcomes as long as $\lnorm{\noise} / \sqrt{m} \leq \noiseparam / \sqrt{m} = \asympfaster{1}$, which is a well-known observation, e.g., see \cite[Cor.~3.5]{tropp2014convex}.

Finally, we would like to highlight that our approach is based on approximations in the signal domain $\R^n$.
This is fundamentally different from the predominant viewpoint in the literature which suggests working in the space of analysis coefficients $\R^N$ instead.
While those stability results often lead to convenient error bounds, based on the ordinary \emph{best $S$-term approximation error} of the coefficient vector $\aop\grtr$ (e.g., see \eqref{eq:appl:wavelets:err_kabanava2015}), they suffer from disregarding the coherence structure of $\aop$, which is particularly problematic in highly redundant scenarios; see Figure~\ref{fig:appl:wavelets:bestSterm} in Subsection~\ref{subsubsec:appl:wavelets:blocks} for more details.
%Working in the space of analysis coefficients $\R^N$ instead is actually more common in the literature (see Subsection~\ref{subsec:discussion:sparse}), but suffers from disregarding the coherence structure of $\aop$, which is in turn problematic in highly redundant scenarios.
On the other hand, Theorem~\ref{thm:results:robuststable} can be still connected to the coefficient domain, as the following bound for the approximation error in \eqref{eq:results:robuststable:bound} shows:
\begin{equation}\label{eq:results:robuststable:compressibility}
	\lnorm{\grtr - \grtrsparse} \leq \tfrac{\lnorm{\proj{\projnsp}\grtr}}{\lnorm{\aop\proj{\projnsp}\grtr}[1]} \cdot \lnorm{\aop\proj{\orthcompl{\projnsp}}\grtr}[1] + \lnorm{\proj{\orthcompl{\projnsp}}\grtr} \ ,
\end{equation}
see Proposition~\ref{prop:proofs:framework:compressibility:normbound} for a proof.
If $\aop$ is an orthonormal basis and $\projnsp$ corresponds to the coordinate space of the $S$-largest entries of $\aop\grtr$ in magnitude, it is not hard to see that this estimate is closely related to traditional bounds from compressibility theory (cf. \cite[Thm.~9.13]{foucart2013cs}).

\begin{remark}
\begin{rmklist}
\item
	The scaling factor in \eqref{eq:results:robuststable:grtrsparse} ensures that $\lnorm{\aop\grtr}[1] = \lnorm{\aop\grtrsparse}[1]$.
	Hence, $\grtrsparse$ is not precisely equal to the optimal approximation of $\grtr$ in $\projnsp$, unless $\lnorm{\aop\grtr}[1] = \lnorm{\aop\proj{\projnsp}\grtr}[1]$. This technicality is required for the proof of Theorem~\ref{thm:results:robuststable} (see Corollary~\ref{cor:proofs:framework:guarantee-stable}), 
	but we believe that it is not necessary in general.
\item
	The parameter $R$ in Theorem~\ref{thm:results:robuststable} allows us to fine-tune the trade-off described above: Enlarging $R$ causes a decrease in accuracy in \eqref{eq:results:robuststable:bound} but improves the sampling rate of \eqref{eq:results:robuststable:meas} at the same time.
	In this sense, the fraction $(R+1)/R$ can be also regarded as oversampling factor. % when $m$ is adjusted such that $m = m_0 + 1$.
\item 
	Albeit with a slightly different purpose, comparable trade-off principles concerning the computational and sample complexity have been recently studied in the literature: In \cite{chandrasekaran2013tradeoff}, a hierarchy of algorithms is defined, which is endowed with a statistical characterization and thereby allows for trading computational costs against statistical efficiency. This conception was further investigated in \cite{bruer2015tradeoff} where this trade-off is exploited by varying the degree of smoothing in the underlying optimization problem. Similarly, \cite{giryes2018tradeoff} compromises between the reconstruction error and convergence speed by modifying the iterations of the projected gradient descent algorithm. In this context, it turns out that a significant speed-up can be achieved by a complexity reduction of the underlying signal set; such a step is comparable to the approximation strategy proposed in this subsection. Finally, \cite{oymak2018sharp} establishes recovery guarantees for the (standard) projected gradient descent algorithm, elaborating a fundamental trade-off between accuracy, the number of iterations, and sample size.
		
	But despite various similarities, let us point out an important difference to the approach of Theorem~\ref{thm:results:robuststable}: The trade-offs in \cite{chandrasekaran2013tradeoff,bruer2015tradeoff,giryes2018tradeoff,oymak2018sharp} come along with a modification of the respective algorithmic methods, thereby balancing between accuracy, runtime, and sample size. In contrast, we consider the vanilla analysis basis pursuit and focus on a theoretical explanation of its stability, which can be viewed as a trade-off between the achievable accuracy and the available sample size.
	\qedhere
\end{rmklist}
\end{remark}

\section{Applications and Numerical Experiments}
\label{sec:appl}

The purpose of this section is to thoroughly evaluate the predictive quality of the theorems presented in Section~\ref{sec:results}. A series of experiments based on redundant wavelet frames (Subsection~\ref{subsec:appl:wavelets}), total variation minimization (Subsection~\ref{subsec:appl:tv}), and tight random frames in general position (Subsection~\ref{subsec:appl:onbrandom}) is conducted, which covers a wide range of analysis operators in one and two spatial dimensions (Subsection~\ref{subsec:appl:2d}). 
By creating phase transition plots, we visualize the sampling-rate bound of Theorem~\ref{thm:results:recovery} and investigate to what extent our wish for accuracy  is fulfilled.
Combined with a simple greedy approximation strategy for compressible signals, this also allows for a verification of our result on stable recovery (Theorem~\ref{thm:results:robuststable}).

While our theoretical framework includes a more general class of sub-Gaussian random matrices, we do only consider the benchmark case of noiseless Gaussian measurements, as it is often done in the compressed sensing literature. Let us point out again that the statement of Theorem~\ref{thm:results:recovery} is \emph{non-asymptotic} so that it can be directly compared with the true recovery performance of \eqref{eq:intro:anabpnoiseless}. 
In contrast, a visualization of asymptotic-order results, such as stated in \cite{candes2011csdict,needell2013}, is not compatible with this setup due to unspecified numerical constants.\footnote{These results are still non-asymptotic in the sense that they consider the case of finite $m$, but the sample complexity is only described by an asymptotic (upper) bound.}
The only other practically reasonable bound that we are aware of was established by Kabanava, Rauhut, and Zhang in \cite{rkz2015}, and will therefore often serve as an object of comparison.
%In contrast, a visualization of \emph{asymptotic} results, such as stated in \cite{candes2011csdict,needell2013}, is inappropriate for this setup due to unspecified numerical constants.
%The only other non-asymptotic bound, that we are aware of, was established by Kabanava, Rauhut, and Zhang in \cite{rkz2015}, and will therefore often serve as an object of comparison. 
See also Subsection~\ref{subsec:discussion:further} for a more formal presentation of the arguments presented in \cite{rkz2015}.

The notation of this part again follows Subsection~\ref{subsec:results:setup}, in particular Model~\ref{model:results:setup:meas}. 
% But for the sake of brevity, we will just write $\ssupp \coloneqq \ssupp[\grtr]$ and $\ssuppc \coloneqq \ssuppc[\grtr]$ for the support and cosupport, respectively.
As it is common in the literature (cf. \cite{amelunxen2014edge}), only the sampling-rate function $\grtr \mapsto \samplecompl{\aop,\grtr}$ is reported when analyzing the quality of Theorem~\ref{thm:results:recovery}, whereas the probability parameter $u$ in \eqref{eq:results:recovery:meas} is neglected.
Unless stated otherwise, we are solving the convex program \eqref{eq:intro:anabpnoiseless} using the \texttt{Matlab} software package \texttt{cvx}~\cite{cvx1,cvx2} with the default settings in place and the precision set to \texttt{best}. An outcome $\solu \in \R^n$ is considered to be an ``exact'' recovery of $\grtr \in \R^n$ if $\lnorm{\solu - \grtr}[2] < 10^{-5}$. Indeed, this threshold has proven to produce very stable solutions and seems to reflect the numerical accuracy of \texttt{cvx}.

\subsection{(Highly) Redundant Wavelet Frames}
\label{subsec:appl:wavelets}

A major difficulty in analysis modeling is to select an appropriate operator $\aop \in \R^{N \times n}$ that enables good reconstructions for a whole class of signals.
One of the most popular classes is formed by \emph{piecewise constant functions} (e.g., the \texttt{blocks} signal of~\cite{donoho94}), oftentimes regarded as a prototype system in the literature.
In this context, \emph{translation-invariant Haar wavelet frames} have turned out to be very useful for various signal- and image-processing tasks, most prominently for denoising \cite{coifman1995,donoho94}.
Depending on the number of scaling levels, such wavelet transforms are typically highly redundant, i.e., the total number of wavelet coefficients $N$ may be orders of magnitudes larger than the signal's dimension $n$.
% The linear dependencies within $\aop$ therefore imply that the sparsity of coefficient vectors is usually also much greater than $n$.
However, despite their success and popularity, there exists to the best of our knowledge no sound numerical or theoretical analysis of this important use case in respect of $\l{1}$-analysis recovery. 

We emphasize that the situation of high redundancy is not just a flaw of this specific setup, but is ubiquitous in \emph{applied harmonic analysis}. For instance, when applying \emph{Gabor frames} in time-frequency analysis, or (directional) representation systems like \emph{curvelets} or \emph{shearlets} in imaging, one often has to face problems with $N/n \geq 50$ \cite{kutyniok2016shearlab}.
Similar observations were recently also made for approaches to analysis operator learning from training data: For example, \cite{hawe2013} reported that a redundancy factor of $N/n = 128$ yields the best results for many classical imaging tasks.
Apart from being considerably redundant, these systems usually share the following properties:
\begin{itemize}
\item
	\emph{Linear dependency.} There are (necessarily) linear dependencies among the analysis vectors of the frame. 
\item
	\emph{Translation-invariance.} The analysis vectors are obtained by translating a family of generators.
\item\label{prop:appl:loc}
	\emph{Intrinsic localization.} The frame is well-localized in the sense of~\cite{Groechenig2004,Fornasier2005}. Roughly speaking, this means that the off-diagonal decay of the Gram matrix $\gram = \aop \aop^\T$ is fast.  
\end{itemize}

\subsubsection{Setup and Notation of Wavelet Frames}

Inspired by these observations, we build the simulations of this subsection on a combination of piecewise constant signals and Haar wavelet systems, serving as benchmark example of (highly) redundant frames. % in $\l{1}$-analysis recovery.
Our specific implementation of wavelet transforms relies on the \texttt{Matlab} software package \texttt{spot}~\cite{spot}, which is in turn adapted from the \texttt{Rice Wavelet Toolbox}~\cite{rwt}. The \emph{redundant wavelet transform} $\texttt{rdwt}: \R^n \to \R^N$ of this package is based on the so-called \emph{algorithme \`{a} trous}. Here, the downsampling steps of the discrete wavelet transform are omitted, whereas instead the filter responses are upsampled by padding with zeros, which reminds us of creating holes (\emph{trous} in French) in the non-zero coefficients. 
In the following, let $\aopr \in \R^{N \times n}$ denote the (matricized) analysis operator associated with this frame.
The redundant forward transform of \cite{spot} comes along with an \emph{inverse wavelet transform} $\texttt{irdwt}:\R^N \to \R^n$ satisfying
\begin{equation}\label{eq:appl:wavelets:forwardinverse}
	\texttt{irdwt} \circ \texttt{rdwt} = \Id_{\R^n}.
\end{equation}
From the viewpoint of frame theory, the application of \texttt{irdwt} corresponds to the synthesis operation with a dual frame of $\aopr$.
More precisely, if $\aopi \in \R^{N \times n}$ denotes the analysis operator of this dual frame, the rule \eqref{eq:appl:wavelets:forwardinverse} translates into $\aopi^\T \cdot \aopr = \I{n}$.
It is worth mentioning that $\aopr$ and $\aopi$ are not just academic concepts that are inaccessible in practice. Indeed, the analysis and synthesis operations of both frames are computable via the algorithm \`{a} trous with a complexity of $\asympfaster{n\log (n)}$. 

Finally, if $n$ is a power of two, the \emph{non-redundant discrete wavelet transform} $\texttt{dwt}: \R^n \to \R^n$ is orthogonal. Its associated analysis operator is denoted by $\aopd \in \R^{n \times n}$, forming an orthonormal basis for $\R^n$. Unless stated otherwise, we will always choose a Haar wavelet filter and $6$ decomposition levels for $\aopr$, $\aopi$, and $\aopd$.
For a more detailed introduction to (redundant) wavelet transforms, the interested reader is referred to~\cite{mallat09,fowler05}.

\subsubsection{Recovery of \texttt{blocks} Signal}
\label{subsubsec:appl:wavelets:blocks}

Our first experiment revisits the simulation of Figure~\ref{fig:intro:1d} in Subsection~\ref{subsec:intro:issues}, where the reconstruction of the classical \texttt{blocks} signal \cite{donoho94} was investigated for different choices of Haar wavelet frames.  
Technically, we fix the ambient dimension $n=256$, initialize $\grtr \in \R^n$ as \texttt{blocks} (see Figure~\ref{fig:intro:1d:signal}), and run Experiment~\ref{exp:appl:wavelets:intro} below for $\aop \in \{\aopr,\allowbreak \aopi,\allowbreak \aopd\}$ and $M = \left\{ 1,\dots,n\right\}$. The plots of Figure~\ref{fig:intro:1d:pt_frame} and Figure~\ref{fig:intro:1d:pt_onb} show the empirical success rates\footnote{As pointed out above, ``success'' means that $\grtr$ is recovered up to an accuracy of $\lnorm{\grtr - \solu} < 10^{-5}$.} for $\aop \in \{\aopr, \aopi\}$ and $\aop = \aopd$, respectively.

\begin{experiment}[Recovery of a fixed signal]\leavevmode
\label{exp:appl:wavelets:intro}
\vspace{-.25\baselineskip}\hrule\vspace{.5\baselineskip}

\myhangindent{Input: \ }
\expkwd{Input:} Fixed ambient dimension $n$, signal vector $\grtr \in \R^n$, analysis operator $\aop \in \R^{N \times n}$, range of measurements $M \subset \{1,2,\dots,n\}$.

\vspace{.5\baselineskip}
\expkwd{Compute:} Repeat the following procedure $50$ times for every $m \in M$:
\begin{expstep}
	 \item 
		Draw a standard i.i.d.\ Gaussian random matrix $\A \in \R^{m\times n}$ and determine the measurement vector $\y = \A\grtr$.
	\item 
		Solve the analysis basis pursuit \eqref{eq:intro:anabpnoiseless} to obtain an estimator $\solu \in \R^n$. 
	\item
		Compute and store the recovery error $\lnorm{\grtr - \solu}$.
\end{expstep}
\end{experiment}

Figure~\ref{fig:intro:1d:pt_frame} reveals that a perfect reconstruction of $\grtr$ with $\aopr$ requires almost $m\approx n$ measurements. 
This observation is remarkable, since using $\aopr$ for denoising via soft thresholding has proven to be very effective \cite{coifman1995}, whereas it seems to be quite inappropriate for the $\l{1}$-analysis formulation.
In contrast, the choice of $\aopi$ performs significantly better, succeeding from $m\approx 105$ onward.
Interestingly, the same superiority remains true for analysis basis pursuit denoising (i.e., applying \eqref{eq:intro:anabp} with $\A = \I{n}$), which indicates that the analysis approach is fundamentally different from classical denoising via soft thresholding.

The subsampled orthonormal basis $\aopd$ behaves slightly worse than $\aopi$ and allows for retrieval with $m\approx 125$ onward (see Figure~\ref{fig:intro:1d:pt_onb}).
This gives further evidence that the redundancy of a sparsifying transform can help to improve the recovery capability of analysis-based priors.

A closer look at the structure of $\aopr$ and $\aopi$ reveals that they actually correspond to the same frame up to a reweighting of their analysis vectors. In fact, we have
\begin{equation}\label{eq:appl:wavelets:reweighting}
\avec_{\texttt{rdwt},k} = \frac{\avec_{\texttt{irdwt},k}}{\lnorm{\avec_{\texttt{irdwt},k}}[2]} \in \R^n \quad \text{for all $k=1,\dots,N$,}
\end{equation}
which particularly implies that the analysis supports $\supp(\aopr\grtr)$ and $\supp(\aopi\grtr)$ coincide. Following traditional approaches based on analysis sparsity \cite{kabanava2015,candes2011csdict} or cosparse modeling \cite{nam2013}, one would therefore conclude that the number of samples required for exact recovery should be equal as well. 

Since this is far from being true, one might be tempted to take the perspective of analysis compressibility, instead of insisting upon perfect sparsity.
However, the sharp phase transition in Figure~\ref{fig:intro:1d:pt_frame} as well as the decay of the analysis coefficients in Figure~\ref{fig:intro:1d:decay} show that this principle does not provide a valid explanation. 
Indeed, even the most accurate and explicit bound known from the literature \cite[Thm.~2]{kabanava2015} predicts an error rate of 
\begin{equation}\label{eq:appl:wavelets:err_kabanava2015}
	\lnorm{\grtr - \solu}[2] \leq \frac{2 \sigma_S(\aop \grtr)_1}{\sqrt{\fbl \cdot S}},
\end{equation}
supposed that, roughly speaking, $m \geq 10 \cdot (\fbu / \fbl) \cdot S \cdot\log(e \cdot N/S)$. Here, $\fbl$ and $\fbu$ are lower and upper frame bounds of $\aop$, respectively, and $\sigma_S(\aop \grtr)_1$ denotes the best $S$-term approximation error of $\aop\grtr$ (cf. Subsection~\ref{subsec:intro:notation}\ref{item:intro:notation:lp}).
The visualization of \eqref{eq:appl:wavelets:err_kabanava2015} in Figure~\ref{fig:appl:wavelets:bestSterm} indicates that a reliable reconstruction can be only expected if the sparsity $S$ exceeds the ambient dimension $n$ by a factor of $3$, which in turn requires a sampling rate $m \gg n$.
Moreover, these plots do even suggest a superiority of $\aopr$ over $\aopi$, whereas Figure~\ref{fig:intro:1d:pt_frame} shows that the opposite is true.
%The visualization of \eqref{eq:appl:wavelets:err_kabanava2015} in Figure~\ref{fig:appl:wavelets:bestSterm} indicates that a reliable reconstruction can be only expected if $m$ exceeds the ambient dimension $n$ by a factor of $3$, contradicting the actual performance of \eqref{eq:intro:anabpnoiseless}.
%Moreover, the plots do even certify a superiority of $\aopr$ over $\aopi$, whereas Figure~\ref{fig:intro:1d:pt_frame} shows that the complete opposite is true.
We finally would like to point out that an argumentation based on the matrix condition is also misleading in this setup because both matrices $\aopr$ and $\aopi$ actually have the same condition number (which equals the square root of the frame bound ratio $\fbu / \fbl$); see Table~\ref{tab:appl:wavelets}.
% Finally, in this regard, one could argue that $\aopr$ is not as well conditioned as $\aopi$. However, as Tabular \ref{tab:appl:wavelets} shows, the ratio of the frame bounds (which equals the condition number) is actually the same for both frames.     
% We conclude from this discussion that the existing bounds based on sparsity or cosparsity do not qualitatively reflect the recovery behavior using $\aopr$ and $\aopi$. Furthermore, they also don't allow for a quantitative prediction for a a sufficient number of measurements.

\begin{figure}[!t]
	\centering
	\includegraphics[width=0.5\textwidth]{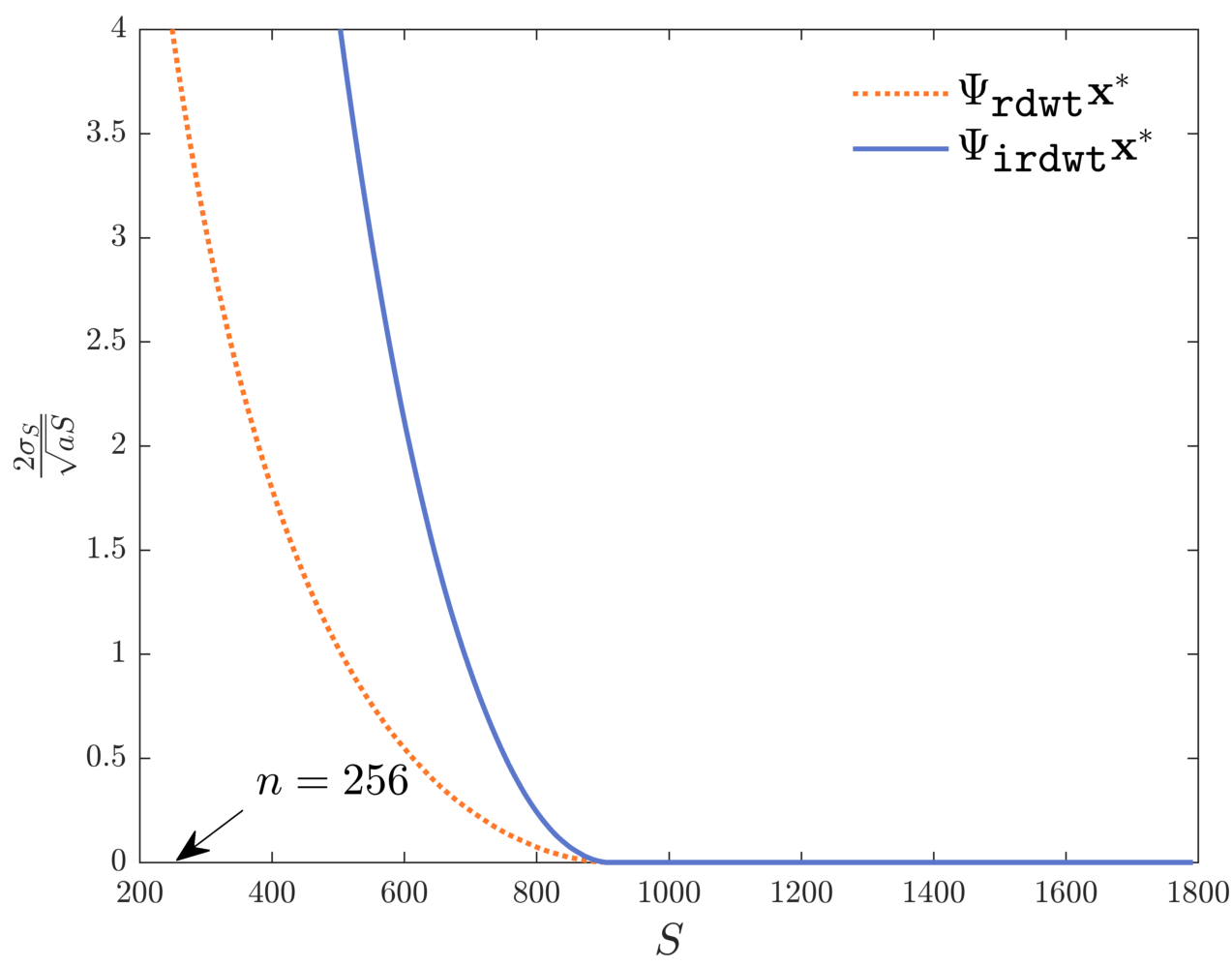}	
	\caption{Predicted error decay of \eqref{eq:appl:wavelets:err_kabanava2015} proposed by \cite[Thm.~2]{kabanava2015}.}
	\label{fig:appl:wavelets:bestSterm}
\end{figure}

Turning towards the predictive quality of Theorem~\ref{thm:results:recovery}, we observe in Figure~\ref{fig:intro:1d:pt_frame} that our bound on the sampling rate very precisely captures the location of the phase transitions for both $\aopr$ and $\aopi$.
While $\samplecompl{\aopr, \grtr}$ almost perfectly hits the center of the transition curve ($50\%$ success rate), $\samplecompl{\aopi, \grtr}$ is slightly more pessimistic, located at a rate of $75\%$ success. There seems to be a common pattern in all of our examples: The higher the sample complexity\footnote{The true sample complexity is numerically computed by means of the statistical dimension \cite{amelunxen2014edge}, see Appendix~\ref{subsec:implementation:ed} for details.} $\ed{\aop,\grtr}$, the more accurately it is approximated by $\samplecompl{\aop,\grtr}$.
On the other hand, the prediction for $\aopd$ in Figure~\ref{fig:intro:1d:pt_onb} is perfect, which is due to the fact that our bound is provably tight in the orthonormal case (cf. Proposition~\ref{prop:results:scfuncsimplyfied}).
Tabular~\ref{tab:appl:wavelets} lists a number of key quantities for all previously discussed choices of $\aop$. For the sake of completeness, in the last line, we have also documented the same parameters for a version of $\aopi$ with $3$ scales instead of $6$.
% Testing various choices for the number of scales for recovering $\grtr$ reveals that the optimal number of Scales is 5, which we have omitted to report since the quantities are very similar to having 6 scales. 

In a nutshell, we can conclude our discussion as follows:
\begin{highlight}
	The capacity of the analysis basis pursuit using highly redundant frames is not solely captured by \mbox{(co-)sparsity}. 
	In contrast, the sampling-rate bound proposed by Theorem~\ref{thm:results:recovery} reliably states when recovery is possible and when not. 
\end{highlight}

\begin{remark}\label{rmk:appl:wavelets:blocks:reweighting}
The main reason for choosing $\aopi$ and $\aopr$ as a basis of comparison in this part is twofold: firstly, both analysis operators are used in practice and allow for efficient computations, and secondly, their atoms are scalewise reweighted versions of each other; cf. \eqref{eq:appl:wavelets:reweighting}. However, one might be tempted to regard both instances as the same analysis operator and simply trace back the observed performance discrepancy to a differently weighted $\l{1}$-norm in the analysis basis pursuit; note that $\aopi$ and $\aopr$ actually have the same mutual coherence (see Table~\ref{tab:appl:wavelets}). But in fact, this gap goes beyond a weighting of the $\l{1}$-norm. To this end, let us consider an exemplary analysis operator $\aopart = [\avecart_{1} \dots \avecart_{N}]^\T$, which is (adaptively) constructed from $\aopi$ and $\grtr$ in following way: If $k \in \supp(\aopi\grtr)$, then the associated analysis vector $\avecart_{k}$ is generated as a random linear combination from elements in the set $\{ \avec_{\texttt{irdwt},k'} \suchthat k' \in \supp(\aopi\grtr)\}$; and similarly for $k \not\in \supp(\aopi\grtr)$. It is then clear that $\aopart$ and $\aopi$ are structurally very different, and in particular, their mutual coherence parameters deviate strongly from each other (see Table~\ref{tab:appl:wavelets}). By construction, however, it still holds that $\supp(\aopi\grtr) = \supp(\aopart\grtr)$ almost surely. But despite a perfectly matching analysis (co-)support, the parameters in Table~\ref{tab:appl:wavelets} reveal that also in this case $\ed{\aopart,\grtr}$ differs considerably from $\ed{\aopi,\grtr}$, while our sampling-rate bound still reflects this discrepancy.
\end{remark}

\begin{table}
\smaller{
\begin{center}
\begin{tabular}{c|cccccccccccc|c}
$\aop$ & Scales & $\fbl$ & $\fbu$ & $\nicefrac{\fbu}{\fbl}$ & $N$  & $S$ & $L$  & $\spparam{\grtr}$ & $\cospparam{\grtr}$ & $\cospparamdg{\grtr}$ & $\mu(\aop)$ & $\samplecompl{\aop, \grtr}$ & $\ed{\aop,\grtr}$\\
\hline\hline
$\aopd$ & $6$ & $1$ & $1$ & $1$ & $256$ & $41$ & $215$ & $41$ & $215$ & $215$ & $0$ & $114$ & $114$ \\
$\aopr$ & $6$ & $2$ & $64$ & $32$ & $1792$ & $906$ & $886$ & $28462$ & $5180$ & $886$ & $0.98$ & $241$ & $240$ \\
 $\aopi$ & $6$ & $\nicefrac{1}{64}$ &  $\nicefrac{1}{2}$ & $32$ & $1792$ & $906$ & $886$ & $33$ & $197$ & $ 212$ & $0.98$ & $ 100$ & $84$ \\
 $\aopi$ & $3$ & $\nicefrac{1}{8}$ & $\nicefrac{1}{2}$ & $4$ & $1024$ & $306$ & $718$ & $39$ & $194$ & $209$ & $0.88$ &  $109$ & $95$ \\
 $\aopart$ & - & $0.05$ & $31.3$ & $591$ & $1792$ & $906$ & $886$ & $1546$ & $8509$ & $276$ & $0.66$ &  $250$ & $235$ \\
\end{tabular}
\end{center}}
\caption{Characteristic parameters of the simulation in Figure~\ref{fig:intro:1d}. 
Here, $\fbl$ and $\fbu$ denote the lower and upper frame bound of $\aop$, respectively. The \mbox{(co-)sparsity} of the analysis coefficients $\aop\grtr$ are $S = \cardinality{\ssupp[\grtr]}$ and $L = \cardinality{\ssuppc[\grtr]} = N - S$, whereas $\spparam{\grtr}$, $\cospparam{\grtr}$, and $\cospparamdg{\grtr}$ correspond to the generalized sparsity terms from Definition~\ref{def:results:generalsparsity}. The \emph{normalized mutual coherence} of $\aop$ is computed by $\mu(\aop) = \max_{k \neq k'} {\abs{\sp{\avec_k}{\avec_{k'}}}/(\lnorm{\avec_k}\cdot\lnorm{\avec_{k'}})}$. The sampling-rate function $\samplecompl{\aop, \grtr}$ is defined according to Definition~\ref{def:results:samplingratefct}. The true sample complexity is denoted by $\ed{\aop,\grtr}$, see Appendix~\ref{subsec:implementation:ed}.
Note that all values are rounded.}
\label{tab:appl:wavelets}
\end{table}

\subsubsection{Asymptotic Behavior of the Sampling Rate in $n$}
\label{subsubsec:appl:wavelets:asymp}

In the setup of the previous subsection, we have fixed both the ambient dimension $n$ and the ground truth $\grtr \in \R^n$. %, while analyzing various choices of $\aop$. The sampling-rate function $\samplecompl{\aop,\grtr}$ however heavily depends on the value of $n$, which explicitly appears as \emph{additive} term in its definition. 
% Note that this is not the case in traditional bounds of the form \eqref{eq:intro:boundliterature}, where $n$ just affects a logarithmic factor. 
Compared to that, changing $n$ does also require an adaptation of the analysis operator $\aop \in \R^{N \times n}$, so that the value of $\samplecompl{\aop,\grtr}$ is affected in a highly non-trivial way. 
For these reasons, we next investigate the accuracy of $\samplecompl{\aop,\grtr}$ if $n$ is varied, while $\grtr$ and $\aop$ are just adapted in terms of different resolution levels.
Since the theoretical guarantee of \cite[Thm.~3.3]{rkz2015} is of a similar type, it is also reported in our simulation below.
% 
%  Therefore, in this section, we want to preliminarily investigate how the accuracy of $\samplecompl{\aop,\grtr}$ scales in $n$. To the best of our knowledge, the only other bound that is of a similar form as Equation~\eqref{eq:results:recovery:meas}, containing $n$ as a summand, is~\cite{rkz2015} and will therefore also be reported.

\begin{figure}[!t]
	\centering
	\begin{subfigure}[t]{0.45\textwidth}
		\centering
		\includegraphics[width=\textwidth]{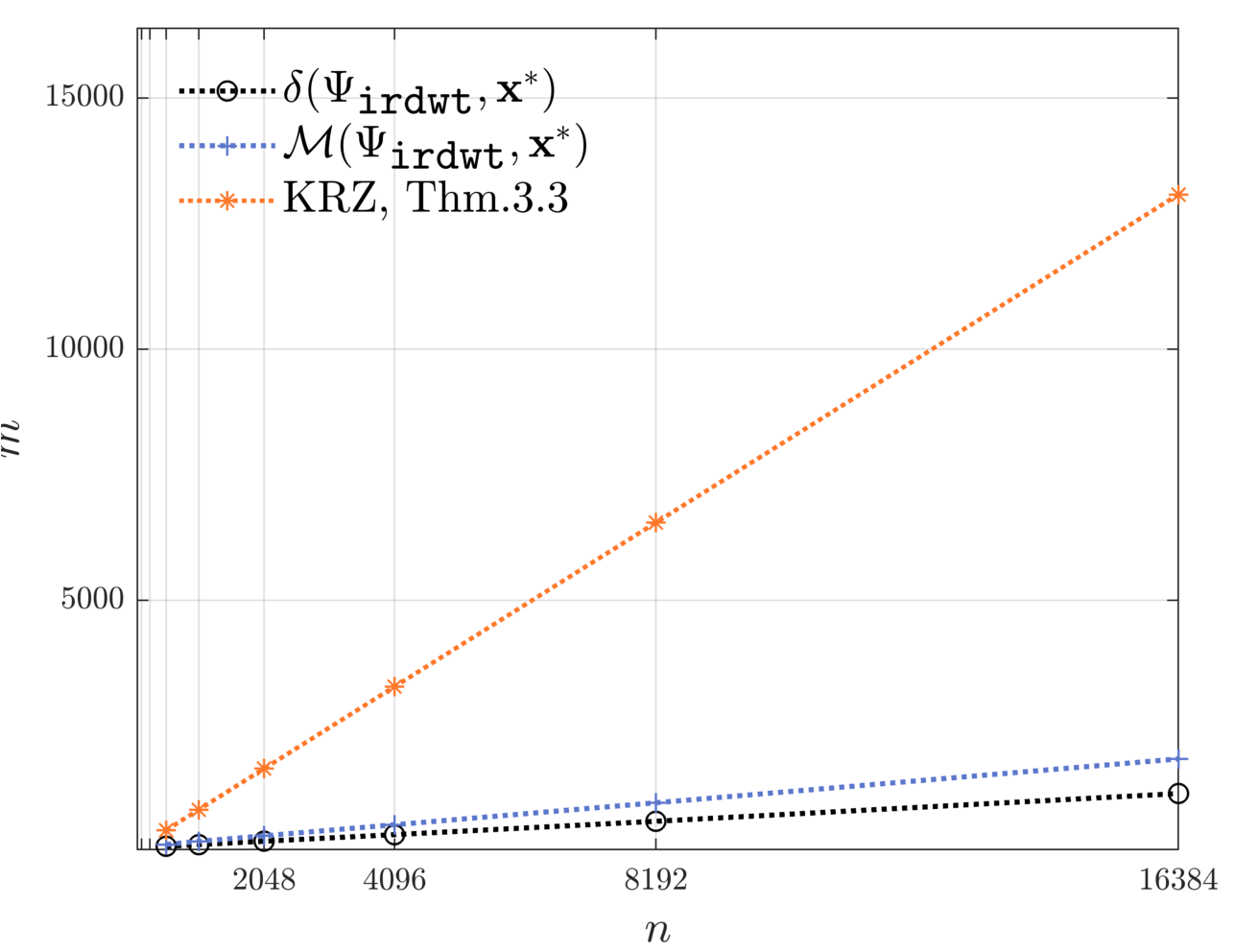}
		\caption{}
		\label{fig:appl:wavelets:scaling:bounds_i}
	\end{subfigure}%
	\qquad
	\begin{subfigure}[t]{0.44\textwidth}
		\centering
		\includegraphics[width=\textwidth]{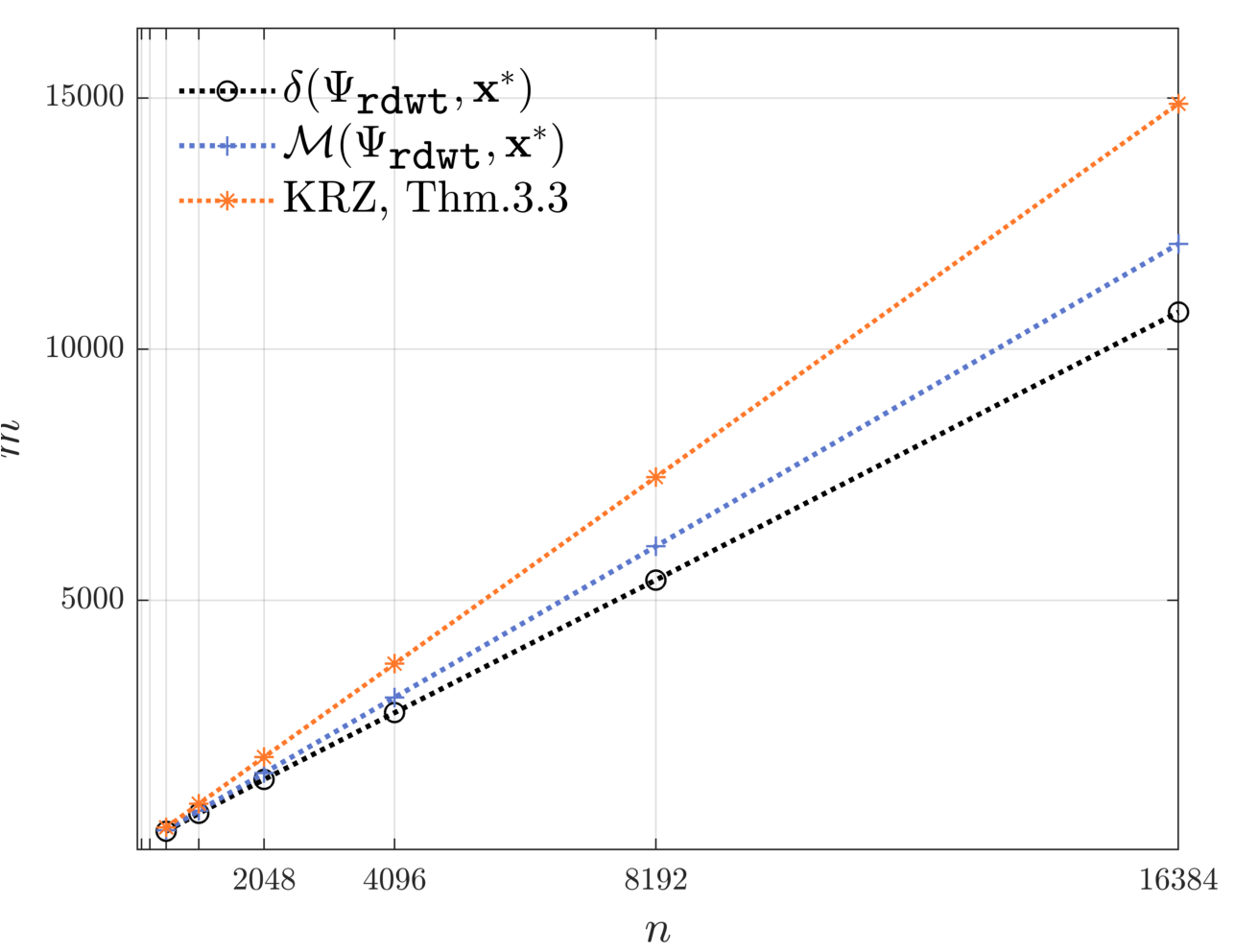}
		\caption{}
		\label{fig:appl:wavelets:scaling:bounds_r}
	\end{subfigure}%	
		\qquad
	\begin{subfigure}[t]{0.45\textwidth}
		\centering
		\includegraphics[width=\textwidth]{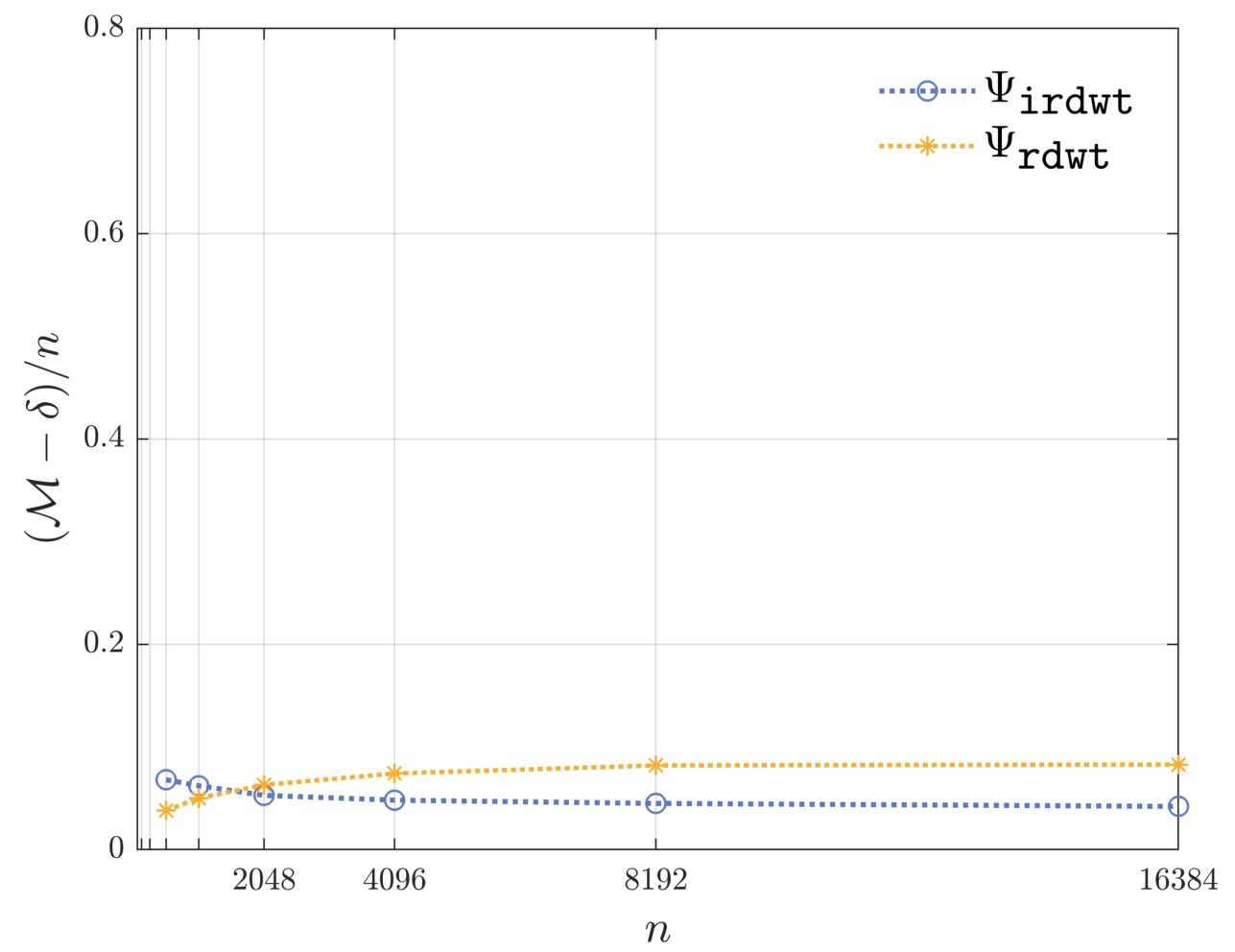}
		\caption{}
		\label{fig:appl:wavelets:scaling:RelE1}
	\end{subfigure}%
			\qquad
	\begin{subfigure}[t]{0.45\textwidth}
		\centering
		\includegraphics[width=\textwidth]{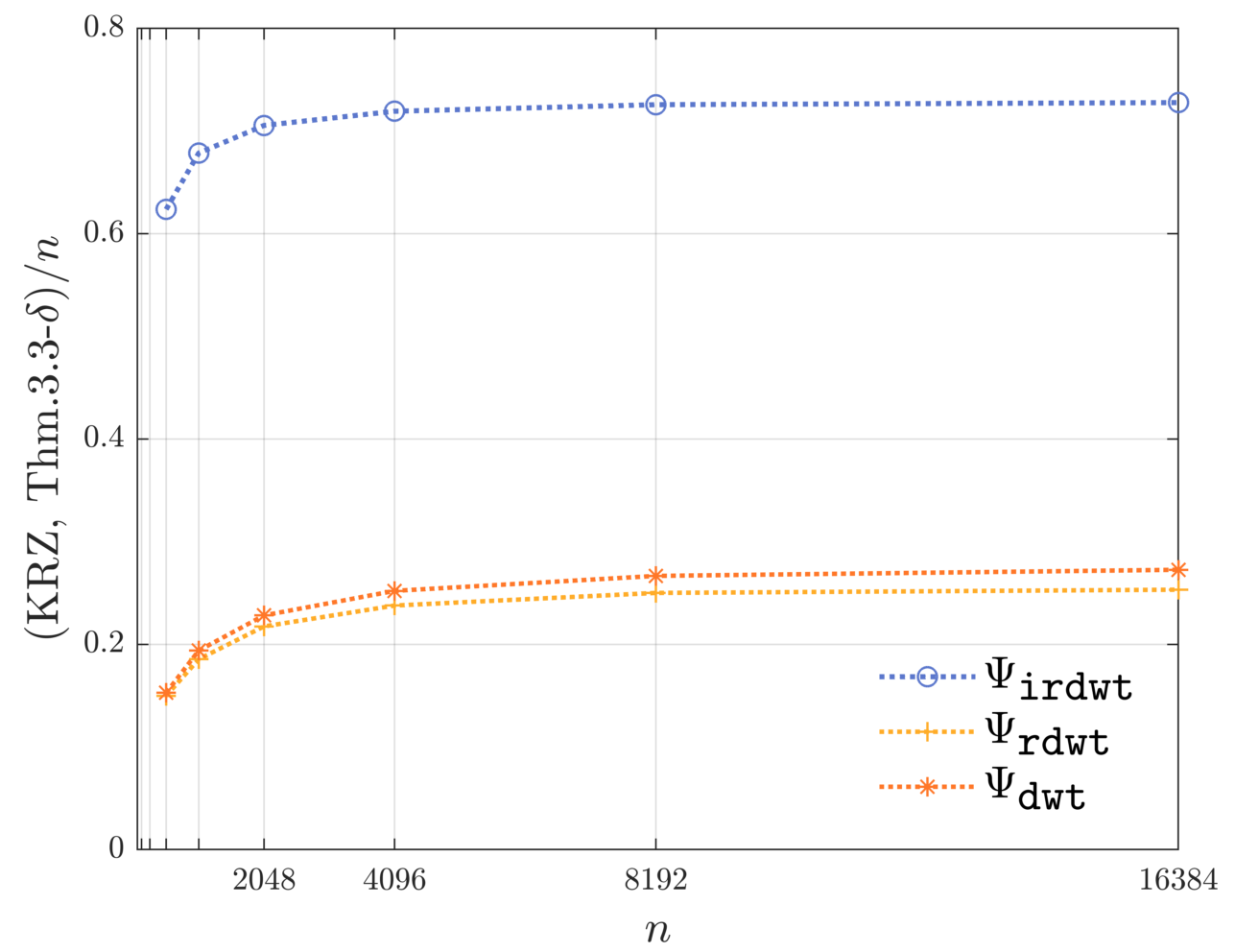}
		\caption{}
		\label{fig:appl:wavelets:scaling:RelE2}
	\end{subfigure}%
 	\caption{Subplots~\subref{fig:appl:wavelets:scaling:bounds_i} and \subref{fig:appl:wavelets:scaling:bounds_r} visualize how the predicted sampling-rate bounds scale with the ambient dimension $n$, where $\grtr \in \R^n$ is the \texttt{blocks} signal. Subplots~\subref{fig:appl:wavelets:scaling:RelE1}  and \subref{fig:appl:wavelets:scaling:RelE2} report the relative errors with respect to $n$.}
	\label{fig:appl:wavelets:scaling}
\end{figure}

In order to create the plots of Figure~\ref{fig:appl:wavelets:scaling}, we again consider a combination of $\grtr \in \R^n$ as the \texttt{blocks} signal \cite{donoho94} and the analysis operators $\aopi, \aopr \in \R^{N \times n}$, and $\aopd \in \R^{n \times n}$, where $n \in \{2^9, 2^{10},\dots, 2^{14}\}$.
Figures~\ref{fig:appl:wavelets:scaling:bounds_i} and \ref{fig:appl:wavelets:scaling:bounds_r} reveal that the true sample complexity $\ed{\cdot,\grtr}$ and our prediction by $\samplecompl{\cdot,\grtr}$ both scale linearly with the ambient dimension $n$, with almost the same slope. We have omitted the corresponding plot for $\aopd$ here (which would also exhibit a linear growth with $n$), since $\samplecompl{\aopd,\grtr}$ yields a perfect prediction of $\ed{\aopd,\grtr}$; cf. Proposition~\ref{prop:results:scfuncsimplyfied}. While the upper bound of \cite[Thm.~3.3]{rkz2015} also scales linearly with $n$, its slope differs significantly from that of the true sample complexity, especially for $\aopi$.
% Although the latter range is restricted to powers of two for technical reason, Figure~\ref{fig:appl:wavelets:scaling:bounds} indicates that both $\samplecompl{\aopi,\grtr}$ and $\samplecompl{\aopd,\grtr} = \ed{\aopd,\grtr}$ scale almost logarithmically in $n$. Compared to that, the prediction of \cite[Thm.~3.3]{rkz2015} rather exhibits a linear growth.

To assess the mismatch of our bound with the true sample complexity, Figure~\ref{fig:appl:wavelets:scaling:RelE1} visualizes the relative error 
\begin{equation}
	\frac{\samplecompl{\cdot,\grtr} - \ed{\cdot,\grtr}}{n},
\end{equation}
showing that the error curve drops to $\sim 4\%-8\%$ for increasingly higher dimensions. The relative errors of \cite[Thm.~3.3]{rkz2015}, shown in Figure~\ref{fig:appl:wavelets:scaling:RelE2}, are substantially larger, in particular also for the orthonormal basis $\aopd$.

We emphasize that the setup of Figure~\ref{fig:appl:wavelets:scaling} considers an \emph{asymptotic} model:
While the ``resolution'' of $\aop$ gets increasingly finer as $n \to \infty$, the signal $\grtr \in \R^n$ is always generated by discretizing a piecewise constant function consisting of exactly $11$ discontinuities.
%\sout{Hence, one would expect that the complexity of $\grtr$ is only slightly affected, and indeed, $\ed{\aopi,\grtr}$ and $\samplecompl{\aopd,\grtr}$ grow both much more slowly than $n$.}
In a certain sense, this experiment pushes our theoretical approach to its limits because it is rather designed for a non-asymptotic setting, where $n$ remains fixed.
Let us briefly summarize our findings:
% Although a more thorough investigation of the asymptotic behavior might reveal further details, this experiment Nevertheless allows us to draw the following preliminary conclusion:
\begin{highlight}
The proposed prediction of Theorem~\ref{thm:results:recovery} as well as the true sample complexity both scale linearly with the ambient dimension. 
%\sout{resembling the asymptotics of traditional bounds from compressed sensing} 
Hence, the relative deviation from the true sample complexity remains constant for larger values of $n$.
\end{highlight}

\subsubsection{Stable Recovery of Compressible Signals}
\label{subsec:appl:compr}

As discussed in the course of Subsection~\ref{subsec:results:compressible}, the significance of Theorem~\ref{thm:results:recovery} relies on the assumption that sufficiently many analysis coefficients are zero. However, perfect sparsity is rarely satisfied in practice, since real-world signals are usually only compressible (cf. Figure~\ref{fig:results:compresible}). Although one cannot expect exact recovery in these cases, it still makes sense to ask for an approximate reconstruction of $\grtr \in \R^n$ if $m$ is reasonably large, corresponding to a smooth error decay. In this part, we intend to verify this heuristic numerically and demonstrate how the stability result of Theorem~\ref{thm:results:robuststable} can be applied to such situations.

For this purpose, we continue to use $\aopi \in \R^{N \times n}$ as analysis operator, whereas the \texttt{blocks} signal clearly needs to be modified in order to achieve a compressible coefficient vector with respect to $\aopi$.
Our adapted choice of the source signal, denoted by $\compr \in \R^n$, is displayed in Figure~\ref{fig:appl:wavelets:compr:evolution} for $n=256$. More precisely, it is constructed by replacing one of the piecewise constant segments of \texttt{blocks} by a smooth curve, such that the number of non-zero coefficients increases from $906$ to $1140$, i.e., $\lnorm{\aopi\compr}[0] = 1140$. Figure~\ref{fig:appl:wavelets:compr:decay} shows that this modification particularly results in a larger region of slowly decaying analysis coefficients (between $800$ and $1140$). For a more appropriate comparison, we do not consider the original \texttt{blocks} signal as ``sparse'' counterpart, but define a simplified version $\sparse \in \R^n$ which does not have any discontinuities within the curvy segment of $\compr$, see Figure~\ref{fig:appl:wavelets:compr:evolution} for $\sparse$ and Figure~\ref{fig:intro:1d:signal} for \texttt{blocks}.
% The choice of the source signal $\grtr$ will be essential for the purpose of this section. We have seen in the previous Sections that the combination of $\aopi$ and the \texttt{blocks} signal yields exactly sparse analysis coefficients and  is well suited for a recovery form Gaussian measurements in the analysis formulation. For the sake of consistency, we choose to stick with $\aopi$ as analysis operator and modify the \texttt{blocks} signal for obtaining compressible analysis coefficients instead. Hence, our choice of the source signal in dimension $n=256$ is displayed in Figure~\ref{fig:appl:wavelets:compr:evolution} and shall be denoted by $\compr$. By replacing one of the piecewise constant sections by a smoother segment the number of non-zero coefficients increases from 906 to 1140 and it can be seen in Plot~\ref{fig:appl:wavelets:compr:decay} that this modification also results in a passage of rapidly decaying analysis coefficients. For a fair comparison we do not compare with the original \texttt{blocks} signal, but a further simplified version $\sparse$, for which the discontinuities in the previously smoothed section are deleted. 

\begin{figure}[!t]
	\centering
	\begin{subfigure}[t]{0.45\textwidth}
		\centering
		\includegraphics[width=\textwidth]{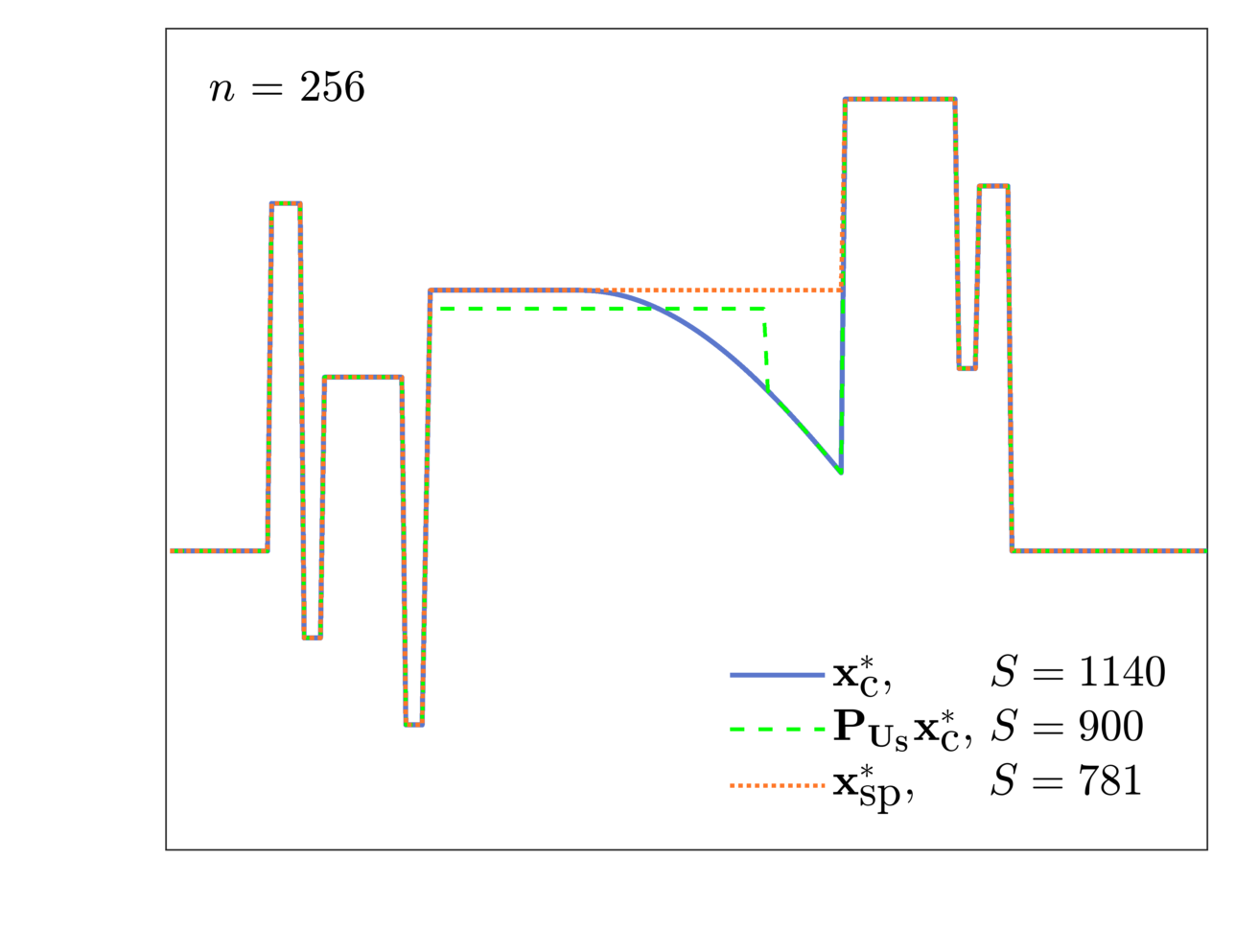}
		\caption{}
		\label{fig:appl:wavelets:compr:evolution}
	\end{subfigure}%
	\qquad
	\begin{subfigure}[t]{0.45\textwidth}
		\centering
		\includegraphics[width=\textwidth]{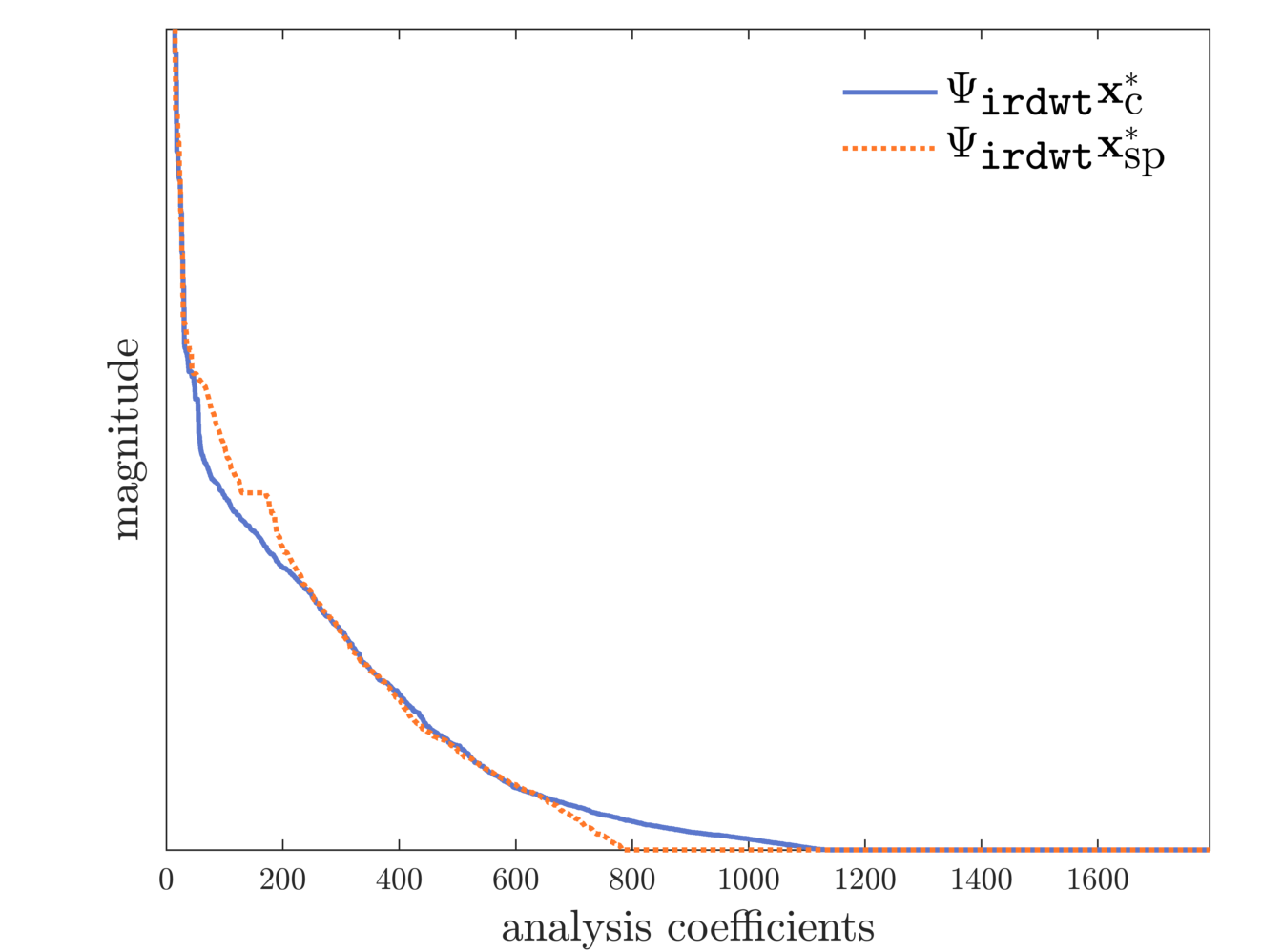}
		\caption{}
		\label{fig:appl:wavelets:compr:decay}
	\end{subfigure}%
	\qquad
	\begin{subfigure}[t]{0.45\textwidth}
		\centering
		\includegraphics[width=\textwidth]{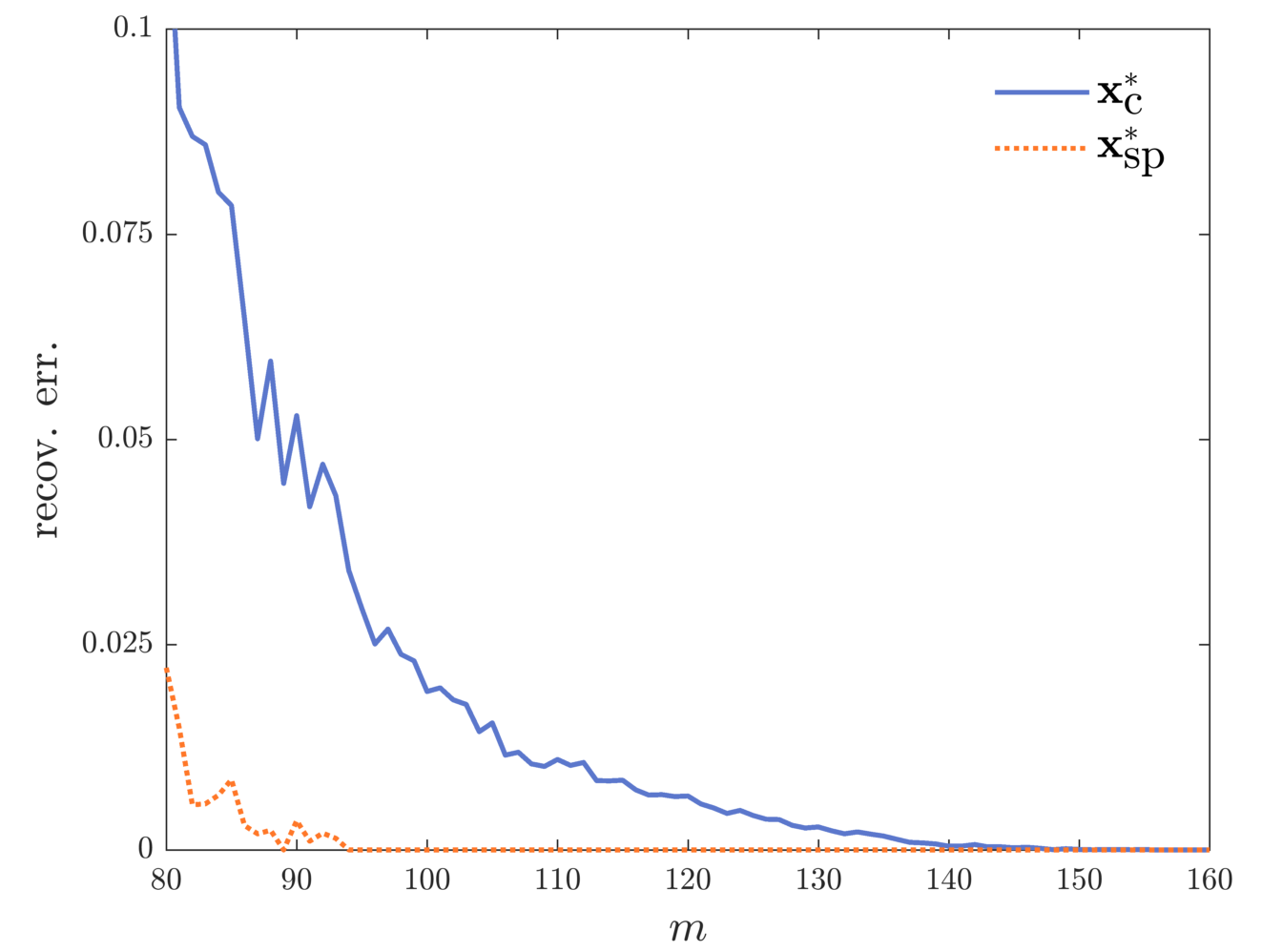}
		\caption{}
		\label{fig:appl:wavelets:compr:err1}
	\end{subfigure}%	
	\qquad
	\begin{subfigure}[t]{0.45\textwidth}
		\centering
		\includegraphics[width=\textwidth]{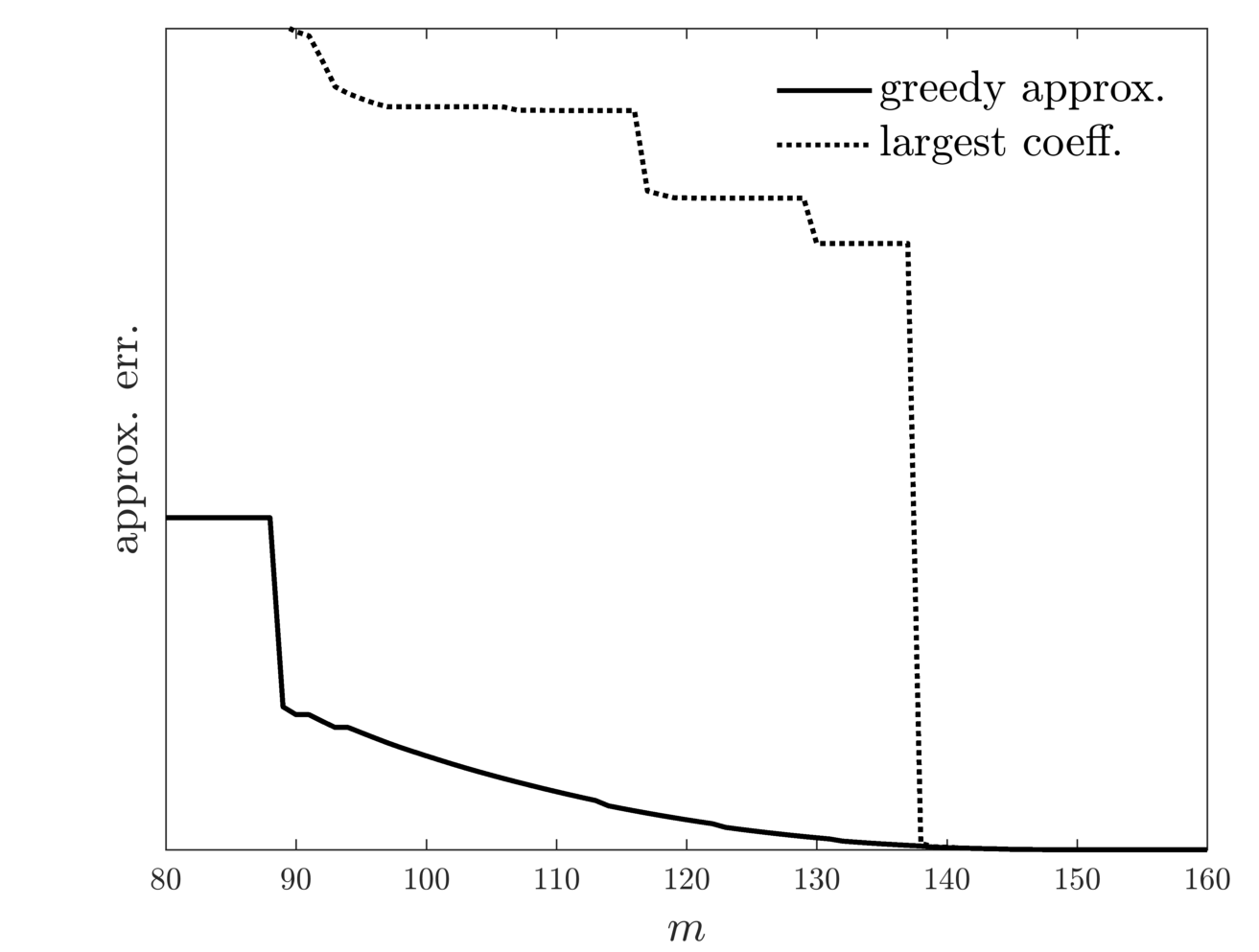}
		\caption{}
		\label{fig:appl:wavelets:compr:err2}
	\end{subfigure}%
	\qquad
	\begin{subfigure}[t]{0.45\textwidth}
		\centering
		\includegraphics[width=\textwidth]{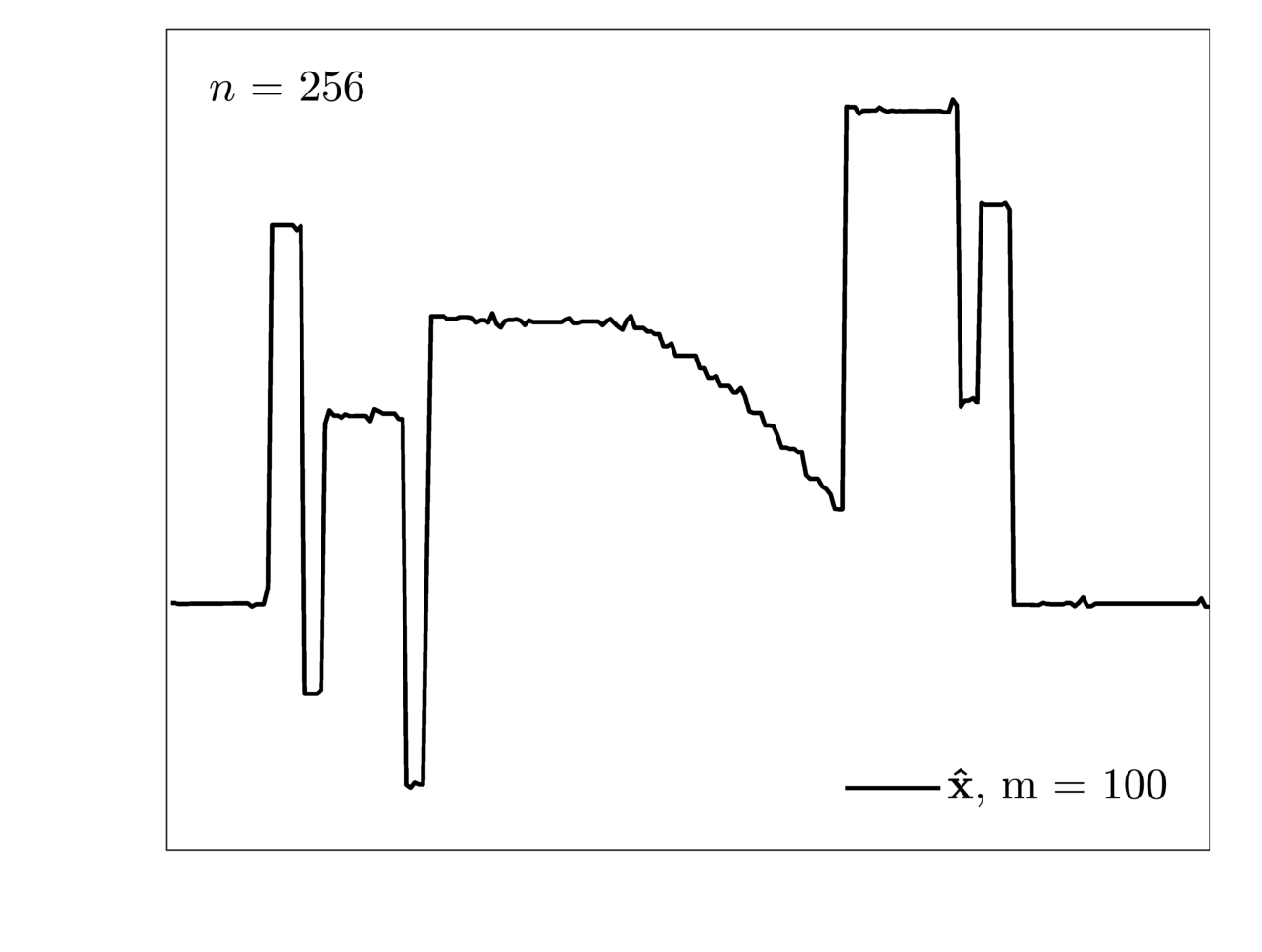}
		\caption{}
		\label{fig:appl:wavelets:compr:rec110}
	\end{subfigure}%
	\qquad
	\begin{subfigure}[t]{0.45\textwidth}
		\centering
		\includegraphics[width=\textwidth]{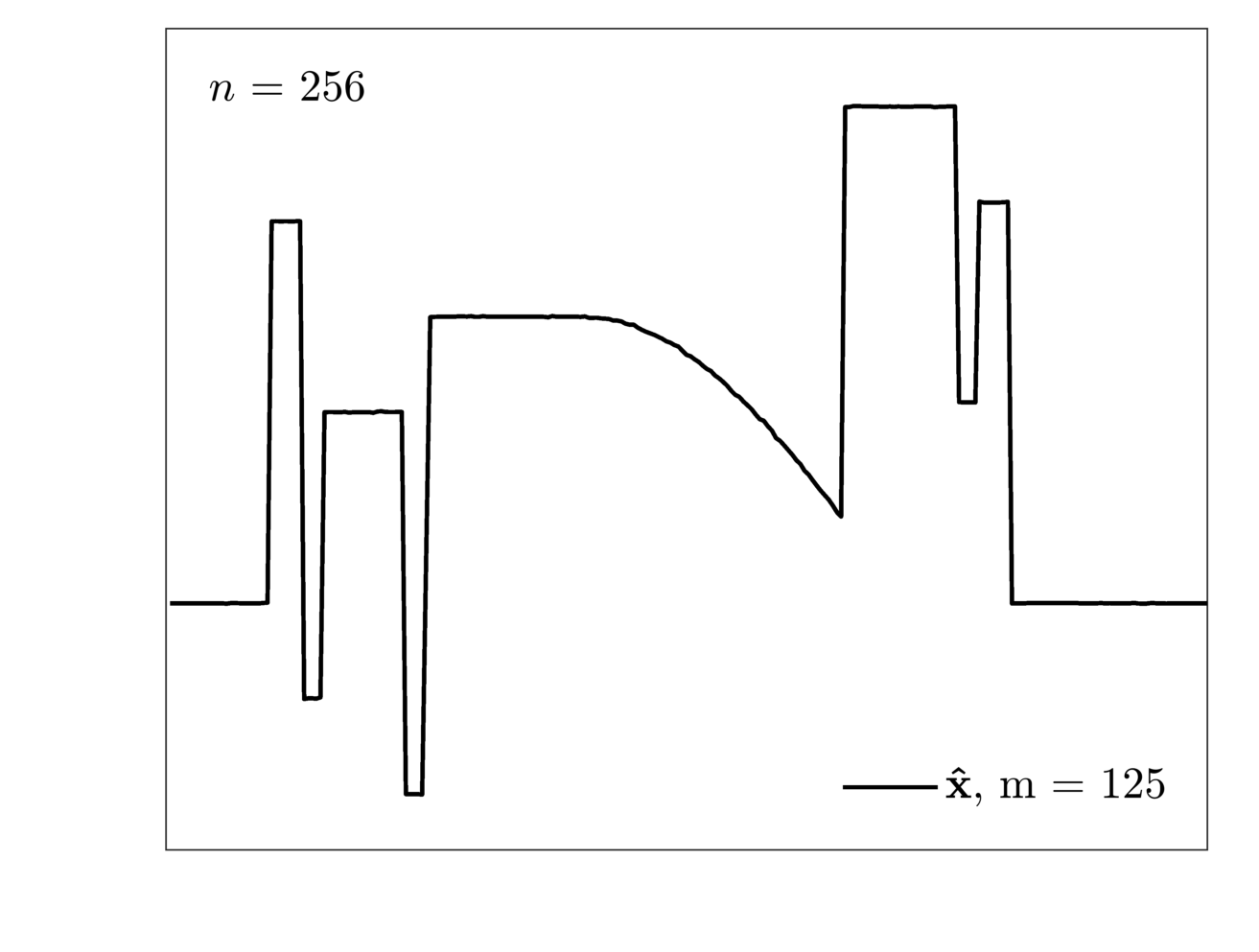}
		\caption{}
		\label{fig:appl:wavelets:compr:rec125}
	\end{subfigure}%
	
	\caption{\subref{fig:appl:wavelets:compr:evolution}~Compressible signal $\compr$ and its piecewise constant counterpart $\sparse$. The projection $\proj{\projnsp[S]}\compr$ yields a low-complexity approximation of $\compr$. \subref{fig:appl:wavelets:compr:decay}~Sorted $\l{1}$-normalized magnitudes of the analysis coefficients $\aopi\compr$ and $\aopi\sparse$. \subref{fig:appl:wavelets:compr:err1}~Averaged recovery error via \protect{\refabpnoiseless{\noiseparam=0}{\aopi}}. \subref{fig:appl:wavelets:compr:err2}~Relative recovery error predicted by Theorem~\ref{thm:results:robuststable} (see Experiment~\ref{exp:appl:wavelets:compr}). \subref{fig:appl:wavelets:compr:rec110} + \subref{fig:appl:wavelets:compr:rec125}~Minimizers $\solu$ of \refabpnoiseless{\noiseparam=0}{\aopi} with input $\y = \A \compr \in \R^m$.}
	\label{fig:appl:wavelets:compr}
\end{figure}
%\enlargethispage{-\baselineskip}

The performance of the analysis basis pursuit \refabpnoiseless{\noiseparam=0}{\aopi} is visualized in Figure~\ref{fig:appl:wavelets:compr:err1}: Here, we have applied the instructions of Experiment~\ref{exp:appl:wavelets:intro} to both $\compr$ and $\sparse$ with $M = \left\{80,81,\dots,160\right\}$ and report the averaged reconstruction errors. 
Similar to Figure~\ref{fig:intro:1d:pt_frame}, the estimation error of $\sparse$ drops to zero at $m \approx 95$. %, that is, exact recovery. 
% Note that $\sparse$ has fewer jumps than the original \texttt{blocks} signal. 
In contrast, $\compr$ is not exactly recovered before $m \approx 150$, but starting from $m \approx 100$, the error curve however smoothly tends to zero. This underpins our intuition that one cannot expect a perfect, but still an accurate, outcome for compressible signals. For illustration, we have also plotted two exemplary reconstructions of $\compr$ via \refabpnoiseless{\noiseparam=0}{\aopi} in Figure~\ref{fig:appl:wavelets:compr:rec110} and Figure~\ref{fig:appl:wavelets:compr:rec125} for $m=100$ and $m=125$, respectively. 
	
To invoke our theoretical framework from Subsection~\ref{subsec:results:compressible}, we first need to come up with a meaningful approximation strategy for $\compr$. More precisely, our task is to identify a subspace $\projnsp[S] \subset \R^n$, such that the rescaled projection (cf. \eqref{eq:results:robuststable:grtrsparse})
\begin{equation}
	\comprsparse \coloneqq \frac{\lnorm{\aopi\compr}[1]}{\lnorm{\aopi\proj{\projnsp[S]}\compr}[1]} \cdot \proj{\projnsp[S]} \compr \in \R^n
\end{equation}
is both of lower complexity and the resulting approximation error $\lnorm{\compr - \comprsparse}$ is small. If $\projnsp[S]$ is appropriately chosen, Theorem~\ref{thm:results:robuststable} certifies accurate recovery via \refabpnoiseless{\noiseparam=0}{\aopi}, while the involved sampling-rate function $\samplecompl{\aopi, \comprsparse}$ is of moderate size.
As already pointed out in the course of \eqref{eq:results:compressible:optimalSterm:subspace}, finding a best $S$-analysis-sparse approximation of $\compr$ with respect to $\aopi$ is not straightforward and may lead to an NP-hard problem. 
Hence, we rather propose a greedy method in Appendix~\ref{subsec:implementation:greedy}: Roughly speaking, Algorithm~\ref{exp:impl:algo} provides a subspace $\projnsp[S] \subset \R^n$ for a given $S \in [N]$ that tries to meet the above two criteria, even though this choice might be suboptimal. 
Figure~\ref{fig:appl:wavelets:compr:evolution} shows the ``evolution'' of our approach by plotting an intermediate approximation $\proj{\projnsp[S]} \compr$ with $S = \lnorm{\aopi\proj{\projnsp[S]} \compr}[0] = 900$.
The algorithm indeed first targets the area of highest curvature, approximating $\compr$ by piecewise constant segments, so that the analysis coefficients of $\comprsparse$ become sparser. However, we suspect that an optimal projection would rather replace the smooth segment in a more ``zig-zag'' like fashion.

To verify the statement of Theorem~\ref{thm:results:robuststable}, we apply the steps of Experiment~\ref{exp:appl:wavelets:compr} below.
Note that this template permits any subspace $\projnsp[S]$ with $\lnorm{\aopi\proj{\projnsp[S]} \compr}[0] \leq S$.

\begin{experiment}[Approximation of a compressible signal]\leavevmode
\label{exp:appl:wavelets:compr}
\vspace{-.25\baselineskip}\hrule\vspace{.5\baselineskip}

\myhangindent{Input: \ }
\expkwd{Input:} Signal vector $\compr \in \R^{256}$, analysis operator $\aopi \in \R^{1792 \times 256}$, approximation subspace $\projnsp[S] \subset \R^n$ for every $S \in [N]$.

\vspace{.5\baselineskip}
\expkwd{Compute:} Repeat the following procedure $50$ times for every $m \in \{ 80,81,\dots,160\}$:
\begin{expstep}
	 \item 
		Starting with $S = 1140$, decrease $S$ until $\samplecompl{\aopi,\proj{\projnsp[S]}\compr} \leq m$ is satisfied.
	\item 
		Determine $\comprsparse \coloneqq \frac{\lnorm{\aopi\compr}[1]}{\lnorm{\aopi\proj{\projnsp[S]}\compr}[1]} \cdot \proj{\projnsp[S]} \compr$ according to \eqref{eq:results:robuststable:grtrsparse}.
	\item
		Compute and store the approximation error $\lnorm{\compr - \comprsparse}[2]$. 
\end{expstep}
\end{experiment}

% Having such a formalism at hand that provides us for every $S$ with a subspace $\projnsp[S]$ such that $\# \supp (\proj{\projnsp[S]} \compr) = S$, we will now visualize the relationship of $m$ and the recovery error predicted by Theorem~\ref{thm:results:robuststable}. 
The results of this simulation are presented in Figure~\ref{fig:appl:wavelets:compr:err2}, plotting the error curve over the selected range for $m$.
% For a comparison, we repeat the same instructions but this time $\projnsp[S]$ is chosen in a ``naive'' manner, i.e., we pick the $S$ largest coefficients $\ssupp$ of $\aopi \compr$ and set $\projnsp[S] \coloneqq \ker\aopiSc$. 
To illustrate the benefit of our greedy method, Algorithm~\ref{exp:impl:algo}, we have repeated the instructions of Experiment~\ref{exp:appl:wavelets:compr} using a standard approximation strategy, i.e., one sets $\projnsp[S] \coloneqq \ker[(\aopi)_{\ssuppc}]$ where $\ssupp$ just corresponds to the $S$ \emph{largest coefficients} of $\aopi \compr$ in magnitude.
Note that the plots of Figure~\ref{fig:appl:wavelets:compr:err2} are only suited for a \emph{relative} comparison between both approaches. In fact, we have disregarded the impact of the tuning parameter $R$ in Theorem~\ref{thm:results:robuststable}, so that the respective curves do rather describe the qualitative behavior of the error bound \eqref{eq:results:robuststable:bound}. This particularly explains why the labels on the vertical axis are omitted in Figure~\ref{fig:appl:wavelets:compr:err2}.

For $m \geq 100$ measurements, the predicted curve based on the greedy choice of $\projnsp[S]$ strongly resembles the true recovery error in Figure~\ref{fig:appl:wavelets:compr:err1}. This stands in contrast to the naive approach (``largest coeff.''), which is apparently not able to exploit the compressibility of $\compr$.
The unfavorable ``jump'' at $m \approx 138$ is due to the ignorance of linear dependencies within $\aopi$:
For example, when enforcing that 
\begin{equation}
	\sp{\proj{\projnsp[S]}\compr}{\avec_k} = 0
\end{equation}
for a high-scale wavelet frame vector $\avec_k \in \R^n$ of $\aopi$, the same orthogonality relation does automatically hold true for many more coefficients at lower scales.
Such a clustering of coefficients eventually leads to poor $\l{2}$-approximations of $\compr$. 
Our greedy selection procedure avoids this drawback by implicitly respecting the multilevel structure of wavelets; see also \cite{dragotti2003}.
% This shows that there are fundamental differences depending on how an approximation of lower analysis complexity is chosen. 
% Furthermore, although not being optimal, the proposed greedy method captures the qualitative error decay in $m$.
Let us again draw a brief conclusion:
% \pagebreak
\begin{highlight}
	The recovery error of analysis compressible signals typically tends smoothly to zero as $m$ grows. 
	The stability result of Theorem~\ref{thm:results:robuststable}, combined with an appropriate greedy approximation scheme, adequately reflects this decay behavior. 
\end{highlight}

\subsubsection{Phase Transition for Piecewise Constant Signals}
\label{subsubsec:appl:wavelets:pt}

Up to now, we have only analyzed a specific signal (\texttt{blocks}) with respect to certain Haar wavelet frames.
Although the resolution of $\grtr \in \R^n$ was adapted to different dimensions $n$, its geometric features (the number of piecewise constant segments) remained untouched.
On the other hand, it is also of interest to assess how well our main result, Theorem~\ref{thm:results:recovery}, predicts the phase transition of \refabpnoiseless{\noiseparam=0}{\aop} if the ``complexity'' of $\grtr$ is changed. 
In the classical situation of an orthonormal analysis operator (e.g., $\aopd$), the success of recovery is completely determined by the sparsity $S = \lnorm{\aop\grtr}[0]$, implying that it is quite natural to create phase transition plots over $S$, see \cite{amelunxen2014edge}.
But such a simple indicator for complexity does often not exist for (highly) redundant systems, which is in fact one of the key findings of this work.
However, recalling that wavelets were specifically developed for the detection of singularities, one may argue that the number of jumps characterizes a signal much better when working, for example, with $\aopi$ . 
The following experiments show that this heuristic of \emph{total variation sparsity} (\emph{TV-sparsity})\footnote{The TV-sparsity (or gradient sparsity) of $\grtr \in \R^n$ is given by $S_{\text{TV}} \coloneqq \lnorm{\aoptv \grtr}[0]$, where $\aoptv \in \R^{(n-1)\times n}$ is a \emph{finite difference operator in 1D}. For a precise definition, see Subsection~\ref{subsec:appl:tv}.} indeed serves as an appropriate surrogate of analysis sparsity in the context of piecewise constant functions.

Our first simulation is generated according to the following template with $n = 128$, $\aop = \aopi$, and $\mathscr{S} = \left\{0,1,\dots,127\right\}$:
\begin{experiment}[Phase transition for piecewise constant signals]\leavevmode
\label{exp:appl:wavelets:pt}
\vspace{-.25\baselineskip}\hrule\vspace{.5\baselineskip}

\myhangindent{Input: \ }
\expkwd{Input:} Fixed ambient dimension $n$, analysis operator $\aop \in \R^{N \times n}$, range of TV-sparsity\newline$\mathscr{S} \subset \{0,1,\dots,n-1\}$.

\vspace{.5\baselineskip}
\expkwd{Compute:} Repeat the following procedure $50$ times for every $S_{\text{TV}} \in \mathscr{S}$:
\begin{expstep}
	\item 
	Select a random set $\ssupp[\text{TV}] \subset [n-1]$ with $\cardinality{\ssupp[\text{TV}]} = S_{\text{TV}}$ and determine an orthonormal basis $\basis$ of $\ker[(\aoptv)_{\ssuppc[\text{TV}]}]$. Then draw a standard Gaussian random vector $\vec{c}$ and set $\grtr \coloneqq \basis\vec{c}$. 
	\item 
	Repeat the following steps $10$ times for each $m \in \{1,2,\dots,n\}$:
	\begin{expsubstep}
		\item 
		Draw a standard i.i.d. Gaussian random matrix $\A \in \R^{m\times n}$ and determine the measurement vector $\y = \A\grtr$.
		\item 
		Solve the analysis basis pursuit \eqref{eq:intro:anabpnoiseless} to obtain an estimator $\solu \in \R^n$. 
		\item
		Declare success if $\lnorm{\grtr - \solu}[2] < 10^{-5}$.	
	\end{expsubstep}
	% 	\item 
	% 		Determine the values $\cardinality{\supp} (\aopi \grtr)$, $\samplecompl{\aopi,\grtr}$ as well as the bound presented in~\cite{rkz2015}. 
\end{expstep}
\end{experiment}
Two examples of the signal-generation step in Experiment~\ref{exp:appl:wavelets:pt} are shown in Figure~\ref{fig:appl:wavelets:signal}. The actual phase transition plot of Figure~\ref{fig:intro:pt:128} is then created by computing the empirical mean of the success rates from Experiment~\ref{exp:appl:wavelets:pt}.
The annotated curves visualize the sampling-rate bounds proposed by Theorem~\ref{thm:results:recovery} and \cite[Thm.~3.3]{rkz2015}, respectively.
For the second run in Figure~\ref{fig:intro:pt:512}, we have invoked Experiment~\ref{exp:appl:wavelets:pt} with higher dimension $n = 512$ and $\aop = \aopi$, but using a slightly coarser grid $\mathscr{S}=\{1,5,9,\dots,509\}$ to reduce the computational burden. 

\begin{figure}[!t]
	\centering
	\begin{subfigure}[t]{0.45\textwidth}
		\centering
		\includegraphics[width=\textwidth]{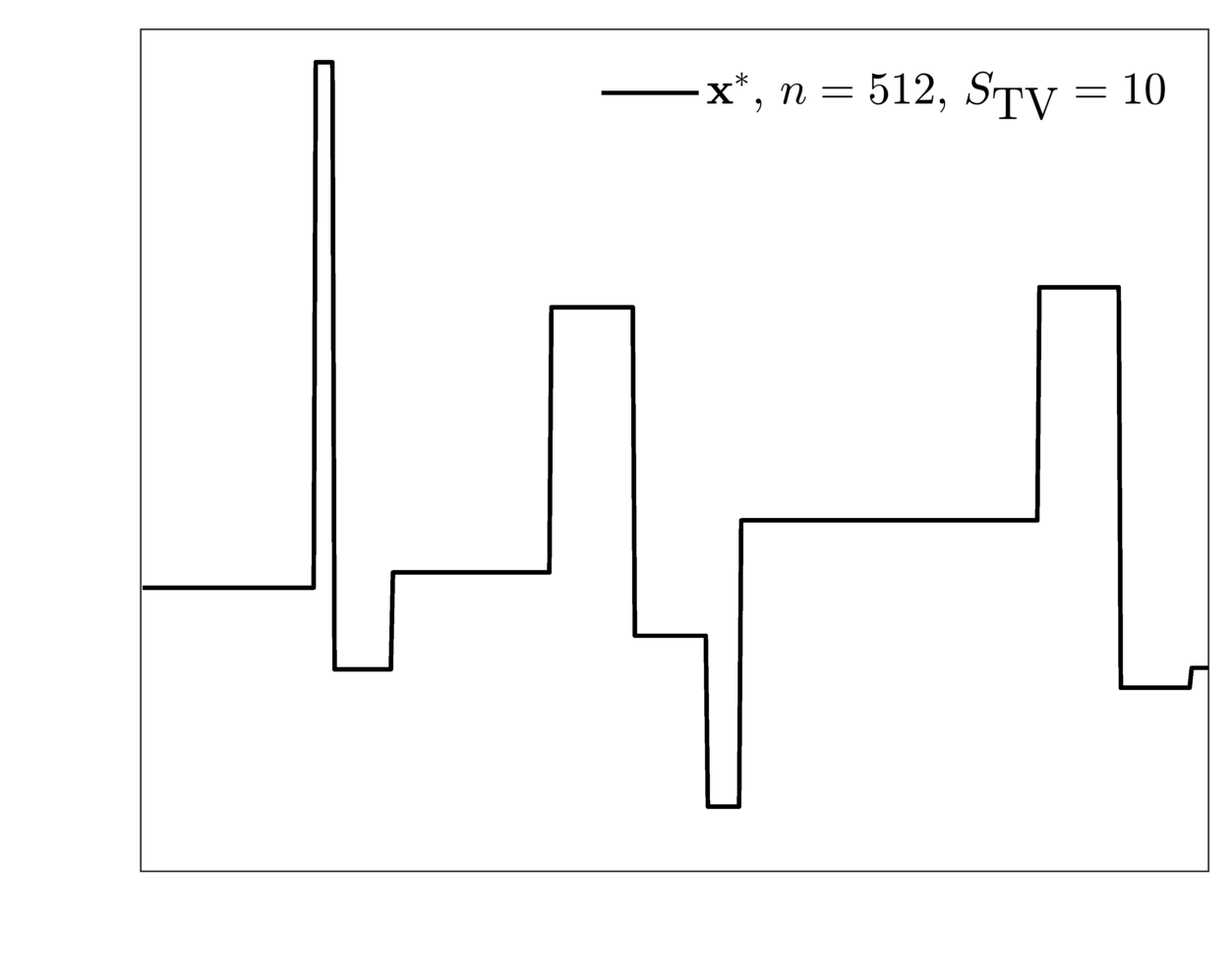}
		\caption{}
		\label{fig:appl:wavelets:pt:signal1}
	\end{subfigure}%
	\qquad
	\begin{subfigure}[t]{0.45\textwidth}
		\centering
		\includegraphics[width=\textwidth]{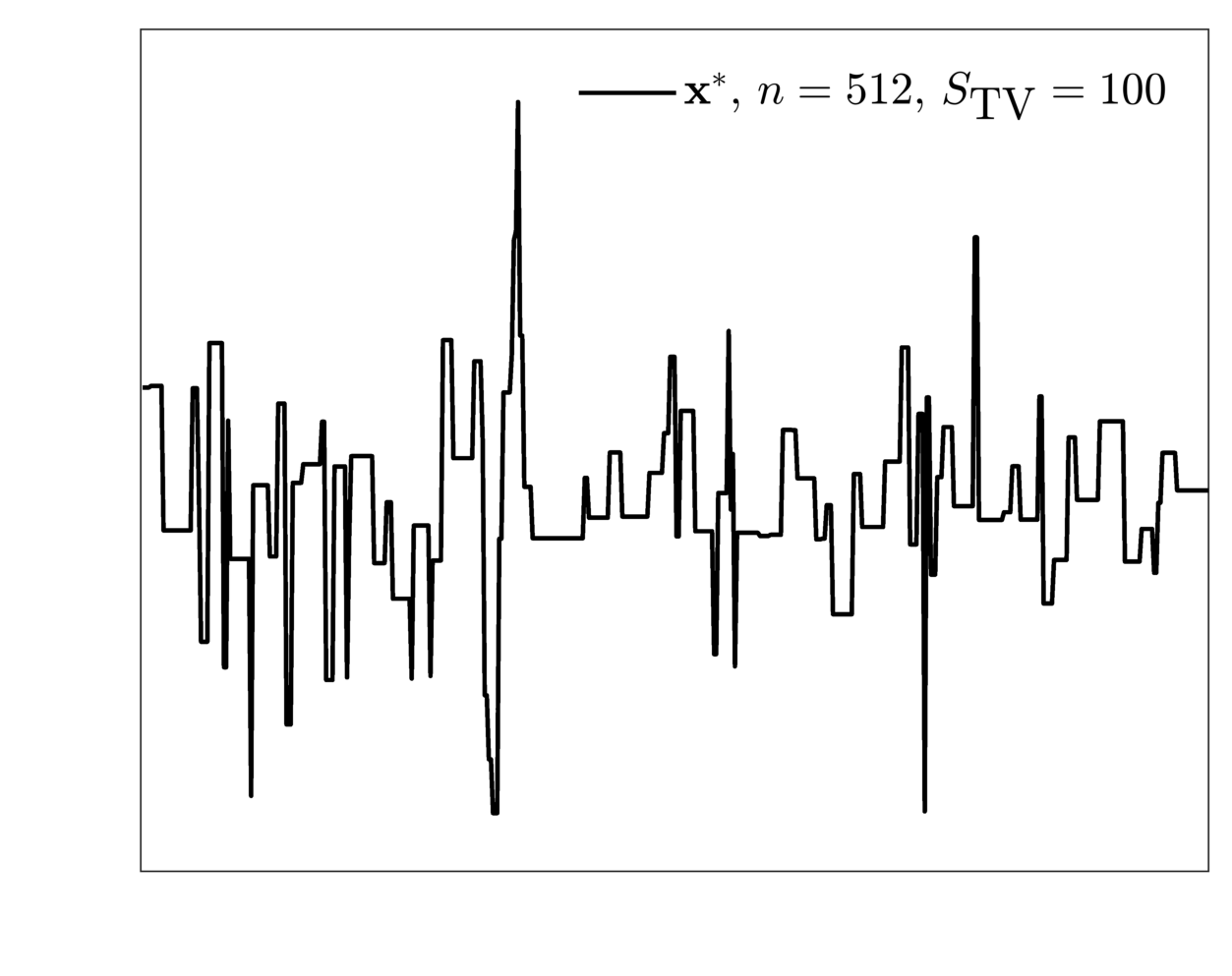}
		\caption{}
		\label{fig:appl:wavelets:pt:signal2}
	\end{subfigure}%	
	\caption{Two signal vectors $\grtr \in \R^n$ generated by the first step of Experiment~\ref{exp:appl:wavelets:pt} with different choices of TV-sparsity.}
	\label{fig:appl:wavelets:signal}
\end{figure}

\begin{figure}[!t]
	\centering
	\begin{subfigure}[t]{0.45\textwidth}
		\centering
		\includegraphics[width=\textwidth]{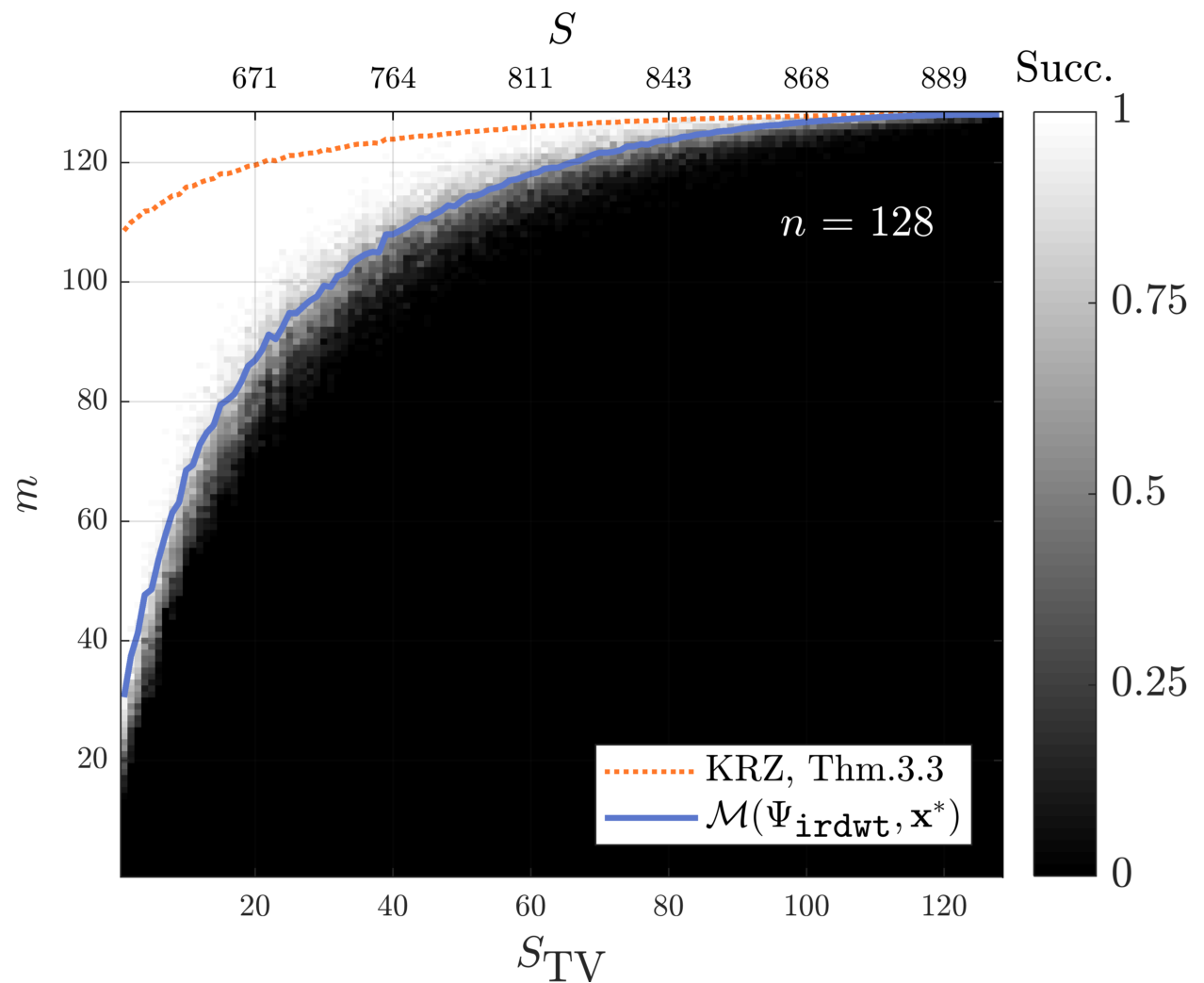}
		\caption{}
		\label{fig:intro:pt:128}
	\end{subfigure}%
	\qquad
	\begin{subfigure}[t]{0.45\textwidth}
		\centering
		\includegraphics[width=\textwidth]{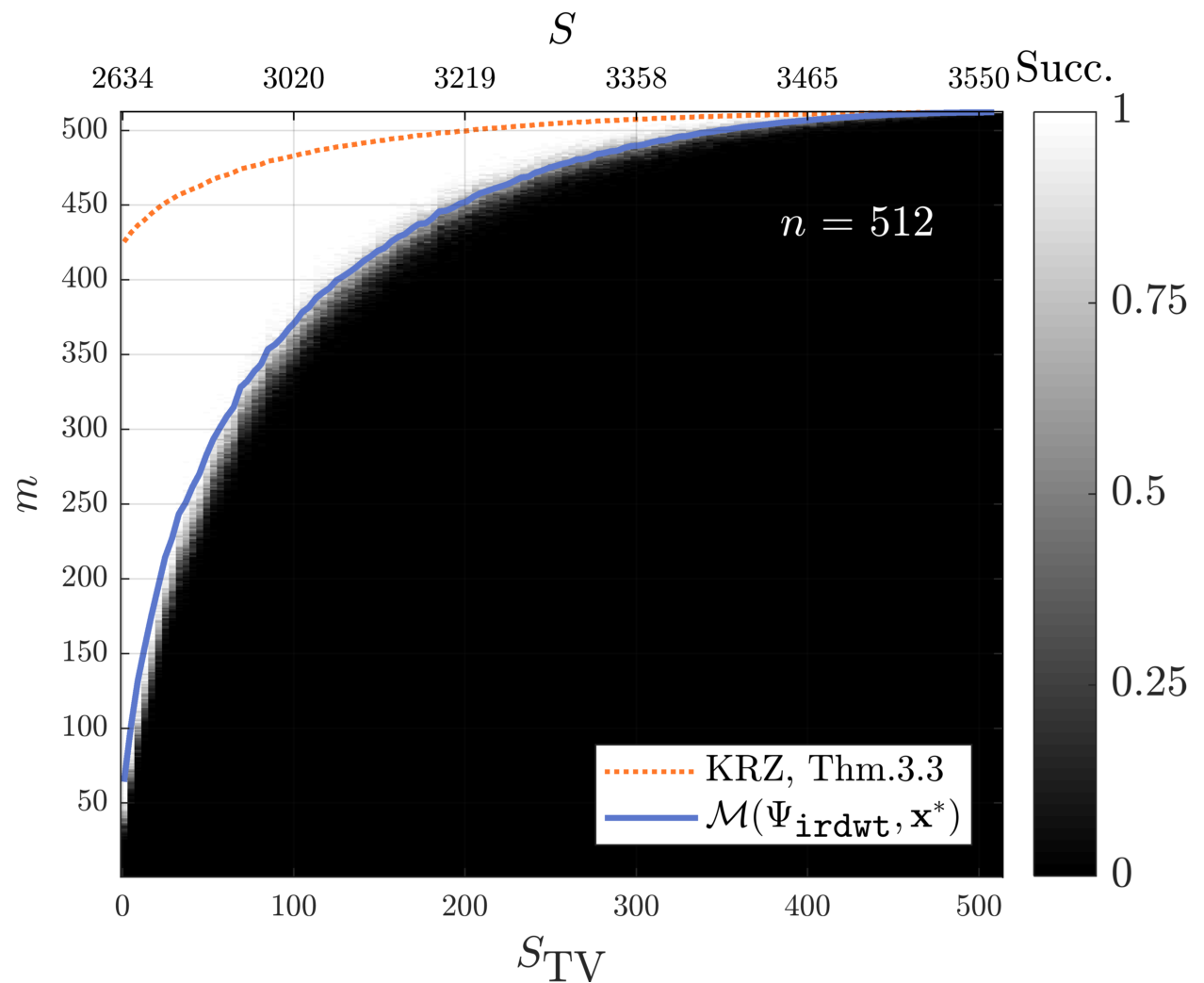}
		\caption{}
		\label{fig:intro:pt:512}
\end{subfigure}%
	
	\caption{Phase transition plots for piecewise constant signals analyzed with $\aopi$ ($S_{\text{TV}}$ is the number of discontinuities of $\grtr \in \R^n$).
	The blue curve is obtained by computing the empirical mean of the sampling-rate function $\samplecompl{\aopi,\grtr}$ within each iteration step $S_{\text{TV}} \in \mathscr{S}$ of Experiment~\ref{exp:appl:wavelets:pt}. The same procedure was performed for the orange curve corresponding to \cite[Thm.~3.3]{rkz2015}.}
	\label{fig:intro:pt}
\end{figure}

Interestingly, the results of Figure~\ref{fig:intro:pt} do strongly resemble classical phase transitions, e.g., as reported in \cite{amelunxen2014edge}.
This observation is somewhat surprising because the (averaged) coefficient sparsity $S = \lnorm{\aopi\grtr}[0]$, displayed on the top of the plots in Figure~\ref{fig:intro:pt}, appears detached in our setting.
Regarding accuracy, we can conclude that $\samplecompl{\aopi,\grtr}$ captures the location of the phase transition fairly well, whereas \cite[Thm.~3.3]{rkz2015} provides a much worse prediction. 

\begin{figure}[!t]
	\centering
	\begin{subfigure}[t]{0.45\textwidth}
		\centering
		\includegraphics[width=\textwidth]{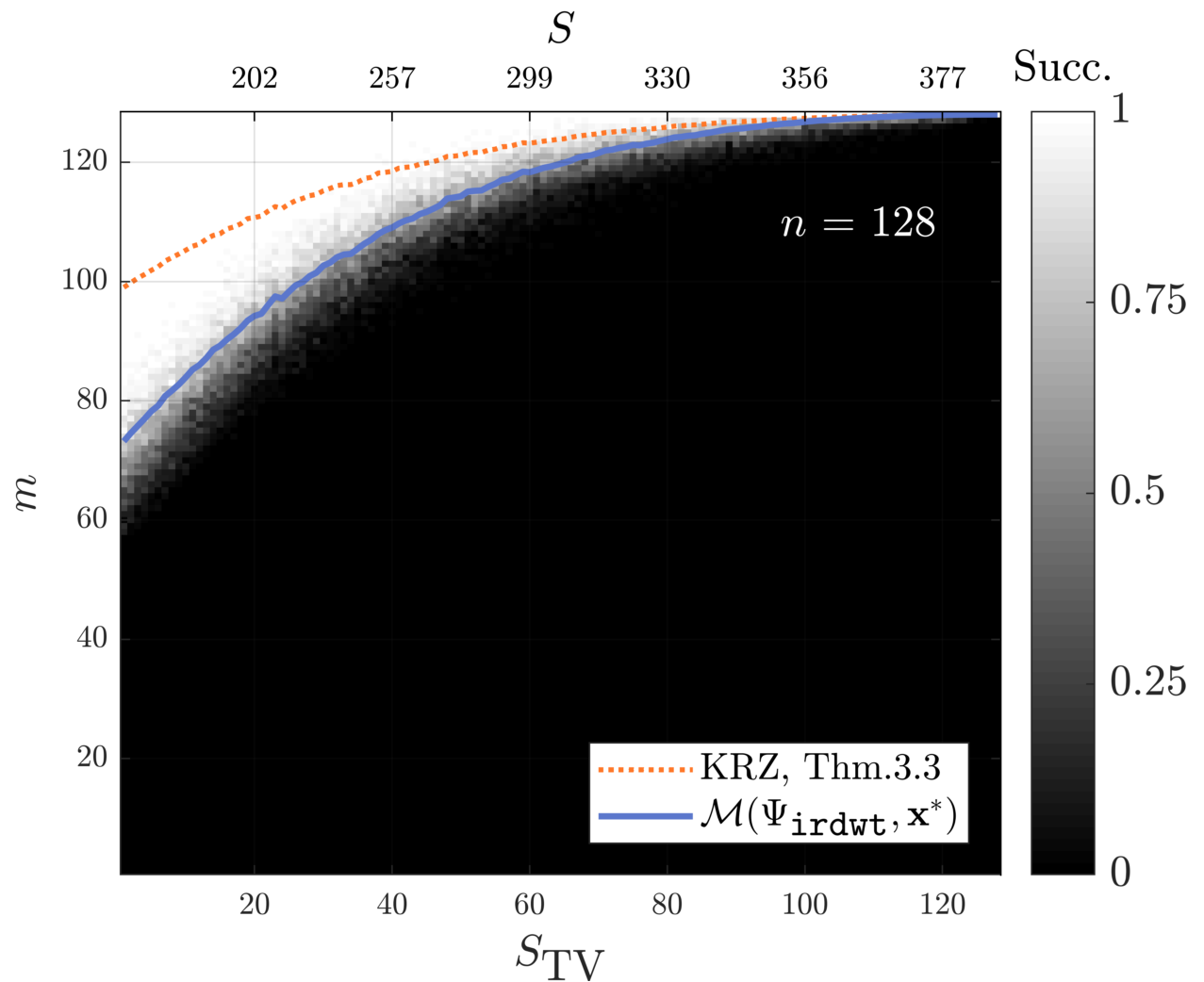}
		\caption{}
		\label{fig:intro:pt2:128}
	\end{subfigure}%
	\qquad
	\begin{subfigure}[t]{0.45\textwidth}
		\centering
		\includegraphics[width=\textwidth]{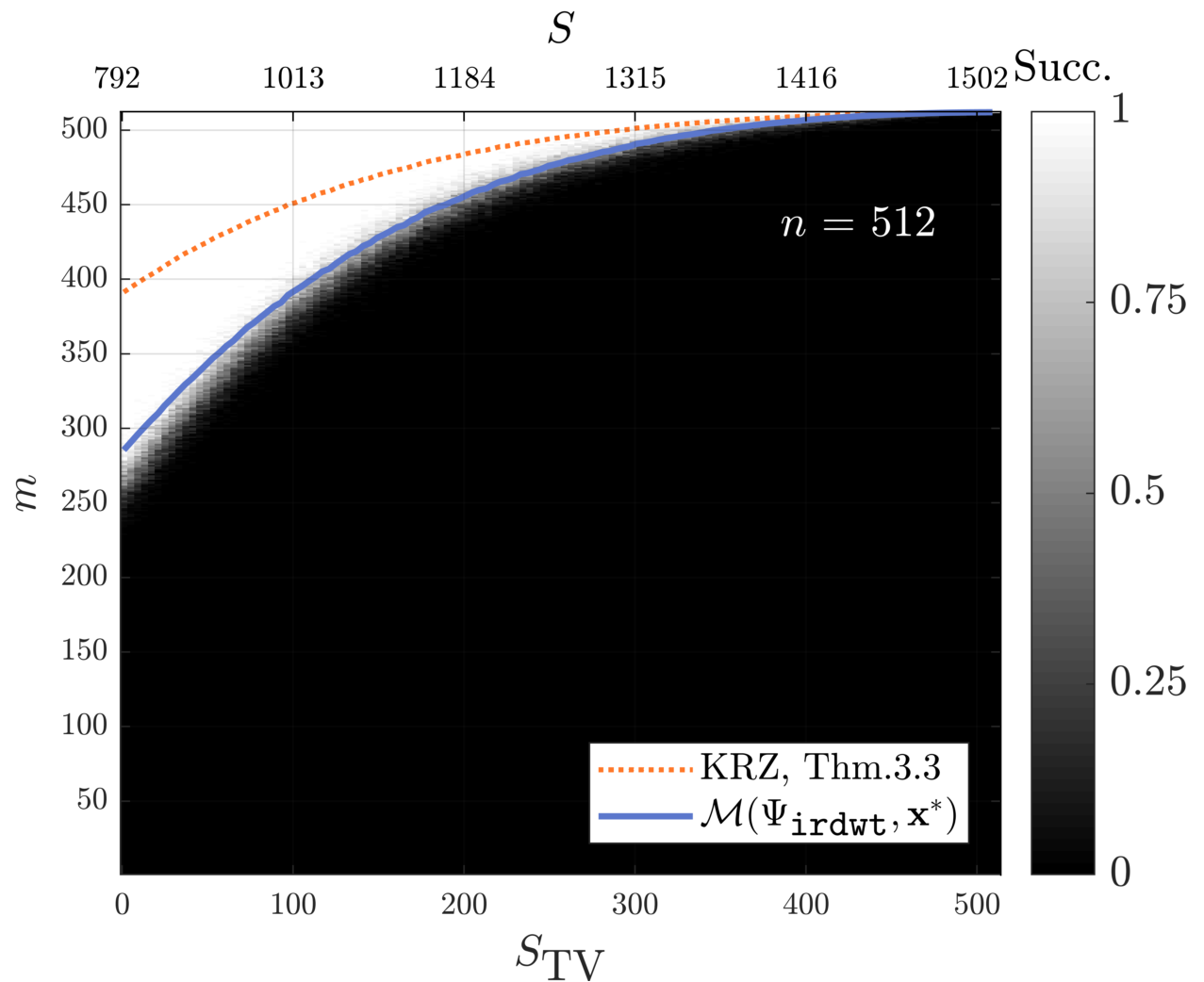}
		\caption{}
		\label{fig:intro:pt2:512}
	\end{subfigure}%
	
	\caption{Phase transition plots for piecewise constant signals similar to Figure~\ref{fig:intro:pt}, but this time $\aop = \aopi$ has only $2$ decomposition levels instead of $6$.}
	\label{fig:intro:pt2}
\end{figure}

In order to complement the previous experiment, we have repeated the same simulations, but this time $\aop = \aopi$ does only have $2$ decomposition levels instead of $6$. Note that the resulting analysis operators are now less redundant and have a different coherence structure as well as different characteristic parameters; cf.~Table \ref{tab:appl:wavelets}. The resulting phase transitions are shown in Figure~\ref{fig:intro:pt2}, and as before, we observe that the sampling-rate function $\samplecompl{\aopi,\grtr}$ captures their location almost perfectly. Moreover, as one might expect, the reconstruction capacity decreases significantly by computing fewer decomposition levels, i.e., distinguishing fewer scales and having less redundancy. Let us emphasize that the coefficient sparsity $S = \lnorm{\aopi\grtr}[0]$ does not reflect this behavior; indeed, the sparsity displayed on top of Figure~\ref{fig:intro:pt2} is considerably smaller than the corresponding values in Figure~\ref{fig:intro:pt}.

\enlargethispage{1.5\baselineskip}
In a nutshell, the key message of this subsection reads as follows:
\begin{highlight}
	The number of discontinuities governs the complexity of piecewise constant signals, enabling us to create appropriate phase transition plots for redundant Haar wavelet frames.
	The resulting transition curves are very accurately described by our sampling-rate bound in Theorem~\ref{thm:results:recovery}.
\end{highlight}

\begin{remark}
	At this point, it is also worth revisiting Subsection~\ref{subsubsec:appl:wavelets:asymp}: The simulation of Figure~\ref{fig:appl:wavelets:scaling} investigates the behavior of the sample complexity when the TV-sparsity of $\grtr \in \R^n$ remains constant while $n$ grows.
	Conceptually, this corresponds to invoking Experiment~\ref{exp:appl:wavelets:pt} for very large values of $n$ and studying the left end of the resulting phase transition plot. The relative error in Figure~\ref{fig:appl:wavelets:scaling:RelE1} therefore particularly reflects the small deviation of our prediction from the truth in Figure~\ref{fig:intro:pt} if $S_{\text{TV}}$ is small.
	In practice, however, the choice of $\aop$ is usually adapted to the signal's resolution level, implying that one is rather interested in those vertical cross sections of transition plots for which the ratio $S_{\text{TV}} / n$ is of moderate size.
\end{remark}

\subsubsection{Phase Transition for Continuous, Piecewise Linear Signals}
\label{subsubsec:appl:wavelets:pt2}

In all previous experiments of this section, we have studied a combination of piecewise constant signals and different types of Haar wavelet operators. 
The purpose of this subsection is to demonstrate that the sampling-rate bound proposed by Theorem~\ref{thm:results:recovery} is also capable of accurate predictions for more sophisticated wavelet operators and different signal classes. To this end, let us consider the same experimental setup as in Subsection~\ref{subsubsec:appl:wavelets:pt}, but now $\aop = \aopi$ is  associated with a \emph{Daubechies wavelet} with two vanishing moments and $6$ decomposition levels. Furthermore, we adapt the signal regularity to the increased number of vanishing moments by analyzing continuous, piecewise linear signals instead of piecewise constant ones. More specifically, we follow the instructions of Experiment~\ref{exp:appl:wavelets:pt} where the first step is replaced by
\begin{expstep}
\item 
	Select a random set $\ssupp[\text{kink}] \subset [n-1]$ with $\cardinality{\ssupp[\text{kink}]} = S_{\text{TV}}$ and determine an orthonormal basis $\basis$ of $\ker[(\aoptv)_{\ssuppc[\text{kink}]}]$. Then draw a standard Gaussian random vector $\vec{c}$ and set $\grtr \coloneqq \texttt{cumsum}(\basis\vec{c})$, i.e., $x_j^\ast = \sum_{l = 1}^j (\basis\vec{c})_l$ for $j = 1, \dots, n$.
\end{expstep}
In other words, a ground truth signal $\grtr$ is generated by first drawing a random piecewise constant signal and afterwards computing the cumulative sum; note that $S_{\text{TV}}$ then determines the number of ``kinks'' but not the number of jumps.
The resulting phase transitions for ambient dimensions $n=128$ and $n=512$ are shown in Figure~\ref{fig:intro:pt_daub:128} and Figure~\ref{fig:intro:pt_daub:512}, respectively.  
We observe that the number of kinks indeed governs the sample complexity of the considered signal class, enabling us to create phase transition plots for redundant Daubechies wavelets with more than one vanishing moments. More importantly, it turns out that $\samplecompl{\aopi,\grtr}$ again captures the locations of transition fairly well, whereas \cite[Thm.~3.3]{rkz2015} provides a much more pessimistic prediction.  

Let us conclude our discussion by pointing out that the recovery of more regular signals requires more vanishing moments. In fact, when just using the Haar wavelet operator from the previous subsection, perfect reconstruction would only be possible with nearly $m\approx n$ measurements. On the other hand, since the support size of the wavelet atoms grows with more vanishing moments, the latter quantity should not be chosen too large. We refer to \cite{mallat09} for a related discussion of this trade-off in a general context.

\begin{figure}[!t]
	\centering
	\begin{subfigure}[t]{0.45\textwidth}
		\centering
		\includegraphics[width=\textwidth]{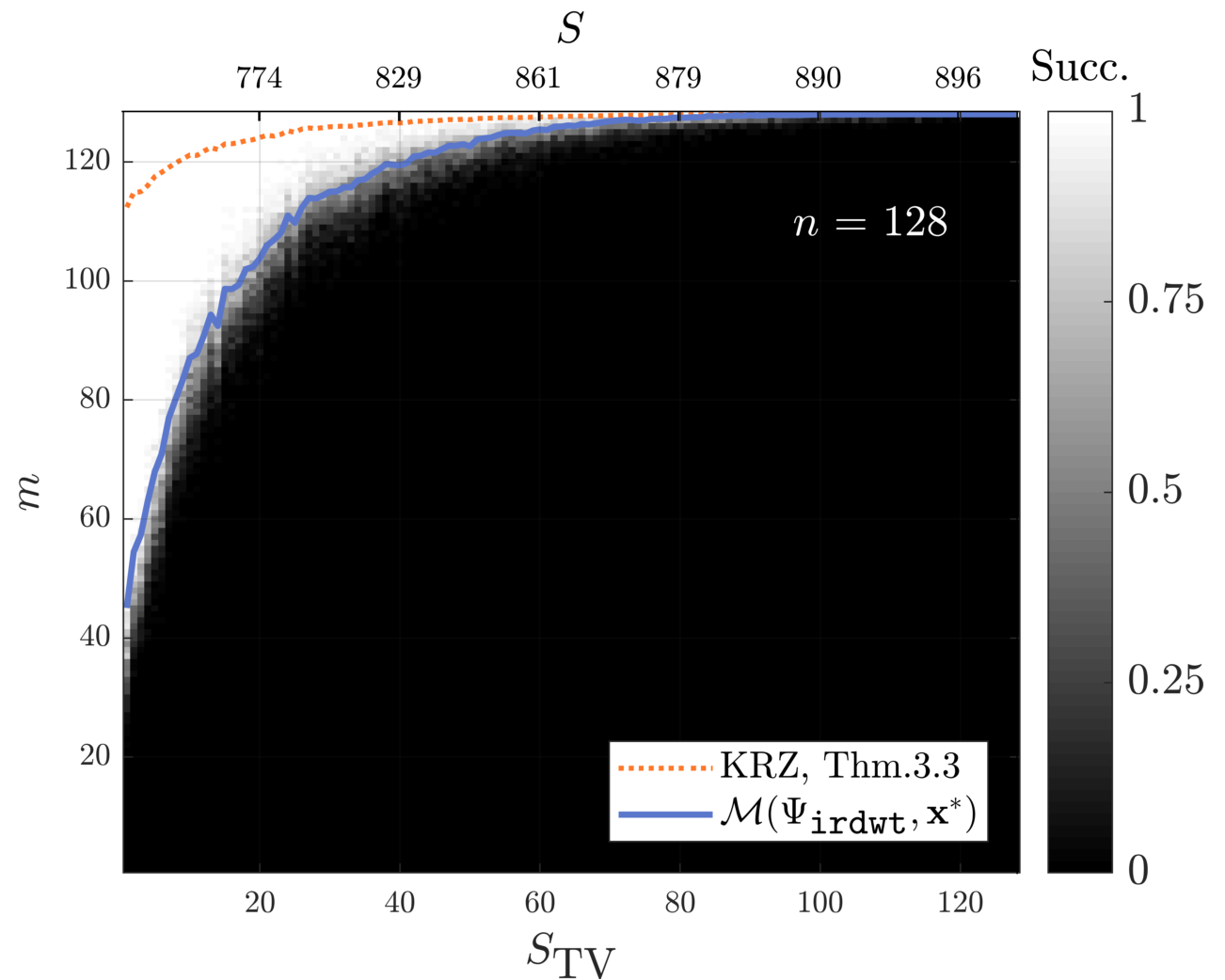}
		\caption{}
		\label{fig:intro:pt_daub:128}
	\end{subfigure}%
	\qquad
	\begin{subfigure}[t]{0.45\textwidth}
		\centering
		\includegraphics[width=\textwidth]{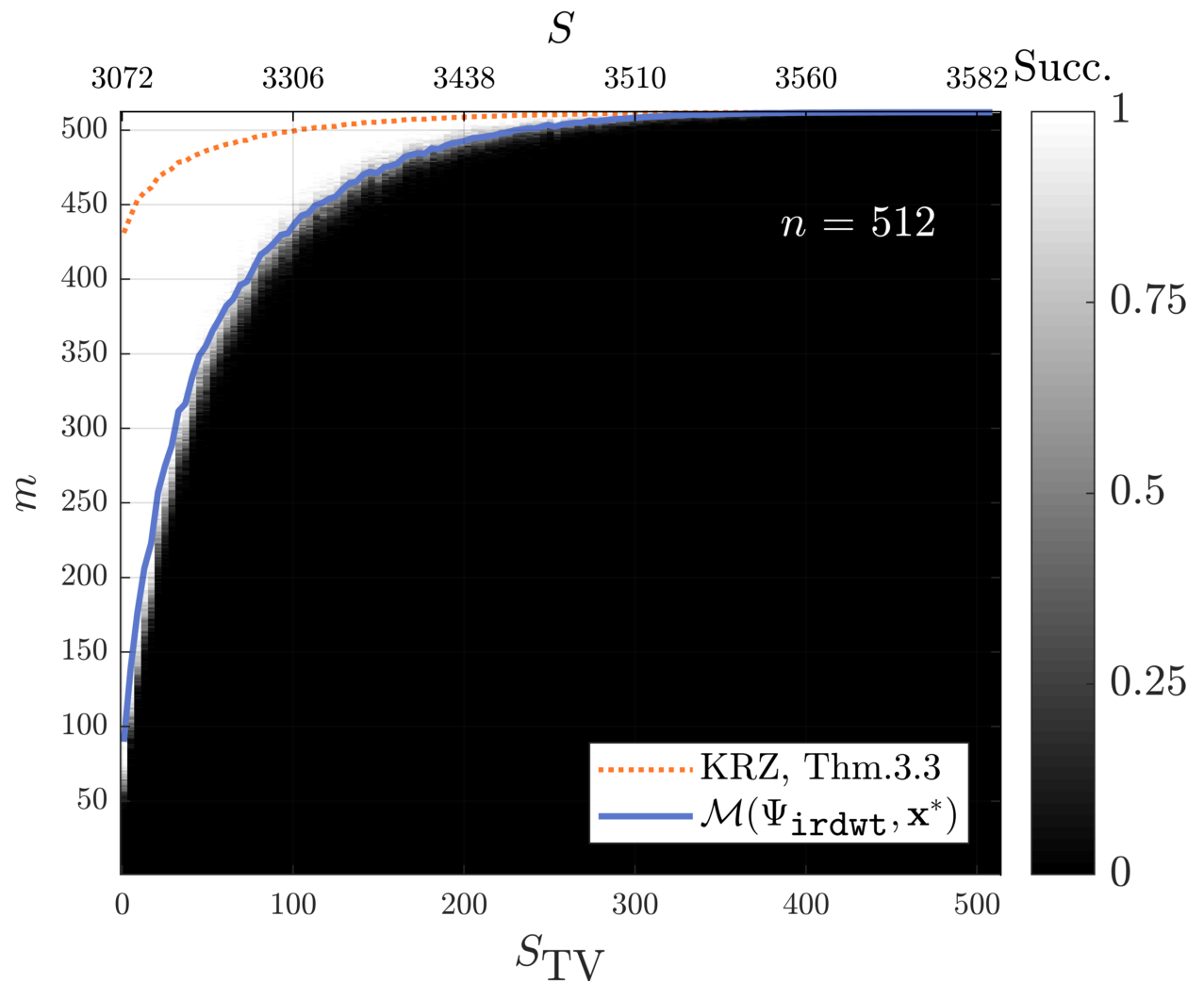}
		\caption{}
		\label{fig:intro:pt_daub:512}
	\end{subfigure}%
	
	\caption{Phase transition plots for continuous, piecewise linear signals analyzed with $\aopi$, where $\aopi$ is associated with a Daubechies wavelet with two vanishing moments ($S_{\text{TV}}$ here denotes the number of ``kinks'' of $\grtr \in \R^n$). The blue curve is obtained by computing the empirical mean of the sampling-rate function $\samplecompl{\aopi,\grtr}$ within each iteration step $S_{\text{TV}} \in \mathscr{S}$ of Experiment~\ref{exp:appl:wavelets:pt}. The same procedure was performed for the orange curve corresponding to \cite[Thm.~3.3]{rkz2015}.}
	\label{fig:intro:pt_daub}
\end{figure}

\subsection{Total Variation}
\label{subsec:appl:tv}

In this section, we consider a fundamentally different, yet classical example of an analysis operator, namely \emph{total variation in 1D (TV-1)}. It originates from the seminal work of \cite{rudin1992}, investigating the problem of signal denoising. Although conceptually quite simple, total variation has proven to be a highly effective prior in regularizing inverse problems and therefore became one of the most popular analysis operators used in practice. %We refer to \cite{krahmer2017} and the references therein for a more detailed discussion on the use of total variation in compressed sensing. 

% \pagebreak
Perhaps, the most striking structural difference to the wavelet-based approach of the previous subsection is that the \emph{finite difference operator}
\begin{align}
	 \aoptv  \coloneqq \matr{
	  -1 & 1 & 0 & \dots & 0 \\
	  0 & -1 & 1 & 0 & 0\\
	  \vdots & & \ddots &\ddots & \vdots \\
	  0 & \dots & 0 & -1 & 1} \in \R^{(n-1) \times n}
\end{align}
does not constitute a frame for $\R^n$. 
Note that there exist numerous variants of total variation imposing different boundary conditions. The above choice uses forward differences with von Neumann boundary conditions and is common in the field of compressed sensing.

To assess the quality of Theorem~\ref{thm:results:recovery}, it makes sense to follow the same strategy as in Subsection~\ref{subsubsec:appl:wavelets:pt}.
Indeed, the TV-based basis pursuit \refabpnoiseless{\noiseparam=0}{\aoptv} promotes piecewise constant output signals, which indicates that TV-sparsity is again the correct quantity to study.
Our first phase transition plot in Figure~\ref{fig:appl:tv:random} is generated according to Experiment~\ref{exp:appl:wavelets:pt}, with $n=300$, $\aop = \aoptv$ and $\mathscr{S} = \left\{1,4,\dots,298 \right\}$.
The second simulation in Figure~\ref{fig:appl:tv:dense} repeats the same procedure, but the first step of Experiment~\ref{exp:appl:wavelets:pt} is now replaced by
\begin{expstep}
\item 
	Generate $\grtr \in \R^n$ by setting $x_j^\ast = (-1)^{j-1}$ for $j=1,2,\dots,S_{\text{TV}}$ and $x_j^\ast = 0$ otherwise.  
\end{expstep}
Even though the respective TV-sparsity coincides in both signal-generation steps, the two transition curves of Figure~\ref{fig:appl:tv} look somewhat different.
This observation has a remarkable implication: While TV-sparsity is usually considered to be a heuristic measure of complexity for piecewise constant functions, it does not fully characterize the capacity of \refabpnoiseless{\noiseparam=0}{\aoptv}.
Similar to the example of wavelets in Subsection~\ref{subsec:appl:wavelets}, we therefore suspect that a sound analysis of the finite difference operator must be signal-dependent to a certain extent, depending on the specific structure of $\grtr$.
In fact, the sampling-rate function $\grtr \mapsto \samplecompl{\aoptv,\grtr}$ seems to meet this desire for signal-dependence, since the shape of the blue curves in Figure~\ref{fig:appl:tv} adapts to each of the cases. But we have to clearly confess that our predictions are only reliable in the regimes of mid- and high-level sparsity.
In contrast, the bound of \cite[Thm.~3.1]{rkz2015} is less significant and particularly remains unchanged for both signal classes.
% While our bound, as well as the one presented in~\cite{rkz2015}, is rather off in a low-sparsity regime, for mid-level and high sparsity we are able to capture the location of the phase transition quite precisely. Furthermore, the proposed bound of the effective dimension reflects the non-uniformity of TV-minimization, whereas~\cite{rkz2015} is only depending on $S$ and therefore gives uniform predictions.
Let us conclude our discussion as follows:
\begin{highlight}
	Total variation minimization exhibits a signal-dependent recovery behavior, which is also captured by the sampling-rate bound of Theorem~\ref{thm:results:recovery}.
	However, if the TV-sparsity is very small, our prediction is much less accurate than for redundant wavelet frames.
\end{highlight}

\begin{figure}[!t]
	\centering
	\begin{subfigure}[t]{0.45\textwidth}
		\centering
		\includegraphics[width=\textwidth]{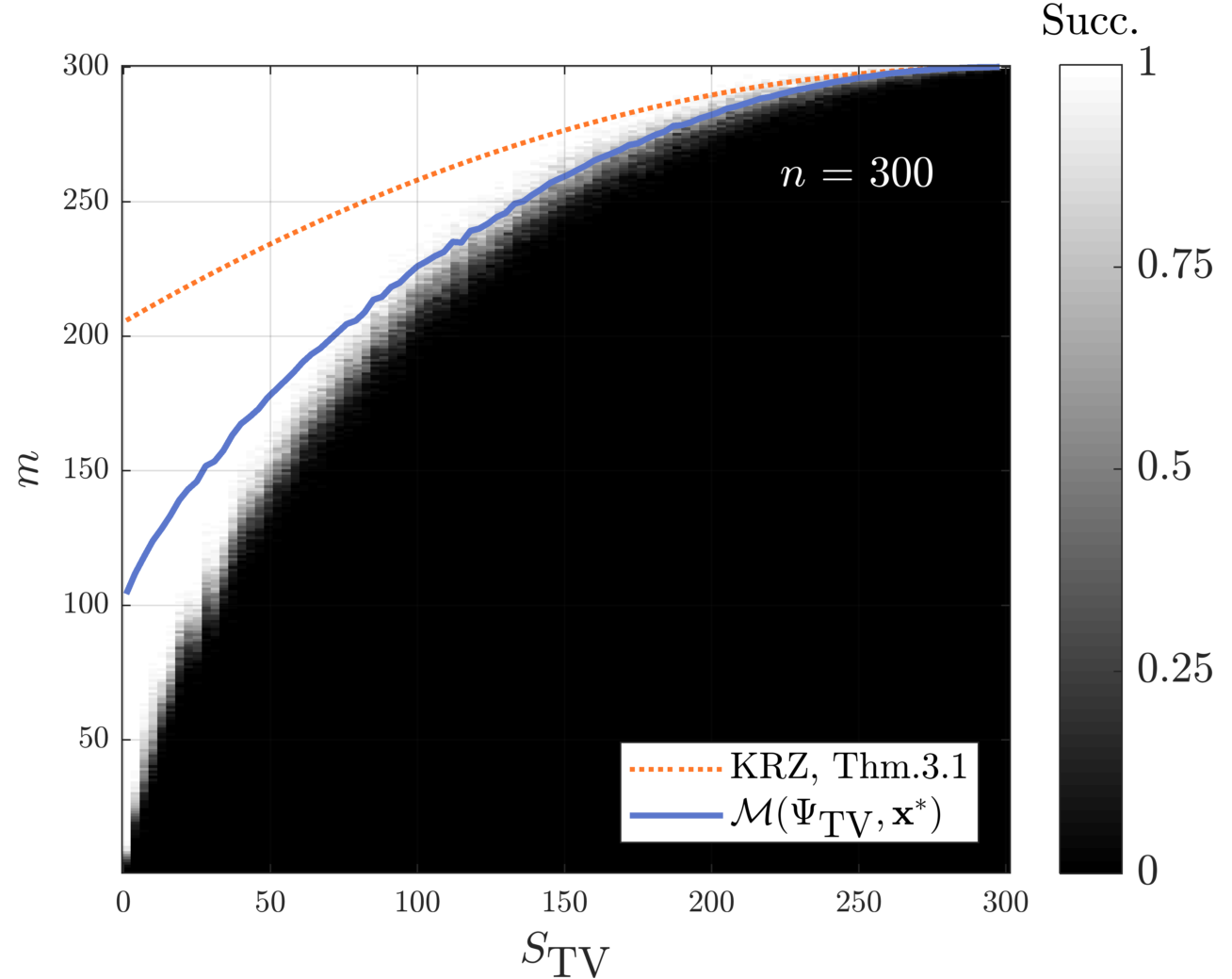}
		\caption{}
		\label{fig:appl:tv:random}
	\end{subfigure}%
	\qquad
	\begin{subfigure}[t]{0.45\textwidth}
		\centering
		\includegraphics[width=\textwidth]{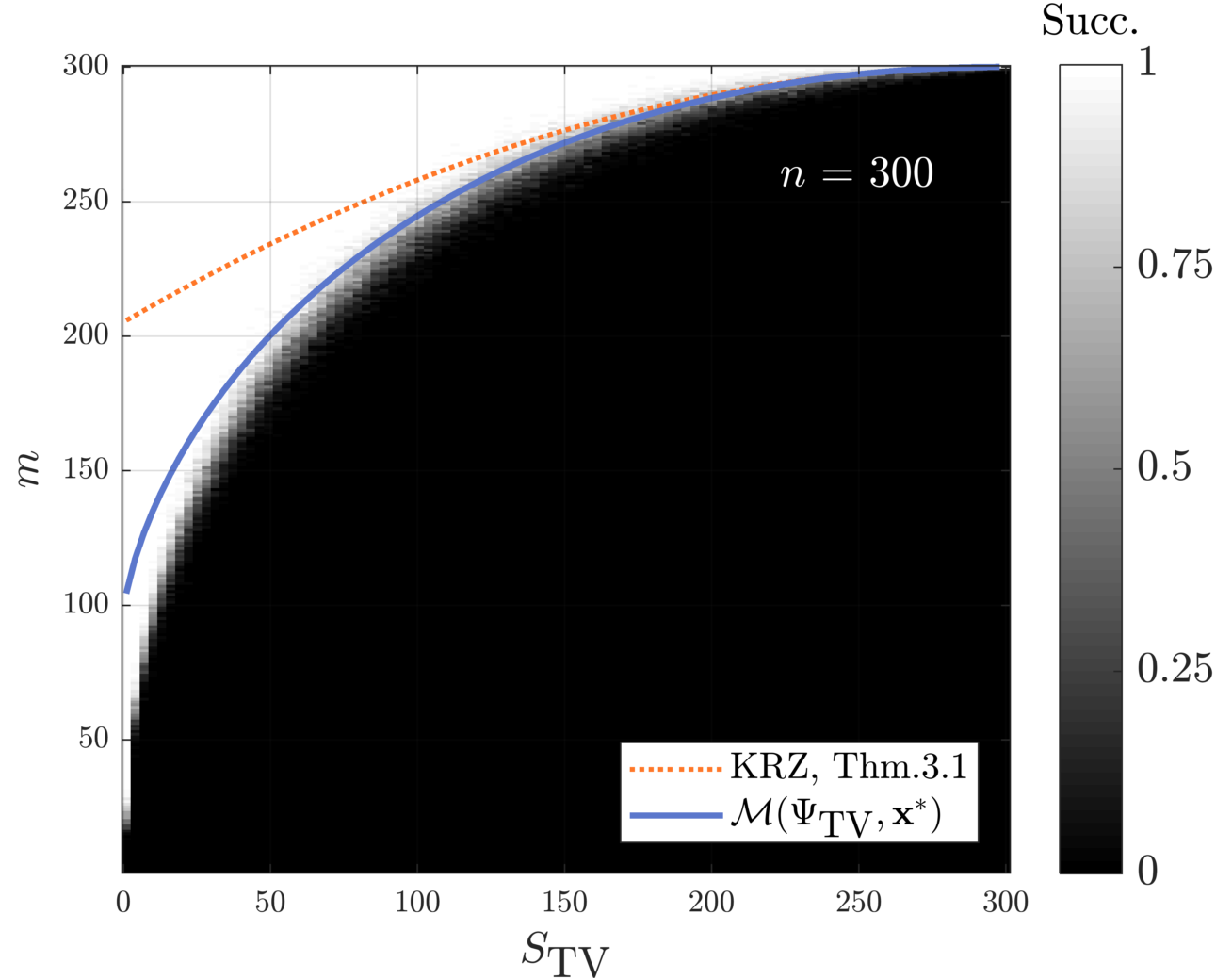}
		\caption{}
		\label{fig:appl:tv:dense}
	\end{subfigure}%	
	\caption{Phase transition plots for \refabpnoiseless{\noiseparam=0}{\aoptv} using~\subref{fig:appl:tv:random} signals built of random piecewise constant segments and~\subref{fig:appl:tv:dense} signals built of ``dense jumps.'' The blue curve is obtained by computing the empirical mean of $\samplecompl{\aoptv,\grtr}$ for every iteration step $S_{\text{TV}} \in \mathscr{S}$ of Experiment~\ref{exp:appl:wavelets:pt}, and the orange curve corresponds to \cite[Thm.~3.1]{rkz2015}.}
	\label{fig:appl:tv}
\end{figure}

\begin{remark}[Related Literature]
	Compared to redundant wavelet frames, the $\l{1}$-analysis formulation of total variation is actually much better theoretically understood.
	For example, the work of Cai and Xu \cite{cai2015} proves that the optimal sampling rate of \refabpnoiseless{\noiseparam=0}{\aoptv} is given by $m = \asympeq{\sqrt{S_{\text{TV}} \cdot n}}$ for Gaussian measurements, where $S_{\text{TV}} = \lnorm{\aoptv\grtr}[0]$ and log-terms are ignored.
	The corresponding proofs are also based on estimating the Gaussian mean width, although using very different techniques that are specifically tailored to $\aoptv$.
	The guarantees of \cite{cai2015} underpin once again the fundamental role of TV-sparsity, but as pointed out above, a refined \emph{non-asymptotic} analysis would eventually rely on additional geometric properties of the ground truth signal.
	For a more extensive discussion of total variation in compressed sensing, see \cite{krahmer2017} and the references therein.
\end{remark}

\subsection{Tight Random Frames}
\label{subsec:appl:onbrandom}

Proposition~\ref{prop:results:scfuncsimplyfied} confirms the optimality of our sampling-rate bound for orthonormal bases, and indeed, $\grtr \mapsto \samplecompl{\aop,\grtr}$ precisely describes the phase transition curve of \eqref{eq:intro:anabpnoiseless} in this case, see \cite[Prop.~4.5]{amelunxen2014edge}.
However, the situation already becomes more complicated for a closely related class of analysis operators, namely those generated from standard Gaussian matrices. These constructions almost surely yield (tight) frames in general position\footnote{A matrix $\aop \in \R^{N \times n}$ with $N \geq n$ is in general position if every subset of $n$ rows is linearly independent.} and have been widely studied in the literature as a benchmark example of the analysis formulation \cite{nam2013,kabanava2015}. 	

In this part, we roughly follow a construction of tight random frames from \cite{nam2013}: First, draw an $N\times n$ Gaussian random matrix and compute its singular value decomposition $\vec{U}\vec{\Sigma}\vec{V}^\T$. If $N \geq n$, we replace $\vec{\Sigma}$ by the matrix $[ \I{n}, \vnull]^\T \in \R^{N\times n}$, which yields a tight frame
\begin{equation}
	\aoprand \coloneqq \vec{U} ([\I{n}, \vnull]^\T) \vec{V}^\T \in \R^{N \times n}.
\end{equation}
If $N<n$, we replace $\vec{\Sigma}$ by $[ \I{N}, \vec{0}] \in \R^{N\times n}$ and set $\aoprand \coloneqq \vec{U} [\I{N}, \vnull] \vec{V}^\T \in \R^{N \times n}$.
But note that this operator does not form a frame, and in particular, $\aoprand^\T \aoprand \neq \I{n}$.  % In order to obtain also non-tight random frames with a specific ratio of the frame bounds $\alpha = A/B$, we replace S by $\left[ \diag{\vec{d}},\mathbf{0} \right]^\T$ for $\vec{d} \in \R^n$ and $\max (\vec{d})/\min (\vec{d}) = \sqrt{\alpha}$. 

Our first phase transition plot is created according to Experiment~\ref{exp:appl:onb:pt} below with $n = 300$, $N = 350$, and
\begin{equation}
	\mathscr{S} = \{ (51, N - 51), (54, N - 54), (57, N - 57),\dots, (297,N - 297) \}.
\end{equation}
The set of tuples is chosen such that for every pair $(S, L) \in \mathscr{S}$, it holds $S + L = N = 350$. Moreover, we do only consider $S > 50$, since otherwise $\ker[(\aoprand)_{\ssuppc}] = \{\vnull\}$ due to the general position property of $\aoprand$. %, and therefore $\grtr = \vnull$.
The experimental result in Figure~\ref{fig:appl:onbrandom:parse} shows that the sampling-rate bound of Theorem~\ref{thm:results:recovery} almost perfectly hits the transition curve. % in the situation of randomly generated frames.
\begin{experiment}[Phase transition for random analysis operators]\leavevmode
\label{exp:appl:onb:pt}
\vspace{-.25\baselineskip}\hrule\vspace{.5\baselineskip}

\myhangindent{Input: \ }
\expkwd{Input:} Fixed ambient dimension $n$, range of sparsity-cosparsity tuples $\mathscr{S} \subset \N_0 \times \N_0$.

\vspace{.5\baselineskip}
\expkwd{Compute:} Repeat the following procedure $5$ times for each pair $(S,L) \in \mathscr{S}$:% $S=1,4,\dots,100$:
\begin{expstep}
	 \item 
		Set $N = S + L$ and construct a random operator $\aoprand \in \R^{N \times n}$ as described above. 
		Select a random set $\ssupp \subset [N]$ with $\cardinality{\ssupp} = S$ and determine an orthonormal basis $\basis$ of $\ker[(\aoprand)_{\ssuppc}]$.
		Then draw a standard Gaussian random vector $\vec{c}$ and set $\grtr \coloneqq \basis\vec{c}$.
	\item 
		Repeat the following steps $10$ times for each $m \in \{1,2,\dots,n\}$:
	\begin{expsubstep}
	\item 
		Draw a standard i.i.d.\ Gaussian random matrix $\A \in \R^{m\times n}$ and determine the measurement vector $\y = \A\grtr$.
	\item 
		Solve the analysis basis pursuit \refabpnoiseless{\noiseparam=0}{\aoprand} to obtain an estimator $\solu \in \R^n$. 
	\item
		Declare success if $\lnorm{\grtr - \solu}[2] < 10^{-5}$.	
	\end{expsubstep}
% 	\item 
% 		Determine the value $\samplecompl{\aoprand,\grtr}$ and the bound of~\cite{rkz2015}.
\end{expstep}
\end{experiment}

\begin{figure}[!t]
	\centering
		\includegraphics[width=0.6\textwidth]{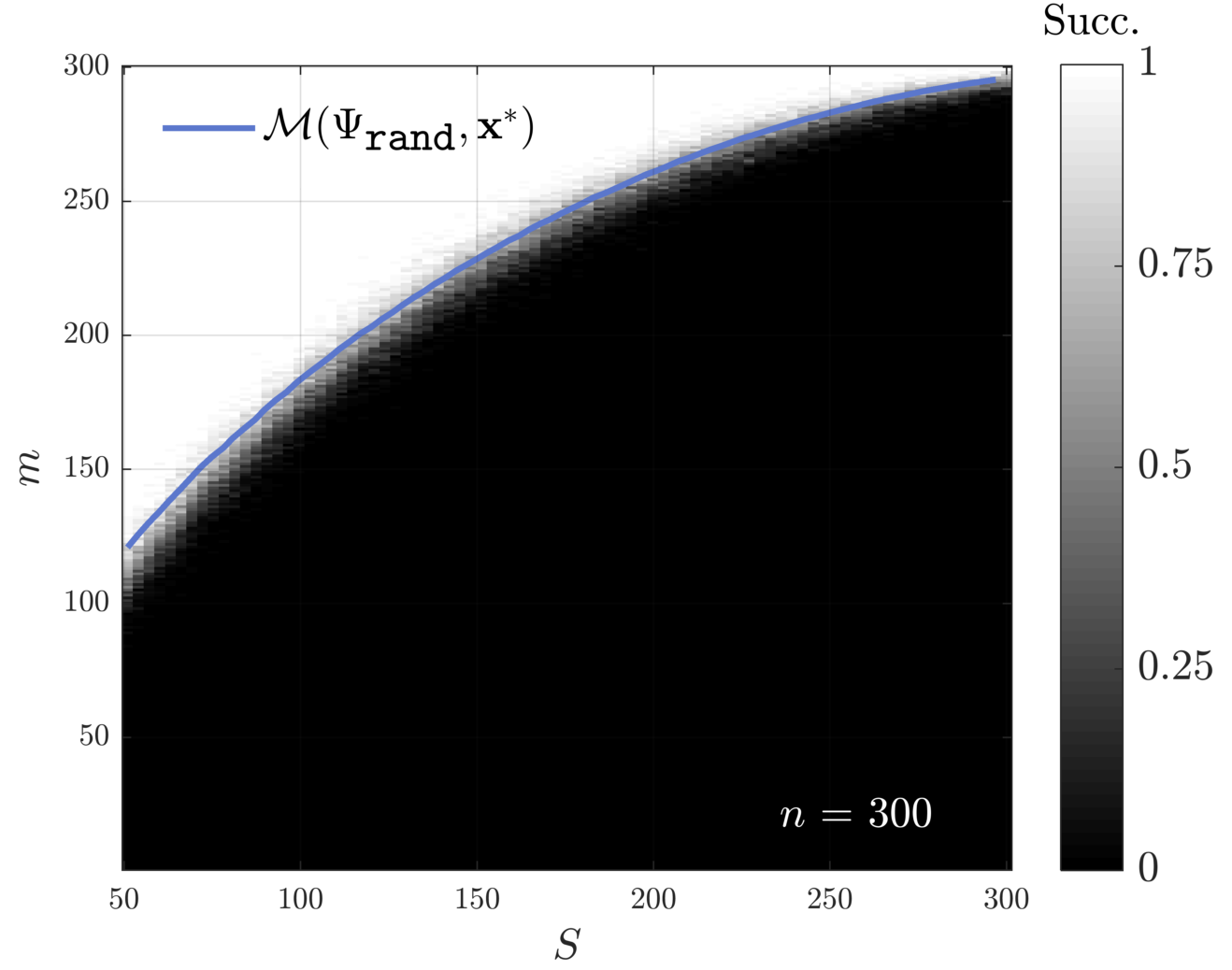}
		\caption{Phase transition of \refabpnoiseless{\noiseparam=0}{\aoprand} for tight random frames $\aoprand \in \R^{N \times n}$. The blue curve is obtained by computing the empirical mean of $\samplecompl{\aoprand,\grtr}$ for every iteration step $(S,L) \in \mathscr{S}$ of Experiment~\ref{exp:appl:onb:pt}. The labels of the horizontal axis do only display the first component of $\mathscr{S}$.}
		\label{fig:appl:onbrandom:parse}
\end{figure}

\begin{figure}[!t]
	\centering
	\begin{subfigure}[t]{0.45\textwidth}
		\centering
		\includegraphics[width=\textwidth]{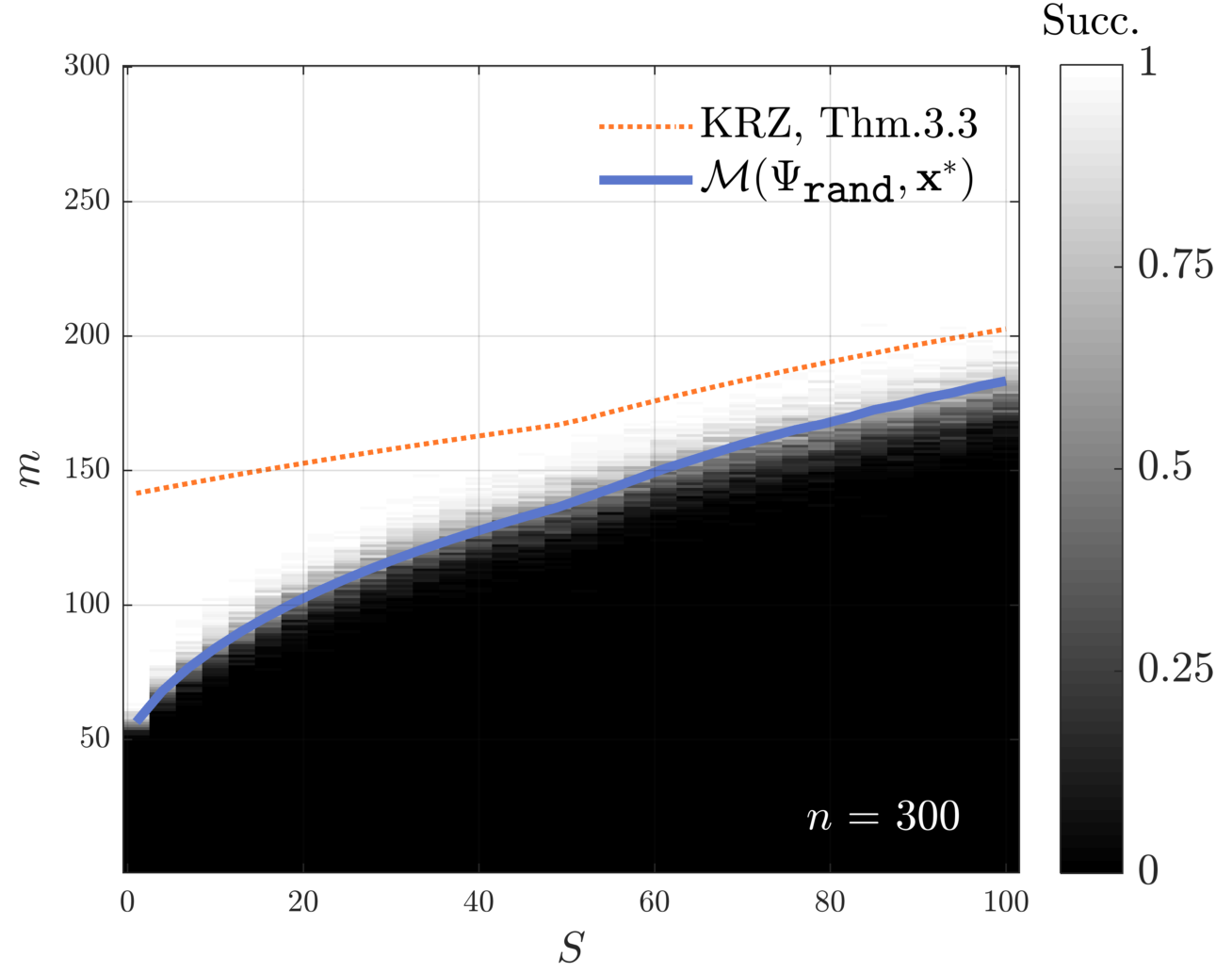}
		\caption{}
		\label{fig:appl:onbrandom:S}
	\end{subfigure}%
	\qquad
	\begin{subfigure}[t]{0.45\textwidth}
		\centering
		\includegraphics[width=\textwidth]{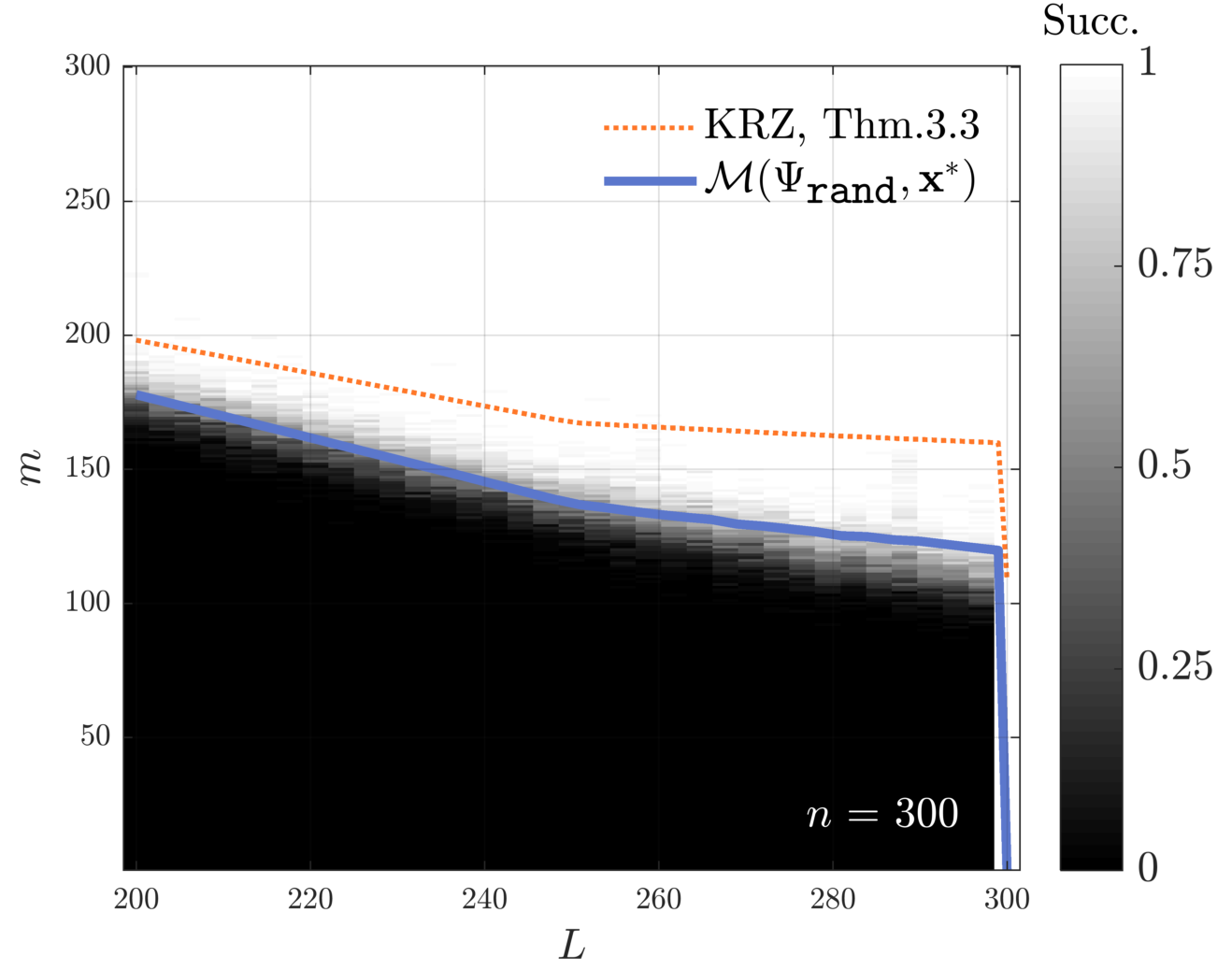}
		\caption{}
		\label{fig:appl:onbrandom:L}
	\end{subfigure}%	
	\caption{Phase transition of \refabpnoiseless{\noiseparam=0}{\aoprand} for random frames $\aoprand \in \R^{N \times n}$ with either $L$ or $S$ fixed. 
	\subref{fig:appl:onbrandom:S}~$L = 250$ is fixed and the labels of the horizontal axis do only display the sparsity of $\mathscr{S}$. \subref{fig:appl:onbrandom:L}~$S = 50$ is fixed and the labels of the horizontal axis do only display the cosparsity of $\mathscr{S}$.}
	\label{fig:appl:onbrandom:varySL}
\end{figure}

In Figure~\ref{fig:appl:onbrandom:parse}, we have selected the pairs $(S, L)$ such that both $n$ and $N$ remain fixed. 
The purpose of the second and third run of Experiment~\ref{exp:appl:onb:pt} is to highlight the non-trivial impact of both parameters by varying either one of them (and adapting $N$).
Figure~\ref{fig:appl:onbrandom:S} is obtained by fixing $L = 250$ and applying Experiment~\ref{exp:appl:onb:pt} with $n = 300$
and
\begin{equation}
	\mathscr{S} = \{ (1, L), (4, L), (7, L),\dots, (100,L) \}.
\end{equation}
In Figure~\ref{fig:appl:onbrandom:L}, we have fixed $S = 50$ and invoked Experiment~\ref{exp:appl:onb:pt} with $n = 300$ and
\begin{equation}
	\mathscr{S} = \{ (S, 200), (S, 203), (S, 206),\dots, (S,299) \}.
\end{equation}

The transition curves of Figure~\ref{fig:appl:onbrandom:S} and Figure~\ref{fig:appl:onbrandom:L} are obviously not constant in $S$ and $L$, respectively, which verifies once again that neither sparsity nor cosparsity can thoroughly quantify the recovery performance of $\l{1}$-analysis minimization.  
Note that the plots of Figure~\ref{fig:appl:onbrandom:varySL} do actually consist of two regions for which the shapes of the curves are somewhat different. This is due to the case distinction between $N \geq n$ and $N < n$, i.e., whether $\aoprand \in \R^{N \times n}$ forms a frame or not.
Interestingly, Figure~\ref{fig:appl:onbrandom:L} also reveals another drawback of the traditional bound in \eqref{eq:intro:boundliterature}: 
If $S$ remains fixed, the analysis dimension $N = S + L$ increases with $L$, so that one would expect a logarithmic growth of the number of required measurements. 
But Figure~\ref{fig:appl:onbrandom:L} shows that the true sample complexity even decreases for larger values of $L$. 
Regarding Theorem~\ref{thm:results:recovery}, we can again conclude that the sampling-rate function $\samplecompl{\aoprand,\grtr}$ captures the phase transition almost perfectly in both scenarios. The prediction of \cite[Thm.~3.3]{rkz2015} also reflects the shapes of the transition curves, but is still worse than ours.\footnote{Note that, although \cite[Thm.~3.3]{rkz2015} is only stated for frames, it literally holds true for any choice of $\aop \in \R^{N \times n}$ where the upper frame bound just needs to be replaced by $\lnorm{\aop}[2\to 2]^2$.}

\subsection{Analysis Operators for 2D Signals}
\label{subsec:appl:2d}

We have only considered one-dimensional signals up to now, in particular, the class of piecewise constant functions in Subsection~\ref{subsec:appl:wavelets} and Subsection~\ref{subsec:appl:tv}.
However, signals in higher dimensions typically exhibit a richer geometric structure, such as anisotropic features. 
For example, it has been observed in the literature that the sample complexity of total variation minimization in 2D scales very differently compared to its counterpart in 1D \cite{needell2013,cai2015, krahmer2017}.
Let us therefore also conduct a simple experiment in 2D, using finite differences, redundant Haar wavelets as well as the 2D discrete cosine transform (2D-DCT) as analysis operators.
In order to reduce the immense computational burden of creating a phase transition plot, we just restrict ourselves to a specific image signal here.
Figure~\ref{fig:appl:2D:X} illustrates our choice of $\grtr$, which is a realistic brain phantom for magnetic resonance imaging~\cite{GLPU} at a resolution of $100 \times 100$.

Recovery based on total variation in 2D aims at retrieving images from compressed measurements by promoting sparse gradients, i.e., piecewise constant signals. 
As in 1D, there exist multiple variants of total variation minimization, depending on different boundary conditions and on how the $\l{1}$-norm of the discrete gradient is calculated. To keep the exposition as brief as possible, we focus on the simple case where an anisotropic finite difference operator, based on periodic forward differences, is used and $\grtr$ corresponds to an $\tilde{n} \times \tilde{n}$ grayscale image. Thus, up to modifications of the boundary values, the entries of the discrete (forward) gradient in 2D are defined by
\begin{equation}
(\gradient \grtr)_{j,j'} \coloneqq \matr{ (\gradient_1 \grtr)_{j,j'} \\ (\gradient_2 \grtr)_{j,j'}} \coloneqq \matr{ x_{j+1,j'}^\ast - x_{j,j'}^\ast \\ x_{j,j'+1}^\ast - x_{j,j'}^\ast}.
\end{equation}
For the sake of convenience, we will now identify the signal domain $\R^{\tilde{n} \times \tilde{n}}$ with its canonical columnwise vectorization $\R^{\tilde{n}^2}$, i.e., the ambient space is of dimension $n \coloneqq \tilde{n}^2$ and $\grtr \in \R^{n}$. 
% The measurements are, as usually, denoted by $\y = \A \grtr$, for a Gaussian matrix $\A \in \R^{m \times n^2}$ and 
The associated 2D finite difference operator then takes the following form:
\begin{align}
 \aoptvD \coloneqq [ \aoptv \otimes \I{\tilde{n}}, \I{\tilde{n}} \otimes \aoptv ]^\T \in \R^{2(\tilde{n}-1)\tilde{n} \times n}.
\end{align}
To define a redundant Haar wavelet frame in 2D, we again rely on the (two-dimensional) inverse wavelet transform with $2$ scales provided by the software package \texttt{spot}~\cite{spot}, and denote it by $\aoprD \in \R^{7n \times n}$.
Finally, we use the standard \texttt{Matlab} implementation to compute the 2D-DCT transform on $4\times 4$ patches, omitting the DC coefficient. The resulting analysis operator $\aopdct \in \R^{15n \times n}$ is then applied to the entire image via a common sliding-window technique for obtaining the coefficients on all $n$ patches. Note that such a convolutional model for computing analysis operators is frequently used in the literature, in particular also for learned analysis operators, e.g., see \cite{hawe2013}.  Although the DCT filters are not a perfect match for the piecewise constant image signal $\grtr$, we consider them as a classical representative that is well suited to demonstrate the predictive power of our framework.

We assess the recovery capability very similarly to Figure~\ref{fig:appl:wavelets:compr:err1}:
Experiment~\ref{exp:appl:wavelets:intro} is invoked for $\aop = \aoptvD$, $\aop = \aoprD$, and $\aop = \aopdct$, where $n = \tilde{n}^2 = 10000$ and $M = \{1,26,51,\dots\}$;
note that only $12$ repetitions were performed here instead of $50$.
The plots of Figure~\ref{fig:appl:2D:TV}, \ref{fig:appl:2D:wave}, and~\ref{fig:appl:2D:dct} visualize the averaged error curves for \refabpnoiseless{\noiseparam=0}{\aoptvD}, \refabpnoiseless{\noiseparam=0}{\aoprD}, and \refabpnoiseless{\noiseparam=0}{\aopdct}, respectively. Unfortunately, solving these minimization problems via \texttt{cvx} with $n=10000$ is not feasible anymore, so that we have used an implementation based on the \emph{alternating direction method of multipliers} instead; see \cite{boyd2011} and references therein. 

\begin{figure}[!t]
	\centering
	\begin{subfigure}[t]{0.45\textwidth}
		\centering
		\includegraphics[height=0.65\textwidth]{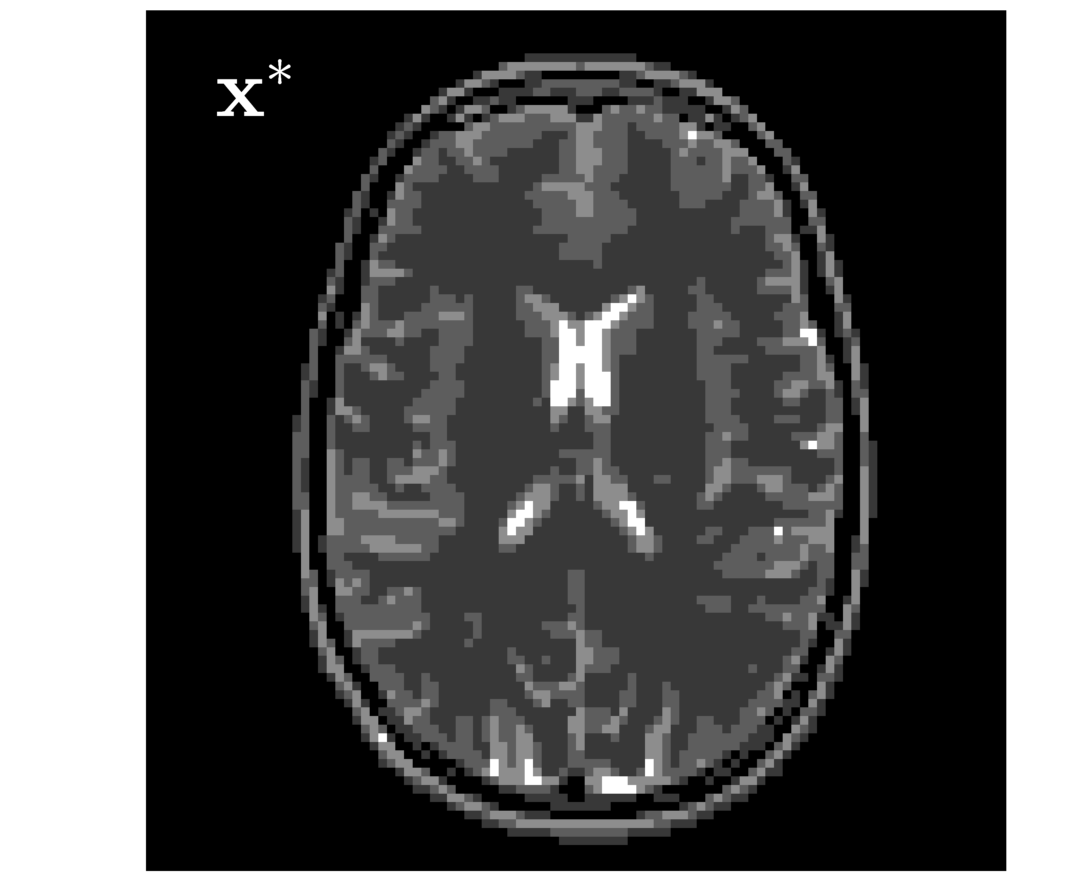}
		\caption{}
		\label{fig:appl:2D:X}
	\end{subfigure}%
	\begin{subfigure}[t]{0.45\textwidth}
		\centering
		\includegraphics[height=0.63\textwidth]{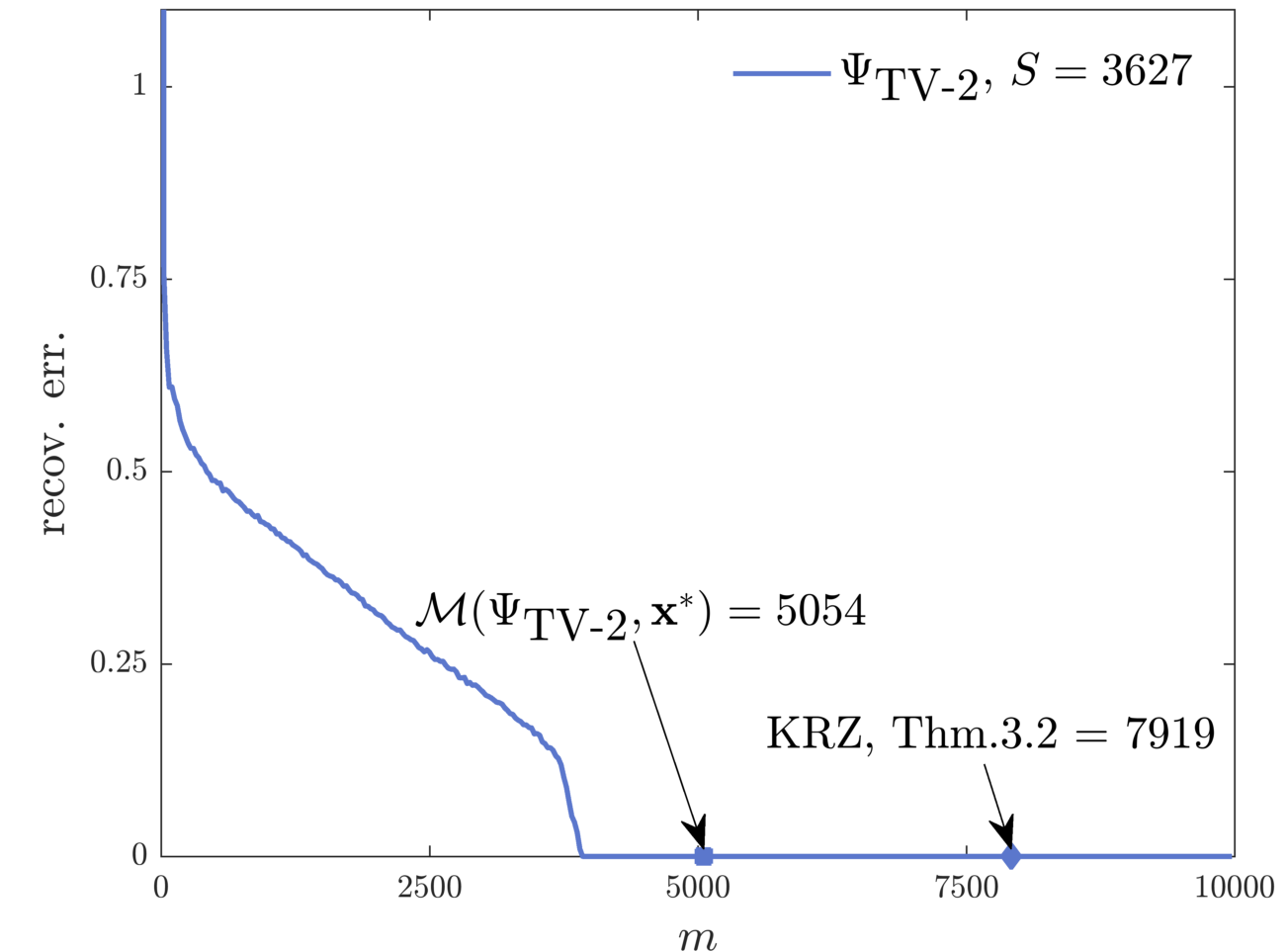}
		\caption{}
		\label{fig:appl:2D:TV}
	\end{subfigure}%

	\begin{subfigure}[t]{0.45\textwidth}
		\centering
		\includegraphics[height=0.63\textwidth]{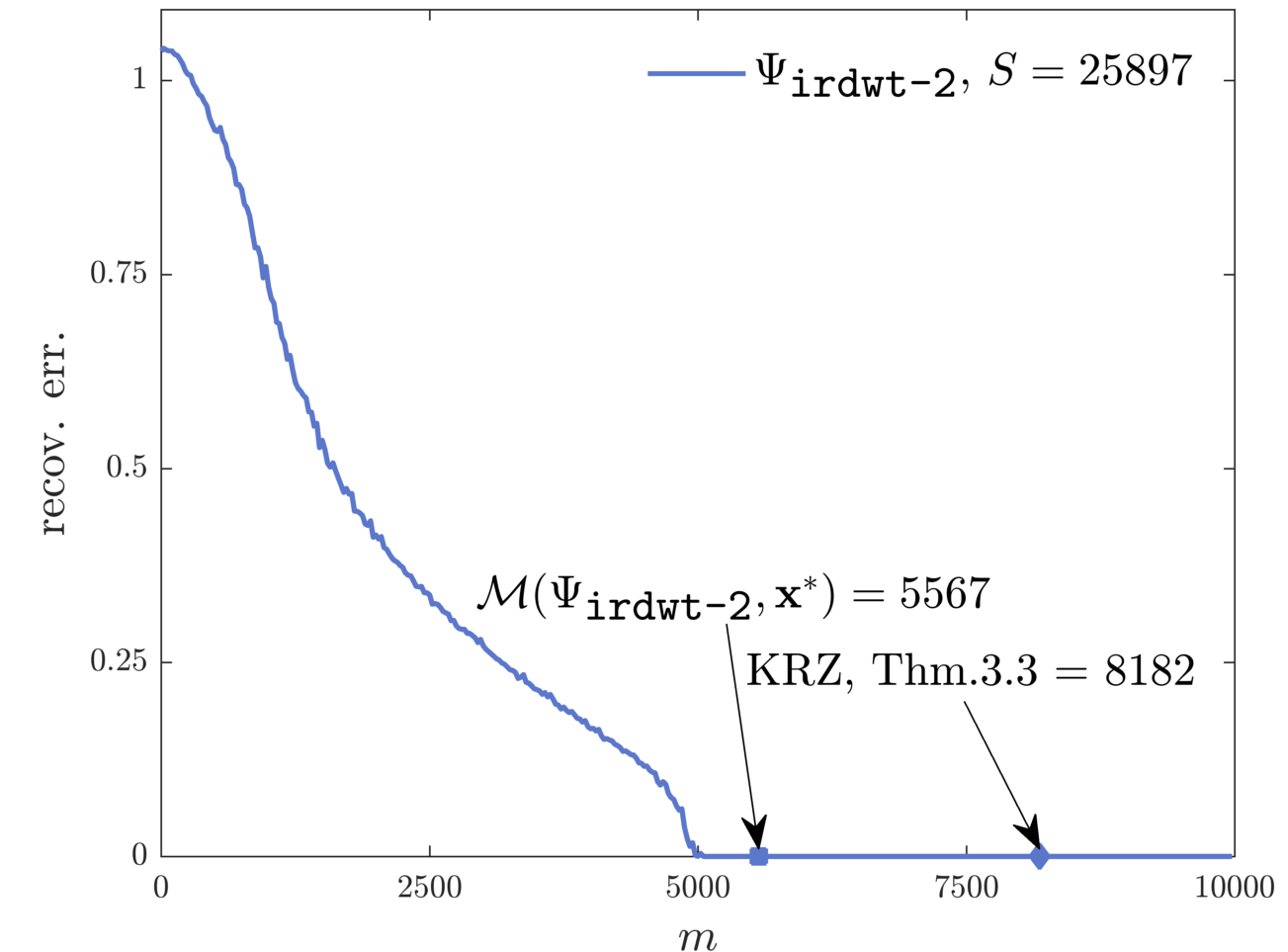}
		\caption{}
		\label{fig:appl:2D:wave}
	\end{subfigure}%	
    \begin{subfigure}[t]{0.46\textwidth}
		\centering
		\includegraphics[height=0.62\textwidth]{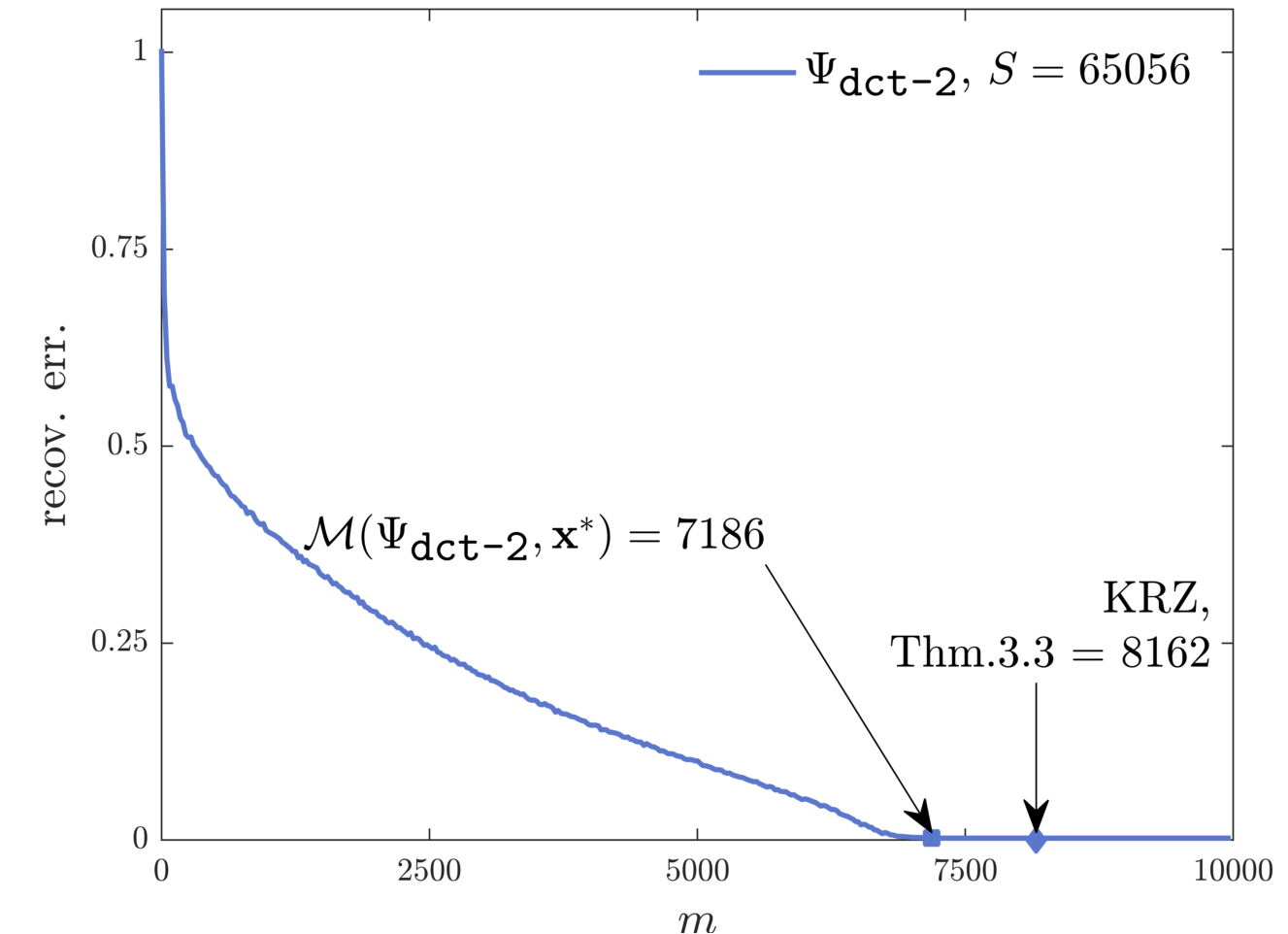}
		\caption{}
		\label{fig:appl:2D:dct}
	\end{subfigure}%	
	\caption{\subref{fig:appl:2D:X}~Visualization of the ground truth image signal $\grtr \in \R^{100^2} \cong \R^{100\times 100}$ (``brain phantom'' \cite{GLPU}).
	\subref{fig:appl:2D:TV}+\subref{fig:appl:2D:wave}+\subref{fig:appl:2D:dct}~Average recovery error achieved by \refabpnoiseless{\noiseparam=0}{\aoptvD},~\refabpnoiseless{\noiseparam=0}{\aoprD}, and \refabpnoiseless{\noiseparam=0}{\aopdct}, respectively.}
	\label{fig:appl:2D}
\end{figure}

Similar to the 1D case, it can be observed that the coefficient sparsity $\lnorm{\aoprD \grtr}[0] = 25897$ is substantially larger than the ambient dimension $n = 10000$. For the $15$-times redundant operator $\aopdct$ this observation is even more striking with  $\lnorm{\aopdct \grtr}[0] = 65056$.
Nevertheless, there appears to be a sharp phase transition at $m\approx 5000$ measurements in Figure~\ref{fig:appl:2D:wave} and a transition at $m \approx 7100$ in Figure~\ref{fig:appl:2D:dct}. The estimate of Theorem~\ref{thm:results:recovery} perfectly captures the phase transition for the 2D-DCT operator. In the wavelet case, the prediction $m\approx 5500$  is still quite accurate, corresponding to a relative error of about $5\%$. In contrast, our bound is slightly more pessimistic for 2D total variation (relative error about $10\%$) but still quite reliable. % significantly smaller than $n$. 
This is obviously not the case for \cite{rkz2015}, whose predictions deviate much more strongly from the true sample complexity. % are relatively close to the vacuous case of $m = n$. 
Finally, it is worth mentioning that the above experimental setup may not be suited for a comparison between these operators in respect of analysis modeling. 
We suspect for instance that redefining $\aoprD$ with more scales and a more sophisticated weighting would considerably improve the outcomes of \refabpnoiseless{\noiseparam=0}{\aoprD}.

% \begin{highlight}
% 	The sampling rate for total variation follows different laws in 1D or in higher dimensions. Thus, we also conduct experiments investigating the prediction accuracy in 2D, which demonstrate that Theorem~\ref{thm:results:recovery} also allows for reasonable predictions in this setup. 
% \end{highlight}

\section{Related Literature}
\label{sec:discussion}

In this part, we survey existing theoretical approaches to the analysis formulation and put them into context with our findings. 
Our discussion starts in Subsection~\ref{subsec:discussion:anavssyn} with a comparison to the \emph{synthesis formulation}, which is also widely used in compressed sensing and promotes a somewhat different viewpoint on sparse representations.
In Subsection~\ref{subsec:discussion:sparse}, we return to our initial concern from Subsection~\ref{subsec:intro:issues}, presenting several results relying on traditional analysis sparsity (cf. \eqref{eq:intro:boundliterature}).
Subsection~\ref{subsec:discussion:cosparse} then points out the importance of cosparse modeling.
Finally, more details on the recent work of Kabanava, Rauhut, and Zhang \cite{rkz2015} are provided in Subsection~\ref{subsec:discussion:further}.

\subsection{Analysis versus Synthesis}
\label{subsec:discussion:anavssyn}

One of the cornerstones in the literature on sparsity priors is \cite{Elad2006}, which was the first work systematically studying the relationship between the predominant synthesis formulation and the analysis formulation. Instead of solving~\eqref{eq:intro:anabp}, the \emph{synthesis formulation} rather considers the convex program
\begin{equation}
 \label{eq:disc:synthesis}\tag{$\text{BPS}_{\noiseparam}^{\vec{D}}$}
 \vec{D} \cdot \left(\argmin_{\vec{z} \in \R^N} \lnorm{\vec{z}}[1] \quad \text{subject to \quad $\lnorm{\A \vec{D}\vec{z} -\y}[2] \leq \noiseparam$} \right),
\end{equation}
where $\vec{D} \in \R^{n \times N}$ is a (possibly redundant) \emph{dictionary} in $\R^n$. The rationale behind this approach is that the signal vector $\grtr \in \R^n$ possesses a \emph{sparse representation} by $\vec{D}$, i.e., $\grtr = \vec{D} \vec{z}$ for a coefficient vector $\vec{z} \in \R^N$ with $\lnorm{\vec{z}}[0] \ll n$. 

By investigating the respective polytopal geometry of \eqref{eq:intro:anabp} and \eqref{eq:disc:synthesis}, the authors of~\cite{Elad2006} develop a theoretical model that describes the differences between both methods. In particular, they point out the following fundamental issue: While the synthesis approach \eqref{eq:disc:synthesis} seems to benefit from a higher redundancy of the dictionary $\vec{D}$, it is not clear how \eqref{eq:intro:anabp} is influenced by the redundancy degree of $\aop$. Theorem~\ref{thm:results:recovery} shows that this structural property indeed has a wide impact, which is however highly non-trivial. 
We hope that our results could give rise to further progress in this matter.

The numerical simulations of \cite{Elad2006} certify a considerable gap between both strategies and the observed recovery performances indicate that the analysis formulation outperforms its synthe\-sis-based counterpart in many situations of interest. Although never stated as a general conclusion, the superiority of analysis-based priors is often confirmed in the literature, e.g., see \cite{Figueiredo2009}. Moreover, depending on the redundancy of $\aop$ and $\vec{D}$, one may also argue that solving \eqref{eq:intro:anabp} is computationally more appealing, since the actual optimization takes place in $\R^n$. In contrast, \eqref{eq:disc:synthesis} operates on the possibly much higher dimensional coefficient space of $\R^N$. We refer the interested reader to \cite{rauhut2008} for more information on synthesis recovery. 

\subsection{Approaches Based on Analysis Sparsity}
\label{subsec:discussion:sparse}

The first compressed-sensing-based approach to the analysis formulation was undertaken by the seminal work of Candès, Eldar, Needell, and Randall \cite{candes2011csdict}.
While previous results for redundant dictionaries did rather study coefficient recovery in the synthesis formulation, e.g., based on incoherent frame atoms~\cite{rauhut2008}, a major breakthrough of \cite{candes2011csdict} was that such assumptions can be avoided by switching over to the analysis perspective.
% Previous recovery results for redundant dictionaries are formulated in terms of coefficient recovery in the synthesis formulation and rely . In~\cite{candes2011csdict} on the other hand, it is argued by Cand\`{e}s et al. that  one can avoid this assumption by switching over to the analysis formulation instead.  
The theoretical analysis of \cite{candes2011csdict} relies on the so-called $\vec{D}$-RIP, which is an adaptation of the classical \emph{restricted isometry property} to sparse representations in dictionaries:
\begin{definition}[$\vec{D}$-RIP, \cite{candes2011csdict}]
Let $\vec{D} \in \R^{n\times N}$ be a tight frame\footnote{In contrast to our convention in Subsection~\ref{subsec:intro:notation}\ref{item:intro:notation:frames}, the \emph{columns} of $\vec{D}$ form a frame for $\R^n$ here, meaning that $\vec{D}$ is actually the associated synthesis operator.} for $\R^n$ and $\A \in \R^{m \times n}$. Then, $\A$ satisfies the \emph{$\vec{D}$-RIP} with parameters $\delta > 0$ and $S \in [N]$ if
\begin{equation}
 (1-\delta) \lnorm{\vec{Dz}}[2]^2 \leq \lnorm{\A \vec{D} \vec{z}}[2]^2 \leq (1 + \delta) \lnorm{\vec{Dz}}[2]^2
\end{equation}
holds true for all $S$-sparse vectors $\vec{z} \in \R^N$.  
\end{definition}
Supposed that $\A$ fulfills the $\vec{D}$-RIP with constant $\delta < 0.08$, it was shown in~\cite{candes2011csdict} that any minimizer $\solu \in \R^n$ of \refabpnoiseless{\noiseparam}{\aop} with $\aop \coloneqq \vec{D}^\T$ obeys
\begin{equation} \label{eq:disc:error}
%  \lnorm{\solu - \grtr}[2] \leq C_0 \cdot \noiseparam + C_1 \cdot \frac{\lnorm{\vec{D}^*\grtr - \left(\vec{D}^T \grtr\right)_S}[1]}{\sqrt{S}},
	\lnorm{\solu - \grtr}[2] \leq C_0 \cdot \noiseparam + C_1 \cdot \frac{\sigma_S(\aop \grtr)_1}{\sqrt{S}},
\end{equation}
where $C_0,C_1 > 0$ are numerical constants and $\sigma_S(\aop \grtr)_1$ denotes the best $S$-term approximation error of the coefficient vector $\aop\grtr$.
% $\left(\vec{D}^T\grtr\right)_S$ denotes the vector consisting of the largest $S$ entries of $\vec{D}^T \grtr$ in magnitude. 
This result was the starting point of several generalizations and refinements: For instance, \cite{foucart2014} has extended the above setup to arbitrary frames and Weibull matrices, incorporating an adapted robust null space property. The work \cite{liu2012} modifies the concept of $\vec{D}$-RIP in such a way that \eqref{eq:intro:anabp} can be applied to arbitrary dual operators of frames. 
Several achievements towards practical applicability were also made by \cite{krahmer2015}, allowing for structured measurements if a certain incoherence condition on the dictionary $\vec{D}$ and the subsampled sensing matrix is satisfied.
% In~\cite{krahmer2015}, a step was made towards applications by extending the result to structured measurements by controlling a localization factor of $\vec{D}$ and the incoherence between the measurement system and the frame. 
Guarantees of a similar flavor were unified in a general framework~\cite{Lee2018Ana} and even proven in an infinite-dimensional setting~\cite{poon2017}, both using different proof techniques. Finally, results involving the $\vec{D}$-RIP were also obtained for various modifications of the analysis basis pursuit \eqref{eq:intro:anabp}, e.g., see \cite{Lin2014ana,Shen2015ana,Tan2014ana}.

Regarding the measurement model of this work (Model~\ref{model:results:setup:meas}), it has turned out that a standard Gaussian matrix $\A \in \R^{m \times n}$  fulfills the $\vec{D}$-RIP with high probability provided that the number of observations obeys $m \gtrsim S \cdot \log (e \cdot N/S)$.
In this case, we again obtain a uniform error bound of the form \eqref{eq:disc:error}. 
A similar sampling rate for Gaussian measurements was also recently achieved by Kabanava and Rauhut in~\cite{kabanava2015}.
Their proof is however based on yet another statistical tool, namely a modified version of \emph{Gordon's Escape Through a Mesh}---note that we also make use of this fundamental principle in Theorem~\ref{thm:prelim:framework:minsingval}. In a nutshell, the first (non-uniform) guarantee of \cite{kabanava2015} reads as follows: Given $\grtr \in \R^n$ with $\lnorm{\aop\grtr}[0] \leq S$, a Gaussian sensing matrix $\A \in \R^{m \times n}$ satisfying
\begin{equation} \label{eq:disc:kabanava}
 \frac{m^2}{m + 1} \geq \frac{2 \fbu}{\fbl} \cdot S \cdot \bigg(\sqrt{\log \left( \tfrac{e \cdot N}{S} \right)} + \sqrt{\tfrac{\fbl \cdot \log (\varepsilon^{-1}) }{\fbu \cdot S}} \bigg)^2
\end{equation}
enables recovery of $\grtr$ via~\eqref{eq:intro:anabpnoiseless} with probability at least $1-\varepsilon$.
It is worth mentioning that, compared to \eqref{eq:disc:error}, this bound does not involve any unspecified numerical constants.
% and is therefore well-suited for comparison to our results, see \eqref{eq:appl:wavelets:err_kabanava2015} and Figure~\ref{fig:appl:wavelets:bestSterm}.
Moreover, by introducing a generalized null space property, the statement of \eqref{eq:disc:kabanava} can be even extended to uniform recovery \cite[Thm.~9]{kabanava2015}. 

Compared to our sampling-rate bound in Theorem~\ref{thm:results:recovery}, a condition of the form $m \gtrsim S \cdot \log (e \cdot N/S)$ is quite attractive and intuitive from an aesthetic viewpoint, since it mimics traditional results in compressed sensing. 
This resemblance is not very astonishing, since most proofs in the literature do actually ``operate'' on the space of analysis coefficients $\R^N$ and then just pull back the corresponding estimates to the signal domain $\R^n$.
% into the signal domain, however, they essentially carry on the information of the higher dimensional coefficient space. 
For various examples of less redundant analysis operators, e.g., tight random frames, such a strategy indeed seems to provide accurate predictions of the required number of measurements. 

However, as already pointed out in the course of Subsection~\ref{subsec:intro:issues}, the situation becomes much more delicate for highly redundant and coherent systems, in which the above approaches do by far not explain the capacity of~\eqref{eq:intro:anabp}. 
The theoretical framework of Section~\ref{sec:results} and the preceding experiments in Section~\ref{sec:appl} demonstrate, for instance, that the mutual coherence structure of an operator $\aop$ deserves special attention as well. %, while this important property was only marginally addressed in previous works.
Apart from that, our results do not assume that $\aop$ forms a spanning system for $\R^n$, which in turn was an important hypothesis in \cite{candes2011csdict,kabanava2015}.
This observation particularly gives rise to doubts that conditioning-related quantities, such as the ratio $\fbu/\fbl$ in \eqref{eq:disc:kabanava}, need to  appear explicitly in a sampling-rate bound.

Let us finally turn to the issue of stability: The error estimate of \eqref{eq:disc:error} relies on the best $S$-term approximation error. Since this parameter is completely determined by the analysis coefficient sequence $\aop\grtr$, the linear dependencies within $\aop$ are again completely disregarded.
The compressibility experiment of Figure~\ref{fig:appl:wavelets:compr} shows that this methodology can be problematic for redundant systems.
On the other hand, our approach in Subsection~\ref{subsec:results:compressible} is based on approximations in the signal space $\R^n$, and when combined with an appropriate greedy scheme (e.g., Algorithm~\ref{exp:impl:algo}), it yields much more reliable outcomes.

\subsection{Approaches Based on Cosparse Modeling}
\label{subsec:discussion:cosparse}

A very influential contribution to analysis-based priors in inverse problem theory is the work of Nam, Davies, Elad, and Gribonval~\cite{nam2013}, which has coined the term \emph{cosparse analysis model}. Perhaps, this recent branch of research has grown out of the same observation as we have made above:
\begin{quote}
``The fact that this representation may contain many non-zeroes (and especially so when $N \gg n$) should be of no consequence to the efficiency of the analysis model.'' \cite[p. 35]{nam2013}
\end{quote}
In fact, a key insight of this paper was that the viewpoints of analysis and synthesis sparsity are actually based on completely different signal models and therefore also draw their strength from different features. While the synthesis approach puts its emphasis on the number of non-zero coefficients, \cite{nam2013} argues that the analysis model is rather specified by the number of vanishing coefficients, that is, the \emph{cosparsity}. Indeed, the class of $L$-cosparse signals can be simply written as a union of subspaces
\begin{equation} \label{eq:disc:union}
 \bigcup_{\substack{\lam \subset [N], \\ \cardinality{\lam} = L}} W_{\lam},
\end{equation}
where $W_{\lam} \coloneqq \ker\aop_{\lam}$.  
A careful comparison of both setups reveals that one should not---as suggested by the approaches presented in Subsection~\ref{subsec:discussion:sparse}---regard the analysis formulation from a synthesis perspective, thereby treating ``the analysis operator as a \emph{poor man's} sparse synthesis representation''~\cite[p. 35]{nam2013}. 
For example, the subspace collection of \eqref{eq:disc:union} may become extremely complicated if $\aop$ corresponds to a redundant frame with strong linear dependencies, and it is quite evident that its geometric arrangement is then not just reflected by the value of $S$ (or~$L$).

Towards a theoretical analysis, \cite[Thm.~7]{nam2013} establishes the following equivalence: Every $\grtr \in \R^n$ with cosupport $\lam = \ssuppc[\grtr] \subset [N]$ is a unique minimizer of~\refabpnoiseless{\noiseparam=0}{\aop} if and only if
\begin{equation}\label{eq:disc:cospequiv}
	\sup_{\substack{\x \in \R^n \\ \aop_{\lam} \x =  \vnull}} \abs{\sp{\aop_{\setcompl{\lam}} \vec{z}}{\sign (\aop_{\setcompl{\lam}} \x )} } < \lnorm{\aop_{\lam} \vec{z}}[1] \quad \text{for all $\vec{z} \in \ker(\A)\setminus\{ \vnull \}$.}
\end{equation}
We note that similar conditions were also studied for stable recovery in~\cite{vaiter2013}. 
% Moreover, it is worth mentioning that the sign vector $\sign(\aop_{\setcompl{\lam}} \x)$ from \eqref{eq:disc:cospequiv} plays a central role in our approach as well, since it specifies the generalized sparsity term $\spparam{\x}$ (see Definition~\ref{def:results:generalsparsity}).
Moreover, it is worth mentioning that the vector $\sign(\aop_{\setcompl{\lam}} \x)$ plays a central role in \eqref{eq:disc:cospequiv}, which reminds us of our generalized sparsity term $\spparam{\x}$ from Definition~\ref{def:results:generalsparsity}.
In contrast, such an expression is missing in the concept of $\vec{D}$-RIP.

Regarding sample complexity, it has turned out that the quantity
\begin{equation}
	\kappa_{\aop}(L) \coloneqq \max_{\cardinality{\lam} \geq L} \dim W_{\lam}
\end{equation}
essentially determines the (optimal) number of measurements required for noiseless recovery via a combinatorial search over all $L$-cosparse signals \cite[Sec. 3]{nam2013}. This finding is particularly consistent with the philosophy of compressed sensing after which an estimation succeeds as long as the sampling rate slightly exceeds the signal's manifold dimension. 
Somewhat surprisingly, the recent work of Giryes, Plan, and Vershynin \cite{giryes2015} has shown that the situation is unfortunately not that simple in $\l{1}$-analysis minimization: When asking for robustness or a tractable algorithm instead of combinatorial searching, the sample complexity is not appropriately captured by $\kappa_{\aop}(L)$ anymore. 
For that reason, we hope that our theoretical framework sheds more light on this issue, featuring the complex relationship between sparsity and cosparsity in the study of \eqref{eq:intro:anabp}. 
% Finally, we would like to point out an alternative strategy by \cite{giryes2014}, where greedy-like algorithms are proposed to tackle the drawbacks of \cite{nam2013}.

% An insightful study regarding this principle for analysis sparsity has been recently provided by Giryes et al.~\cite{giryes2015}: surprisingly, demanding stability to noise, or applying a tractable algorithm like~\eqref{eq:intro:anabp} instead of a combinatorial search, is found to require a number of measurements that is detached from the underlying manifold dimension. We hope that our analysis for (sub-)Gaussian measurements sheds new light on this discussion and demonstrates the complex interaction of sparsity an cosparsity for determining a sufficient number of measurements for a recovery with~\eqref{eq:intro:anabp}. 
% 
% Finally, building on the work in~\cite{nam2013}, more greedy-like algorithms for the cosparse analysis model have been proposed and analyzed in~\cite{giryes2014}. 

\subsection{Non-Asymptotic Sampling-Rate Bounds}
\label{subsec:discussion:further}

To the best of our knowledge, the only work that allows for a direct comparison to our results is by Kabanava, Rauhut, and Zhang \cite{rkz2015}, hence often serving as a ``competitor'' in Section~\ref{sec:appl}.
In fact, it is also motivated by a similar observation as we have made in this paper: Bounds on the number of needed measurements for frame-based analysis operators (cf. \eqref{eq:disc:kabanava}) and total variation \cite{cai2015} may render vacuous statements if the sparsity level is too high. 

Focusing on the case of Gaussian measurements, Theorem~3.3 in \cite{rkz2015} provides a non-asymptotic sampling-rate bound when $\aop \in \R^{N \times n}$ is a frame: Any minimizer $\solu \in \R^n$ of \eqref{eq:intro:anabp} satisfies $\lnorm{\solu - \grtr}[2] \leq 2\noiseparam / \theta$ with probability at least $1-\varepsilon$, provided that
\begin{equation} \label{eq:disc:rkz}
	\frac{m^2}{m+1} \geq \bigg(\sqrt{\vphantom{\bigg(} n - \tfrac{2(\sum_{i \in \ssuppc[\grtr]} \lnorm{\avec_i}[2])^2}{\pi \cdot \fbu \cdot N}} + \sqrt{2 \log (\varepsilon^{-1})} + \theta \bigg)^2,
\end{equation}
where $\fbu$ is an upper frame bound of $\aop$.
An analogous statement holds true if $\aop$ is the finite difference operator in 1D (cf. Subsection~\ref{subsec:appl:tv}), where \eqref{eq:disc:rkz} is replaced by
\begin{equation}
 \frac{m^2}{m+1} \geq \bigg(\sqrt{n \cdot \left(1 - \tfrac{1}{\pi} \left(1-\tfrac{S}{n} \right)^2 \right)} + \sqrt{2 \log (\varepsilon^{-1})} + \theta \bigg)^2,
\end{equation}
see Theorem~3.1 in \cite{rkz2015} for details. Ignoring terms of lower order, both bounds involve the dimension $n$ of the ambient space as an additive term. %, implying that recovery is ensured with $m \leq n$. 
This resemblance to the sampling-rate function $\samplecompl{\aop,\grtr}$ from Definition~\ref{def:results:samplingratefct} is actually not very surprising, since the proof strategy of \cite{rkz2015} builds upon the same principle as in this work, namely estimating the conic mean width. % \cite{chandrasekaran2012geometry,amelunxen2014edge}.
% Besides this similarity with our sampling rate of Definition~\ref{def:results:samplingratefct}, both bounds are non-asymptotic and allow for a direct comparison with our results as it is done for many examples in Section~\ref{sec:appl}. 

% \pagebreak
However, there are several substantial differences between the approach of \cite{rkz2015} and ours:
\begin{itemize}
\item 
	The results of~\cite{rkz2015} are not consistent with classical compressed sensing theory: If $\aop \in \R^{n \times n}$ is an orthonormal basis, \eqref{eq:disc:rkz} basically degenerates into
	\begin{equation}
		m \geq \Big( 1 - \tfrac{2}{\pi}\Big) \cdot n + \tfrac{4S}{\pi} - \tfrac{2S^2}{\pi \cdot n} \ ,
	\end{equation}
	which is clearly suboptimal for small values of $S$; see also Figure~\ref{fig:appl:wavelets:scaling:RelE2}.
\item 
	The bounds on the required number of measurements in \cite{rkz2015} are specific to Gaussian measurements and do not include stable recovery that would allow for coping with compressible signals.
\item 
	In all cases examined in Section~\ref{sec:appl}, our bounds are vastly superior in terms of predicting the phase transition behavior of \refabpnoiseless{\noiseparam=0}{\aop}; see Figure~\ref{fig:intro:pt}, Figure~\ref{fig:appl:tv}, and Figure~\ref{fig:appl:onbrandom:varySL}.
% 	In all of the examples that we have analyzed, it does not provide an accurate prediction of the needed number of measurements; see in particular Figure~\ref{fig:intro:pt} or Figure~\ref{fig:appl:wavelets:scaling}.
\end{itemize}

Comparing the sampling-rate bound \eqref{eq:results:recovery:meas} of Theorem~\ref{thm:results:recovery} with \eqref{eq:disc:rkz}, we observe that the latter does neither incorporate the cross-correlations between the individual analysis vectors nor the sign vector $\anasign{\grtr} = \sign(\aop\grtr)$. 
On the other hand, both features seem to be essential for achieving sound recovery results. 
A careful study of the corresponding proofs reveals the underlying reason for the superiority of our approach:
While Step 1 in the proof of Theorem~\ref{thm:proofs:mainresults:scbound} in Subsection~\ref{subsec:proofs:scbound} is rather standard (cf. \cite{amelunxen2014edge}), our argumentation in Step 2 is quite different. Indeed, the dual vector $\dv'$ in \cite{rkz2015} is chosen such that it just maximizes the term $T_1$ in \eqref{eq:proofs:scbound:bound:dvgeneral}. Our choice of $\dv'$ in \eqref{eq:proofs:scbound:dvgeneral} is more sophisticated, particularly taking the variable $\tau$ into account. 
This appears to be a minor issue at first sight but actually affects the remainder of the proofs drastically. For example, we had to come up with a highly non-trivial estimate for $T_2$ in Lemma~\ref{lem:proofs:scbound:clippedcovar}, eventually leading to the generalized cosparsity term $\cospparam{\grtr}$ (see Definition~\ref{def:results:samplingratefct}). Furthermore, the optimization over $\tau$ is very challenging in our case, whereas it is just a simple quadratic problem in \cite{rkz2015}.

\section{Discussion and Outlook}
\label{sec:concl}

One of the key findings of this work was that the widely-used concept of analysis sparsity is oftentimes not capable of explaining the success (and failure) of $\l{1}$-analysis recovery. 
This observation defies conventional wisdom that would justify the superiority of an operator $\aop \in \R^{N \times n}$ only by means of sparsity, i.e., smaller values of $S = \lnorm{\aop\grtr}[0]$ lead to better outcomes.
In contrast, our three generalized sparsity parameters from Definition~\ref{def:results:generalsparsity} are more carefully designed and particularly involve the Gram matrix $\gram = \aop \aop^\T \in \R^{N \times N}$.
%The definition of our three generalized sparsity parameters (see Definition~\ref{def:results:generalsparsity}), in contrast, is somewhat more delicate and particularly involves the Gram matrix $\gram = \aop \aop^\T \in \R^{N \times N}$.
In that way, the mutual coherence structure of~$\aop$ is also taken into account, which is in turn a missing feature of traditional approaches.
% This observation defies conventional wisdom that would substantiate the superiority of an analysis operator $\aop \in \R^{N \times n}$ by the sparsity of its representations, i.e., $S = \lnorm{\aop\grtr}[0]$ is supposed to be small for a signal-of-interest $\grtr \in \R^n$.\todocontent{Eine gewisse Sparsity ist aber schon wichtig... Das sollte vielleicht noch hervorgehoben werden.}
% In contrast, our approach relies on three generalized sparsity parameters $\spparam{\grtr}$, $\cospparam{\grtr}$, and $\cospparamdg{\grtr}$ (see Definition~\ref{def:results:generalsparsity}). By computing weighted sums of the Gram matrix $\gram = \aop \aop^\T \in \R^{N \times N}$, these notion do also account for the mutual coherence structure of $\aop$, which is in fact a missing feature of traditional analysis (co-)sparsity.
This important refinement enabled us to prove much more reliable bounds on the sample complexity of the analysis basis pursuit \eqref{eq:intro:anabp}, even when a coefficient sequence $\aop\grtr \in \R^N$ is just compressible (see Subsection~\ref{subsec:results:compressible}).
Hence, we can conclude that our main results, Theorem~\ref{thm:results:recovery} and Theorem~\ref{thm:results:robuststable}, provide fairly general answers to the initial challenges of \ref{quest:intro:samplecompl} and \ref{quest:intro:compressibility}, as stated in Subsection~\ref{subsec:intro:issues}.

A more technical yet crucial role in these statements is played by the sampling-rate function $\samplecompl{\aop,\grtr}$, which essentially determines the minimal number of measurements needed to invoke the actual error estimates.
Our numerical simulations in Section~\ref{sec:appl} have demonstrated that this quantity indeed very accurately localizes the phase transition of many different types of analysis operators, such as redundant wavelet systems, total variation, or random frames.
This variety certifies once again that the proposed bounds meet the desiderata of \ref{item:intro:desiderata:accurate}--\ref{item:intro:desiderata:generic} from Subsection~\ref{subsec:intro:contrib} in many situations of interest.
Thus, let us make the following closing remark:
\begin{highlight}
	The functional $\grtr \mapsto \samplecompl{\aop,\grtr}$ can be regarded as a surrogate measure of complexity, quantifying how well $\aop$ captures the low-dimensional structure of a signal $\grtr$.
	Compared to the number of non-zero analysis coefficients, its value has turned out to be an appropriate indicator for the success or failure of $\l{1}$-analysis minimization.
\end{highlight}

\subsection{Practical Scope and Guidelines}
\label{subsec:concl:practical}

The previous sections have rather focused on the predictive power of our theoretical framework, whereas its practical implications were only marginally addressed so far.
Therefore, we shall now investigate the issues of \ref{quest:intro:interpretability} in greater detail.
Taking the perspective of a practitioner, these problem statements are typically associated with the following tasks: \emph{How to come up with a good analysis operator $\aop$ for my specific application? What properties and design rules are of particular relevance?}
Unfortunately, we are not able to answer these questions in full generality here, but our theoretical approach could at least give rise to several novel solution strategies.
The above discussion in fact suggests using $\samplecompl{\aop,\grtr}$ as a measure of quality that allows for an assessment and comparison of different choices of $\aop$.
Let us a recall a simple example from Subsection~\ref{subsubsec:appl:wavelets:blocks}: In the course of relation \eqref{eq:appl:wavelets:reweighting}, it has turned out that, although $\aopr$ and $\aopi$ correspond to the same frame except from a scale-wise reweighting of their analysis vectors, the predicted sampling rates differ dramatically, as illustrated in Figure~\ref{fig:intro:1d:pt_frame}.
This observation indicates that the recovery performance might be greatly improved by \emph{weighting}, which is a technique closely related to \emph{weighted sparsity} in compressed
sensing theory, e.g., see \cite{khajehnejad2009weighted,friedlander2012recovering,oymak2012recovery}.
The numerical results of Section~\ref{sec:appl} even show that such a rule-of-thumb could be applied to various other structural properties of analysis operators, for example:
\begin{itemize}
\item
	\emph{Intrinsic localization of the Gram matrix.} The off-diagonal decay of $\gram = \aop\aop^\T$ may strongly affect the recovery capability of $\aop$ (cf. Remark~\ref{rmk:proofs:scbound:kkt}). 
\item
	\emph{Conditioning.} The impact of the linear dependency structure and frame bounds of $\aop$ is expected to be highly non-trivial.
\item
	\emph{Redundancy.} Highly redundant frames are sometimes superior over orthonormal bases (see Figure~\ref{fig:intro:1d}).
\end{itemize}
While analyzing $\samplecompl{\aop,\grtr}$ with respect to certain characteristic quantities yields a first approach to \ref{quest:intro:interpretability}, a systematic methodology for designing good operators is still missing.

%We suspect that the viewpoint of \emph{statistical learning theory} might provide a remedy. 
A promising alternative to tackle this challenge is based on the viewpoint of \emph{statistical learning theory}.
Indeed, recent works on \emph{analysis operator learning} have shown that it is possible to infer effective analysis operators from a collection of training samples, e.g., see \cite{rubinstein2013ksvd,hawe2013,chen2014learning,ravishankar2013,yaghoobi2013learning}. In \cite{bian2016learning}, an interesting connection between the sparse null space problem and analysis dictionary learning has been established, which also sheds more light on the underlying complexity of the analysis model. All these methods have in common that the proposed learning procedure minimizes a cost function based on the $\l{1}$- or $\l{0}$-norm of the analysis coefficients. In Appendix~\ref{sec:learning}, we outline a mathematical learning formalism that suggests minimizing $\samplecompl{\cdot,\cdot}$ instead. A detailed elaboration of this topic is however deferred to future research.

\subsection{Open Problems}

Beside the emerging challenge of the previous subsection, we would like to state some further open issues which might be studied in future works as well:

\begin{itemize}
\item
	\emph{Sharper bounds and optimality.}
	Our analysis of the finite difference operator in Subsection~\ref{subsec:appl:tv} has demonstrated that the sampling-rate bound of Theorem~\ref{thm:results:recovery} could be somewhat inaccurate in situations of low sparsity.
	While there are indeed several non-sharp estimates in the proof of Theorem~\ref{thm:proofs:mainresults:scbound} (e.g., see Proposition~\ref{prop:proofs:scbound:mwpolar} and Lemma~\ref{lem:proofs:scbound:clippedcovar}), the most heuristic step is actually our choice of the dual vector $\dv'$ in \eqref{eq:proofs:scbound:dvgeneral}, see also Remark~\ref{rmk:proofs:scbound:kkt}.
	This leaves certain space for refinements, but note that, no matter how $\dv'$ is chosen, one eventually has to face a non-trivial minimization problem over $\tau$ in \eqref{eq:proofs:scbound:bound:polar}.
	Apart from that, it would be interesting to specify conditions on $\aop$ and $\grtr$ under which our bounds are already optimal, at least in an asymptotic sense.
% 	In the case of orthonormal bases, for example, we have already verified the optimality in Proposition~\ref{prop:results:scfuncsimplyfied}
% 	According to Remark~\ref{rmk:proofs:scbound:kkt}, one could regard our results as a perturbation theory for Gram matrix of $\aop$, in the sense that accurate predictions of the sample complexity are still possible when it deviates from the identity.
\item
	\emph{Algorithms.} The analysis basis pursuit is just a very popular example of methods that try to exploit analysis sparsity in signal estimation. In fact, there exist numerous alternatives to solving \eqref{eq:intro:anabp}, for instance, the \emph{greedy analysis pursuit} \cite{nam2013} or those of \cite{giryes2014}.
	We believe that our proof techniques can be adapted to many of these methods, ultimately leading to similar recovery guarantees than we have obtained for the analysis basis pursuit.
	
	At this point, it is worth mentioning that a key ingredient of our statistical analysis is a sophisticated upper bound for the conic Gaussian mean width of the $\l{1}$-analysis (semi-)norm; see Theorem~\ref{thm:proofs:mainresults:scbound}.
	Hence, our findings are also compatible with those sampling-rate bounds in the literature which employ the conic Gaussian mean width as measure of complexity.
	For example, one may easily carry over our recovery guarantees to the projected gradient descent method studied by Oymak et al.\ in \cite{oymak2018sharp}.
	% Furthermore, the recent work \cite{oymak2018sharp} establishes a convergence analysis for solving least-squares recovery problems, constrained by the sublevel set of a penalty function. When solving a suitable modification of \eqref{eq:intro:anabp}, these results would allow to characterize the rate of convergence of the projected gradient descent algorithm in terms of $\samplecompl{\aop,\grtr}$.}
\item
	\emph{Structured measurements.} Our argumentation in Section~\ref{sec:proofs} does strongly rely on the row-independence of the sensing matrix $\A \in \R^{m \times n}$.
	This assumption is unfortunately not satisfied for many types of structured measurement schemes, so that our statistical tools from Appendix~\ref{subsec:prelim:framework} do not apply anymore.
	In particular, the concept of Gaussian mean width might become inappropriate, implying that most parts of our proofs break down.
	For such situations, one probably would have to come up with a completely new proof strategy.
% 	However, we suspect that the actual problem (``Does sparsity explain everything?'') will stay the same, since the analysis sparsity does not depend on $\A$.
% \item
% 	Compressibility: Even if we have found the ``best'' subspace $\projnsp$, it is not clear that our choice of $\grtrsparse$ (a scaled version of $\proj{\projnsp} \grtr$) is optimal.
% 	There are different strategies that may work better and in general we need to solve a complicated minimization problem on the boundary of the analysis polytope.
% 	Again, our numerics show that in many cases, our strategy is good enough. Using more sophisticated methods, we may loose the interpretability of the bounds.
\item
	\emph{Analysis versus synthesis.} We have already discussed the crucial difference between analysis- and synthesis-based priors in Subsection~\ref{subsec:discussion:anavssyn}, but a fundamental question remains widely open: Which of the two formulations is more appropriate for a given signal class and measurement ensemble?
% 	Following the philosophy of our results, it could be worth investigating the sample complexity of signal (not coefficient) recovery via \eqref{eq:disc:synthesis}.
% 	This would at least allow for a theoretical comparison between both approaches, ultimately leading to practical rules-of-thumb as discussed in Subsection~\ref{subsec:concl:practical}.
\end{itemize}

\section{Proofs of Main Results}
\label{sec:proofs}

Before presenting all proofs in detail, let us briefly sketch the roadmap of this section:
In Subsection~\ref{subsec:proofs:framework}, we first provide an abstract framework, which even applies to a generalized version of the basis pursuit.
The associated recovery results rely on variants of the \emph{Gaussian mean width}, which forms a key concept in analyzing the sample complexity of many convex (signal) estimation problems, e.g., see \cite{rudelson2008sparse,chandrasekaran2012geometry,chandrasekaran2013tradeoff,plan2013robust,amelunxen2014edge,tropp2014convex,oymak2016sharpmse,vershynin2018hdp}.
These quantities are however very implicit and hard to compute in general. For that reason, we show in Subsection~\ref{subsec:proofs:mainresults} (Theorem~\ref{thm:proofs:mainresults:scbound}) that the square of the Gaussian mean width of the descent cone of $\x \mapsto \lnorm{\aop \x}[1]$ at $\grtr$ is at least upper bounded by our sampling-rate function $\samplecompl{\aop, \grtr}$, introduced in Definition~\ref{def:results:samplingratefct}.
Combining this key step with the guarantees from Subsection~\ref{subsec:proofs:framework}, we will immediately obtain the statements of Theorem~\ref{thm:results:recovery} and Theorem~\ref{thm:results:robuststable}.
It is worth emphasizing that the actual challenge of our approach is the proof of Theorem~\ref{thm:proofs:mainresults:scbound} in Subsection~\ref{subsec:proofs:scbound}. 
Apart from various technicalities, a particular endeavor is to find the right balance between accurate and explicit estimates of the Gaussian mean width.

% The more technical proof of Theorem~\ref{thm:proofs:mainresults:scbound} is then postponed to Subsection~\ref{subsec:proofs:scbound}.

% 	It will turn our that we prove more general statement at several points that leave space for extensions of our results. We shall comment one this later on.

\subsection{Structured Signal Recovery and Gaussian Mean Width}
\label{subsec:proofs:framework}

In this part, we analyze the \define{generalized basis pursuit}
\begin{equation}\label{eq:proofs:framework:bpgeneral}\tag{$\text{BP}_{\noiseparam}^{\objfunc}$}
	\min_{\x \in \R^n} \objfunc(\x) \quad \text{subject to \quad $\lnorm{\A \x - \y} \leq \noiseparam$,}
\end{equation}
\newcommand{\refbpgeneralnoiseless}[2]{(\hyperref[eq:proofs:framework:bpgeneral]{$\text{BP}_{#1}^{#2}$})}%
where $\objfunc \colon \R^n \to \R$ is a convex function. Note that for our purposes, it would be sufficient to consider the $\l{1}$-analysis functional $\objfunc(\x) = \lnorm{\aop\x}[1]$, but in fact, each of the following arguments does literally hold true for the general case.
Conceptually, the optimization problem of \eqref{eq:proofs:framework:bpgeneral} seeks for the most structured signal $\x \in \R^n$ that is consistent with the measured data $\y = \A \grtr + \noise$. 
The optimal sampling rate of this approach heavily relies on how well the objective function $\objfunc$ captures the complexity of the ground truth signal $\grtr \in \R^n$.
In order to formalize this intuition, let us first introduce the notion of (Gaussian) mean width:
\begin{definition} \label{def:proofs:framework:meanwidth}
Let $\sset \subset \R^n$ be a set.
\begin{deflist}
\item
	The \define{(global) mean width} of $\sset$ is defined as
	\begin{equation}\label{eq:proofs:framework:meanwidth:global}
		\meanwidth{\sset} \coloneqq \mean[\sup_{\h \in \sset} \sp{\gaussian}{\h}],
	\end{equation}
	where $\gaussian \distributed \Normdistr{\vnull}{\I{n}}$ is a standard Gaussian random vector.
\item
	The \define{conic mean width} of $\sset$ is given by
	\begin{equation}\label{eq:proofs:framework:meanwidth:conic}
		\meanwidth[\conic]{\sset} \coloneqq \meanwidth{\cone{\sset} \intersec \S^{n-1}}.
	\end{equation}
\item
	The \define{local mean width} of $\sset$ at scale $t > 0$ is defined as
	\begin{equation}\label{eq:proofs:framework:meanwidth:local}
		\meanwidth[t]{\sset} \coloneqq \meanwidth{\tfrac{1}{t}\sset \intersec \S^{n-1}}.
	\end{equation}
\end{deflist}
\end{definition}

%\enlargethispage{\baselineskip}
\begin{remark}\label{rmk:proofs:framework:meanwidth}
	If $\sset$ is convex and $\vnull \in \sset$, the mapping $t \mapsto \meanwidth[t]{\sset}$ is non-increasing, and we also have $\meanwidth[t]{\sset} \leq \meanwidth[\conic]{\sset}$, since $\tfrac{1}{t}\sset \subset \cone{\sset}$ for all $t > 0$.
	Moreover, it holds
	\begin{equation}
		\meanwidth[t]{\cone{\sset}} = \meanwidth[\conic]{\cone{\sset}} = \meanwidth[\conic]{\sset}.
	\end{equation}
\end{remark}

A particularly interesting choice of $\sset \subset \R^n$ is the set of all descent directions of $\objfunc$ at a certain ``anchor point'':
\begin{definition}\label{def:proofs:framework:descentset}
	Let $\objfunc \colon \R^n \to \R$ be a convex function and let $\grtr \in \R^n$.
	The \define{descent set} of $\objfunc$ at $\grtr$ is given by
	\begin{equation}
		\descset{\objfunc, \grtr} \coloneqq \{ \h \in \R^n \suchthat \objfunc(\grtr + \h) \leq \objfunc(\grtr) \},
	\end{equation}
	and its corresponding \define{descent cone} is denoted by $\desccone{\objfunc, \grtr} \coloneqq \cone{\descset{\objfunc, \grtr}}$.
	In this context, we refer to $\meanwidth[\conic]{\descset{\objfunc, \grtr}}$ as the \emph{conic mean width of $\objfunc$ at $\grtr$}.
\end{definition}

Note that, since $\objfunc$ is assumed to be convex, both $\descset{\objfunc, \grtr}$ and $\desccone{\objfunc, \grtr}$ are convex sets.
The geometric idea behind these definitions is quite simple: Every minimizer of \eqref{eq:proofs:framework:bpgeneral} is close to~$\grtr$ if, and only if, the intersection between the constraint set $\{\x \in \R^n \suchthat \lnorm{\A \x - \y} \leq \noiseparam \}$ and $\grtr + \descset{\objfunc, \grtr}$ has a small diameter.
It was therefore a remarkable observation of \cite{chandrasekaran2012geometry} that, for Gaussian measurements, the occurrence of this event is completely characterized by a single parameter, namely the conic mean width of $\objfunc$ at $\grtr$.
For this reason, $\meanwidth[\conic]{\descset{\objfunc, \grtr}}$ may be also regarded as a measure of complexity for a signal $\grtr$ with respect to a given structure-imposing functional $\objfunc$.

Our first theoretical guarantee actually extends this fundamental principle by incorporating the notion of \emph{local mean width}. %(cf. \cite[Cor.~3.3]{chandrasekaran2012geometry} and \cite[Cor.~3.5]{tropp2014convex})
We will see below that this refinement indeed allows us to establish stable recovery, which is an essential feature of Theorem~\ref{thm:results:robuststable}.
\begin{theorem}\label{thm:proofs:framework:guarantee}
	Let Model~\ref{model:results:setup:meas} be satisfied and let $\objfunc \colon \R^n \to \R$ be a convex function.
	There exists a numerical constant $C > 0$ such that, for every fixed $t > 0$ and $\probsuccess > 0$, the following holds true with probability at least $1 - e^{-\probsuccess^2/2}$:
	Assumed that 
	\begin{equation}\label{eq:proofs:framework:guarantee:cond}
		t \geq \frac{2\noiseparam}{\pospart{\sqrt{m-1} - C \cdot \subgparam^2 \cdot \big(\meanwidth[t]{\descset{\objfunc, \grtr}} + \probsuccess\big)}} \ ,
	\end{equation}
	any minimizer $\solu$ of \eqref{eq:proofs:framework:bpgeneral} satisfies $\lnorm{\solu - \grtr} \leq t$.
	In the Gaussian case, $\a \distributed \Normdistr{\vnull}{\I{n}}$, we have $C = \subgparam = 1$.
\end{theorem}
A proof of Theorem~\ref{thm:proofs:framework:guarantee} and all its corollaries below can be found in Appendix~\ref{subsec:prelim:framework}.
Roughly speaking, the error bound of Theorem~\ref{thm:proofs:framework:guarantee} states the following: To achieve a desired reconstruction accuracy $t > 0$ via \eqref{eq:proofs:framework:bpgeneral}, increase $m$ until \eqref{eq:proofs:framework:guarantee:cond} is fulfilled.
% In fact, this procedure leads to a constraint on the required number of samples:
% \begin{equation}\label{eq:proofs:framework:guarantee:meas}
% 	m > \Big( C \cdot \subgparam^2 \cdot (\meanwidth[t]{\descset{\objfunc, \grtr}} + \probsuccess) \Big)^2 + 1.
% \end{equation}
While this procedure always succeeds for sufficiently large values of $m$, the condition of \eqref{eq:proofs:framework:guarantee:cond} is still quite implicit, since both sides depend on $t$.
Fortunately, due to Remark~\ref{rmk:proofs:framework:meanwidth}, one may just estimate the local mean width by its conic counterpart.
This leads to our next result, which is well-known from the literature, e.g., see \cite[Cor.~3.5]{tropp2014convex}.
\begin{corollary}\label{cor:proofs:framework:guarantee-conic}
	Let Model~\ref{model:results:setup:meas} be satisfied and let $\objfunc \colon \R^n \to \R$ be convex.
	There exists a numerical constant $C > 0$ (the same as in Theorem~\ref{thm:proofs:framework:guarantee}) such that, for every $\probsuccess > 0$, the following holds true with probability at least $1 - e^{-\probsuccess^2/2}$:
	Any minimizer $\solu$ of \eqref{eq:proofs:framework:bpgeneral} satisfies
	\begin{equation}\label{eq:proofs:framework:guarantee-conic:bound}
		\lnorm{\solu - \grtr} \leq \frac{2\noiseparam}{\pospart{\sqrt{m-1} - C \cdot \subgparam^2 \cdot \big(\meanwidth[\conic]{\descset{\objfunc, \grtr}} + \probsuccess\big)}} \ .
	\end{equation}
	In particular, if $\noiseparam = 0$ and
	\begin{equation}\label{eq:proofs:framework:guarantee-conic:meas}
		m > C^2 \cdot \subgparam^4 \cdot \big(\meanwidth[\conic]{\descset{\objfunc, \grtr}} + \probsuccess\big)^2 + 1,
	\end{equation}
	one has $\solu = \grtr$.
\end{corollary}
As already briefly discussed at the beginning of Subsection~\ref{subsec:results:samplecompl}, the condition of \eqref{eq:proofs:framework:guarantee-conic:meas} is not only sufficient but essentially also necessary to achieve exact recovery in the Gaussian case \cite{amelunxen2014edge}.\footnote{The authors of \cite{amelunxen2014edge} are focusing on the so-called \emph{statistical dimension} rather than on the Gaussian mean width, but these notions are actually equivalent (see \cite[Prop.~10.2]{amelunxen2014edge}).}
Hence, the conic mean width is the crucial quantity to study when determining the sample complexity of \refbpgeneralnoiseless{\noiseparam=0}{\objfunc}.

While Corollary~\ref{cor:proofs:framework:guarantee-conic} already includes robustness against noise, the wish for stable reconstructions via \eqref{eq:proofs:framework:bpgeneral}, as requested in Subsection~\ref{subsec:results:compressible}, remains unfulfilled.
A major difficulty is that the mapping $\grtr \mapsto \meanwidth[\conic]{\descset{\objfunc, \grtr}}$ is typically discontinuous, since it involves forming the conic hull of the descent set $\descset{\objfunc, \grtr}$.
If $\meanwidth[\conic]{\descset{\objfunc, \grtr}} \approx \sqrt{n}$, the statement of Corollary~\ref{cor:proofs:framework:guarantee-conic} looses its significance, even when there exists a vector close to $\grtr$ whose conic mean width is much smaller.
For example, such a situation might occur if $\grtr$ just possesses a compressible coefficient sequence which is not too sparse; see also Subsection~\ref{subsec:results:compressible} as well as \cite[Sec.~III.D]{genzel2017highdim} for more details.
Our second corollary of Theorem~\ref{thm:proofs:framework:guarantee} resolves this issue by analyzing the local mean width at a specific scale $t > 0$.
It shows that one may investigate the complexity of $\objfunc$ with respect to a \emph{surrogate vector} $\grtrsparse \in \R^n$ by sacrificing estimation accuracy in the order of the approximation error $\lnorm{\grtr - \grtrsparse}$.
% \begin{draft}
% \item
% 	An appropriate choice of $t$ in Theorem~\ref{thm:proofs:framework:guarantee} does indeed allows us to derive such a result.
% 	The key task is to understand $t \mapsto \meanwidth[t]{\descset{\objfunc, \grtr}}$.
% 	Since this map is non-increasing, we may sacrifice a bit of accuracy (by choosing $t$ larger) but in turn require (much less) measurements (see \eqref{eq:proofs:framework:guarantee:meas}).
% 	We shall see now that this idea indeed allows us to show stable recovery for signals with possibly large complexity, which are sufficiently close to signals of lower complexity.
% \end{draft}
\begin{corollary}\label{cor:proofs:framework:guarantee-stable}
	Let Model~\ref{model:results:setup:meas} be satisfied and let $\objfunc \colon \R^n \to \R$ be convex. Moreover, fix a vector $\grtrsparse \in \R^n$ with $\objfunc(\grtrsparse) = \objfunc(\grtr)$.
	There exists a numerical constant $C > 0$ (the same as in Theorem~\ref{thm:proofs:framework:guarantee}) such that, for every $R > 0$ and $\probsuccess > 0$, the following holds true with probability at least $1 - e^{-\probsuccess^2/2}$:
	Any minimizer $\solu$ of \eqref{eq:proofs:framework:bpgeneral} satisfies
	\begin{equation}\label{eq:proofs:framework:guarantee-stable:bound}
		\lnorm{\solu - \grtr} \leq \max\Big\{ R \lnorm{\grtr - \grtrsparse}, \tfrac{2\noiseparam}{\pospart{\sqrt{m-1} - \sqrt{\tilde{m}_0-1}}} \Big\},
	\end{equation}
	where
	\begin{equation}\label{eq:proofs:framework:guarantee-stable:meas}
		\tilde{m}_0 \coloneqq C^2 \cdot \subgparam^4 \cdot \Big(\tfrac{R + 1}{R} \cdot \big[\meanwidth[\conic]{\descset{\objfunc, \grtrsparse}} + 1\big] + \probsuccess \Big)^2 + 1.
	\end{equation}
	In particular, if $\noiseparam = 0$, we have $\lnorm{\solu - \grtr} \leq R \lnorm{\grtr - \grtrsparse}$ provided that $m > \tilde{m}_0$.
\end{corollary}
In a nutshell, Corollary~\ref{cor:proofs:framework:guarantee-stable} certifies a good performance of \eqref{eq:proofs:framework:bpgeneral} if both $\grtrsparse$ approximates $\grtr$ sufficiently well and $\meanwidth[\conic]{\descset{\objfunc, \grtrsparse}}$ is small.
For a more detailed discussion of this surrogate principle, we again refer to Subsection~\ref{subsec:results:compressible}, which considers the special case of $\objfunc(\cdot) = \lnorm{\aop(\cdot)}[1]$.
In this course, we do also demonstrate a general selection strategy for $\grtrsparse$: If there exists a subspace $\projnsp \subset \R^n$ which consists of low-complexity vectors according to $\x \mapsto \meanwidth[\conic]{\descset{\objfunc, \x}}$, it is natural to set $\grtrsparse$ as the renormalized orthogonal projection of $\grtr$ onto $\projnsp$, such that $\objfunc(\grtrsparse) = \objfunc(\grtr)$.
If $\objfunc$ is a semi-norm, one may even invoke the following upper bound on the approximation error, leading to a more explicit expression in \eqref{eq:proofs:framework:guarantee-stable:bound}:
\begin{proposition}\label{prop:proofs:framework:compressibility:normbound}
	Assume that $\objfunc = \norm{\cdot}$ is a semi-norm on $\R^n$.
	Let $\grtr \in \R^n$ and let $\projnsp \subset \R^n$ be a linear subspace such that $\norm{\proj{\projnsp}\grtr} \neq 0$.
	Setting $\grtrsparse \coloneqq \tfrac{\norm{\grtr}}{\norm{\proj{\projnsp}\grtr}} \cdot \proj{\projnsp} \grtr$, we have $\norm{\grtr} = \norm{\grtrsparse}$ and
	\begin{align}
		\lnorm{\grtr - \grtrsparse} %&\leq \Big( \tfrac{\lnorm{\proj{\projnsp}\grtr}^2}{\norm{\proj{\projnsp}\grtr}^2} \cdot \norm{\proj{\orthcompl{\projnsp}}\grtr}^2 + \lnorm{\proj{\orthcompl{\projnsp}}\grtr}^2 \Big)^{1 / 2} \\
		&\leq \tfrac{\lnorm{\proj{\projnsp}\grtr}}{\norm{\proj{\projnsp}\grtr}} \cdot \norm{\proj{\orthcompl{\projnsp}}\grtr} + \lnorm{\proj{\orthcompl{\projnsp}}\grtr}.
	\end{align}
\end{proposition}
\begin{proof}
	Since $\grtr = \proj{\projnsp}\grtr + \proj{\orthcompl{\projnsp}}\grtr$, we have
	\begin{align}
		\lnorm{\grtr - \grtrsparse}^2 &= \lnorm[\Big]{\Big( 1 - \tfrac{\norm{\grtr}}{\norm{\proj{\projnsp}\grtr}} \Big) \proj{\projnsp}\grtr + \proj{\orthcompl{\projnsp}}\grtr}^2 \\
		&= \abs[\Big]{1 - \tfrac{\norm{\grtr}}{\norm{\proj{\projnsp}\grtr}}}^2 \cdot \lnorm{ \proj{\projnsp}\grtr }^2 + \lnorm{\proj{\orthcompl{\projnsp}}\grtr}^2 \\
		&= \abs[\Big]{\tfrac{\norm{\proj{\projnsp}\grtr} - \norm{ \proj{\projnsp}\grtr + \proj{\orthcompl{\projnsp}}\grtr}}{\norm{\proj{\projnsp}\grtr}}}^2 \cdot \lnorm{ \proj{\projnsp}\grtr }^2 + \lnorm{\proj{\orthcompl{\projnsp}}\grtr}^2 \\
		&\leq \abs[\Big]{\tfrac{\norm{\proj{\orthcompl{\projnsp}}\grtr}}{\norm{\proj{\projnsp}\grtr}}}^2 \cdot \lnorm{ \proj{\projnsp}\grtr }^2 + \lnorm{\proj{\orthcompl{\projnsp}}\grtr}^2 \\
		&\leq \Big( \tfrac{\norm{\proj{\orthcompl{\projnsp}}\grtr}}{\norm{\proj{\projnsp}\grtr}} \cdot \lnorm{ \proj{\projnsp}\grtr } + \lnorm{\proj{\orthcompl{\projnsp}}\grtr} \Big)^2.\qedhere
	\end{align}
\end{proof}

\subsection{Proofs of Theorem~\ref{thm:results:recovery} and Theorem~\ref{thm:results:robuststable}}
\label{subsec:proofs:mainresults}

Let us now focus on the $\l{1}$-analysis \mbox{(semi-)norm} $\objfunc(\cdot) = \lnorm{\aop(\cdot)}[1]$.
The abstract results of the previous part indicate that the conic mean width $\meanwidth[\conic]{\descset{\lnorm{\aop(\cdot)}[1], \grtr}}$ plays a key role in analyzing \eqref{eq:intro:anabp}. While this quantity may indeed yield an accurate bound on the sampling rate, our desiderata \ref{item:intro:desiderata:computable}--\ref{item:intro:desiderata:generic} from Subsection~\ref{subsec:intro:contrib} are still hardly met, since the definition of the mean width is quite abstract and non-informative.
Therefore, we need to come up with a more explicit expression or at least a meaningful upper bound on $\meanwidth[\conic]{\descset{\lnorm{\aop(\cdot)}[1], \grtr}}$.
This is precisely the purpose of the following theorem, which forms the main technical ingredient of our approach. Its proof is postponed to Subsection~\ref{subsec:proofs:scbound}.
\begin{theorem}\label{thm:proofs:mainresults:scbound}
	Let $\aop \in \R^{N \times n}$ be an analysis operator and let $\grtr \in \R^n$ with $\grtr \not\in \ker\aop$.
	Then
	\begin{equation}
		\effdim[\conic]{\descset{\lnorm{\aop(\cdot)}[1], \grtr}} \leq \samplecompl{\aop, \grtr}.
	\end{equation}
\end{theorem}

Our two main results are now a direct consequence of this bound and the guarantees from Subsection~\ref{subsec:proofs:framework}:
\begin{proof}[Proof of Theorem~\ref{thm:results:recovery}]
	By Theorem~\ref{thm:proofs:mainresults:scbound} and the condition \eqref{eq:results:recovery:meas}, we obtain
	\begin{equation}
		m > \Big( \sqrt{\samplecompl{\aop, \grtr}} + \probsuccess \Big)^2 + 1 \geq C^2 \cdot \subgparam^4 \cdot \Big( \meanwidth[\conic]{\descset{\lnorm{\aop(\cdot)}[1], \grtr}} + \probsuccess \Big)^2 + 1,
	\end{equation}
	implying that in the ``in particular''-part of Corollary~\ref{cor:proofs:framework:guarantee-conic} can be applied with $\objfunc(\cdot) = \lnorm{\aop(\cdot)}[1]$.
	Note that $C^2 \cdot \subgparam^4 = 1$ here, due to the Gaussianity of $\A$. This already yields the claim of Theorem~\ref{thm:results:recovery}.
\end{proof}

\begin{proof}[Proof of Theorem~\ref{thm:results:robuststable}]
	By the definition of $\grtrsparse \in \R^n$ in \eqref{eq:results:robuststable:grtrsparse}, we have $\grtrsparse \not\in \ker\aop$ and therefore Theorem~\ref{thm:proofs:mainresults:scbound} gives
	\begin{equation}
		\effdim[\conic]{\descset{\lnorm{\aop(\cdot)}[1], \grtrsparse}} \leq \samplecompl{\aop, \grtrsparse}.
	\end{equation}
	Moreover, we observe that $\lnorm{\aop\grtr}[1] = \lnorm{\aop\grtrsparse}[1]$. Hence, Corollary~\ref{cor:proofs:framework:guarantee-stable} is applicable with $\objfunc(\cdot) = \lnorm{\aop(\cdot)}[1]$ and we particularly have $\tilde{m}_0 \leq m_0 < m$; note that the constants $C$ in the definitions of $\tilde{m}_0$ (in \eqref{eq:proofs:framework:guarantee-stable:meas}) and $m_0$ (in \eqref{eq:results:robuststable:meas}) just differ in taking the square.
	The claim now follows from the error bound of \eqref{eq:proofs:framework:guarantee-stable:bound}.
\end{proof}

\subsection{Proof of Theorem~\ref{thm:proofs:mainresults:scbound} (Bounding the Conic Mean Width)}
\label{subsec:proofs:scbound}

To complete the proof, it remains to verify Theorem~\ref{thm:proofs:mainresults:scbound}.
For this, we loosely follow \cite[Recipe~4.1]{amelunxen2014edge}, but note again that we do not attempt to exactly compute $\meanwidth[\conic]{\descset{\lnorm{\aop(\cdot)}[1], \grtr}}$, but rather establish a sophisticated upper bound.

\subsubsection*{Step 1: A polar bound on the conic mean width}

% 	The proof loosely follows the concept of \cite[Recipe~4.1]{amelunxen2014edge} and is postponed to Subsection~\ref{subsec:proofs:scbound}.
% 	It involves several highly non-trivial estimates on the expectation of non-linearly deformed Gaussian random variables.

% 	Interestingly, we even proof a more general statement (in Theorem~\ref{thm:proofs:scbound:generalbound}) that leaves space for extensions and refinements of our approach.

Our first proof step is based on a well-known polarity argument, providing an upper bound on the mean width of a descent cone.
This result even holds true in the very general setup Subsection~\ref{subsec:proofs:framework}.
\begin{proposition}[\protect{\cite[Prop.~4.1]{amelunxen2014edge}}]\label{prop:proofs:scbound:mwpolar}
Let $\objfunc \colon \R^n \to \R$ be convex and let $\grtr \in \R^n$.
Then, the \define{subdifferential} of $\objfunc$ at $\grtr$, given by
\begin{equation}\label{eq:proofs:scbound:defsubd}
	\subd{\objfunc}(\grtr) \coloneqq \{ \vec{z} \in \R^n \suchthat \objfunc(\x) \geq \objfunc(\grtr) + \sp{\vec{z}}{\x - \grtr} \text{ for all $\x \in \R^n$} \},
\end{equation}
is well-defined.
If $\subd{\objfunc}(\grtr)$ is non-empty, compact, and does not contain the origin, we have\footnote{Hereafter, the expectation is always taken with respect to the Gaussian random vector $\gaussian$.}
\begin{equation}\label{eq:proofs:scbound:mwpolar:bound}
% 	\effdim[\conic]{\descset{\objfunc, \grtr}} \leq \inf_{\tau > 0} \mean[\distsq{\gaussian}{\tau \cdot \subd{\objfunc}(\grtr)}][\Big], 
	\effdim[\conic]{\descset{\objfunc, \grtr}} \leq \inf_{\tau > 0} \mean[\inf_{\vec{z} \in \subd{\objfunc}(\grtr)} \lnorm{\gaussian - \tau \cdot \vec{z}}^2][\Big], 
\end{equation}
where $\gaussian \distributed \Normdistr{\vnull}{\I{n}}$.
\end{proposition}
\begin{remark}
For several special cases, such as norms \cite[Thm.~4.3]{amelunxen2014edge} or total variation \cite{zhang2016precise}, it is known that \eqref{eq:proofs:scbound:mwpolar:bound} provides a very accurate bound.
For a more detailed discussion of the sharpness of Proposition~\ref{prop:proofs:scbound:mwpolar}, we refer to \cite[Sec.~4.2]{amelunxen2014edge}.
\end{remark}

In order to invoke Proposition~\ref{prop:proofs:scbound:mwpolar}, we need to specify the subdifferential of $\objfunc(\cdot) = \lnorm{\aop(\cdot)}[1]$ at $\grtr$.
For the sake of brevity, we now just write $\ssupp \coloneqq \ssupp[\grtr] = \supp(\aop\grtr)$, $\ssuppc \coloneqq \ssuppc[\grtr]$ as well as $\anasign{} \coloneqq \anasign{\grtr} = \sign(\aop \grtr)$.
\begin{lemma}\label{lem:proofs:scbound:subd}
	We have
	\begin{equation}\label{eq:proofs:scbound:subdanalysisnorm}
		\subd{\lnorm{\aop (\cdot)}[1]}(\grtr) = \aop^\T \anasign{}  + \Big\{  \aop^\T \dv \suchthat \dv \in \R^N, \supp(\dv) \subset \ssuppc, \lnorm{\dv}[\infty] \leq 1 \Big\}.
	\end{equation}
	In particular, $\emptyset \neq \subd{\lnorm{\aop (\cdot)}[1]}(\grtr) \subset \aop^\T \intvcl{-1}{1}^N$ and the following equivalences hold true:
	\begin{equation}\label{eq:proofs:scbound:subd:equiv}
		\vnull \not\in \subd{\lnorm{\aop (\cdot)}[1]}(\grtr) \quad \iff \quad \grtr \not\in \ker\aop \quad \iff \quad \spparam{} = \lnorm{\aop^\T \anasign{}}^2 > 0.
	\end{equation}
\end{lemma}
\begin{proof}
	First, we apply the chain rule for subdifferentials (see \cite[Thm.~23.9]{rockafellar2015convex})
	\begin{equation}
		\subd{\lnorm{\aop (\cdot)}[1]}(\grtr) =  \aop^\T \left(\subd{\lnorm{\cdot}[1]}(\aop \grtr)\right),
	\end{equation}
	and together with
	\begin{equation}
	\subd{\lnorm{\cdot}[1]}(\aop \grtr) = \Big\{\vec{z}  \in \R^N \suchthat \vec{z}_{\ssupp} = \underbrace{\sign (\aop \grtr)}_{= \anasign{}}, \lnorm{\vec{z}_{\ssuppc}}[\infty] \leq 1  \Big\},
	\end{equation}
	we obtain that
	\begin{equation}
	\subd{\lnorm{\aop (\cdot)}[1]}(\grtr) = \aop^\T \anasign{}  + \Big\{  \aop^\T \dv \suchthat \dv \in \R^N, \supp(\dv) \subset \ssuppc, \lnorm{\dv}[\infty] \leq 1 \Big\}.
	\end{equation}
	
	To verify the ``in particular''-part, suppose that $\vnull \in \subd{\lnorm{\aop (\cdot)}[1]}(\grtr)$. 
	Then there exists a $\dv \in \R^N$ with $\supp(\dv) \subset \ssuppc$ and $\lnorm{\dv}[\infty] \leq 1$ such that
	$\vnull = \aop^\T \anasign{} + \aop^\T \dv$. Thus,
	\begin{equation}\label{eq:proofs:scbound:subd:l1anazero}
		0 = \sp{\grtr}{\aop^\T \anasign{} + \aop^\T \dv} = \underbrace{\sp{\aop\grtr}{\anasign{}}}_{ = \lnorm{\aop\grtr}[1]} + \underbrace{\sp{\aop\grtr}{\dv}}_{ = 0} = \lnorm{\aop\grtr}[1],
	\end{equation}
	which means $\grtr \in \ker\aop$. On the other hand, if $\grtr \in \ker\aop$, we have $\anasign{} = \sign(\aop\grtr) = \vnull$ and therefore $\aop^\T \anasign{} = \vnull$. Selecting $\dv \coloneqq \vnull$ shows that $\vnull \in \subd{\lnorm{\aop (\cdot)}[1]}(\grtr)$.
	This proves the first equivalence.
	
	By the definition of $\spparam{}$ (see Definition~\ref{def:results:generalsparsity}), it holds 
	\begin{equation}
		\spparam{} = \sp{\anasign{}}{\gram\anasign{}} = \sp{\anasign{}}{\aop\aop^\T\anasign{}} = \sp{\aop^\T\anasign{}}{\aop^\T\anasign{}} = \lnorm{\aop^\T \anasign{}}^2.
	\end{equation}
	Thus, the second equivalence of \eqref{eq:proofs:scbound:subd:equiv} follows from
	\begin{equation}
		\spparam{} = \lnorm{\aop^\T \anasign{}}^2 = 0 \quad \iff \quad \aop^\T \anasign{} = \vnull \quad \iff \quad \anasign{} = \sign(\aop\grtr) = \vnull \quad \iff \quad \aop\grtr = \vnull,
	\end{equation}
	where we have again made use of \eqref{eq:proofs:scbound:subd:l1anazero} with $\dv = \vnull$. % in the second last equivalence.
\end{proof}

In conclusion, if $\grtr \not\in \ker\aop$, Proposition~\ref{prop:proofs:scbound:mwpolar} and Lemma~\ref{lem:proofs:scbound:subd} imply that
\begin{align}
	\effdim[\conic]{\descset{\lnorm{\aop (\cdot)}[1], \grtr}} &\leq \inf_{\tau > 0} \mean{}\Big[ \inf_{\substack{\dv\in \R^N,\\ \dv_{\ssupp} = \vnull, \lnorm{\dv_{\ssuppc}}[\infty] \leq 1}} \lnorm{\gaussian - \tau \aop^\T \anasign{} - \tau \aop^\T \dv }[2]^2\Big]. \label{eq:proofs:scbound:bound:polar}
\end{align}
	
\subsubsection*{Step 2: Selecting the $\l{1}$-dual vector $\dv$.}

Evaluating the right-hand side of \eqref{eq:proofs:scbound:bound:polar} particularly asks us to study the optimization problem inside of the expected value.
It can be reformulated as a \emph{least-squares problem with box-con\-straint}:
\begin{equation}\label{eq:proofs:scbound:boxconstraint}
	\min_{\dv \in \R^N} \lnorm{\gaussian - \tau \aop^\T \anasign{} - \tau \aop^\T \dv}[2]^2 \qquad \text{subject to} \qquad \substack{\displaystyle\text{$w_k = 0$ for $k \in \ssupp$} \\ \displaystyle\text{and $\abs{w_k} \leq 1$ for $k \in \ssuppc$.}}
\end{equation}
Unfortunately, \eqref{eq:proofs:scbound:boxconstraint} does not seem to have a closed-form minimizer in general.
Hence, instead of seeking for an exact solution, we rather consider an ansatz vector $\dv' = (w_1', \dots, w_N') \in \R^N$ of the form
\begin{equation}\label{eq:proofs:scbound:dvgeneral}
	w_k' \coloneqq \begin{cases}
		0, & k \in \ssupp, \\
		\clip{\tau^{-1} \dvparam_k \sp{\avec_k}{\gaussian}}{1}, & k \in \ssuppc,
		 \end{cases}
\end{equation}
where $\dvparam_k > 0$ are flexible tuning parameters that will be specified later on.
Note that $\dv' \in \intvcl{-1}{1}^N$ depends on both $\gaussian$ and $\tau$.
In Remark~\ref{rmk:proofs:scbound:kkt} below, we will also state the KKT-conditions of \eqref{eq:proofs:scbound:boxconstraint}, showing that the choice of \eqref{eq:proofs:scbound:dvgeneral} is at least plausible (although not always optimal) if the Gram matrix of $\aop$ is well-localized. %, i.e., $\gram = \aop \aop^\T$ is almost diagonal.

Since $w_k'$ is a symmetric transformation of $\sp{\avec_k}{\gaussian} \distributed \Normdistr{0}{\lnorm{\avec_k}^2}$ for every $k \in \ssuppc$, we observe that $\mean[\dv'] = \vnull$.
Now, we can proceed with \eqref{eq:proofs:scbound:bound:polar} as follows:
\begin{align}
	&\mean{}\Big[ \inf_{\substack{\dv\in \R^N,\\ \dv_{\ssupp} = \vnull, \lnorm{\dv_{\ssuppc}}[\infty] \leq 1}} \lnorm{\gaussian - \tau \aop^\T \anasign{} - \tau \aop^\T \dv }[2]^2\Big] \\
	\leq{} & \mean[{\lnorm{\gaussian - \tau \aop^\T \anasign{} - \tau \aop^\T \dv' }[2]^2}][\Big] \\
	={} & \mean[{\lnorm{\gaussian - \tau \aop^\T \dv' }^2}][\Big] + \mean{}\Big[\lnorm{\tau \aop^\T \anasign{}}[2]^2\Big] - \mean[2 \sp{\gaussian}{\tau \aop^\T \anasign{}}] + \mean[2\sp{\tau \aop^\T \dv'}{\tau \aop^\T \anasign{}}] \\
	={} & \mean[{\lnorm{\gaussian - \tau \aop^\T \dv' }^2}][\Big] + \tau^2 \underbrace{\lnorm{\aop^\T \anasign{}}[2]^2}_{\stackrel{\eqref{eq:proofs:scbound:subd:equiv}}{=} \spparam{}} -2 \sp{\underbrace{\mean[\gaussian]}_{= \vnull}}{\tau \aop^\T \anasign{}} + 2 \sp{\tau \aop^\T \underbrace{\mean[\dv']}_{= \vnull}}{\tau \aop^\T \anasign{}} \\
	={} & \mean[\lnorm{\gaussian}^2][\Big] - 2 \tau \mean[\sp{\gaussian}{\aop^\T \dv'}] + \tau^2 \mean[{\lnorm{\aop^\T \dv'}^2}][\Big] + \tau^2 \cdot \spparam{} \\
	={} & n - \underbrace{2 \tau \mean[\sp{\aop\gaussian}{\dv'}]}_{\eqqcolon T_1} + \underbrace{\tau^2 \mean[\sp{\dv'}{\gram\dv'}]}_{\eqqcolon T_2} + \tau^2 \spparam{} = n - T_1 + T_2 + \tau^2 \cdot \spparam{}. \label{eq:proofs:scbound:bound:dvgeneral}
\end{align}
Hence, it remains to resolve the expected values of $T_1$ and $T_2$.

\subsubsection*{Step 3: A general bound on the mean width.}

Let us start with the term $T_1$:
Using the definition of $\dv'$, we obtain
\begin{align}
	T_1 = 2 \tau \mean[\sp{\aop\gaussian}{\dv'}] &= 2\tau \sum_{k \in \ssuppc} \mean[\sp{\avec_k}{\gaussian} \cdot w_k' ] \\
	&= 2\tau \sum_{k \in \ssuppc} \mean[\sp{\avec_k}{\gaussian} \cdot \sign(\sp{\avec_k}{\gaussian}) \cdot \min\{\tau^{-1} \dvparam_k \abs{\sp{\avec_k}{\gaussian}}, 1 \} ][\Big] \\
	&= 2 \sum_{k \in \ssuppc} \mean[\abs{\sp{\avec_k}{\gaussian}} \cdot \min\{\dvparam_k \abs{\sp{\avec_k}{\gaussian}}, \tau \} ][\Big],
\end{align}
and since $\sp{\avec_k}{\gaussian} \distributed \Normdistr{0}{g_{k,k}}$, an elementary calculation shows that
\begin{equation}
	T_1 = 2 \sum_{k \in \ssuppc} \dvparam_k \cdot g_{k,k} \cdot \erf\big(\tfrac{\tau}{\dvparam_k \sqrt{2g_{k,k}}}\big).
\end{equation}

Handling the term $T_2$ turns out to be more complicated:
% Since $\sqrt{g_{k,k}} = \lnorm{\avec_k}$, we can bring $T_2$ into a ``normalized'' \mm{explanation?} form: 
Since $\sqrt{g_{k,k}} = \lnorm{\avec_k}$, we first rewrite~$T_2$ such that the clipped variables are standard Gaussians: 
\begin{align}
	T_2 &= \tau^2 \sum_{k,k' \in \ssuppc} g_{k, k'} \cdot \mean[w_{k}' \cdot w_{k'}'] = \sum_{k,k' \in \ssuppc} g_{k, k'} \cdot \mean[\clip{\dvparam_k \sp{\avec_k}{\gaussian}}{\tau} \cdot \clip{\dvparam_{k'} \sp{\avec_{k'}}{\gaussian}}{\tau}][\Big] \\
	&= \sum_{k,k' \in \ssuppc} g_{k, k'} \cdot \dvparam_k \cdot \dvparam_{k'} \cdot \sqrt{g_{k,k} \cdot g_{k',k'}}  \cdot \mean[{\clip[auto]{\sp{\tfrac{\avec_k}{\lnorm{\avec_k}}}{\gaussian}}{\tfrac{\tau}{\dvparam_k \sqrt{g_{k,k}}}} \cdot \clip[\Big]{\sp{\tfrac{\avec_{k'}}{\lnorm{\avec_{k'}}}}{\gaussian}}{\tfrac{\tau}{\dvparam_{k'} \sqrt{g_{k',k'}}}}}][\Big]. \\ \label{eq:proofs:scbound:T2}
\end{align}
In general, there does not exist a closed-form expression for the covariance between clipped Gaussian variables.
However, we may invoke the following bound, which is based on Price's Theorem and proven in Appendix~\ref{subsec:prelim:clippedcovar}:
\begin{lemma}\label{lem:proofs:scbound:clippedcovar}
	Let $\v_1, \v_2 \in \S^{n-1}$ and let $\gaussian \distributed \Normdistr{\vnull}{\I{n}}$.
	Fix numbers $\beta_1, \beta_2 \geq 0$ and set $\beta^{\min} \coloneqq \min \{\beta_1, \beta_2\}$, $\beta^{\max} \coloneqq \max \{\beta_1, \beta_2\}$.
	Then, we have
	\begin{align}
		&\abs[\Big]{\mean[\clip{\sp{\v_1}{\gaussian}}{\beta_1} \cdot \clip{\sp{\v_2}{\gaussian}}{\beta_2}][\Big]} \\
		\leq{} & \abs{\sp{\v_1}{\v_2}} \cdot \Big[ \erf\big(\tfrac{\beta^{\min}}{\sqrt{2}}\big) + \Big( 1 - \erf\big(\tfrac{\beta^{\max}}{\sqrt{2}}\big) - \sqrt{\tfrac{2}{\pi}} \tfrac{e^{-(\beta^{\max})^2 / 2}}{\beta^{\max}} \Big) \cdot \beta^{\min} \cdot \beta^{\max}\Big] \\
		={} & \abs{\sp{\v_1}{\v_2}} \cdot \Big[ \erf\big(\tfrac{\beta^{\min}}{\sqrt{2}}\big) - h(\beta^{\max}) \cdot \beta^{\min} \cdot \beta^{\max} \Big], \label{eq:proofs:scbound:clippedcovar:bound}
	\end{align}
	where $h(\tau) = \sqrt{\tfrac{2}{\pi}} \tfrac{e^{-\tau^2/2}}{\tau} + \erf(\tfrac{\tau}{\sqrt{2}}) - 1$; see \eqref{eq:results:scfunc:h}. If $\v_1 = \v_2$, \eqref{eq:proofs:scbound:clippedcovar:bound} holds true with equality.
\end{lemma}
Let us apply Lemma~\ref{lem:proofs:scbound:clippedcovar} to equation \eqref{eq:proofs:scbound:T2} by setting
\begin{alignat}{2}
	&\v_1 \coloneqq \tfrac{\avec_k}{\lnorm{\avec_k}} \ , \quad && \v_2 \coloneqq \tfrac{\avec_{k'}}{\lnorm{\avec_{k'}}} \ , \quad \text{and} \\
	&\beta_1 \coloneqq \tfrac{\tau}{\dvparam_k \sqrt{g_{k,k}}} \ , \quad && \beta_2 \coloneqq \tfrac{\tau}{\dvparam_{k'} \sqrt{g_{k',k'}}} \ .
\end{alignat}
Observing that $\tfrac{g_{k,k'}}{\sqrt{g_{k,k} \cdot g_{k',k'}}} = \sp{\v_1}{\v_2}$, we end up with
\begin{align}
	T_2 \leq \sum_{k,k' \in \ssuppc} g_{k, k'}^2 \cdot \dvparam_k \cdot \dvparam_{k'} \cdot \Big( \erf\big(\tfrac{\tau\beta_{k,k'}^{\min}}{\sqrt{2}}\big) - h(\tau\beta_{k,k'}^{\max}) \cdot \tau^2 \cdot \beta_{k,k'}^{\min} \cdot \beta_{k,k'}^{\max} \Big),
\end{align}
where
\begin{equation}\label{eq:proofs:scbound:bound:betaminmax}
	\beta_{k,k'}^{\min} \coloneqq \min\Big\{ \tfrac{1}{\dvparam_{k} \sqrt{g_{k,k}}}, \tfrac{1}{\dvparam_{k'} \sqrt{g_{k',k'}}} \Big\} \quad \text{and} \quad \beta_{k,k'}^{\max} \coloneqq \max\Big\{ \tfrac{1}{\dvparam_{k} \sqrt{g_{k,k}}}, \tfrac{1}{\dvparam_{k'} \sqrt{g_{k',k'}}} \Big\}.
\end{equation}

Combining our findings for $T_1$ and $T_2$ with the estimates of \eqref{eq:proofs:scbound:bound:polar} and \eqref{eq:proofs:scbound:bound:dvgeneral} yields the following interim result:
\begin{proposition}\label{prop:proofs:scbound:generalbound}
	Let $\aop \in \R^{N \times n}$ be an analysis operator and let $\grtr \in \R^n$ with $\grtr \not\in \ker\aop$.
	For fixed numbers $\dvparam_k > 0$, $k \in \ssuppc$, we have
	\begin{multline}
		\effdim[\conic]{\descset{\lnorm{\aop(\cdot)}[1], \grtr}} \leq \inf_{\tau > 0} \Big\{  n + \tau^2 \cdot \spparam{} - 2 \sum_{k \in \ssuppc} \dvparam_k \cdot g_{k,k} \cdot \erf\big(\tfrac{\tau}{\dvparam_k \sqrt{2g_{k,k}}}\big) \\
		 + \sum_{k,k' \in \ssuppc} g_{k, k'}^2 \cdot \dvparam_k \cdot \dvparam_{k'} \cdot \Big( \erf\big(\tfrac{\tau\beta_{k,k'}^{\min}}{\sqrt{2}}\big) - h(\tau\beta_{k,k'}^{\max}) \cdot \tau^2 \cdot \beta_{k,k'}^{\min} \cdot \beta_{k,k'}^{\max} \Big)\Big\}, \label{eq:proofs:scbound:generalbound}
	\end{multline}
	where $\beta_{k,k'}^{\min}$ and $\beta_{k,k'}^{\max}$ are defined according to \eqref{eq:proofs:scbound:bound:betaminmax}.
\end{proposition}
The bound of Proposition~\ref{prop:proofs:scbound:generalbound} is in fact quite general and leaves much space for selecting the parameters $\dvparam_k$.
However, we still need to determine the optimal value of $\tau$, which seems to be a very challenging task in general.
Our final step is therefore to make a specific choice of $\dvparam_k$, eventually leading to the explicit bound of Theorem~\ref{thm:proofs:mainresults:scbound}.

\subsubsection*{Step 4: Selecting the tuning parameters $\dvparam_k$.}

In order to simplify the optimization problem of \eqref{eq:proofs:scbound:generalbound}, we would like to factor out all terms depending on $\tau$.
Such a separation is easily obtained with $\dvparam_k \coloneqq \tfrac{\lambda}{\sqrt{g_{k,k}}}$ for $k \in \ssuppc$ and some $\lambda > 0$ which is fixed later on.
With this choice of $\dvparam_k$, \eqref{eq:proofs:scbound:generalbound} reads as follows:
% \pagebreak
\begin{multline}
	\effdim[\conic]{\descset{\lnorm{\aop(\cdot)}[1], \grtr}} \leq \inf_{\substack{\tau > 0 \\ \lambda > 0}} \Big\{  n + \tau^2 \cdot \spparam{} - 2\lambda \cdot \Big( \underbrace{\sum_{k \in \ssuppc} \sqrt{g_{k,k}}}_{\stackrel{\eqref{eq:results:generalsparsity:cospparam}}{=} \cospparamdg{}} \Big) \cdot \erf\big(\tfrac{\tau}{\lambda\sqrt{2}}\big) \\* + \lambda^2 \cdot \Big( \underbrace{\sum_{k,k' \in \ssuppc} \tfrac{g_{k, k'}^2}{\sqrt{g_{k,k} \cdot g_{k',k'}}} }_{\stackrel{\eqref{eq:results:generalsparsity:cospparam}}{=} \cospparam{}} \Big) \cdot \Big( \erf\big(\tfrac{\tau}{\lambda\sqrt{2}}\big) - h(\tfrac{\tau}{\lambda}) \cdot (\tfrac{\tau}{\lambda})^2  \Big)\Big\},
\end{multline}
and substituting $\tau$ by $\lambda \cdot \tau$ leads to
\begin{multline}
	\effdim[\conic]{\descset{\lnorm{\aop(\cdot)}[1], \grtr}} \\
	\leq{}  \inf_{\substack{\tau > 0 \\ \lambda > 0}} \Big\{ \underbrace{ n + \big( \spparam{} - \cospparam{} \cdot h(\tau)  \big) \cdot \tau^2 \cdot \lambda^2 + \big(\lambda^2 \cdot \cospparam{} - 2 \lambda \cdot \cospparamdg{}\big) \cdot \erf\big(\tfrac{\tau}{\sqrt{2}} \big) }_{\eqqcolon F(\tau,\lambda)} \Big\}. \label{eq:proofs:scbound:bound:finaltarget}
\end{multline}
It is not hard to see that $(\tau,\lambda) \mapsto F(\tau,\lambda)$ defines a smooth function that can be continuously extended to $\intvclop{0}{\infty} \times \intvclop{0}{\infty}$.
The next lemma shows that $F$ has in fact a unique global minimum. % that particularly minimizes \eqref{eq:proofs:scbound:bound:finaltarget}.
% and therefore, it suffices to minimize $(\tau,\lambda) \mapsto F(\tau,\lambda)$ on $\intvclop{0}{\infty} \times \intvclop{0}{\infty}$.\todocontent{Mention that F is cont.}
% In fact, $F$ has a unique global minimum as the following Lemma shows. 
% Its technical proof relies on basic calculus and is deferred to Appendix~\ref{subsec:prelim:scbound:globalmin}.
Due to its technical nature, the proof is deferred to Appendix~\ref{subsec:prelim:scbound:globalmin}.
\begin{lemma}\label{lem:proofs:scbound:globalmin}
	Assuming that $\spparam{}, \cospparam{}, \cospparamdg{} > 0$, the function $F \colon \intvclop{0}{\infty} \times \intvclop{0}{\infty} \to \R$, defined according to \eqref{eq:proofs:scbound:bound:finaltarget}, attains its unique global minimum at
	\begin{align}
		(\tau', \lambda') \coloneqq \Big( h^{-1}\big( \tfrac{\spparam{}}{\cospparam{}}\big), \tfrac{\cospparamdg{}}{\cospparam{}} \Big).
	\end{align}
\end{lemma}

Now, we can easily derive the statement of Theorem~\ref{thm:proofs:mainresults:scbound}:
\begin{proof}[Proof of Theorem~\ref{thm:proofs:mainresults:scbound}]
First, suppose that $\ssuppc{} \neq \emptyset$. This implies that $\cospparam{}, \cospparamdg{} > 0$, and since $\grtr \not\in \ker\aop$, we also have $\spparam{} > 0$ by Lemma~\ref{lem:proofs:scbound:subd}.
Lemma~\ref{lem:proofs:scbound:globalmin} and \eqref{eq:proofs:scbound:bound:finaltarget} then yield
\begin{align}
	& \effdim[\conic]{\descset{\lnorm{\aop(\cdot)}[1], \grtr}} \leq F(\tau', \lambda') \\
	={} & n + \underbrace{\big( \spparam{} - \cospparam{} \cdot h(\tau')  \big)}_{= 0} \cdot (\tau')^2 \cdot (\lambda')^2 + \underbrace{\big((\lambda')^2 \cdot \cospparam{} - 2 \lambda' \cdot \cospparamdg{}\big)}_{= - (\cospparamdg{})^2 / \cospparam{}} \cdot \erf\big(\tfrac{\tau'}{\sqrt{2}} \big) \\
	={} & n - \tfrac{(\cospparamdg{})^2}{\cospparam{}} \cdot \erf\Big( \tfrac{1}{\sqrt{2}} \cdot h^{-1} \big( \tfrac{\spparam{}}{\cospparam{}} \big) \Big) \stackrel{\eqref{eq:results:scfunc:Phi}}{=} n - \tfrac{(\cospparamdg{})^2}{\cospparam{}} \cdot \Phi\Big( \tfrac{\spparam{}}{\cospparam{}}\Big) \stackrel{\eqref{eq:results:scfunc:samplecompl}}{=} \samplecompl{\aop, \grtr}. \label{eq:proofs:mainresults:scbound:final}
\end{align}
On the other hand, if $\ssuppc{} = \emptyset$, we simply observe that
\begin{align}
	\effdim[\conic]{\descset{\lnorm{\aop(\cdot)}[1], \grtr}} &\leq \lim_{\substack{\tau,\lambda \searrow 0}} F(\tau, \lambda) = F(0, 0) = n = \samplecompl{\aop, \grtr}.
\end{align}
\end{proof}

\begin{remark}[KKT-Conditions and Optimality] \label{rmk:proofs:scbound:kkt}
	A key step in the above proof was to select the dual vector $\dv'$ in \eqref{eq:proofs:scbound:dvgeneral}, which eventually enabled us to derive a \emph{closed-form} bound on the conic mean width $\meanwidth[\conic]{\descset{\lnorm{\aop(\cdot)}[1], \grtr}}$.
	Such a choice is however not always optimal. 
	Indeed, any minimizer $\dv \in \R^N$ of \eqref{eq:proofs:scbound:boxconstraint} fulfills the following KKT-conditions (cf. \cite{mead2010}):
	\begin{align}
		\left(\tau^2 \aop_{\ssuppc} (\aop_{\ssuppc})^\T + \operatorname{diag}(\vec{\mu}_{\ssuppc}) \right) \dv_{\ssuppc} & = \tau \aop_{\ssuppc} \gaussian - \tau^2 \aop_{\ssuppc} \aop^\T \anasign{}, \\*
		\mu_k & \geq 0, \quad k \in [N], \\
		\mu_{k} \left(1 - w_{k} \right) & = 0, \quad k \in \ssuppc,\\
		w_{k}^ 2 & \leq 1, \quad k \in \ssuppc,\\*
		w_{k} & = 0, \quad k \in \ssupp,
	\end{align}
% 	\begin{align}
% 		\left(\tau^2 \aop \aop^\T + \vec{\mu} \right) \dv & = \tau \aop \gaussian, \\
% 		\mu_k & \geq 0, \quad k \in [N], \\
% 		\mu_{k} \left(1 - w_{k} \right) & = 0, \quad k \in \ssuppc,\\
% 		w_{k}^ 2 & \leq 1, \quad k \in \ssuppc,\\
% 		w_{k} & = 0, \quad k \in \ssupp,
% 	\end{align}
	for some $\vec{\mu} = (\mu_1, \dots, \mu_N) \in \R^{N}$.
	Supposed that $\gram = \aop \aop^\T = \operatorname{diag}(\dvparam_1^{-1},\dots,\dvparam_N^{-1})$ for certain numbers $\dvparam_k > 0$, it is not hard to see that $\dv = \dv'$ from \eqref{eq:proofs:scbound:dvgeneral} precisely satisfies these conditions.
	This observation particularly explains why our sampling-rate bound is optimal in the case of orthonormal bases (cf. Proposition~\ref{prop:results:scfuncsimplyfied}), where $\gram = \I{n}$.
	
	On the other hand, if the operator $\aop$ is redundant, i.e., $N > n$, the Gram matrix cannot be diagonal anymore and the solution of \eqref{eq:proofs:scbound:boxconstraint} may take a different form. But as long as $\gram$ is well-localized, in the sense that $\gram = \aop \aop^\T \approx \operatorname{diag}(\dvparam_1^{-1},\dots,\dvparam_N^{-1})$, one can expect that \eqref{eq:proofs:scbound:dvgeneral} is at least relatively close to the true minimizer.
	While this heuristic actually suggests setting $\dvparam_k \coloneqq 1 / g_{k,k}$ for $k \in \ssuppc$, it is not clear how to proceed with \eqref{eq:proofs:scbound:generalbound} in that situation.
	Therefore, we have made a slightly different choice of $\dvparam_k$ in Step 4 above. But as our numerical simulations in Section~\ref{sec:appl} indicate, this strategy still works surprisingly well for many different classes of analysis operators.
\end{remark}

\subsection{Proof of Proposition~\ref{prop:results:scfuncsimplyfied} (Simplified Sample Complexity Bound)}
\label{subsec:proofs:scfuncsimplyfied}

We prove that $\samplecompl{\aop, \grtr}$ is upper bounded by each of the two expressions on the right-hand side of \eqref{eq:results:scfuncsimplyfied:bound}.
Recall that the left branch of the minimum is supposed to mimic the (asymptotic) behavior of $\samplecompl{\aop, \grtr}$ if the ratio $\spparam{\grtr} / \cospparam{\grtr}$ is small, whereas the right branch is appropriate as $\spparam{\grtr} / \cospparam{\grtr}$ becomes large.

\begin{proof}[Proof of Proposition~\ref{prop:results:scfuncsimplyfied}]
	By Lemma~\ref{lem:proofs:scbound:subd} and $\grtr \not\in \ker\aop$, we know that $S = \spparam{\grtr} > 0$, while  the assumption $\ssuppc[\grtr] \neq \emptyset$ implies that $L = \cospparam{\grtr} > 0$ and $\bar{L} = \cospparamdg{\grtr} > 0$.
	Therefore, by Lemma~\ref{lem:proofs:scbound:globalmin} and \eqref{eq:proofs:mainresults:scbound:final}, we have
	\begin{equation}
		\samplecompl{\aop, \grtr} = F(\tau', \lambda') \leq F(\tau, \lambda')
	\end{equation}
	for all $\tau > 0$. Since $h(\tau)\geq 0$ and $\erf(\tau/\sqrt{2}) \geq 1 - e^{-\tau^2/2}$ for all $\tau > 0$ (see \cite[Prop.~7.5]{foucart2013cs}), one can bound $F(\tau, \lambda')$ as follows:
	\begin{align}
		F(\tau, \lambda') &= n + \big( S - L \cdot h(\tau)  \big) \cdot \tau^2 \cdot (\tfrac{\bar{L}}{L})^2 - \tfrac{\bar{L}^2}{L} \cdot \erf\big(\tfrac{\tau}{\sqrt{2}} \big) \\
		&\leq n + S \cdot (\tfrac{\bar{L}}{L})^2 \cdot \tau^2 - \tfrac{\bar{L}^2}{L} \cdot (1 - e^{-\tau^2/2}) \\
		&= n - \tfrac{\bar{L}^2}{L} + (\tfrac{\bar{L}}{L})^2 \cdot (S \cdot \tau^2 + L \cdot e^{-\tau^2/2}).
	\end{align}
	Finally, setting $\tau = \tau'' \coloneqq \sqrt{2 \log(\tfrac{S + L}{S})} > 0$, we obtain
	\begin{align}
		\samplecompl{\aop, \grtr} \leq F(\tau'', \lambda') &= n - \tfrac{\bar{L}^2}{L} + (\tfrac{\bar{L}}{L})^2 \cdot \Big(2 S \cdot \log(\tfrac{S + L}{S}) + \underbrace{\tfrac{L \cdot S}{S + L}}_{\leq S}\Big) \\
		&\leq n - \tfrac{\bar{L}^2}{L} + (\tfrac{\bar{L}}{L})^2 \cdot \Big(2 S \cdot \log(\tfrac{S + L}{S}) + S\Big).
	\end{align}
	
	It remains to verify that $\samplecompl{\aop, \grtr} \leq n - \tfrac{2}{\pi} \cdot \tfrac{\bar{L}^2}{S + L}$ holds true as well.
	By the definition of $\samplecompl{\aop, \grtr}$ in \eqref{eq:results:scfunc:samplecompl}, this is equivalent to
	\begin{align}
		& \Phi\Big(\tfrac{S}{L}\Big) \stackrel{\eqref{eq:results:scfunc:Phi}}{=} \erf\Big( \tfrac{1}{\sqrt{2}} \cdot h^{-1}(\tfrac{S}{L}) \Big) \geq \frac{2}{\pi} \cdot \frac{L}{S+L} = \frac{2}{\pi} \cdot \frac{1}{\tfrac{S}{L}+1} \\
		\stackrel{\xi\coloneqq h^{-1}(S/L)}{\iff} \quad & \erf(\tfrac{\xi}{\sqrt{2}}) \geq \frac{2}{\pi} \cdot \frac{1}{h(\xi)+1}, \quad \xi > 0 \\
		\iff \quad\quad & G(\xi) \coloneqq \erf(\tfrac{\xi}{\sqrt{2}}) \cdot (h(\xi)+1) \geq \frac{2}{\pi} \ .
	\end{align}
	To see that the inequality in last line is true for all $\xi > 0$, we observe that
	\begin{equation}
		\lim_{\xi \searrow 0} G(\xi) = \lim_{\xi \searrow 0} \Big[ \erf(\tfrac{\xi}{\sqrt{2}}) \cdot \sqrt{\tfrac{2}{\pi}} \tfrac{e^{-\xi^2/2}}{\xi} \Big] = \frac{2}{\pi} \ .
	\end{equation}
	The claim then follows from the fact that $G$ is monotonically increasing, which is in turn a consequence of $G'(\xi) > 0$ for all $\xi > 0$: indeed, we have
	\begin{align}
		& 0 < G'(\xi) = \sqrt{\tfrac{2}{\pi}} e^{-\xi^2/2} \cdot \Big( (\xi^2 - 1) \cdot \erf(\tfrac{\xi}{\sqrt{2}}) + \xi \cdot \sqrt{\tfrac{2}{\pi}} \cdot e^{-\xi^2/2} \Big) \\
		\iff \quad & 0 < \tilde{G}(\xi) \coloneqq (\xi^2 - 1) \cdot \erf(\tfrac{\xi}{\sqrt{2}}) + \xi \cdot \sqrt{\tfrac{2}{\pi}} \cdot e^{-\xi^2/2},
	\end{align}
	while the latter statement holds true due to $\tilde{G}(0) = 0$ and
	\begin{align}
		\tilde{G}'(\xi) &= 2\xi \cdot \erf(\tfrac{\xi}{\sqrt{2}}) > 0, \quad \xi > 0.
	\end{align}
	This concludes the proof.
\end{proof}

% \begin{figure}
% 	\centering
% 	\begin{subfigure}[t]{0.45\textwidth}
% 		\centering
% 		\includegraphics[width=\textwidth]{convex_loss.pdf}
% 		\caption{}
% 		\label{fig:intro:convexloss:nonflat}
% 	\end{subfigure}%
% 	\qquad
% 	\begin{subfigure}[t]{0.45\textwidth}
% 		\centering
% 		\includegraphics[width=\textwidth]{convex_loss_flat.pdf}
% 		\caption{}
% 		\label{fig:intro:convexloss:flat}
% 	\end{subfigure}%	
% 	\caption{\subref{fig:intro:convexloss:nonflat} $\lossemp[\yadv]$ is strongly curved so that small $\Delta\loss$ implies small $\Delta\grtr$. \subref{fig:intro:convexloss:flat} A less desirable situation where $\lossemp[\yadv]$ is relatively ``flat'' in the neighborhood of its minimizer.}
% 	\label{fig:intro:convexloss}
% \end{figure}

\appendix
% \cleardoublepage
% \addappheadtotoc
\section{Proofs of Auxiliary Results}
\label{sec:prelim}

\subsection{Proofs of Recovery Results From Subsection~\ref{subsec:proofs:framework}}
\label{subsec:prelim:framework}

First, let us state an important preliminary result from random matrix theory.
\begin{theorem}[Minimum Spheric Singular Value]\label{thm:prelim:framework:minsingval}
	Assume that Model~\ref{model:results:setup:meas} is satisfied and let $\sset \subset \S^{n-1}$.
	There exists a numerical constant $C > 0$ such that, for every $\probsuccess > 0$, we have
	\begin{equation}
		\inf_{\h \in \sset} \lnorm{\A \h} > \sqrt{m-1} - C \cdot \subgparam^2 \cdot ( \meanwidth{\sset} + \probsuccess)
	\end{equation}
	with probability at least $1 - e^{-\probsuccess^2/2}$.
	If $\a \distributed \Normdistr{\vnull}{\I{n}}$, we have $C = \subgparam = 1$.%\todocontent{Die strikte Ungleichung ist nicht in Tropp. Für Gaussians tritt die Gleichheit habe fast sicher nicht sein.}
\end{theorem}
\begin{proof}
	This statement is a consequence of \cite[Thm.~1.3]{liaw2016randommat}.
	To obtain the optimal constants in the Gaussian case, one can use \cite[Prop.~3.3]{tropp2014convex} instead.
\end{proof}

The Gaussian case of Theorem~\ref{thm:prelim:framework:minsingval} is well-known as \emph{Gordon's Escape Through a Mesh Theorem}, originating from \cite{gordon1988escape}.
The above formulation is typically interpreted as a lower bound on the \emph{minimum singular value} of the random matrix $\A \in \R^{m \times n}$ restricted to $\cone{\sset}$.

With this result at hand, we are now ready to prove Theorem~\ref{thm:proofs:framework:guarantee}:
\begin{proof}[Proof of Theorem~\ref{thm:proofs:framework:guarantee}]
	Any minimizer $\solu$ of \eqref{eq:proofs:framework:bpgeneral} satisfies $\solu - \grtr \in \descset{\objfunc, \grtr}$ as well as $\lnorm{\A \solu - \y} \leq \noiseparam$.
	Our goal is therefore to show that every vector $\x \in \R^n$ with $\lnorm{\x - \grtr} > t$ cannot be a minimizer, if both $\x - \grtr \in \descset{\objfunc, \grtr}$ and $\lnorm{\A \x - \y} \leq \noiseparam$ are fulfilled.

	First, we apply Theorem~\ref{thm:prelim:framework:minsingval} for $\sset = \tfrac{1}{t} \descset{\objfunc, \grtr} \intersec \S^{n-1}$:
	Let $\h \in t \sset$ and assume $\lnorm{\A \x - \y} \leq \noiseparam$ where $\x \coloneqq \grtr + \h$. 
	Without loss of generality, suppose that
	\begin{equation}
		\sqrt{m-1} > C \cdot \subgparam^2 \cdot ( \meanwidth[t]{\descset{\objfunc, \grtr}} + \probsuccess),
	\end{equation}
	since otherwise, the condition \eqref{eq:proofs:framework:guarantee:cond} is not satisfied in any case.	
	On the event of Theorem~\ref{thm:prelim:framework:minsingval}, we have
	\begin{multline}
		\delta \coloneqq \pospart{\sqrt{m-1} - C \cdot \subgparam^2 \cdot ( \meanwidth[t]{\descset{\objfunc, \grtr}} + \probsuccess)} = \sqrt{m-1} - C \cdot \subgparam^2 \cdot ( \meanwidth{\sset} + \probsuccess) \\
		< \tfrac{1}{t} \lnorm{\A \h} \leq \tfrac{1}{t} ( \lnorm{\A\h - \noise} + \lnorm{\noise}) = \tfrac{1}{t} ( \lnorm{\A\x - \y} + \lnorm{\noise}) \leq \frac{2\noiseparam}{t} \ .
	\end{multline}
	Rearranging the latter inequality gives $t < 2\noiseparam / \delta$ which contradicts the assumption \eqref{eq:proofs:framework:guarantee:cond}.
	This implies that $\x = \grtr + \h$ cannot be a minimizer of \eqref{eq:proofs:framework:bpgeneral}.
	
	Now, assume that $\solu \in \descset{\objfunc, \grtr} + \grtr$ is a minimizer of \eqref{eq:proofs:framework:bpgeneral} with $\lnorm{\solu - \grtr} > t$.
	Setting $\solu[\h] \coloneqq \solu - \grtr$, we have $\solu[t] \coloneqq \lnorm{\solu[\h]} > t$. The convexity of $\descset{\objfunc, \grtr}$ and $\vnull = \grtr - \grtr \in \descset{\objfunc, \grtr}$ imply that there exists $\h \in t \sset = \descset{\objfunc, \grtr} \intersec t \S^{n-1}$ such that $\solu[\h] = \tfrac{\solu[t]}{t} \h$.
	In other words, $\h$ is the intersection point between $t \S^{n-1}$ and the line that connects $\solu[\h]$ with $\vnull$.
	Hence, on the event of Theorem~\ref{thm:prelim:framework:minsingval}, we also observe that
	\begin{equation}
		\delta < \tfrac{1}{t} \lnorm{\A \h} = \tfrac{1}{\solu[t]} \lnorm{\A \solu[\h]} \leq \frac{2\noiseparam}{\solu[t]} < \frac{2\noiseparam}{t} \ ,
	\end{equation}
	which is again a contradiction to \eqref{eq:proofs:framework:guarantee:cond}.
	From this, we can conclude that any minimizer $\solu$ of \eqref{eq:proofs:framework:bpgeneral} satisfies $\lnorm{\solu - \grtr} \leq t$.
% 	\mm{this can be moved or deleted} Since, there exists at least one feasible vector for \eqref{eq:proofs:framework:bpgeneral}, namely $\grtr$, we can conclude that any minimizer $\solu$ satisfies $\lnorm{\solu - \grtr} \leq t$.
\end{proof}

\begin{proof}[Proof of Corollary~\ref{cor:proofs:framework:guarantee-conic}]
	First, assume that $\noiseparam > 0$. We would like to apply Theorem~\ref{thm:proofs:framework:guarantee} with
	\begin{equation}
		t \coloneqq \frac{2\noiseparam}{\pospart{\sqrt{m-1} - C \cdot \subgparam^2 \cdot (\meanwidth[\conic]{\descset{\objfunc, \grtr}} + \probsuccess)}} \ .
	\end{equation}
	Since $\meanwidth[\conic]{\descset{\objfunc, \grtr}} \geq \meanwidth[t]{\descset{\objfunc, \grtr}}$ by Remark~\ref{rmk:proofs:framework:meanwidth}, the assumption \eqref{eq:proofs:framework:guarantee:cond} is satisfied and the claim now follows from Theorem~\ref{thm:proofs:framework:guarantee}.
	
	Now, let us consider the noiseless case of $\noiseparam = 0$. Similar to the proof of Theorem~\ref{thm:proofs:framework:guarantee}, we apply Theorem~\ref{thm:prelim:framework:minsingval} with $\sset = \desccone{\objfunc, \grtr} \intersec \S^{n-1}$. 
	Let $\h \in \sset$ and $\x \coloneqq \grtr + \h$ with $\lnorm{\A \x - \y} = \lnorm{\A\h} = 0$.
	On the event of Theorem~\ref{thm:prelim:framework:minsingval}, we have
	\begin{equation}\label{eq:prelim:framework:guarantee-conic:contradict}
		0 \stackrel{\eqref{eq:proofs:framework:guarantee-conic:meas}}{<} \sqrt{m-1} - C \cdot \subgparam^2 \cdot ( \meanwidth[\conic]{\descset{\objfunc, \grtr}} + \probsuccess) < \lnorm{\A \h} = 0,
	\end{equation}
	which is a contradiction.
	If $\h \in \desccone{\objfunc, \grtr} \setminus \{\vnull\}$ with $\lnorm{\A\h} = 0$, we just invoke \eqref{eq:prelim:framework:guarantee-conic:contradict} for $\h / \lnorm{\h}$. Hence, the only feasible point of \eqref{eq:proofs:framework:bpgeneral} is $\grtr$, so that $\solu = \grtr$.
\end{proof}

To prove Corollary~\ref{cor:proofs:framework:guarantee-stable}, we need the following lemma.
It shows how to select a scale $t > 0$ for the local mean width such that it can be bounded by the conic mean width of a surrogate signal:
\begin{lemma}\label{lem:prelim:framework:loceffdimstable}
	Let $\grtr \in \R^n$ and let $\sset \subset \R^n$ be convex. For every $t > 0$ and $\grtrsparse \in \sset$, we have
	\begin{equation}
		\meanwidth[t]{\sset - \grtr} \leq \tfrac{t + \lnorm{\grtrsparse - \grtr}}{t} \cdot (\meanwidth[\conic]{\sset-\grtrsparse } + 1).
	\end{equation}
	In particular, $\meanwidth[R\lnorm{\grtrsparse - \grtr}]{\sset-\grtr} \leq \tfrac{R+1}{R} \cdot (\meanwidth[\conic]{\sset-\grtrsparse } + 1)$ for every $R > 0$.
\end{lemma}
\begin{proof}
	We set $t' \coloneqq \lnorm{\grtrsparse - \grtr}$. For every $\x \in  \sset \intersec (\grtr + t\S^{n-1}) $, we have
	\begin{equation}
		\lnorm{\x - \grtrsparse} \leq \lnorm{\x - \grtr} + \lnorm{\grtr - \grtrsparse} \leq t + t',
	\end{equation}
	and therefore, $\x \in \sset_{\grtrsparse,t+t'} \coloneqq \sset \intersec (\grtrsparse + (t+t')\ball[2][n])$. This implies %\enlargethispage{\baselineskip}
	\begin{align}
		\meanwidth[t]{\sset-\grtr} &= \mean[\sup_{\h \in  \tfrac{1}{t}(\sset - \grtr) \intersec \S^{n-1} } \sp{\gaussian}{\h}] = \tfrac{1}{t} \mean[\sup_{\x \in  \sset \intersec (\grtr + t\S^{n-1}) } \sp{\gaussian}{\x - \grtr}] \\
		&\leq \tfrac{1}{t} \mean[\sup_{\x \in \sset_{\grtrsparse,t+t'}} \sp{\gaussian}{\x - \grtrsparse}] + \underbrace{\tfrac{1}{t}\mean[\sup_{\x \in \sset \intersec (\grtr + t\S^{n-1})} \sp{\gaussian}{\grtrsparse- \grtr}]}_{= 0} \\
		&= \tfrac{t+t'}{t} \cdot \tfrac{1}{t+t'} \mean[\sup_{\x \in \sset_{\grtrsparse,t+t'}} \sp{\gaussian}{\x - \grtrsparse}] \\
		&= \tfrac{t+t'}{t} \cdot \mean[\sup_{\h \in \tfrac{1}{t+t'}(\sset - \grtrsparse)\intersec \ball[2][n]} \sp{\gaussian}{\h}] \\
		&\stackrel{\mathllap{\grtrsparse \in \sset}}{\leq} \tfrac{t+t'}{t} \cdot \mean[\sup_{\h \in \cone{\sset - \grtrsparse}\intersec \ball[2][n]} \sp{\gaussian}{\h}] \\
		&\stackrel{\mathllap{\text{CSI}}}{\leq} \tfrac{t+t'}{t} \cdot \Big(\mean[(\sup_{\h \in \cone{\sset - \grtrsparse}\intersec \ball[2][n]} \sp{\gaussian}{\h})^2]\Big)^{1/2} \\*
		&\stackrel{(\ast)}{\leq} \tfrac{t+t'}{t} \cdot (\effdim[\conic]{\sset- \grtrsparse } + 1)^{1/2} \leq \tfrac{t+t'}{t} \cdot (\meanwidth[\conic]{\sset- \grtrsparse } + 1),
	\end{align}
	where $(\ast)$ follows from \cite[Prop.~3.1(5) and Prop. 10.2]{amelunxen2014edge}.
\end{proof}

\begin{proof}[Proof of Corollary~\ref{cor:proofs:framework:guarantee-stable}]
	We would like to apply Theorem~\ref{thm:proofs:framework:guarantee} with
	\begin{equation}
		t = \max\left\{ R \lnorm{\grtr - \grtrsparse}, \frac{2\noiseparam}{\pospart{\sqrt{m-1} - C \cdot \subgparam^2 \cdot ( \tfrac{R + 1}{R} \cdot (\meanwidth[\conic]{\descset{\objfunc, \grtrsparse}} + 1) + \probsuccess)}} \right\}.
	\end{equation}
	Thus, it suffices to verify that the assumption \eqref{eq:proofs:framework:guarantee:cond} is satisfied.
	For this, we first observe that
	\begin{align}
		\meanwidth[t]{\descset{\objfunc, \grtr}} &= \meanwidth[t]{\descset{\objfunc, \grtr} + \grtr - \grtr} \\
		&\stackrel{\mathllap{t \geq  R \lnorm{\grtr - \grtrsparse}}}{\leq} \meanwidth[R\lnorm{\grtr - \grtrsparse}]{(\descset{\objfunc, \grtr} + \grtr) - \grtr} \\
		&\stackrel{\mathllap{\text{Lemma~\ref{lem:prelim:framework:loceffdimstable}}}}{\leq} \tfrac{R+1}{R} \cdot (\meanwidth[\conic]{(\descset{\objfunc, \grtr} + \grtr) -\grtrsparse } + 1) \\
		& = \tfrac{R+1}{R} \cdot (\meanwidth[\conic]{\descset{\objfunc, \grtrsparse}} + 1),
	\end{align}
	where we have used the monotony of the local mean width in the second step, and the last step follows from
	\begin{align}
		\descset{\objfunc, \grtr} + \grtr -\grtrsparse &= \{ \h \in \R^n \suchthat \objfunc(\grtr + \h) \leq \underbrace{\objfunc(\grtr)}_{= \objfunc(\grtrsparse)} \}  + \grtr - \grtrsparse \\
		&\stackrel{\mathllap{\h' \coloneqq \h + \grtr - \grtrsparse}}{=} \{ \h' \in \R^n \suchthat \objfunc(\h' + \grtrsparse) \leq \objfunc(\grtrsparse) \} = \descset{\objfunc, \grtrsparse}.
	\end{align}
	This implies\footnote{This argument does actually not cover the case of $\noiseparam = 0$, but here we can use that $t = R\lnorm{\grtr - \grtrsparse} > 0$ if $\grtr \neq \grtrsparse$.}
	\begin{align}
		t &\geq \frac{2\noiseparam}{\pospart{\sqrt{m-1} - C \cdot \subgparam^2 \cdot ( \tfrac{R + 1}{R} \cdot (\meanwidth[\conic]{\descset{\objfunc, \grtrsparse}} + 1) + \probsuccess)}} \\
		&\geq \frac{2\noiseparam}{\pospart{\sqrt{m-1} - C \cdot \subgparam^2 \cdot (\meanwidth[t]{\descset{\objfunc, \grtr}} + \probsuccess)}} \ ,
	\end{align}
	so that the claim again follows from Theorem~\ref{thm:proofs:framework:guarantee}.
\end{proof}

\subsection{Proof of Lemma~\ref{lem:proofs:scbound:clippedcovar}}
\label{subsec:prelim:clippedcovar}

% \begin{draft}
% \item
% 	Explain how one can derive this bound intuitively: Compute the integral in the case of perfect correlation and show by Price's theorem that it is monotonically decreasing in terms of the correlation.
% \end{draft}

The main idea behind the proof of Lemma~\ref{lem:proofs:scbound:clippedcovar} is to regard the \emph{clipped covariance}
\begin{equation}
	\mean[\clip{\sp{\v_1}{\gaussian}}{\beta_1} \cdot \clip{\sp{\v_2}{\gaussian}}{\beta_2}]
\end{equation}
as a function in the correlation between $\sp{\v_1}{\gaussian}$ and $\sp{\v_2}{\gaussian}$, which is given by $\mean[\sp{\v_1}{\gaussian} \cdot \sp{\v_2}{\gaussian}] = \sp{\v_1}{\v_2}$.
In fact, we will even see that it is monotone and \eqref{eq:proofs:scbound:clippedcovar:bound} holds with equality if $\abs{\sp{\v_1}{\v_2}} = 1$.
To make this argument precise, we need the following
\begin{lemma}\label{lem:prelim:clippedcovar:corrfct}
	Let $\beta_1, \beta_2 \geq 0$ and define the function
	\begin{equation}
		\Theta \colon \intvcl{-1}{1} \to \R, \ \rho \mapsto \mean{}_{(\gaussianuniv_1, \gaussianuniv_2) \distributed \Normdistr[\Big]{\vnull}{\smallmatr{1 & \rho \\ \rho & 1 }}} \Big[ \clip{\gaussianuniv_1}{\beta_1} \cdot \clip{\gaussianuniv_2}{\beta_2} \Big].
	\end{equation}
	Then, $\Theta$ is continuous, $\Theta(0) = 0$, and $\Theta$ is convex on $\intvcl{0}{1}$. In particular, $\Theta(\rho) \leq \rho \cdot \Theta(1)$ for all $0 \leq \rho \leq 1$.
\end{lemma}

The proof of Lemma~\ref{lem:prelim:clippedcovar:corrfct} is a consequence of Price's Theorem \cite{price1958theorem,mcmahon1964price}, which is a powerful tool to calculate the derivative of non-linearly distorted, correlated Gaussian random variables. The following version of Price's Theorem was established by Voigtlaender \cite{voigtlaender2017price} and even applies to tempered distributions:\footnote{For a detailed introduction to Schwartz functions and tempered distributions, see \cite[Chap.~8+9]{folland1999analysis} for example.}
\begin{theorem}[Price, \protect{\cite[Cor.~2]{voigtlaender2017price}}]\label{thm:prelim:clippedcovar:price}
	For $\rho \in \intvop{-1}{1}$, let 
	\begin{equation}
		\phi_\rho \colon \R^2 \to \intvop{0}{\infty}, \ (x_1, x_2) \mapsto \tfrac{1}{2\pi \sqrt{1 - \rho^2}} \cdot \exp \Big( - \tfrac{x_1^2 + x_2^2 - 2 \rho \cdot x_1 \cdot x_2}{2 (1- \rho^2)} \Big)
	\end{equation}
	be the density function of a random vector $\gaussian \distributed \Normdistr[\Big]{\vnull}{\smallmatr{1 & \rho \\ \rho & 1 }}$.
	If $f \in \Tempdistr[\R^2]$ is a tempered distribution, then the function
	\begin{equation}
		\mathscr{F} \colon \intvop{-1}{1} \to \C, \ \rho \mapsto \sp{f}{\phi_\rho}_{\Tempdistr,\Schwartz}
	\end{equation}
	is infinitely often differentiable with
	\begin{equation}
		\mathscr{F}^{(k)}(\rho) = \sp[\Big]{\tfrac{\partial^{2k}}{\partial x_1^k \partial x_2^k} f}{\phi_\rho}_{\Tempdistr,\Schwartz} \quad \text{for all $\rho \in \intvop{-1}{1}$ and $k \geq 0$.}
	\end{equation}
	In particular, if $f \in \Leb[\R^2]{1,\text{loc}}$, we have $\mathscr{F}(\rho) = \mean[f(\gaussian)]$ for $\gaussian \distributed \Normdistr[\Big]{\vnull}{\smallmatr{1 & \rho \\ \rho & 1 }}$.
\end{theorem}

\begin{proof}[Proof of Lemma~\ref{lem:prelim:clippedcovar:corrfct}]
	We consider the function
	\begin{equation}
		f_{\beta_1,\beta_2}\colon \R^2 \to \R, \ (x_1, x_2) \mapsto \clip{x_1}{\beta_1} \cdot \clip{x_2}{\beta_2}.
	\end{equation}
	It is not hard to see that $f_{\beta_1,\beta_2}$ is $1$-Lipschitz and therefore $f_{\beta_1,\beta_2} \in \Tempdistr[\R^2]$.
	The (weak) derivative of $f_{\beta_1,\beta_2}$ is given by the set indicator function of the box $\intvop{-\beta_1}{\beta_1} \times \intvop{-\beta_2}{\beta_2} \subset \R^2$:
	\begin{equation}
		\tfrac{\partial^{2}}{\partial x_1 \partial x_2} f_{\beta_1,\beta_2} = \indset{\intvop{-\beta_1}{\beta_1} \times \intvop{-\beta_2}{\beta_2}}.
	\end{equation}
	For any $\phi \in \Schwartz[\R^2]$, Fubini's theorem and the fundamental theorem of calculus now yield
	\begin{align}
		\sp[\Big]{\tfrac{\partial^{4}}{\partial x_1^2 \partial x_2^2} f_{\beta_1,\beta_2}}{\phi}_{\Tempdistr,\Schwartz} &= \sp[\Big]{\indset{\intvop{-\beta_1}{\beta_1} \times \intvop{-\beta_2}{\beta_2}}}{\tfrac{\partial^2}{\partial x_1 \partial x_2} \phi}_{\Tempdistr,\Schwartz} \\
		&= \int_{-\beta_2}^{\beta_2} \int_{-\beta_1}^{\beta_1} \tfrac{\partial^2}{\partial x_1 \partial x_2} \phi (x_1, x_2) dx_1 dx_2 \\
		&= \phi(\beta_1, \beta_2) - \phi(\beta_1, -\beta_2) - \phi(-\beta_1, \beta_2) +  \phi(-\beta_1, -\beta_2).
	\end{align}
	Hence, by Theorem~\ref{thm:prelim:clippedcovar:price}, the map $\Theta$ is smooth on $\intvop{-1}{1}$ and
	\begin{align}
		\Theta''(\rho) &= \phi_\rho(\beta_1, \beta_2) - \phi_\rho(\beta_1, -\beta_2) - \phi_\rho(-\beta_1, \beta_2) +  \phi_\rho(-\beta_1, -\beta_2) \\
		&= 2 (\phi_\rho(\beta_1, \beta_2) - \phi_\rho(\beta_1, -\beta_2)) \label{eq:prelim:clippedcovar:corrfct:sndder}
	\end{align}
	for $\rho \in \intvop{-1}{1}$, where we have also used the symmetry of $\phi_\rho$.
	Now, let us show the continuity of $\Theta$ on the interval $\intvcl{-1}{1}$. For this, let $\gaussian \coloneqq (\gaussianuniv, \gaussianuniv') \distributed \Normdistr{\vnull}{\I{2}}$ and set $\gaussianuniv_\rho \coloneqq \rho \cdot \gaussianuniv + \sqrt{1 - \rho} \cdot \gaussianuniv'$ for $\rho \in \intvcl{-1}{1}$. Then $(\gaussianuniv, \gaussianuniv_\rho) \distributed \Normdistr[\Big]{\vnull}{\smallmatr{1 & \rho \\ \rho & 1 }}$, so that
	\begin{align}
		\abs{\Theta(\rho) - \Theta(\rho')} &\leq \mean[]_{\gaussian} [\abs{f_{\beta_1,\beta_2}(\gaussianuniv, \gaussianuniv_\rho) - f_{\beta_1,\beta_2}(\gaussianuniv, \gaussianuniv_{\rho'})}] \\
		&\stackrel{\mathllap{\text{$f_{\beta_1,\beta_2}$ is $1$-Lipschitz}}}{\leq} \mean[]_{\gaussian} [\abs{\gaussianuniv_\rho - \gaussianuniv_{\rho'}}] \\
		&\leq \abs{\rho - \rho'} \cdot \mean[\abs{\gaussianuniv}] + \abs{\sqrt{1 - \rho} - \sqrt{1 - \rho'}} \cdot \mean[\abs{\gaussianuniv'}] \xrightarrow{\rho' \to \rho} 0.
	\end{align}
	Moreover, we observe that
	\begin{equation}
		\Theta(0) = \mean[]_{\gaussian}[f_{\beta_1,\beta_2}(\gaussianuniv, \gaussianuniv')] = \underbrace{\mean[\clip{\gaussianuniv}{\beta_1}]}_{= 0} \cdot \underbrace{\mean[\clip{\gaussianuniv'}{\beta_1}]}_{= 0} = 0.
	\end{equation}
	Finally, let us verify the convexity of $\Theta$ on $\intvcl{0}{1}$. It is sufficient to show $\Theta''(\rho) \geq 0$ for $\rho \in \intvop{0}{1}$.
	Indeed, by \eqref{eq:prelim:clippedcovar:corrfct:sndder}, we have
	\begin{align}
		& \Theta''(\rho) = 2 (\phi_\rho(\beta_1, \beta_2) - \phi_\rho(\beta_1, -\beta_2)) \geq 0 \\
		\iff \quad & \exp \Big( - \tfrac{\beta_1^2 + \beta_2^2 - 2 \rho \cdot \beta_1 \cdot \beta_2}{2 (1- \rho^2)} \Big) \geq \exp \Big( - \tfrac{\beta_1^2 + \beta_2^2 + 2 \rho \cdot \beta_1 \cdot \beta_2}{2 (1- \rho^2)} \Big) \\
		\iff \quad & \beta_1^2 + \beta_2^2 - 2 \rho \cdot \beta_1 \cdot \beta_2 \leq \beta_1^2 + \beta_2^2 + 2 \rho \cdot \beta_1 \cdot \beta_2 \quad \iff \quad 0 \leq 4 \rho \cdot \beta_1 \cdot \beta_2,
	\end{align}
	where the latter inequality is obviously fulfilled for every $\rho \in \intvop{0}{1}$.
\end{proof}

Now, using Lemma~\ref{lem:prelim:clippedcovar:corrfct}, we can easily derive the statement of Lemma~\ref{lem:proofs:scbound:clippedcovar}:
\begin{proof}[Proof of Lemma~\ref{lem:proofs:scbound:clippedcovar}]
	First, we observe that
	\begin{equation}
		(\sp{\v_1}{\gaussian}, \sp{\v_2}{\gaussian}) \distributed \Normdistr[\Big]{\vnull}{\smallmatr{1 & \sp{\v_1}{\v_2} \\ \sp{\v_1}{\v_2} & 1 }}.
	\end{equation}
	If $\sp{\v_1}{\v_2} \geq 0$, Lemma~\ref{lem:prelim:clippedcovar:corrfct} implies
	\begin{equation}
		\mean[\clip{\sp{\v_1}{\gaussian}}{\beta_1} \cdot \clip{\sp{\v_2}{\gaussian}}{\beta_2}] = \Theta(\sp{\v_1}{\v_2}) \leq \sp{\v_1}{\v_2} \cdot \Theta(1),
	\end{equation}
	and an elementary calculation shows that
	\begin{align}
		\Theta(1) &= \mean{}[\clip{\underbrace{\gaussianuniv}_{\mathclap{\distributed\Normdistr{0}{1}}}}{\beta_1} \cdot \clip{\gaussianuniv}{\beta_2}] \\
		&= \Big( 1 - \erf\big(\tfrac{\beta^{\max}}{\sqrt{2}}\big) - \sqrt{\tfrac{2}{\pi}} \tfrac{e^{-(\beta^{\max})^2 / 2}}{\beta^{\max}} \Big) \cdot \beta^{\min} \cdot \beta^{\max} + \erf\big(\tfrac{\beta^{\min}}{\sqrt{2}}\big).
	\end{align}
	On the other hand, if $\sp{\v_1}{\v_2} \leq 0$, we have
	\begin{align}
		\mean[\clip{\sp{\v_1}{\gaussian}}{\beta_1} \cdot \clip{\sp{\v_2}{\gaussian}}{\beta_2}] &= - \mean[\clip{\sp{-\v_1}{\gaussian}}{\beta_1} \cdot \clip{\sp{\v_2}{\gaussian}}{\beta_2}] \\
		&= - \Theta(\sp{-\v_1}{\v_2}) \\
		&\stackrel{\mathllap{\sp{-\v_1}{\v_2} \geq 0}}{\geq} - \sp{-\v_1}{\v_2} \cdot \Theta(1) \\
		&= \sp{\v_1}{\v_2} \cdot \Theta(1),
	\end{align}
	and by taking the absolute value, the claim follows.
\end{proof}

\subsection{Proof of Lemma~\ref{lem:proofs:scbound:globalmin}}
\label{subsec:prelim:scbound:globalmin}

Before proving the actual statement of Lemma~\ref{lem:proofs:scbound:globalmin}, we recall the definition of the function $h$ defined in \eqref{eq:results:scfunc:h}:
\begin{equation}
	h\colon \intvop{0}{\infty} \to \intvop{0}{\infty}, \ \tau \mapsto \sqrt{\tfrac{2}{\pi}} \tfrac{e^{-\tau^2/2}}{\tau} + \erf(\tfrac{\tau}{\sqrt{2}}) - 1.
\end{equation}
The well-definedness of $h$, i.e., $h(\tau) > 0$ for all $\tau > 0$, is a consequence of standard bounds on the error function, e.g., see \cite[Prop.~7.5]{foucart2013cs}.
Since $\lim_{\tau \searrow 0} h(\tau) = \infty$, $\lim_{\tau \to \infty} h(\tau) = 0$, and
\begin{equation}
	h'(\tau) = -\sqrt{\tfrac{2}{\pi}} \tfrac{e^{-\tau^2/2}}{\tau^2} < 0
\end{equation}
for all $\tau > 0$, we can conclude that $h^{-1}\colon \intvop{0}{\infty} \to \intvop{0}{\infty}$ is also well-defined.

\begin{proof}[Proof of Lemma~\ref{lem:proofs:scbound:globalmin}]
	Adapting the notation from Proposition~\ref{prop:results:scfuncsimplyfied}, let us set $S \coloneqq \spparam{}$, $L \coloneqq \cospparam{}$, $\bar{L} \coloneqq \cospparamdg{}$, and for the sake of convenience, let us also restate the definition of $F$:
	\begin{equation}
		F(\tau,\lambda) = n + \big( S - L \cdot h(\tau)  \big) \cdot \tau^2 \cdot \lambda^2 + \big(\lambda^2 \cdot L - 2 \lambda \cdot \bar{L}\big) \cdot \erf\big(\tfrac{\tau}{\sqrt{2}} \big), \quad \tau, \lambda \geq 0.
	\end{equation}
	First, we calculate the partial derivatives of $F$ for $\tau, \lambda > 0$:
	\begin{align}
		\partder{F}{\tau} (\tau,\lambda) &= 2 \lambda^2 \cdot \tau \cdot \big( S - L \cdot h(\tau)  \big) + 2 \lambda \cdot \sqrt{\tfrac{2}{\pi}}e^{-\tau^2/2} \cdot (\lambda \cdot L - \bar{L}), \\
		\partder{F}{\lambda} (\tau,\lambda) &= 2 \lambda \cdot \tau^2 \cdot \big( S - L \cdot h(\tau)  \big) + 2 \erf(\tfrac{\tau}{\sqrt{2}}) \cdot (\lambda \cdot L - \bar{L}).
	\end{align}
	Hence, the gradient of $F$ takes the following form:
	\begin{equation}
		\gradient F(\tau,\lambda) = \underbrace{\matr{2 \lambda^2 \cdot \tau & 2 \lambda \cdot \sqrt{\tfrac{2}{\pi}}e^{-\tau^2/2} \\ 2 \lambda \cdot \tau^2 & 2 \erf(\tfrac{\tau}{\sqrt{2}})}}_{\eqqcolon \vec{M}} \cdot \matr{S - L \cdot h(\tau) \\ \lambda \cdot L - \bar{L}}.
	\end{equation}
	Since $S,L,\bar{L} > 0$ and
	\begin{equation}
		\det(\vec{M}) = 4 \lambda^2 \cdot \tau \cdot \Big( \erf(\tfrac{\tau}{\sqrt{2}}) - \tau \cdot \sqrt{\tfrac{2}{\pi}}e^{-\tau^2/2} \Big) > 0
	\end{equation}
	for all $\tau, \lambda > 0$, the only critical point of $F$ on $\intvop{0}{\infty} \times \intvop{0}{\infty}$ is attained if
	\begin{equation}
		0 = S - L \cdot h(\tau) = \lambda \cdot L - \bar{L} \quad \iff \quad (\tau, \lambda) = (\tau', \lambda') = \Big( h^{-1}\big( \tfrac{S}{L}\big), \tfrac{\bar{L}}{L} \Big).
	\end{equation}
	Computing the value of $F$ at this point, we observe that (cf. \eqref{eq:proofs:mainresults:scbound:final})
	\begin{align}
		F(\tau', \lambda') = n - \tfrac{(\bar{L})^2}{L} \cdot \erf\Big( \tfrac{1}{\sqrt{2}} \cdot h^{-1} \big( \tfrac{S}{L} \big) \Big) < n = F(\tau, 0) = F(0,\lambda)
	\end{align}
	for all $\lambda, \tau \geq 0$, implying that $F$ cannot attain a global minimum on the boundary of $\intvclop{0}{\infty} \times \intvclop{0}{\infty}$.
	Consequently, it suffices to seek for a global minimum on $\intvop{0}{\infty} \times \intvop{0}{\infty}$.
% 	Now, suppose that there is some point $(\tau_0, \lambda_0) \in \intvop{0}{\infty} \times \intvop{0}{\infty}$ such that $F(\tau_0, \lambda_0) \leq F(\tau', \lambda')$.
% 	We will show that in fact $\tau_0 = \tau'$ and $\lambda_0 = \lambda'$, which verifies that $(\tau', \lambda')$ is the unique minimizer.
	For this purpose, let us first consider the quadratic function $\lambda \mapsto F_{\tau}(\lambda) \coloneqq F(\tau, \lambda)$ for $\lambda > 0$ and a fixed $\tau > 0$.
	It is not hard to see that $F_{\tau}$ has its unique minimum at
	\begin{equation}
		\lambda_0(\tau) \coloneqq \frac{\bar{L} \cdot \erf(\tfrac{\tau}{\sqrt{2}})}{ L \cdot \erf(\tfrac{\tau}{\sqrt{2}}) + \tau^2 \cdot (S - L \cdot h(\tau))} \ .
	\end{equation}
	Note that $\lambda_0(\tau) > 0$, since $L \cdot (\erf(\tau / \sqrt{2}) - \tau^2 \cdot h(\tau)) > 0$ for all $\tau > 0$.
	Next, we minimize the parameterized mapping
	\begin{equation}
		\tau \mapsto \tilde{F}(\tau) \coloneqq F(\tau, \lambda_0(\tau)) = n - \frac{\bar{L}^2 \cdot \erf(\tfrac{\tau}{\sqrt{2}})^2}{ L \cdot \erf(\tfrac{\tau}{\sqrt{2}}) + \tau^2 \cdot (S - L \cdot h(\tau))}
	\end{equation}
	on $\intvop{0}{\infty}$. An elementary calculation shows that
	\begin{align}
		0 = \tilde{F}'(\tau) \quad &\iff \quad 0 = (S - L \cdot h(\tau)) \cdot \underbrace{\Big( \tau \cdot \sqrt{\tfrac{2}{\pi}}e^{-\tau^2/2} - \erf(\tfrac{\tau}{\sqrt{2}}) \Big)}_{< 0} \\
		&\iff \quad \tau = \tau' = h^{-1}\big( \tfrac{S}{L}\big),
	\end{align}
	and $\tilde{F}''(\tau') > 0$. Thus, $\tilde{F}$ attains its unique minimum at $\tau'$.
	
	Finally, let $(\tau, \lambda) \in \intvop{0}{\infty} \times \intvop{0}{\infty}$ be arbitrary.
	By the above observations,
	\begin{equation}
		F(\tau, \lambda) \geq F(\tau,\lambda_0(\tau)) \geq F(\tau',\underbrace{\lambda_0(\tau')}_{= \lambda'}) = F(\tau', \lambda'),
	\end{equation}
	where equality holds if and only if $\tau = \tau'$ and $\lambda = \lambda_0(\tau') = \lambda'$. Therefore, $F$ has indeed its unique global minimum at $(\tau',\lambda')$.
\end{proof}

\section{Details on Implementation}
\label{sec:implementation}

\subsection{Greedy Approach to Constructing Sparse Representations}
\label{subsec:implementation:greedy}

\newcommand{\paras}{S_0} % groundtruth signal

In this section, we present a greedy strategy for finding a good approximation space $\projnsp \subset \R^n$ according to the argumentation of Subsection~\ref{subsec:results:compressible}. Recall that our basic geometric idea is to select $\projnsp \subset \R^n$ in such a way that the orthogonal projection of $\proj{\projnsp} \grtr$ enjoys a small value of $\samplecompl{\aop,\proj{\projnsp}\grtr}$ while not being too distant from $\grtr$ at the same time. 
This is precisely the purpose of Algorithm~\ref{exp:impl:algo} below, where we fix a target analysis sparsity $S$ in advance and then identify a subspace $\projnsp$ with $\lnorm{\aop \proj{\projnsp} \grtr}[0] \leq S$ and small approximation error $\lnorm{\grtr - \proj{\projnsp}\grtr}[2]$.

\begin{algorithm}[Greedy subspace selection]\leavevmode
\label{exp:impl:algo}
\vspace{-.25\baselineskip}\hrule\vspace{.5\baselineskip}

\myhangindent{Input: \ }
\expkwd{Input:} Analysis operator $\aop \in \R^{N \times n}$, signal vector $\grtr \in \R^n$, target level of sparsity $S\in [N]$, (machine) precision parameter $\varepsilon \geq 0$.

\vspace{.5\baselineskip}
\expkwd{Output:} Linear approximation space $\projnsp \subset \R^n$.

\vspace{.5\baselineskip}
\expkwd{Initialize:} $\paras \coloneqq N$, $\ssupp[0] \coloneqq [N], \projnsp \coloneqq \{ \vnull \}$.
\begin{expstep}
	\item \expkwd{While} $\paras > S$: 
	\begin{enumerate}[label=(\arabic*)]
		\item Select an index $k \in \ssupp[0]$ such that $\lnorm{\proj{\projnsp[k]} \grtr}[2]$ becomes minimal, where $\projnsp[k] \coloneqq \spann \{ \avec_k\}$.\label{step:implementation:step1}
		% \item Set $\proj{} \coloneqq \proj{\projnsp[k]}$;
		\item Determine all indices $k' \in \ssupp[0]$ with $|\sp{\avec_{k'}}{\proj{\orthcompl{\projnsp[k]}} \grtr}| \leq \varepsilon$ and exclude them from $\ssupp[0]$.\label{step:implementation:step2}
		\item Set $\projnsp \coloneqq  \orthcompl{\spann( \{\avec_k \suchthat k \in \ssuppc[0] \})} $.
		\item Update $\paras \coloneqq \lnorm{ \aop \proj{\projnsp} \grtr}[0]$.
	\end{enumerate}
	\item \expkwd{Return} $\projnsp$. 
\end{expstep}
\end{algorithm}

Since $\lnorm{\grtr - \proj{\orthcompl{\projnsp[k]}} \grtr}[2] = \lnorm{\proj{\projnsp[k]} \grtr}[2]$, the actual intention of Step~\ref{step:implementation:step1} in Algorithm~\ref{exp:impl:algo} is to identify an atom $\avec_k \in \R^n$ of $\aop$ that yields a best approximation of $\grtr$ while enforcing $\sp{\avec_k}{\proj{\orthcompl{\projnsp[k]}} \grtr} = 0$. 
The coherence structure of $\aop$ is then incorporated by Step~\ref{step:implementation:step2}, where all (almost) vanishing analysis coefficients of $\proj{\orthcompl{\projnsp[k]}} \grtr$ are excluded from the index set $\ssupp[0]$.
Let us emphasize that this procedure could be suboptimal in the sense of best $S$-analysis-sparse approximations (cf. \eqref{eq:results:compressible:optimalSterm:subspace}), but it still may allow for accurate predictions by means of Theorem~\ref{thm:results:robuststable}, as demonstrated in Subsection~\ref{subsec:appl:compr}.
Finally, let us mention that the statement of Theorem~\ref{thm:results:robuststable} actually considers a scaled version of $\proj{\projnsp}\grtr$ in \eqref{eq:results:robuststable:grtrsparse}, which is merely an artifact of its proof and therefore ignored here. 

\subsection{Numerical Estimation of the Sample Complexity}
\label{subsec:implementation:ed}

To validate the quality of our sampling-rate predictions via $\samplecompl{\aop,\grtr}$, we have also reported the true sample complexity $\ed{\aop,\grtr}$ of \refabpnoiseless{\noiseparam=0}{\aop} in Subsection~\ref{subsec:appl:wavelets} for several choices of $\aop$ and $\grtr$.
Technically, this actually relies on computing the so-called \emph{statistical dimension}, which is known to capture the sample complexity of many convex programs \cite{amelunxen2014edge}.

Let us sketch how this quantity can be numerically calculated. Recalling that $\descset{\lnorm{\aop(\cdot)}[1], \grtr}$ denotes the descent set of the $\l{1}$-analysis objective $\lnorm{\aop(\cdot)}[1]$ at $\grtr$ (see Definition~\ref{def:proofs:framework:descentset}), the statistical dimension of the associated descent cone $\desccone{\lnorm{\aop(\cdot)}[1], \grtr}$ is given by
\begin{align}
	\ed{\aop,\grtr} \coloneqq \mean[ \Big(\sup_{\h \in \desccone{\lnorm{\aop(\cdot)}[1], \grtr} \intersec \ball[2][n]} \sp{\gaussian}{\h} \Big)^2][\Big],
\end{align}
where $\gaussian \distributed \Normdistr{\vnull}{\I{n}}$; see \cite[Prop.~3.1(5)]{amelunxen2014edge}.
Since $\descset{\lnorm{\aop(\cdot)}[1], \grtr}$ is a polytope containing the origin, it is not hard to see that there exists $t_0 > 0$, such that
\begin{equation}
	t^{-1} \descset{\lnorm{\aop(\cdot)}[1], \grtr} \intersec \ball[2][n] = \desccone{\lnorm{\aop(\cdot)}[1], \grtr} \intersec \ball[2][n]
\end{equation}
for all $0 < t \leq t_0$, and in particular, 
\begin{equation}\label{eq:implementation:statdimequiv}
	\ed{\aop,\grtr} = \mean[ \Big(\sup_{\h \in t^{-1} \descset{\lnorm{\aop(\cdot)}[1], \grtr} \intersec \ball[2][n]} \sp{\gaussian}{\h} \Big)^2][\Big].
\end{equation}
The inner expression of this expected value can be rewritten as a simple convex program:
\begin{align} \label{eq:implementation:prog}
 \sup_{\h \in \R^n} \sp{\gaussian}{\h} \quad \text{subject to \quad $\lnorm{\aop (\grtr + t \h)}[1] \leq \lnorm{\aop \grtr}[1]$ and $\lnorm{\h}[2] \leq 1$,} 
\end{align}
where we have used that $\h \in t^{-1}\descset{\lnorm{\aop(\cdot)}[1], \grtr}$ is equivalent to $\norm{\aop (\grtr + t \h)}_1 \leq \norm{\aop \grtr}_1$.
Thus, to approximate the expectation in \eqref{eq:implementation:statdimequiv}, we first select a sufficiently small value of $t$ (typically $t = 0.01$ works well) and then solve  problem~\eqref{eq:implementation:prog} for independent samples $\gaussian_1,\dots,\gaussian_k \distributed \Normdistr{\vnull}{\I{n}}$, using the \texttt{Matlab} software package \texttt{cvx}. Due to the concentration behavior of empirical Gaussian processes, the arithmetic mean of $k = 200$ samples already yields a good estimate of $\ed{\aop,\grtr}$. 

Note that the notion of statistical dimension is essentially equivalent to the conic Gaussian mean width considered in Section~\ref{sec:proofs}. More precisely, we have (see \cite[Prop.~10.2]{amelunxen2014edge})
\begin{equation}
	\effdim[\conic]{\descset{\lnorm{\aop(\cdot)}[1], \grtr}} \leq \ed{\aop,\grtr} \leq \effdim[\conic]{\descset{\lnorm{\aop(\cdot)}[1], \grtr}} + 1.
\end{equation}
For numerical evaluations, the statistical dimension however seems to be more appealing due to the convexity of \eqref{eq:implementation:prog}.

\section{An Excursion to Analysis Operator Learning}
\label{sec:learning}

According to Theorem~\ref{thm:results:recovery}, the value of $\samplecompl{\aop,\grtr}$ determines the number of measurements required to ensure recovery.
The number of measurements, in turn, can be regarded as the ``costs'' of successfully applying \refabpnoiseless{\noiseparam=0}{\aop}.
In that sense, the mapping
\begin{equation}
	\mathcal{M} \colon \R^{N \times n} \times \R^n \to \R, \ (\aop, \grtr) \mapsto \samplecompl{\aop,\grtr}
\end{equation}
actually plays the role of a \emph{loss function} (or \emph{cost function}).
This perspective enables us to formalize the procedure of Subsection~\ref{subsec:concl:practical} by means of a (machine) learning problem:
% The above discussion, for instance, suggests to exploit the fact that $\samplecompl{\aop,\grtr}$ quantifies how well $\aop$ reflects the ``low-dimensional'' structure of a signal $\grtr$.
% Consequently, it is quite natural to use $\aop \mapsto \samplecompl{\aop,\grtr}$ as objective functional when designing analysis operators.
% Mathematically, this simple idea leads us to the following
\begin{problem}[Operator Learning -- Expected Loss Minimization]\label{prob:concl:expriskmin}
	Let $\mathcal{X} \subset \R^n$ be a collection of \emph{feasible signal vectors} and assume that $\mu$ is a probability measure on $\mathcal{X}$.
	Then, solve the minimization problem
	\begin{equation}\label{eq:concl:expriskmin}
		\min_{\aop \in \mathcal{H}} \ \int_{\mathcal{X}} \samplecompl{\aop, \grtr} d\mu(\grtr),
	\end{equation}
	where $\mathcal{H} \subset \R^{N \times n}$ is called \emph{hypothesis set} (or \emph{parameter set}), containing all potential choices of analysis operators.
\end{problem}

Intuitively, the purpose of \eqref{eq:concl:expriskmin} is to select an operator $\aop \in \mathcal{H}$ which minimizes the average costs of reconstructing signals from $\signalclass$.
We would like to emphasize that the significance of this approach clearly depends on both the considered signal class $\signalclass$ and the hypothesis set $\mathcal{H}$.
For instance, if $\signalclass$ is just a singleton, i.e., $\signalclass = \{\grtr\}$, solving \eqref{eq:concl:expriskmin} can lead to vacuous minimizers, being strongly adapted to $\grtr$.
Of more interest are larger classes $\signalclass$ whose elements share certain (geometric) characteristics.
In fact, we have already studied a prototypical example in Section~\ref{sec:appl}, namely discrete piecewise constant functions; more formally, $\signalclass$ might consist of all those vectors $\grtr \in \R^n$ with a fixed number of discontinuities (that is, $\lnorm{\aoptv\grtr}[0] \leq S_{\text{TV}}$ for some $S_{\text{TV}} > 0$).

The choice of $\mathcal{H}$ is typically driven by the belief in what properties of $\aop$ are appropriate to capture the structural features of $\signalclass$. Once again, Subsection~\ref{subsubsec:appl:wavelets:blocks} provides an interesting example: In the scenario of Haar wavelet systems, $\mathcal{H}$ could contain all those $\aop \in \R^{N\times n}$ arising from a scale-wise weighting of $\aopi$. Then, a minimizer of \eqref{eq:concl:expriskmin} would yield an optimally weighted version of $\aopi$ (with respect to the class of piecewise constant signals).
It is also worth mentioning that the dimension of the parameter space $\R^{N \times n}$ is considerably reduced in that way, since the degrees of freedom of $\mathcal{H}$ would just equal the number of scales.

The optimization task of \eqref{eq:concl:expriskmin} is however often unfeasible in practice because the underlying probability measure $\mu$ is usually unknown, e.g., when the $\grtr$ are supposed to be natural images.
Instead, one has only access to a collection of independent samples drawn from a random vector in $\R^n$ distributed according to $\mu$.
This leads to an empirical version of Problem~\ref{prob:concl:expriskmin}:
\begin{problem}[Operator Learning -- Empirical Loss Minimization]\label{prob:concl:empriskmin}
	Let $\mathcal{X} \subset \R^n$ be the set of feasible signals.
	Suppose that $\grtr_1, \dots, \grtr_M \in \signalclass$ are independent samples drawn according to a probability measure $\mu$ on $\signalclass$.
	Then, solve the minimization problem
	\begin{equation}\label{eq:concl:empriskmin}
		\min_{\aop \in \mathcal{H}} \ \tfrac{1}{M} \sum_{l = 1}^M \samplecompl{\aop, \grtr_l},
	\end{equation}
	where $\mathcal{H} \subset \R^{N \times n}$ is again a hypothesis set.
\end{problem}
The idea behind \eqref{eq:concl:empriskmin} is quite simple: By computing the empirical mean instead of an integral, one may end up with a tractable problem, whose solutions are still very close to minimizers of \eqref{eq:concl:expriskmin}.
In fact, solving and analyzing optimization problems of the type \eqref{eq:concl:empriskmin} is one of the key objectives in \emph{statistical learning theory} \cite{vapnik2013nature,mohri2012foundations,shalev2014understanding}.
This field of research offers a rich toolbox that allows us to investigate many related issues, such as regularization, non-convex optimization, iterative algorithms, or sample complexity.
We hope that this methodology, combined with our refined sampling-rate bounds, could be a starting point of novel design rules and guidelines in analysis-based modeling.
On the other hand, due to the non-convexity and discontinuity of $\samplecompl{\cdot, \cdot}$, the specific problem of \eqref{eq:concl:empriskmin} is expected to be (algorithmically) challenging and may require several sophisticated adaptions.
A detailed study would go beyond the scope of this paper and is therefore deferred to future research.
%\revision{\sout{Let us finally point out that similar learning strategies were recently proposed in the literature as well, but using very different kinds of loss functions, e.g., see \cite{yaghoobi2013learning,hawe2013,ravishankar2013,seibert2016}}}.

%%% acknowledgements
\section*{Acknowledgments}
{\smaller
The authors thank Claire Boyer, Peter Jung, Jackie Ma, and Pierre Weiss for fruitful discussions.
Our special thanks goes to Ali Hashemi, who referred us to Price's theorem, and Felix Voigtlaender, who assisted us by working out a fairly general version of it \cite{voigtlaender2017price}, which is also used in this work.
% Moreover, the author greatly appreciates the comments of the anonymous reviewer, which have helped to improve this work.
M.G. is supported by the European Commission Project DEDALE (contract no. 665044) within the H2020 Framework Program as well as the Einstein Center for Mathematics Berlin.
G.K acknowledges partial support by the Einstein Foundation Berlin, the DFG Collaborative Research Center TRR~109 ``Discretization in Geometry and Dynamics,'' and the European Commission Project DEDALE (contract no. 665044) 
within the H2020 Framework Program, DFG Grant KU~1446/18 as well as DFG-SPP~1798 Grants KU~1446/21 and 
KU 1446/23.
M.M. is supported by the DFG Priority Programme DFG-SPP~1798.
Finally, the authors would like to thank the anonymous referees for their useful comments and suggestions which have helped to improve the original manuscript.

}

%%% Biblography
\renewcommand*{\bibfont}{\smaller}
% \nocite{*} % auch nichtzitierte Werke im Verzeichnis
\begin{refcontext}[sorting=nyt]
	\printbibliography[heading=bibintoc]
\end{refcontext}
% \nocite{*}
% \bibliographystyle{bst/amsplain}
% \bibliography{references.bib}

\newpage
\listoftodos

\end{document}